\documentclass{itpthesis}

\usepackage{qittools}

\usepackage{graphicx}

\begin{document}

\ethdissnumber{21180}

\title{Quantum Side Information: Uncertainty Relations, Extractors, Channel Simulations}

\author{Mario Andrea Berta}
\previousdegree{Dipl.\ Phys.\ ETH}
\authorinfo{born April 27, 1985, in Winterthur\\ citizen of Selma, GR, Switzerland}

\referees{Prof.\ Dr.\ Matthias Christandl, examiner\\ Prof.\ Dr.\ Patrick Hayden, co-examiner\\ Prof.\ Dr.\ Renato Renner, co-examiner}
\degreeyear{2013}

\maketitle





\begin{acknowledgements}
First of all I would like to thank my supervisor Matthias Christandl for his support and guidance. We had a very pleasant and fruitful collaboration, and I could not imagine a better environment for doing my PhD. I would also like to thank Patrick Hayden and Renato Renner for taking their valuable time and being my co-examiners.

In the last four years I visited Stephanie Wehner at QCT Singapore and Reinhard Werner at Leibniz University Hanover several times. I profited a lot from these research visits, and I always had a great time. Thank you for this possibility.

During my time as a PhD student I was very fortunate to collaborate with many brilliant researchers. The results presented in this thesis are based on past and ongoing collaborations, and herewith I would like to sincerely thank all my co-authors: Fernando Brand\~ao, Matthias Christandl, Roger Colbeck, Patrick Coles, Fr\'ed\'eric Dupuis, Omar Fawzi, Torsten Franz, Fabian Furrer, Anthony Leverrier, Huei Ying Nelly Ng, Joseph Renes, Renato Renner, Volkher Scholz, Oleg Szehr, Marco Tomamichel, Stephanie Wehner, Reinhard Werner, and Mark Wilde. 

This thesis improved substantially from the feedback of Omar Fawzi, Fabian Furrer, Patrick Hayden, Manuela Omlin, Volkher Scholz, and Mark Wilde; thank you for all your valuable comments and helpful suggestions. Special thanks goes to Volkher Scholz with whom I collaborated intensively, and who helped me a lot with various technical issues. I am also grateful to Marco Tomamichel for providing me with his ETH thesis template, and answering many questions about smooth entropies. Last but not least, I would like to thank all the former and current members of the quantum information theory group at ETH Zurich for many interesting and stimulating discussion.
\end{acknowledgements}



\begin{abstract}
Any theory of information processing always depends on an underlying physical theory. Information theory based on classical physics is known as classical information theory, and information theory based on quantum mechanics is called quantum information theory. It is the aim of this thesis to improve our understanding of the similarities and differences between classical and quantum information theory by discussing various information theoretic problems from the perspective of an observer with the ability to perform quantum operations. This thesis is divided into four parts, each is self contained and can also be read separately.

In the first part, we discuss the algebraic approach to classical and quantum physics and develop information theoretic concepts within this setup. This approach has the advantage of being mathematically very general, and with this, allows to describe virtually all classical and quantum mechanical systems. Moreover, it enables us to include physical situations with infinitely many degrees of freedom involved.

In the second part, we discuss the uncertainty principle in quantum mechanics. The principle states that even if we have full classical information about the state of a quantum system, it is impossible to deterministically predict the outcomes of all possible measurements. In comparison, the perspective of a quantum observer allows to have quantum information about the state of a quantum system. This then leads to an interplay between uncertainty and quantum correlations. It turns out that even in the presence of this additional information, we exhibit nontrivial bounds on the uncertainty contained in the measurements of complementary observables. We provide an information theoretic analysis by discussing so-called entropic uncertainty relations with quantum side information.

In the third part, we discuss the concept of randomness extractors. Classical and quantum randomness are an essential resource in information theory, cryptography, and computation. However, most sources of randomness exhibit only weak forms of unpredictability, and the goal of randomness extraction is to convert such weak randomness into (almost) perfect randomness. We discuss various constructions for classical and quantum randomness extractors, and we examine especially the performance of these constructions relative to an observer with quantum side information.

In the fourth part, we discuss channel simulations. Shannon's noisy channel theorem determines the capacity of classical channels to transmit classical information, and it can be understood as the use of a noisy channel to simulate a noiseless one. Channel simulations as we want to consider them here are about the reverse problem: simulating noisy channels from noiseless ones. Starting from the purely classical case (the classical reverse Shannon theorem), we develop various kinds of quantum channel simulation results. We achieve this by exploiting quantum correlations, and using classical and quantum randomness extractors that also work with respect to quantum side information. Finally, we discuss implications to channel coding theory, quantum physics, and quantum cryptography. 
\end{abstract}


\newpage
\startnumbering

\addtocontents{toc}{~\hfill\textbf{Page}\par}
\addtocontents{lof}{~\hfill\textbf{Page}\par}
\addtocontents{lot}{~\hfill\textbf{Page}\par}

\cleardoublepage
\phantomsection
\addcontentsline{toc}{chapter}{Contents}

\tableofcontents







\chapter{Introduction}\label{ch:intro}

The theory of information processing is very closely related to physics. On the one hand, performing physical experiments is nothing other than extracting information, and on the other hand all information has physical representation, and hence physical concepts are needed~\cite{Landauer92}. Information theory based on classical physics is known as classical information theory, whereas information theory based on non-relativistic quantum mechanics is called quantum information theory. The connection between physics and information theory has led to fundamental new insights. In the case of quantum information theory, this concerns first of all quantum mechanics itself, since information theoretic concepts are still helping to understand the fundamental differences between quantum and classical correlations. In addition, quantum information theory continues to provide new ideas to areas such as strongly correlated systems, quantum statistical mechanics, thermodynamics, to name a few. An essential finding in information theory is that information cannot be defined independently of an underlying physical theory.

Quantum information is fundamentally different from classical information~\cite{Bell87}, and it is the goal of this work to improve our understanding of the differences and similarities between classical and quantum information theory. In particular, we would like to explore these in bi- and multipartite settings, where some of the parties are classical and some of them are quantum. It turns out that even purely classical information theoretic problems can change if we allow for an observer that behaves quantumly (see, e.g., \cite{Gavinsky07}). As an example, we might consider a bipartite setup with a classical party Alice and a second party Bob who might possess classical or quantum systems. Now, if Alice has some information which is partially correlated to Bob, it turns out that a quantum Bob has intrinsically more power than a classical Bob to predict Alice's information. Of course, if the goal is to transmit the full message to Bob, this is an advantage. But if Alice wants to use the message as a key in some cryptographic scheme (the idea being that Bob's knowledge about the message is almost zero), this is a disadvantage. In this thesis, we study various classical and quantum information theoretic problems involving an observer who is of a quantum nature, or in other words, an observer who has quantum side information. The analyzed problems have connections with computer science, cryptography, physics, as well as mathematics (all of which we will discuss). In the following, we give an outline and briefly discuss our contributions.

This work consists of four parts, each of which is self contained and can also be read separately. As we will see, however, each chapter provides applications for the subsequent chapters (the common denominator being the concept of quantum side information). In general, the following introduction will be rather brief, and we refer to the corresponding chapters for a more extensive review and references to the literature on the subject.


\paragraph{Chapter 2 - Preliminaries.} In prior PhD theses about classical and quantum information theory, researchers usually start with an extensive review about finite-dimensional classical and quantum mechanics, and then discuss information theoretic concepts within this setup. Since we are assuming that any potential reader of this thesis knows this framework very well anyway, we would not like to do this.\footnote{For an introduction to finite-dimensional quantum information theory, see, e.g., the excellent textbook by Nielsen and Chuang~\cite{Nielsen00}.} Instead, we would like to discuss the algebraic approach to classical and quantum mechanics, and then develop information theoretic concepts within this setup. This has the advantage of being mathematically more general, and with this, being able to describe infinite-dimensional systems as well. The algebraic approach is also quite elegant, and we believe that it is illuminating to understand information theoretic concepts in this general framework.

After an introduction to the algebraic approach to quantum theory, we introduce the necessary mathematics about $C^{*}$-algebras and von Neumann algebras in the first part of this chapter (Section~\ref{se:math}). We note that this section is very brief and only discusses the mathematical objects that we actually need in the following. In the second part, we introduce quantum information theory on von Neumann algebras (Section~\ref{se:qm}). This section is aimed at readers familiar with finite-dimensional quantum information theory, but no knowledge about algebraic quantum theory is required. In particular, we always comment on the relation to the usual approach based on finite-dimensional Hilbert spaces. Finally, we end with a section about entropy (Section~\ref{se:entropy}). We mention that we will employ the algebraic approached in Section~\ref{se:two} about entropic uncertainty relations on von Neumann algebras. For the other topics covered in this thesis, we will work in the usual finite-dimensional framework.


\paragraph{Chapter 3 - Entropic Uncertainty Relations.} One of the most fundamental concepts in quantum mechanics is Heisenberg's uncertainty principle~\cite{Heisenberg27}. It states that even if we have full classical information about the state of a system, it is impossible to deterministically predict the outcomes of all possible measurements. Entropic uncertainty relations are an information theoretic way to capture this. They were first derived by Hirschman~\cite{Hirschman57}, but many new results have been achieved in recent years, see the review article~\cite{Wehner09}. Interestingly, there exists a deep connection between uncertainty and another fundamental quantum feature, entanglement. The discussion already started with the famous Einstein, Podolsky, and Rosen paper in 1935~\cite{Einstein35}, but a concrete, quantitative, and operationally useful criterion of how uncertainty and entanglement are related remained missing. In that respect, it is important to realize that uncertainty should not be treated as absolute, but with respect to the prior knowledge of an observer~\cite{Hall95,Cerf02,Boileau09,Berta10,Winter10}. Then, as soon as the observer himself is quantum, this has far reaching consequences. It comes down to a subtle interplay between the observed uncertainty, and the entanglement between the analyzed system and the observer. We were able to quantitatively capture these effects with the help of entropic uncertainty relations with quantum side information. These relations not only unified our understanding about quantum mechanics by discussing the interplay of entanglement and uncertainty, but also led to new methods in quantum cryptography (see, e.g., the thesis of Tomamichel for a discussion of these ideas~\cite{Tomamichel12}). Entropic uncertainty relations with quantum side information are now widely studied and used in the field of quantum information theory and quantum cryptography. Our initial relation~\cite{Berta10} was also experimentally verified in~\cite{Prevedel11,Li11}, and coverage in the popular literature includes~\cite{Ananthaswamy11}. It is the subject of this chapter to discuss entropic uncertainty relations with quantum side information. Our contribution is to generalize many known entropic uncertainty relations to the case of quantum side information.

After an introduction to entropic uncertainty relations in quantum mechanics, we start in Section~\ref{se:several} by discussing entropic uncertainty relations with quantum side information for finite-dimensional quantum systems. Our main result is an uncertainty equality (rather than an inequality), which shows that there exists an exact relation between entanglement and uncertainty. In addition, several other entropic uncertainty relations of interest follow as corollaries, and we also sketch applications to witnessing entanglement and quantum cryptography. In the second part of this chapter (Section~\ref{se:two}), we generalize our initial result~\cite{Berta10} to infinite-dimensional quantum systems and measurements. This is particularly nice because it generalizes the original work on entropic uncertainty relations (like, e.g., by Hirschman~\cite{Hirschman57}) to the case of quantum side information. We also mention applications to continuous variable quantum information theory.

This chapter is aimed at readers from the quantum information theory community and/or physics community with an interest in entropic uncertainty relations. We mention that we will need entropic uncertainty relations with quantum side information in Section~\ref{se:qc} about quantum-classical randomness extractors.


\paragraph{Chapter 4 - Randomness Extractors.} Randomness is an essential resource in information theory, cryptography, and computation. However, most sources of randomness exhibit only weak forms of unpredictability. The goal of randomness extraction is to convert such weak randomness into (almost) uniform random bits. In the setup of classical information theory, an extensive theory was developed in recent years, connecting concepts from graph theory (expander graphs), error-correcting codes (list-decoding codes), and complexity theory (hardness amplifiers). As different as these areas may sound, these ideas can be brought into a unified picture by formulating a theory of pseudorandomness, as is nicely done in the lecture notes of Vadhan~\cite{Vadhan11}. The goal of this chapter is to understand some of these concepts when quantum mechanics comes into play. The starting point is a setup where the primary system is still classical, but the role of an observer with quantum side information is investigated (see, e.g., the thesis of Renner~\cite{Renner05} for an elaboration of these ideas). However, the concept of extractors can itself be quantized, that is, to start with a quantum source and then ask for the extraction of classical or quantum randomness. Quantum randomness extractors with quantum side information are then very fundamental in quantum coding theory, and are known as the decoupling approach to quantum information theory (see, e.g., the thesis of Dupuis for a review of these ideas~\cite{Dupuis09}).

After a general introduction to extractors, we start in Section~\ref{se:cc} with a review about classical randomness extractors. We discuss in particular observers with quantum side information. This section then mainly serves as a preparation for the next section, where we go on to extend the concept of randomness extractors to the quantum setup (Section~\ref{se:qq}). Whereas the idea of quantum extractors is omnipresent in quantum information theory, we believe that we can give a new perspective and a unified view on the matter by developing quantum extractors from their classical analogue. In the third part of this chapter, we introduce quantum-classical extractors, which are intermediate objects between fully classical and fully quantum extractors. So far, quantum-classical extractors are mainly useful for cryptography, and we then also mention an application to two-party quantum cryptography (Section~\ref{sec:storage}). Finally, we end with a conclusion and outlook section (Section~\ref{se:extmore}), where we discuss some open questions, especially about the connection of (quantum) randomness extractors to other pseudorandom objects.

This chapter is rather technical, and is aimed at readers from the classical extractor community with an interest in quantum generalizations, and/or at readers interested in decoupling ideas in quantum information theory. We mention that we will need classical and quantum randomness extractors with classical and quantum side information in Chapter~\ref{ch:channels} about channel simulations.


\paragraph{Chapter 5 - Channel Simulations.} A fundamental result in classical information theory is Shannon's noisy channel theorem~\cite{Shannon48}. It determines the capacity of classical channels to transmit classical information, and it can be understood as the use of a noisy channel to simulate a noiseless one. Channel simulations as we want to consider them here are about the reverse problem: simulating noisy channels from noiseless ones. As useless as this problem might sound at first, channel simulations allow for deep insights into the structure of channels, and are also useful for various applications. Starting from the purely classical case, it is the goal of this chapter to develop various kinds of quantum channel simulation results. We mention that all our proofs are based on the concept of randomness extractors, and that the stability against quantum side information is crucial. For a comprehensive review about classical and quantum channel simulation results we refer to the work of Bennett {\it et al.}~\cite{Bennett09}.

After a general introduction to channel simulations, we start in Section~\ref{se:shannon} with the classical reverse Shannon theorem. The theorem basically subsumes all known classical channel simulation results, and serves as a basis for quantum generalizations. We then go on to discuss the quantum reverse Shannon theorem in the next section (Section~\ref{se:qshannon}). The theorem is the most natural generalization of the classical reverse Shanon theorem, but also has aspects that have no classical counterparts. In Section~\ref{se:meas} we discuss universal measurement compression, an intermediate channel simulation problem between the classical and quantum reverse Shannon theorem. Measurement compression is about simulating quantum measurements, and the theorem quantifies how much information is gained by performing a quantum measurement.\footnote{We note that this problem has quite a long history, going all the way back to the work of Groenewold~\cite{Groenewold71}.} Then, as a purely quantum variant of channel simulations, we discuss the entanglement cost of quantum channels (Section~\ref{se:cost}). Finally, channel simulations also allow to derive (strong) converses for channel capacities for transmitting information, and we discuss this in Section~\ref{se:app}.

This chapter is aimed at readers from the classical and/or quantum coding theory community. In addition, Section~\ref{se:meas} about measurement compression is also addressed to physicists with an interest in quantum measurements.


\chapter{Preliminaries}\label{ch:pre}

The ideas in this chapter have been obtained in collaboration with Matthias Christandl, Fabian Furrer, Volkher Scholz, and Marco Tomamichel, and have appeared partially in~\cite{Berta11_4,Berta13_2}. The introduction is taken from the collaboration~\cite{Berta11_4}, and a more extensive review of these ideas can also be found in the thesis of Furrer~\cite{Furrer12_2}. This chapter is aimed at a reader familiar with standard finite-dimensional quantum mechanics and an idea about quantum information theoretic concepts. It is the goal to motivate and develop the algebraic approach to quantum mechanics.

The usual description of a quantum system in quantum information theory is as follows. The basic object for every physical system is a separable, or often even finite-dimensional, Hilbert space $\cH$, and states are described by density matrices. These are linear, positive semi-definite operators $\rho$ from $\cH$ to itself with trace equal to one. The evolution is described by either a unitary transformation $U:\cH\rightarrow\cH$, or a measurement of an observable, which is given by a bounded, self-adjoint operator $O$ on $\cH$.\footnote{There are more general ways to describe the evolution of a quantum system in the Hilbert space setting (see Section~\ref{se:qm}), but for the following explanatory discussion this is sufficient.} The expectation value of a measurement described by the observable $O$ is computed as $\trace[\rho O]$ if the system is in the state $\rho$. Multipartite systems are described by the tensor product of the Hilbert spaces of the individual systems.

In contrast to this, in general quantum theory, especially in quantum field theory, the basic object for every physical system is a von Neumann algebra $\cM$ of observables. Based on the usual approach as given above, this can for example be motivated as follows. Starting with a Hilbert space $\cH$, it is often the case that the physical system obeys certain symmetry conditions. Such symmetries can be due to the evolution determined by the Hamilton operator (dynamical symmetries) or postulated symmetries of the theory itself (superselection rules~\cite{Haag92}). Typical examples are for instance Ising spin chains at zero temperature which are translational invariant, or bosonic systems which are invariant under particle permutation. Such symmetry constraints are realized by a representation $\pi$ of the symmetry group $G$ on $\cH$, and the requirement that all observables $O$ of the theory are invariant under group actions $g \in G$,
\begin{align}
\pi(g^{-1}) \,O\, \pi(g)=O\ ,
\end{align}
which is the same as saying that all physical states are invariant. If the group acts reducibly on $\cH$ it follows that not all mathematically possible observables can actually be observed. Hence, the set of physical observables form a subalgebra of the linear, bounded operators $\cB(\cH)$ on $\cH$. If $\cH$ is infinite-dimensional, the imposed topological requirements on these subalgebras are, however, more subtle. Take a sequence $O_i$ of observables such that $\trace[\rho O_i]$ is a convergent sequence for all density matrices $\rho$ on $\cH$. It is then physically reasonable to require that a limit observable $O$ exists which gives rise to this value. In other words, we would require that the subalgebra of physical observables is closed with respect to taking expectation values. This topology is usually called the $\sigma$-weak topology, and a $\sigma$-weakly closed *-subalgebra (invariant under taking the adjoint) of some $\cB(\cH)$ is a von Neumann algebra $\cM$. States in the Hilbert space sense, i.e., density matrices $\rho$, now induce normalized, positive, normal (i.e., $\sigma$-weakly continuous) linear functionals $\omega:\cM\rightarrow\mathbb{C}$ via
\begin{align}
\omega(a)=\trace[\rho a]\ ,
\end{align}
for $a\in\cM$. The basic idea of the algebraic approach to quantum theory, is to think of the von Neumann algebra $\cM$ as constructed above, as the fundamental object. States of the system are then described by linear, positive, normal, normalized functionals on $\cM$.

We could also choose the norm topology to complete a *-subalgebra of $\cB(\cH)$. This would then lead us to the definition of a $C^*$-algebra, which is more general than a von Neumann algebra. In particular, one often defines $C^*$-algebras independent of the particular representation, i.e., by just requiring certain invariance properties under group actions. Every representation of this invariance group is then defining a different physical setup, or one could say phase. However, once we choose such a representation, i.e., a Hilbert space $\cH$, it is possible to close the observable algebra in the $\sigma$-weak topology and again end up with a von Neumann algebra. Hence, one could say that $C^*$-algebras are the abstract objects defining the theory, whereas von Neumann algebras correspond to physical realizations of that theory.

Bipartite quantum systems are usually modeled by tensor products of Hilbert spaces, the basic idea being that the observables of the two parties should not influence each other and therefore commute. Thus, bipartite quantum systems in our setting are described by a von Neumann algebra $\cM_{AB}$ with commuting subalgebras $\cM_{A},\cM_{B}$, such that the algebra generated by $\cM_{A}\cup\cM_{B}$ is dense in $\cM_{AB}$. We note that such von Neumann algebras can not always be represented on product Hilbert spaces $\cH_{A}\otimes\cH_{B}$. The question whether the possible correlations of bipartite systems modeled by commuting von Neumann algebras is richer than the one obtainable from systems with tensor product structure is an open question. It is known as Tsirelson's problem~\cite{Tsirelson93,Navascues07,Scholz08,Scholz11,Navascues11}.\footnote{It is known that Tsirelson's problem has an affirmative answer for a large class of physical systems, namely if the system's $C^{*}$-algebra is nuclear and/or if the corresponding von Neumann algebra is hyperfinite. But note that this is in general a non-constructive statement, that is, we only know that there exists a bipartite Hilbert space which reproduces the correlations. In addition the Hilbert space might not be separable (see \cite{Scholz08} for a detailed discussion).} From the perspective of quantum information theory, this means that the set of possible correlations using von Neumann algebras might be strictly larger then the set of possible correlations in the standard Hilbert space approach. With respect to the title of this thesis, quantum side information, it thus makes sense to model quantum systems by von Neumann algebras in order to catch all possible correlations.

The Hilbert space approach can be seen as a special case if the system's von Neumann algebra is isomorphic to a full $\cB(\cH)$.\footnote{Even for finite-dimensional Hilbert spaces, a von Neumann algebra does not have to be a full $\cB(\cH)$ (although it is always a direct sum of finite type I factors). But in quantum information theory, it is common to model every algebra of observables as a full $\cB(\cH)$.} Such a von Neumann algebra is called a factor of type I and in this case - but only then - it is sufficient to work in the Hilbert space approach. However, we note that there are von Neumann algebras occurring in physics which are not of such type. As an example we mention a free boson field of finite temperature~\cite{Araki68}. Here, the invariance under particle permutation is the restricting symmetry. Another example of a non-type I factor are the algebras typically assumed to model the set of observables corresponding to a finite space-time region in algebraic quantum field theory~\cite{Buchholz87}.

We conclude  that the von Neumann algebra approach is mathematically more general than the standard one using Hilbert spaces, and sometimes more desirable because it is a unified framework for regular quantum mechanics, quantum statistical mechanics and quantum field theory. In the first part of this chapter, we briefly discuss some mathematical background (Section~\ref{se:math}), and then give a concise introduction to the algebraic approach to quantum mechanics (Section~\ref{se:qm}). This includes a short discussion of how the usual Hilbert space approach follows as a special case. In the second part of this chapter, we discuss the concept of entropy (Section~\ref{se:entropy}).


\section{Mathematical Background}\label{se:math}

This section is aimed to give a very brief but, as far as possible, self contained mathematical introduction to the theory of von Neumann algebras and the concepts used within this work. For a more sophisticated introduction and further literature we refer to~\cite{Bratteli79,Bratteli81,Takesaki01,Takesaki02,Takesaki02_2}. The following is mostly taken from the collaboration~\cite{Berta11_4}.


\subsection{Operator Algebras}\label{sec:operator_algebras}

\paragraph{$C^{*}$-Algebras.} A $*$-algebra is an algebra $\cA$ which is also a vector space over $\nC$, together with an operation $*$ called involution, which satisfies the property $A^{**}=A$, $(AB)^{*}=B^{*}A^{*}$ and $(\alpha A+\beta B)^{*}=\bar{\alpha}A^{*}+\bar{\beta}B^{*}$ for all $A,B\in\cA$ and $\alpha,\beta\in\nC$. If a $*$-algebra is equipped with a norm for which it is complete, it is called a Banach $*$-algebra.

\begin{definition}[$C^{*}$-algebra]\label{def:cstar}
A $C^{*}$-algebra is a Banach $*$-algebra $\cA$ with the property
\begin{align}
\|A^{*}A\|=\|A\|^{2}\ ,
\end{align}
for all $A\in\cA$.
\end{definition}

Henceforth, $\cA$ always denotes a $C^*$-algebra if nothing else is mentioned. Note that the set of all linear, bounded operators on a Hilbert space $\cH$, denoted by $\cB(\cH)$, is a $C^*$-algebra with the usual operator norm (induced by the norm on $\cH$), and the adjoint operation. Furthermore, each norm closed $*$-subalgebra of $\cB(\cH)$ is a $C^*$-algebra.

A representation of a $C^*$-algebra $\cA$ is a $*$-homomorphism $\pi:\cA\rightarrow\cB(\cH)$ on a Hilbert space $\cH$. A $*$-homomorphism is a linear map compatible with the $*$-algebraic structure, that is, $\pi(AB)=\pi(A)\pi(B)$ and $\pi(A^*)=\pi(A)^*$. We call a representation $\pi$ faithful if it is an isometry, which is equivalent to say that it is a $*$-isomorphism from $\cA$ to $\pi(\cA)$. A basic theorem in the theory of $C^*$-algebras says that each $\cA$ is isomorphic to a norm closed $*$-subalgebra of a $\cB(\cH)$ with suitable $\cH$~\cite[Theorem 2.1.10]{Bratteli79}. Hence, each $C^*$-algebra can be seen as a norm closed $*$-subalgebra of a $\cB(\cH)$.

An element $b \in \cA$ is called positive if  $b=a^*a$ for $a\in\cA$, and the set of all positive elements is denoted by $\cA_+$.
A linear functional $\omega$ in the dual space $\cA^*$ of $\cA$ is called positive if $\omega(a)\geq 0$ for all $a\in\cA_+$. The set of all positive functionals $\cA^*_+$ defines a positive cone in $\cA^*$ with the usual ordering $\omega_{1}\geq\omega_{2}$ if $(\omega_{1}-\omega_{2})\in\cA^*_+$, and we say that $\omega_{1}$ dominates $\omega_{2}$. A positive functional $\omega\in\cA^*$ with $\Vert \omega\Vert = 1$ is called a state. The norm on the dual space of $\cA$ is defined to be
\begin{align}\label{def:normDualspace}
\Vert \omega\Vert =\sup_{\substack{x\in\cA\\\Vert x\Vert\leq1}}\vert\omega(x)\vert\ .
\end{align}
A state $\omega$ is called pure if the only positive linear functionals which are dominated by $\omega$ are given by $\lambda\cdot\omega$ for $0\leq\lambda\leq1$. If $\cA=\cB(\cH)$ we have that the pure states are exactly the functionals $\omega_{\xi}(x)=\bra{\xi}x\ket{\xi}$, where $\ket \xi \in \cH$.

\paragraph{Von Neumann Algebras.} Now we consider a subset of linear, bounded operators $\cT\subset\cB(\cH)$ on a Hilbert space $\cH$. The commutant $\cT'$ of $\cT$ is defined as $\cT'=\{a\in\cB(\cH):[a,x]=0,\forall \, x\in\cT\}$.

\begin{definition}[Von Neumann algebra]\label{def:neumann}
Let $\cH$ be a Hilbert space. A von Neumann algebra $\cM$ acting on $\cH$ is a $*$-subalgebra $\cM\subset\cB(\cH)$ which satisfies $\cM''=\cM$.
\end{definition}

There are two other common characterizations of a von Neumann algebra. One rises from the bicommutant theorem~\cite[Lemma 2.4.11]{Bratteli79}: a $*$-subalgebra $\cM\subset\cB(\cH)$ containing the identity is $\sigma$-weakly closed if and only if $\cM''=\cM$.\footnote{The $\sigma$-weak topology on $\cB(\cH)$ is the locally convex topology induced by the semi-norms $A \mapsto \vert \trace (\tau A)\vert$ for trace-class operators $\tau\in\cB(\cH)$, see~\cite[Chapter 2.4.1]{Bratteli79}. Or in more physics words, $\sigma$-weakly closed means that $\cM$ is closed with respect to taking quantum mechanical expectation values.} From this we can conclude that a von Neumann algebra $\cM$ is also norm closed and therefore a $C^*$-algebra.\footnote{We note that a norm closed subalgebra is not necessarily $\sigma$-weakly closed. Thus a $C^*$-algebra on $\cH$ is not always a von Neumann algebra.} The definition of a von Neumann algebra can even be stated in the category of $C^*$-algebras: a von Neumann algebra $\cM$ is a $C^*$-algebra with the property that it is the dual space of a Banach space. Due to historical reasons this is also called a $W^*$-algebra.

In the following $\cM$ denotes a von Neumann algebra. We call $\cZ(\cM)=\cM\cap\cM'$ the center of $\cM$ and $\cM$ a factor if $\cZ(\cM)$ consists only of multiples of the identity. A representation $\pi$ of a von Neumann algebra $\cM$ is a $*$-representation on a Hilbert space $\cH$ that is $\sigma$-weakly continuous. Thus, the image $\pi(\cM)$ is again a von Neumann algebra. We say that two von Neumann algebras are isomorphic if there exists a faithful representation mapping one into the other.

A linear functional $\omega:\cM\rightarrow\nC$ is called normal if it is $\sigma$-weakly continuous and we denote the set of linear, normal functionals on $\cM$ by $\cP(\cM)$. We equip $\cP(\cM)$ with the usual norm as given in~\eqref{def:normDualspace}. Then, the set $\cP(\cM)$ is a Banach space and moreover it is the predual of $\cM$. The cone of positive elements in $\cP(\cM)$ is denoted by $\cP^{+}(\cM)$. It is worth mentioning that $\|\omega\| =\omega(\id)$ for all $\omega\in\cP^{+}(\cM)$. We call functionals $\omega\in\cP^+(\cM)$ with the property $\|\omega\|\leq1$ sub-normalized states, and denote the set of all sub-normalized states by $\cS_{\leq}(\cM)$. Moreover, we say that $\omega\in\cS_{\leq}(\cM)$ is a normalized state if $\|\omega\|=1$, and set $\cS(\cM)=\{\omega\in\cS_{\leq}(\cM):\|\omega\|=1\} $. A particular example of a state is a vector state $\omega_{\xi}(x)=\braket{\xi|x\xi}$, given by some unit vector $\ket{\xi}\in\cH$. The Gelfand-Naimark-Segal (GNS) construction~\cite[Section 2.3.3]{Bratteli79} asserts that for every state $\omega\in\cS(\cM)$ there exists a Hilbert space $\cH_{\omega}$, together with a unit vector $\ket{\xi_{\omega}}\in\cH_{\omega}$ and a representation $\pi_{\omega}:\cM\ra\cB(\cH)$ such that $\omega=\omega_{\xi_{\omega}}\circ\pi_{\omega}$. Moreover, the vector $\ket{\xi_{\omega}}$ is cyclic, that is, $\cH_{\omega}$ is the closure of $\{\pi_{\omega}(x)\ket{\xi_{\omega}}:x\in\cM\}$.

Given two commuting von Neumann algebras $\cM$ and $\hat{\cM}$ acting on the same Hilbert space $\cH$, we define the von Neumann algebra generated by $\cM$ and $\hat{\cM}$ as $\cM\vee\hat{\cM}=(\cM\cup\hat{\cM})''$, where $\cM\cup\hat{\cM}=\rm{span}\{xy\; ;\; x\in\cM \, , y\in\hat{\cM}\}$. According to the bicommutant theorem \cite[Lemma 2.4.11]{Bratteli79}, $\cM\vee\hat{\cM}$ is just the $\sigma$-weak closure of $\cM\cup\hat{\cM}$.

A linear map $\cE:\cM\ra\bar{\cM}$ is called unital if $\cE(\id)=\id$. It is called positive if $\cE(\cP^{+}(\cM))\subset\cP^{+}(\bar{\cM})$, and it is called completely positive if $(\cE\ot\cI):\cM\ot\cB(\nC^{n})\ra\bar{\cM}\ot\cB(\nC^{n})$ is positive for all $n\in\nN$, where $(\cE\ot\cI)(\sum_{k}a_{k}\ot b_{k})=\sum_{k}\cE(a_{k})\ot b_{k}$.\footnote{From now on, $\ot$ denotes either the von Neumann tensor product or the usual Hilbert space tensor product (depending on the context).}

\paragraph{Relative Modular Operator.} In order to study entropic quantities on von Neumann algebras, we also introduce the relative modular operator. We start with a state $\omega\in\cS_{\leq}(\cM)$, and by the GNS representation (as discussed above) we can assume that the state is given by some vector $\ket{\xi_{\omega}}\in\cH_{\omega}$, and that the span of the vectors of the form $a\ket{\xi_{\omega}}$, $a\in\cM$ is dense in $\cH_{\omega}$. Given $\sigma\in\cP^{+}(\cM)$, we then define the relative modular operator $\Delta(\sigma/\omega)$ as the unique self adjoint operator associated to the bilinear form
\begin{align}\label{eq:2}
B_{\Delta(\sigma/\omega)}\big(a\ket{\xi_{\omega}},b\ket{\xi_{\omega}}\big)=\sigma(b a^{\dagger})\ ,
\end{align}
where $a,b\in\cM$.


\subsection{Some Notation}

For the case $\cM=\cB(\cH)$ with $\cH$ finite-dimensional, we need the following definitions. The set of positive semi-definite operators on $\cH$ is denoted by $\cP^{+}(\cH)$. We define the sets of sub-normalized density matrices $\cS_{\leq}(\cH)=\{\rho\in\cP^{+}(\cH):\trace[\rho]\leq1\}$, and normalized density matrices $\cS(\cH)=\{\rho\in\cP^{+}(\cH):\trace[\rho]=1\}$. Furthermore, we define the set of normalized pure-state density matrices $\cV(\cH)=\{\proj{\rho}\in\cS(\cH):\ket{\rho}\in\cH\}$, and sub-normalized pure-state density matrices $\cV_{\leq}(\cH)=\{\proj{\rho}\in\cS_{\leq}(\cH):\ket{\rho}\in\cH\}$. The support of $\rho\in\cP^{+}(\cH)$ is denoted by $\supp(\rho)$, the projector onto $\supp(\rho)$ is denoted by $\rho^{0}$, and $\trace\left[\rho^{0}\right]=\rank(\rho)$, the rank of $\rho$. For $\rho\in\cP^{+}(\cH)$ the maximum eigenvalue is denoted by $\lambda_{1}(\rho)$.

We also need the $\sigma$-weighted $\alpha$-norms (see, e.g., \cite{Olkiewicz99}).

\begin{definition}[$\sigma$-weighted $\alpha$-norms]\label{def:sigmalpha_norms}
Let $\rho,\gamma\in\cS_{\leq}(\cH)$, and $\sigma\in\cP^{+}(\cH)$. For $\alpha\geq1$ the $\sigma$-weighted $\alpha$-norms are defined as
\begin{align}\label{eq:sigma_weighted}
\|\rho\|_{\alpha,\sigma}=\Big(\trace\big[|\sigma^{\frac{1}{2\alpha}}\rho\sigma^{\frac{1}{2\alpha}}|^{\alpha}\big]\Big)^{\frac{1}{\alpha}}\ ,
\end{align}
and for $\alpha\in(0,1)$ we define the same the same expression (although it is no longer a norm). For $\alpha=2$, the norm is induced by the $\sigma$-weighted Hilbert-Schmidt inner product
\begin{align}\label{eq:sigma_ip}
\braket{\rho|\gamma}_{\sigma}=\trace\Big[\big(\sigma^{\frac{1}{4}}\rho\sigma^{\frac{1}{4}}\big)^{\dagger}\big(\sigma^{\frac{1}{4}}\gamma\sigma^{\frac{1}{4}}\big)\Big]\ .
\end{align}
\end{definition}

In particular, for $\sigma=\id$ these define the usual $\alpha$-norms. We have the following H\"older type inequalities for the $\sigma$-weighted $\alpha$-norms.

\begin{lemma}\cite{Olkiewicz99}\label{lem:hoelder}
Let $\rho,\gamma\in\cS_{\leq}(\cH)$, $\sigma\in\cP^{+}(\cH)$, and $p,q\in[1,\infty]$ satisfying $1/p+1/q=1$. Then, we have that 
\begin{align}
\braket{\rho|\gamma}_{\sigma}\leq\|\rho\|_{p,\sigma}\cdot\|\gamma\|_{q,\sigma}\ .
\end{align}
\end{lemma}


\section{Quantum Information Theory}\label{se:qm}

Based on the mathematical definitions given in the previous section (Section~\ref{se:math}), we now discuss quantum mechanics in this algebraic setup. Throughout the section we comment on the relation to the usual approach based on Hilbert spaces. Even though most of the results in this thesis are about finite-dimensional matrix algebras and finite classical systems (with the exception of Section~\ref{se:two} about entropic uncertainty relations on von Neumann algebras), this section should serve as an illustration that quantum information theoretic concepts can also be investigated in the setup of von Neumann algebras. The following is mostly taken from the collaborations~\cite{Berta11_4,Berta13_2}.


\subsection{Quantum Systems}

We associate to every quantum system a von Neumann algebra $\cM\subset\cB(\cH)$ acting on a Hilbert space $\cH$. The state of a quantum system modeled by a von Neumann algebra $\cM$ is given by a functional $\omega\in\cS(\cM)$, and for technical reasons we also consider sub-normalized states in $\cS_{\leq}(\cM)$. If $\cM=\cB(\cH)$ for some finite-dimensional Hilbert space $\cH$, then $\omega\in\cS_{\leq}(\cM)$ is often represented by a (sub-)normalized density matrix $\rho\in\cS_{\leq}(\cH)$ via
\begin{align}
\omega(x)=\trace[\rho x]\quad\forall x\in\cM\ .
\end{align}
In the Hilbert space approach one then considers $\cH$ as the fundamental object defining the quantum system, and (sub-)normalized states are directly modeled by density matrices $\rho\in\cS_{\leq}(\cH)$.


\subsection{Multipartite Systems}

A multipartite system is a composite of different quantum systems $A,B,\ldots,Z$ associated with mutually commuting von Neumann algebras $\cM_A,\cM_B,\dots,\cM_Z$ acting on the same Hilbert space $\cH$.\footnote{If they act on different Hilbert spaces, we just consider their action on the tensor product of the Hilbert spaces.} The corresponding von Neumann algebra of the multipartite system is given by the von Neumann algebra generated by the individual subsystems
\begin{align}
\cM_{AB\ldots Z}=\cM_A\vee\cM_B\vee\ldots\vee \cM_Z\ .
\end{align}
The considered subsystems are labeled by subscripts, e.g., a state on $\cM_{ABC}$ is denoted by $\omega_{ABC}$ while $\omega_{AB}$ is the restriction of $\omega_{ABC}$ onto $\cM_{AB}$. In the Hilbert space approach, a multipartite quantum system is modeled by the tensor product of the Hilbert spaces $\cH_{A},\cH_{B},\dots,\cH_{Z}$ of the individual systems, that is,
\begin{align}
\cH_{AB\ldots Z}=\cH_{A}\ot\cH_{B}\ot\ldots\ot\cH_{Z}\ .
\end{align}
Again, the considered subsystems are labeled by subscripts, e.g., a state on $\cH_{ABC}$ is denoted by $\rho_{ABC}$ while $\rho_{AB}=\trace_{C}[\rho_{ABC}]$ is the restriction of $\rho_{ABC}$ onto $\cH_{AB}$.


\subsection{Purifications}

An important concept in quantum information theory is purification, which is the completion of a system by adding a purifying system. Assume that we have a state $\omega_A$ on a von Neumann algebra $\cM_A$. This state may be regarded as a state $\omega$ on a bigger von Neumann algebra $\cM$, which contains $\cM_A$ as a subalgebra, provided the restriction of $\omega$ onto $\cM_A$ is $\omega_A$. Now the idea of purification is to choose an extension $\omega$ of $\omega_A$ such that $\omega$ is pure. The name is justified by the property that no further extension of the system shows any correlation with the purification $\omega$~\cite[Section IV, Lemma 4.11]{Takesaki02}: if $\tilde\omega\in\cS(\tilde\cM)$ with $\cM\subset\tilde\cM$ and $\tilde\omega$ restricted to $\cM$ is a pure state $\omega$ on $\cM$, then it follows that $\tilde\omega(xy)=\tilde\omega(x)\tilde\omega(y)$ for all $x\in\cM$ and $y\in\cM'\cap\tilde\cM$.

\begin{definition}[Purification]\label{def:purification}
Let $\omega\in\cS_{\leq}(\cM)$. A purification of $\omega$ is defined as a triple $(\pi,\cH,\ket{\xi})$, where $\pi$ is a representation of $\cM$ on a Hilbert space $\cH$, and $\xi\in\cH$ is such that $\omega(x)=\bra{\xi}\pi(x)\ket{\xi}$ for all $x\in\cM$. Moreover, we call $\pi(\cM)$ the principal and $\pi(\cM)'$ the purifying system.
\end{definition}

The GNS construction as reviewed in Section~\ref{sec:operator_algebras} can be rephrased as every state admits a purification. We say for short that $\omega_{A'B}$ is a purification of $\omega_A\in\cS_{\leq}(\cM_A)$, if there exists a purification $(\pi,\cH,\ket{\xi})$ of $\omega_A$ such that $\cM_{A'}=\pi(\cM_A)$, $\cM_B= \pi(\cM_A)'$ and $\omega_{A'B}(x)=\bra{\xi}x\ket{\xi}$ for all $x\in\cM_{A'B}$. Note that we use a less restrictive notion of purification compared to the one of Woronowicz~\cite{Woronowicz72}, which only applies to factor states. This has the consequence that the state $\omega_{A'B}$ is in general not a pure state for $\cM_{A'B}$ (although the vector state $\omega_{\xi}(x)=\bra{\xi}x\ket{\xi}$ on $\cB(\cH)$ is). Another important property of a purification $(\pi,\cH, \ket \xi)$ of $\omega_{A}\in\cS(\cM_A)$ is that $\pi$ is not required to be faithful on the entire $\cM_A$ but only on the part seen by the state $\omega_A$. This means that $\cM_A$ is in general not isomorphic to $\pi(\cM_A)$ and the systems cannot be identified (wherefore we denoted $\pi(\cM_A)$ by $A'$ instead of $A$). We note that for any von Neumann algebra $\cM$, there exists a representation $\pi$ on a Hilbert space $\cH$, such that every state on $\cM$ has a purification in $\cH$. This specific representation is called the standard form of $\cM$~\cite[Chapter 9]{Takesaki02}. A purification is not unique, but all possible ones are connected by partial isometries.

\begin{lemma}\cite{Berta11_4}\label{lem:purification}
Let $\omega\in\cS_{\leq}(\cM)$, and $(\pi_i,\cH_i,\ket {\xi_i})$ with $i=1,2$ be two purifications of $\omega$. Then, there exists a partial isometry $V:\cH_1\rightarrow\cH_2$ such that $V\ket{\xi_1}=\ket{\xi_2}$, and $V$ intertwines with the representations $\pi_i$, that is, $V\pi_1(x)= \pi_2(x)V$ for all $x\in\cM$.
\end{lemma}

The usual definition of a purification in finite-dimensional quantum mechanics is a special case of Definition~\ref{def:purification}. Given a density matrix $\rho_{A}\in\cS_{\leq}(\cH_{A})$ on a finite-dimensional Hilbert space $\cH_{A}$, a purification of $\rho_{A}$ is a pure state density matrix $\rho_{A'B}\in\cV_{\leq}(\cH_{A'B})$ on some Hilbert space $\cH_{A'B}$, such that $\rho_{A'}=\rho_{A}$. The principal system is $\cH_{A'}$, and the purifying system is $\cH_{B}$. Note that the purification can be chosen to be state-dependent as well, since $|A|\geq|A'|\geq\rank(\rho_{A})$ and not necessarily $|A'|=|A|$. However, in contrast to the general case, the $A'$ system can always be embedded in the $A$ system, and hence $\rho_{A'B}$ can always be thought of as a pure state density matrix on $\cH_{AB}$.


\subsection{Classical Systems}\label{sec:classical_systems}

A classical system is specified by the property that all possible observables commute, and can thus be described by an abelian von Neumann algebra. This perspective allows to use the same notation for states on classical systems as for states on quantum systems. Since classical systems will play a major role in the sequel, we now discuss them in more detail. For the sake of illustration, we start with classical systems which are determined by a finite number of degrees of freedom indexed by $X$. In the following, we denote classical systems by $W,X,Y,Z,K$, and use this to specify the subsystem as well as the range of the classical variable. The corresponding von Neumann algebra of such a system is $\ell^\infty(X)$, that is, the set of functions from $X$ to $\nC$ equipped with the supremums norm. A classical state is then a normalized positive functional $\omega_X$ on $\ell^{\infty}(X)$, which may be identified with a probability distribution on $X$, i.e., $\omega_X\in \ell^1(X)$. Furthermore, we denote the set of non-negative distributions on $X$ by $\ell^{+}(X)$. In the Hilbert space approach, it is often convenient to embed the classical system $\ell^\infty(X)$ into a quantum system with Hilbert space dimension $\vert X\vert$ as the algebra of diagonal matrices with respect to a fixed basis $\{\ket{x}\}_{x\in X}$. A classical state $\omega_X$ can then be represented by a density matrix
\begin{align}\label{eq:rhoX}
\rho_{X}=\sum_{x\in X}\omega_X(x)\proj{x}_{X} \ ,
\end{align}
such that the probability distribution can be identified with the eigenvalues of the corresponding density matrix.

Let us now go a step further and consider classical systems with continuous degrees of freedom. In order to define such systems properly, we require that $(X,\Sigma,\mu)$ is a $\sigma$-finite measure space with $\sigma$-algebra $\Sigma$, and measure $\mu$. The von Neumann algebra of the system is then given by the essentially bounded functions denoted by $L^\infty(X)$. A classical state on $X$ is defined as a normalized, positive, normal functional on $L^\infty(X)$, and may be identified with an element of the integrable complex functions $L^1(X)$, which is almost everywhere non-negative and satisfies
\begin{align}
\int_{X}\omega_X(x)d\mu(x)=1\ .
\end{align}
Such functions in $L^1(X)$ are also called probability distributions on $X$. The most prominent example of a continuous classical system is $X=\mathbb{R}^{n}$ with the usual Lebesgue measure. Note that the case of a discrete classical system is obtained if the measure space $X$ is discrete, and equipped with the equally weighted discrete measure $\mu(I)=\sum_{x\in I} 1$ for $I\subset X$. In the discrete case,~\eqref{eq:rhoX} defines a representation of a classical state as a diagonal matrix of trace one on the Hilbert space with dimension equal to the classical degrees of freedom. However, in the case of continuous variables this is not possible if we demand that the image is a valid density matrix. This is easily seen from the fact that every density matrix is by definition of trace class, and hence has discrete spectrum.


\subsection{Classical-Quantum Systems}\label{sec:sub_cq}

Here, we take a closer look at bipartite systems consisting of a classical part $X$ and a quantum part $B$. For a countable classical part $X$, the combined system is described by the von Neumann algebra $\cM_{XB}=\ell^\infty(X)\vee\cM_B$, which is isomorphic to $\ell^{\infty}(X)\ot\cM_{B}\cong\ell^\infty(X,\cM_B)=\{f:X\rightarrow \cM_B \ : \ \sup_x\Vert f(x)\Vert \leq \infty \}$~\cite[Chapter 6.3]{Murphy90}. A state on $\cM_{XB}$ is called a classical-quantum state, and can be written as $\omega_{XB}=(\omega_B^x)$ with $\omega_{B}^x\in\cS_{\leq}(\cM_B)$ and $\sum_{x\in X}\omega_B^x(\id)=1$. If $\cM_{B}=\cB(\cH_{B})$ for some finite-dimensional Hilbert space $\cH_{B}$, then we can represent $\omega_{XB}$ by the density matrix
\begin{align}
\rho_{XB}=\sum_{x\in X}\proj{x}\ot\rho_{B}^{x}\ .
\end{align}
It is now straightforward to generalize the above introduced classical-quantum systems from countable to continuous classical systems. The von Neumann algebra $\cM_{XB}=L^\infty(X)\vee\cM_B\cong L^\infty(X)\ot\cM_B$ is isomorphic to the space of essentially bounded functions with values in $\cM_B$, denoted by $L^\infty(X,\cM_B)$~\cite[Chapter 6.3]{Murphy90}. The normal functionals on $\cM_{XB}$ are then isometrically isomorphic to elements in $L^1(X,\cP(\cM_B))$, and states can thus be identified with integrable functions $\omega_{XB}$ on $X$ with values in $\cP^{+}(\cM_B)$ satisfying the normalization condition
\begin{align}
\int_{X}\omega_{B}^x(\id)d\mu(x)=1\ .
\end{align}
For the sake of convenience, we write the argument of the function as an upper index. The evaluation of $\omega_{XB}$ on an element $E_{XB}\in L^\infty(X,\cM_{B})$ is then computed by $\omega_{XB}(E_{XB})=\int_X \omega_{B}^x(E_B(x))d\mu(x)$.


\subsection{Channels, Measurements, and Post-Measurement States}

Possible evolutions of a system are called channels. As we work with von Neumann algebras it is convenient to define channels as maps on the observable algebra, which is also called the Heisenberg picture. A channel from system $A$ to system $B$ described by von Neumann algebras $\cM_{A}$ and $\cM_{B}$, respectively, is described by a linear, normal, completely positive, and unital map $\cE:\cM_{B}\rightarrow\cM_{A}$. The identity map is denoted by $\cI$. Note that in general $\cM_A$ and $\cM_B$ can be either a classical or a quantum system. If both systems are quantum (classical), we call the channel a quantum (classical) channel.

The corresponding predual map $\cE_*$ in the Schr\"odinger picture is defined via the relation $\cE_*(\omega)(a)=\omega(\cE(a))$ for all $a\in\cM_B$, and is also completely positive. It maps $\cS(\cM_A)$ into $\cS(\cM_B)$. If we consider $\cM_{A}=\cB(\cH_{A})$ and $\cM_{B}=\cB(\cH_{B})$ with finite-dimensional Hilbert spaces $\cH_{A}$ and $\cH_{B}$, and associate states with density matrices $\rho_{A}\in\cS(\cH_{A})$, the unitality of $\cE$ translates to $\trace\left[\cE_*(\rho_{A})\right]=\trace\left[\rho_{A}\right]$, and is referred to as trace preserving. However, in general a von Neumann algebra does not admit a trace, and this property translates to norm conservation on $\cP^+(\cM)$. In the following, whenever we consider channels on von Neumann algebras, we work in the Heisenberg picture, and whenever we consider channels on finite-dimensional matrix algebras, we work in the Schr\"odinger picture.

We define a measurement with outcome range $X$ as a channel which maps $L^{\infty}(X)$ to a von Neumann algebra $\cM_{A}$. Its predual then maps states of the quantum system $A$ to states of the classical system $X$. We denote the set of all measurements $E:L^{\infty}(X) \rightarrow \cM_A$ by $\Obs(X,\cM_A)$. If $X$ is countable, we can identify a measurement $E:\ell^\infty(X) \rightarrow \cM_A$ by a collection of positive operators $E_x=E(e_x)$ with $x\in X$ satisfying $\sum_{x\in X}E_{x}=\id$ (we denote by $e_x$ the sequence with $1$ at position $x$ and $0$ elsewhere). More generally, given a $\sigma$-finite measure space $(X,\Sigma,\mu)$ and the associated algebra $L^\infty(X)$, the mapping $\cO \to \chi_\cO \to \cE(\chi_\cO)$, for $\cO \in \Sigma$, $\chi_\cO$ being its indicator function, defines a measure on $X$ with values in $\cP^{+}(\cM_A)$. Hence, our definition coincides with the definition of a measurement as a positive operator valued measure~\cite[Chapter 3.1]{Davies70}. 

We define the post-measurement state obtained when measuring the state $\omega_A\in \cS(\cM_A)$ with $E_{X}\in\Obs(X,\cM_A)$ by the concatenation $\omega_X=\omega_A \circ E_{X}$. Hence, $\omega_X$ defines a functional on $L^\infty(X)$, where $\omega_{X}(f)=\omega(E_{X}(f))$ for any $f\in L^\infty(X)$. Since $\omega_A$ and $E_{X}$ are normal, so is $\omega_X$, such that $\omega_X$ is an element of the predual of $L^\infty(X)$, which is $L^1(X)$. Thus, the obtained post measurement state is a proper classical state, and can be represented by a probability distribution on $X$.

In the following, we are particularly interested in the situation where we start with a bipartite quantum system $\cM_{AB}$, and measure the $A$-system with $E_{X}\in\Obs(X,\cM_A)$. The post measurement state is then given by $\omega_{XB}=\omega_{AB}\circ E_{X}$. Similarly as in the case of a trivial $B$-system, we can show that the state $\omega_{XB}$ is a proper classical-quantum state on $L^{\infty}(X,\cM_B$) as introduced in the previous section (Section~\ref{sec:sub_cq}).


\subsection{Distance Measures}

\paragraph{Norm.} A common distance measure for elements in $\cS_{\leq}(\cM)$ is the one induced by the norm defined in~\eqref{def:normDualspace}. If $\cM=\cB(\cH)$ with $\cH$ a finite-dimensional Hilbert space, and the states $\omega,\sigma\in\cS_{\leq}(\cM)$ are represented by density matrices $\rho,\gamma\in\cS_{\leq}(\cH)$, this takes the form $\|\omega-\sigma\|=\|\rho-\gamma\|_{1}$, the metric induced by the 1-norm (see Definition~\ref{def:sigmalpha_norms}). This distance measure is operational in the sense that it quantifies the probability of distinguishing two states with an optimal measurement.

\paragraph{Fidelity.} Other useful distance measures are based on the fidelity. For $\omega,\sigma\in\cS_{\leq}(\cM)$, the fidelity is defined as~\cite{Bures69,Uhlmann76}
\begin{align}\label{eq:fidelity}
F_{\cM}(\omega,\sigma)=\sup_{\pi} |\braket{\xi^{\pi}_{\omega}|\xi_{\sigma}^{\pi}}|^2\ ,
\end{align}
where the supremum is over all representations $\pi$ of $\cM$ for which purifications $\ket{\xi^{\pi}_{\omega}}$ and $\ket{\xi_{\sigma}^{\pi}}$ of $\omega$ and $\sigma$ exist. We suppress the subscript $\cM$ if clear from the context, and simply write $F_{\cM}(\vert \xi_{\omega} \rangle, \sigma)$ instead of $F_{\cM}(\omega,\sigma)$ if $\ket{\xi_\omega}$ is a purification of $\omega$. Various properties are known for the fidelity, among them the monotonicity under application of channels~\cite{Alberti83}
\begin{align}\label{Fidelity-Monotonicity}
F(\omega,\sigma)\leq F(\cE_*(\omega),\cE_*(\sigma))\ ,
\end{align}
which also implies that $F_{\cM}(\omega,\sigma)\leq F_{\cN}(\omega,\sigma)$ for von Neumann algebras $\cN\subset\cM$. Furthermore, we can fix a particular representation $\pi$ on $\cH$ in which $\omega,\sigma$ admit vector states $\ket{\xi_{\omega}},\ket{\xi_{\sigma}}\in\cH$, and get~\cite{Alberti83}
\begin{align}\label{eq:alberti}
F(\omega,\sigma)=\sup_{U\in \pi(\cM)'} |\braket{\xi_{\omega}|U\xi_{\sigma}}|^2\ ,
\end{align}
where the supremum is over all elements $U$ in $\pi(\cM)'$ with $\Vert U \Vert \leq 1$. It is useful to adapt the fidelity to sub-normalized states~\cite{Tomamichel10} (see also~\cite{Berta11_4} for a discussion of the von Neumann algebra setup).

\begin{definition}[Generalized fidelity]\label{def:genFidelity}
Let $\omega,\sigma\in\cS_{\leq}(\cM)$. The generalized fidelity between $\sigma$ and $\omega$ is defined as
\begin{align}
\cF_{\cM}(\omega,\sigma)=\Big(\sqrt{F_{\cM}(\omega,\sigma)} +\sqrt{(1-\omega(\id))(1-\sigma(\id))}\Big)^{2}\ .
\end{align}
\end{definition}

\paragraph{Purified Distance.} One particular distance measure on the state space that is based on the generalized fidelity is the purified distance.

\begin{definition}[Purified distance]\label{def:purified_distance}
Let $\omega,\sigma\in\cS_{\leq}(\cM)$. The purified distance between $\rho$ and $\sigma$ is defined as\footnote{The name purified distance comes from the fact that $P_{\cM}(\omega,\sigma)=\frac{1}{2}\inf_{\pi}\|\ket{\xi^{\pi}_{\omega}}\bra{\xi^{\pi}_{\omega}}-\ket{\xi^{\pi}_{\sigma}}\bra{\xi^{\pi}_{\sigma}}\|_{1}$, where the infimum runs over all representations of $\cM$ in which $\omega$ and $\sigma$ have a purifications denoted by $\ket{\xi^{\pi}_{\omega}}$ and $\ket{\xi_{\sigma}^{\pi}}$, respectively.}
\begin{align}
P_{\cM}(\omega,\sigma)=\sqrt{1-\cF_{\cM}(\omega,\sigma)}\ .
\end{align}
\end{definition}

For convenience, we omit the subscript $\cM$ in the following if the respective von Neumann algebra is clear from the context. For $P(\omega,\sigma)\leq\eps$, we use the notation $\omega\approx_{\eps}\sigma$ and say that $\omega$ and $\sigma$ are $\eps$-close. A detailed discussion of the properties of the purified distance can be found in~\cite[Chapter 3]{Tomamichel12}. Although their scope is restricted to systems described by finite-dimensional Hilbert spaces, many of the properties follow in the same way for general systems. It is for instance easy to see that the purified distance defines a metric on $\cS_{\leq}(\cM)$ that is equivalent to the norm distance on $\cP(\cM)$.

\begin{lemma}\label{lem:pdbounds}
Let $\omega,\sigma\in\cS_{\leq}(\cM)$. Then, we have that
\begin{align}
\frac{1}{2}\cdot\big(\Vert \omega -\sigma \Vert + \vert \omega(\id)-\sigma(\id)\vert\big)\leq P_{\cM}(\sigma,\omega)\leq \sqrt{\Vert \omega -\sigma \Vert + \vert \omega(\id)-\sigma(\id)\vert}\ .
\end{align}
\end{lemma}

Furthermore, the purified distance is monotone under application of channels.

\begin{lemma}\label{lem:pdmono}
Let $\omega_{A},\sigma_{A}\in\cS_{\leq}(\cM_{A})$, and let $\cE:\cM_{B}\rightarrow\cM_{A}$ be a normal, completely positive, and sub-unital map. Then, we have that
\begin{align}
P(\omega_{A},\sigma_{A})\geq P(\cE_{*}(\omega_{A}),\cE_{*}(\sigma_{A}))\ ,
\end{align}
where $\cE_{*}$ denotes the predual map of $\cE$.
\end{lemma}

\paragraph{Diamond Norm.} Here, we define a distance measure for quantum channels. Since we only need this for finite-dimensional spaces we restrict ourselves to this case, and work in the Schr\"odinger picture. We use a norm on the set of quantum channels which measures the bias in distinguishing two such mappings. The norm is known as the diamond norm in quantum information theory~\cite{Kitaev97}. Here, we present it in a formulation which highlights that it is dual to the well-known completely bounded (cb) norm~\cite{Paulsen02}. 

\begin{definition}\label{def:diamond}
Let $\cE:\cB(\cH_{A})\mapsto\cB(\cH_{B})$ be a linear map. The diamond norm of $\cE$ is defined as
\begin{align}
\|\cE\|_{\diamond}=\sup_{k\in\mathbb{N}}\|\cE\ot\cI_{k}\|_{1}\ ,
\end{align}
where $\cI_{k}$ denotes the identity map on states of a $k$-dimensional quantum system, and $\|\cF\|_{1}=\sup_{\sigma}\|\cF(\sigma)\|_{1}$ with $\sigma_{A}\in\cS_{\leq}(\cH_{A})$. 
\end{definition}

The supremum in Definition~\ref{def:diamond} is reached for $k=|A|$~\cite{Kitaev97, Paulsen02}. We call two quantum channels $\eps$-close if they are $\eps$-close in the metric induced by the diamond norm.


\subsection{Discretization of Continuous Classical Systems}

Let us consider a classical system $L^\infty(X)$ with $(X,\Sigma,\mu)$ a $\sigma$-finite measure space, where $X$ is also equipped with a topology. The aim is to introduce a discretization of $X$ into countable measurable sets along which we can show convergence of discrete approximations of entropies for states on continuous classical-quantum systems $L^{\infty}(X)\ot\cM_{B}$. We start with some basic definitions.

\begin{definition}[Partition]\label{def:partition}
Let $(X,\Sigma,\mu)$ be a measure space. We call a countable set $\cP=\{I_k\}_{k\in \Lambda}$ ($\Lambda$ any countable index set) of measurable subsets $I_k\in\Sigma$ a partition of $X$ if $X=\bigcup_k I_k$, $\mu(I_k \cap I_l) = \delta_{kl}\cdot\mu(I_k)$, $\mu(I_{k})<\infty$, and the closure $\bar I_k$ is compact for all $k\in\Lambda$. If $\mu(I_k)=\mu(I_l)$ for all $k,l\in\Lambda$, we call $\cP$ a balanced partition, and denote $\mu(\cP)=\mu(I_k)$.
\end{definition}

For a classical-quantum system $\cM_{XB}=L^\infty(X)\ot \cM_B$, and a partition $\cP=\{I_k\}_{k\in\Lambda}$ of $X$, we can now define the discretized state $\omega_{X_\cP B}\in\cS(\ell^{\infty}(\Lambda)\ot\cM_{B})$ of $\omega_{XB}\in\cS(\cM_{XB})$ by
\begin{equation}\label{eq:DiscState}
\omega_{X_{\cP}B}\big((b_{k})\big) = \sum_{k\in\Lambda}\int_{I_k} \omega_B^x(b_{k}) \, d\mu (x)=\sum_{k\in\Lambda}\omega_{B}^{\cP,k}(b_{k})\ ,
\end{equation}
where $(b_{k})\in\ell^{\infty}(\Lambda\ot\cM_{B})$. The new classical system induced by the partition is denoted by $X_\cP$ and thus, $X_\cP \sim \Lambda$. In a similar way we define the discretization of a measurement $E\in\Obs(X,\cM_A)$ with respect to a partition $\cP=\{I_k\}_{k\in\Lambda}$ as the element $E^{\cP}\in\Obs(X_\cP,\cM_A)$ determined by the collection of positive operators
\begin{equation}\label{eq:DiscMeasurement}
E^{\cP}_k = E(\chi_{I_k})\ ,
\end{equation}
where $\chi_{I_k}$ denotes the indicator function of $I_k$. We note that the concept of a discretized measurement and a discretized state are compatible in the sense that the post-measurement state obtained by the discretized measurement $E^\cP$ is equal to the one which is obtained when we first measure $E$, and then discretize the state. Hence, we have that $\omega_{X_{\cP}B}=\omega_{AB}\circ E^\cP$ if $\omega_{XB}=\omega_{AB}\circ E $.

If for two partitions $\cP_1$, $\cP_2$ all sets of $\cP_1$ are subsets of elements in $\cP_2$, we say that $\cP_2$ is finer than $\cP_1$ and write $\cP_2 \leq \cP_1$.

\begin{definition}[Ordered dense sequence of balanced partitions]\label{def:ordered}
A family of partitions $\{\cP_\alpha\}_{\alpha\in\Delta} $ with $\Delta$ a discrete index set approaching zero such that each $\cP_\alpha$ is balanced, $\cP_\alpha \leq \cP_{\alpha'}$ for $\alpha \leq \alpha'$, $\mu(\cP_\alpha) = \alpha$, and $\bigcup_\alpha \cP_\alpha$ generates $\Sigma$, is called an ordered dense sequence of balanced partitions.
\end{definition}

For simplicity, we usually omit the index set $\Delta$. In the case that $X = \mathbb{R}$, $\Sigma$ the Borel $\sigma$-algebra, and $\mu$ the Lebesgue measure, an ordered dense sequence of balanced partitions can be easily constructed. For a positive real number $\alpha$, let us take a partition $\cP_{\alpha}$ of $\mathbb R$ into intervals $I_{k}= [ k\alpha , (k+1)\alpha]$, $k\in\mathbb Z$, with $\mu(\cP_{\alpha})=\alpha$. Choosing for $\alpha$ the sequence $\frac{1}{n}$, $n \in \mathbb{N}$ then gives rise to an ordered dense sequence of balanced partitions.


\section{Entropy}\label{se:entropy}

A fundamental concept in classical and quantum information theory is entropy. Entropy can be defined via an axiomatic approach~\cite{Shannon48,Renyi60}, or operationally, in the sense that it should quantitatively characterize tasks in information theory~\cite{Shannon48}. It is the aim of this section to define the entropy measures that we need in the following. We also discuss some of their mathematical properties, but in general, we omit most of the proofs and refer to the literature and the appendix (Appendix~\ref{ap:entropy}). We start with defining relative entropies for general quantum systems (Section~\ref{sec:entropy_general}). From this, we develop entropies, conditional entropies, and mutual informations of different types. As we are also interested in classical systems, we then proceed with a section about entropies for classical-quantum systems (Section~\ref{sec:entropy_cq}). Finally, we end with a section about entropies for finite-dimensional systems.\footnote{If we define some entropies only for a restricted class of systems (like finite-dimensional ones), this does not mean that these entropies can not be defined for more general systems (like infinite-dimensional ones).}


\subsection{General Quantum Systems}\label{sec:entropy_general}

\paragraph{Von Neumann Entropy.} If we analyze resources that are independent and identically distributed (iid), the most important measure in the asymptotic limit of infinitely many repetitions turns out to be the von Neumann entropy. In order to define the von Neumann entropy independent from the representation of states in terms of density matrices, we require the relative modular operator as defined in~\eqref{eq:2}. This approach was first given by~\cite{Araki76} and further studied by Petz and various authors, see~\cite{Petz93} and references therein. We start with the definition of quantum relative entropy. For a finite-dimensional Hilbert space $\cH$ and density matrices $\rho\in\cS_{\leq}(\cH)$, $\gamma\in\cP^{+}(\cH)$, the quantum relative entropy is defined as\footnote{All logarithms are base 2 in this thesis.}
\begin{align}
D(\rho\|\gamma)=\trace\big[\rho\log\rho\big]-\trace\big[\rho\log\gamma\big]\ .
\end{align}
To motivate our approach using the relative modular operator, we may rewrite this as
\begin{align}
D(\rho\|\gamma)=-\trace\Big[\rho^{1/2}\log\big(\Delta(\gamma/\rho)\big)\rho^{1/2}\Big]\ ,
\end{align}
with the relative modular operator $\Delta(\gamma/\rho)=L(\gamma)R(\rho^{-1})$ acting on the Hilbert space of Hilbert-Schmidt operators in $\cB(\cH)$, and $L(a)$, $R(a)$ the left respectively right multiplication with $a\in\cB(\cH)$.

\begin{definition}[Quantum relative entropy]\label{def:relative}
Let $\cM\subset\cB(\cH)$, $\omega\in\cS_{\leq}(\cM)$, and $\sigma\in\cP^{+}(\cM)$. The quantum relative entropy of $\omega$ with respect to $\sigma$ is defined as
\begin{align}\label{eq:relative}
D(\omega\|\sigma)=-\bra{\xi}\log\left(\Delta(\sigma/\omega)\right)\xi\rangle\ ,
\end{align}
where $\ket{\xi}\in\cH$ defines a vector state for $\omega$. The logarithm of the unbounded operator $\Delta(\omega/\sigma)$ is defined via the functional calculus.
\end{definition}

An important property of the quantum relative entropy is its monotonicity under application of channels.

\begin{lemma}\cite[Corollary 5.12 (iii)]{Petz93}\label{lem:mono}
Let $\omega_{A}\in\cS_{\leq}(\cM_{A})$, $\sigma_{A}\in\cP^{+}(\cM_{A})$, and let $\cE:\cM_{B}\rightarrow\cM_{A}$ be a normal, completely positive, and unital map. Then, we have that
\begin{align}
D(\omega_{A}\|\sigma_{A})\geq D(\cE_{*}(\omega_{A})\|\cE_{*}(\sigma_{A}))\ ,
\end{align}
where $\cE_{*}$ denotes the predual map of $\cE$.
\end{lemma}

We use the quantum relative entropy to define conditional von Neumann entropy.

\begin{definition}[Conditional von Neumann entropy]\label{def:cond_vN}
Let $\cM_{AB}=\cB(\cH_{A})\ot\cM_{B}$ with $\cH_{A}$ finite-dimensional, and $\omega_{AB}\in\cS_{\leq}(\cM_{AB})$. The conditional von Neumann entropy of $A$ given $B$ is defined as
\begin{align}
H(A|B)_{\omega}=-D(\omega_{AB}\|\tau_{A}\ot\omega_{B})\ ,
\end{align}
where $\tau_A$ denotes the trace on $\cB(\cH_{A})$, that is, $\tau_A(a)=\trace[a]$ for all $a\in\cB(\cH_{A})$.
\end{definition}

The conditional von Neumann entropy can also be defined for infinite-dimensional, separable $\cH_{A}$. However, this requires additional technicalities, and we do not need infinite-dimensional principal quantum systems in this work. For the fully general bipartite setup $\cM_{AB}$, we note that it is a priory unclear how to define conditional entropy, since $\cM_{A}$ does not necessarily admit a trace. For trivial $\cM_{B}=\nC$, Definition~\ref{def:cond_vN} gives the usual von Neumann entropy
\begin{align}
H(A)_{\rho}=-\trace[\rho_{A}\log\rho_{A}]\ ,
\end{align}
for a density matrix representation $\rho_{A}$ of the state $\omega_{A}$.\footnote{If the system is classical, the von Neumann entropy is usually called Shannon entropy.} If $\cM_{B}=\cB(\cH_{B})$ with $\cH_{B}$ finite-dimensional, we may write
\begin{align}
H(A|B)_{\omega}=H(AB)_{\omega}-H(B)_{\omega}\ .
\end{align}
The conditional von Neumann entropy is self-dual in the following sense.

\begin{proposition}\cite[Lemma 30]{Berta13_2}\label{prop:duality_vN}
Let $\cM_{AB}=\cB(\cH_{A})\ot\cM_{B}$ with $\cH_{A}$ finite-dimensional, and $\omega_{AB}\in\cS_{\leq}(\cM_{AB})$. Then, we have that
\begin{align}
H(A|B)_{\omega}=-H(A'|C)_{\omega}\ ,
\end{align}
with $\omega_{A'B'C}$ a purification $(\pi,\cK,\ket{\xi})$ of $\omega_{AB}$ with $\cM_{A'B'}=\pi(\cM_{AB})$ the principal system, and $\cM_{C}=\pi(\cM_{A'B'})'$ the purifying system.
\end{proposition}

We also use quantum relative entropy to define mutual information.

\begin{definition}[Mutual information]
Let $\omega_{AB}\in\cS_{\leq}(\cM_{AB})$. The mutual information between $A$ and $B$ is defined as
\begin{align}
I(A:B)_{\omega}=D(\omega_{AB}\|\omega_{A}\ot\omega_{B})\ .
\end{align}
\end{definition}

If $\cM_{AB}=\cB(\cH_{A})\ot\cB(\cH_{B})$ with $\cH_{A}$ and $\cH_{B}$ finite-dimensional, we may write
\begin{align}
I(A:B)_{\omega}=H(A)_{\omega}+H(B)_{\omega}-H(AB)_{\omega}\ .
\end{align}

\paragraph{Min- and Max-Entropy.} Here, we give the definitions of the min- and max-based measures that we need in the following. For a more detailed discussion of min- and max-entropies, we refer to~\cite{Tomamichel12,Koenig09,Tomamichel09,Tomamichel10,Renner05,Datta09} for finite-dimensional systems, and~\cite{Furrer12_2,Furrer11,Furrer09,Berta11_4} for infinite-dimensional systems (also see Appendices~\ref{app:vN}-\ref{app:Imax}). Following~\cite{Datta09}, we start with a pair of relative entropies.

\begin{definition}[Max- and min-relative entropy]\label{def:maxrelative}
Let $\omega\in\cS_{\leq}(\cM)$, and $\sigma\in\cP^{+}(\cM)$. The max-relative entropy of $\omega$ with respect to $\sigma$ is defined as
\begin{align}
D_{\max}(\omega\|\sigma)=\inf\left\{\mu\in\mathbb{R}:\omega\leq2^{\mu}\cdot\sigma\right\}\ ,
\end{align}
where the infimum of the empty set is defined to be $\infty$. The min-relative entropy of $\omega$ with respect to $\sigma$ is defined as
\begin{align}
D_{\min}(\omega\|\sigma)=-\log F(\omega,\sigma)\ .
\end{align}
where $F(\cdot,\cdot)$ denotes the fidelity~\eqref{eq:fidelity}.
\end{definition}

An important property of the max-relative and min-relative entropy is that they are monotone under application of channels.

\begin{lemma}\cite{Berta11_4}\label{lem:maxmono}
Let $\omega_{A}\in\cS_{\leq}(\cM_{A})$, $\sigma_{A}\in\cP^{+}(\cM_{A})$, and let $\cE:\cM_{B}\rightarrow\cM_{A}$ be a normal, completely positive, and sub-unital map. Then, we have that
\begin{align}
D_{\max}(\omega_{A}\|\sigma_{A})\geq D_{\max}(\cE_{*}(\omega_{A})\|\cE_{*}(\sigma_{A}))\ ,
\end{align}
where $\cE_{*}$ denotes the predual map of $\cE$. Moreover, if $\cE$ is also unital
\begin{align}
D_{\min}(\omega_{A}\|\sigma_{A})\geq D_{\min}(\cE_{*}(\omega_{A})\|\cE_{*}(\sigma_{A}))\ .
\end{align}
\end{lemma}

Based on the max-relative entropy, we define the conditional min-entropy.

\begin{definition}[Conditional min-entropy]\label{def:Hmin}
Let $\cM_{AB}=\cB(\cH_{A})\ot\cM_{B}$ with $\cH_{A}$ finite-dimensional, and $\omega_{AB}\in\cS_{\leq}(\cM_{AB})$. The conditional min-entropy of $A$ given $B$ is defined as
\begin{align}
H_{\min}(A|B)_{\omega}=-\inf_{\sigma_B\in \cS(\cM_{B})}D_{\max}(\omega_{AB}\|\tau_A \ot \sigma_B)\ ,
\end{align}
where $\tau_A$ denotes the trace on $\cB(\cH_{A})$.
\end{definition}

The conditional min-entropy can be defined similarly for infinite-dimensional, separable $\cH_{A}$. But again, this requires additional technicalities, and we do not need infinite-dimensional principal quantum systems in this work. The conditional min-entropy can be written in the following alternative form.

\begin{proposition}\label{prop:HminDualForm}
Let $\cM_{AB}=\cB(\cH_{A})\ot\cM_{B}$ with $\cH_{A}$ finite-dimensional, and $\omega_{AB}\in\cS_{\leq}(\cM_{AB})$. Then, we have that $H_{\min}=-\log\big(|A|\cdot F(A|B)_{\omega}\big)$ with
\begin{align}\label{eq:HminDualForm}
F(A|B)_{\rho}=\max_{\Lambda_{B\ra A'}}F\big(\Phi_{AA'},\cI\ot\Lambda_{*}(\rho_{AB})\big)\ ,
\end{align}
where $\cH_{A'}\cong\cH_{A}$, $\Phi_{AA'}$ is the maximally entangled state, and the maximum is over all channels $\Lambda_{*}:\cB(\cH_{A})\ra\cM_{B}$.
\end{proposition}

For a proof, see~\cite{Koenig09} for finite-dimensional systems, and~\cite{Berta11_4} for more general systems. Based on the relative min-entropy, we define the conditional max-entropy.

\begin{definition}[Conditional max-entropy]\label{def:Hmax}
Let $\cM_{AB}=\cB(\cH_{A})\ot\cM_{B}$ with $\cH_{A}$ finite-dimensional, and $\omega_{AB}\in\cS_{\leq}(\cM_{AB})$. The conditional max-entropy of $A$ given $B$ is defined as
\begin{align}
H_{\max}(A|B)_{\omega}=-\inf_{\sigma_B\in \cS(\cM_{B})}D_{\min}(\omega_{AB}\|\tau_A \ot \sigma_B)\ ,
\end{align}
where $\tau_A$ denotes the trace on $\cB(\cH_{A})$.
\end{definition}

Again, the conditional max-entropy can also be defined for infinite-dimensional, separable $\cH_{A}$, but this is not needed for this work. The conditional max-entropy is dual to the conditional min-entropy in the following sense.

\begin{proposition}\cite{Berta11_4}\label{prop:duality_minmax}
Let $\cM_{AB}=\cB(\cH_{A})\ot\cM_{B}$ with $\cH_{A}$ finite-dimensional, and $\omega_{AB}\in\cS_{\leq}(\cM_{AB})$. Then, we have that
\begin{align}
H_{\min}(A|B)_{\omega}&=-H_{\max}(A'|C)_{\omega}\ ,
\end{align}
with $\omega_{A'B'C}$ a purification $(\pi,\cK,\ket{\xi})$ of $\omega_{AB}$ with $\cM_{A'B'}=\pi(\cM_{AB})$ the principal system, and $\cM_{C}=\pi(\cM_{A'B'})'$ the purifying system.
\end{proposition}

The definition of the max-information is based on the relative max-entropy.

\begin{definition}[Max-information]\label{def:Imax}
Let $\omega_{AB}\in\cS_{\leq}(\cM_{AB})$. The max-information that $B$ has about $A$ is defined as
\begin{align}
I_{\max}(A:B)_{\omega}=\inf_{\sigma_B\in \cS(\cM_{B})}D_{\max}(\omega_{AB}\|\omega_A \ot \sigma_B)\ .
\end{align}
\end{definition}

Note that unlike the mutual information, this definition is not symmetric in its arguments.\footnote{For a further discussion of max-based measures for mutual information, see~\cite{Ciganovic12}.}

\paragraph{Smooth Min- and Max-Entropy.} Smooth entropies emerge from their non-smooth counterparts by a maximization and minimization, respectively, over states close with respect to a suitable distance measure. The choice of the distance measure influences the properties of the smooth entropies crucially. Here, we follow Tomamichel~\cite{Tomamichel12}, and define the smooth entropies using $\eps$-balls with respect to the purified distance. For $\omega\in\cS_{\leq}(\cM)$ and $\eps\geq0$, we define
\begin{align}
\cB^{\eps}_{\cM}(\omega)= \{\sigma\in\cS_{\leq}(\cM):\cP_{\cM}(\omega,\sigma)\leq\eps\}\ .
\end{align}
The set $\cB^{\eps}_{\cM}(\omega)$ is referred to as the smoothing set, and $\eps$ is called the smoothing parameter. In the following, we often omit the indication of the von Neumann algebra $\cM$ whenever it is clear from the context.

\begin{definition}[Smooth min- and max-entropy]\label{def:smooth_entropy}
Let $\cM_{AB}=\cB(\cH_{A})\ot\cM_{B}$ with $\cH_{A}$ finite-dimensional, $\omega_{AB}\in\cS_{\leq}(\cM_{AB})$, and $\eps\geq0$. The $\eps$-smooth conditional min-entropy of $A$ given $B$ is defined as
\begin{align}\label{def:eq1:smooth_entropy}
H_{\min}^{\eps}(A|B)_{\omega}=\sup_{\bar{\omega}_{AB}\in\cB^{\eps}(\omega_{AB})}H_{\min}(A|B)_{\bar{\omega}}\ ,
\end{align}
and the $\eps$-smooth conditional max-entropy of $A$ given $B$ is defined as
\begin{align}\label{def:eq2:smooth_entropy}
H_{\max}^{\eps}(A|B)_{\omega}=\inf_{\bar{\omega}_{AB}\in\cB^{\eps}(\omega_{AB})}H_{\max}(A|B)_{\bar{\omega}}\ .
\end{align}
\end{definition}

Since the purified distance defines a metric on $\cS_{\leq}(\cM_{AB})$, we retrieve the conditional min- and max-entropy for $\eps=0$ (Definitions~\ref{def:Hmin} and~\ref{def:Hmax}). The duality property of the conditional min- and max-entropy (Proposition~\ref{prop:duality_minmax}) can be extended to the smooth case~\cite{Tomamichel10,Berta11_4}. The following is the asymptotic equipartition property (AEP) for smooth conditional min- and max-entropy.

\begin{lemma}\label{lem:aepminmax}
Let $\cM_{AB}=\cB(\cH_{A})\ot\cB(\cH_{B})$ with $\cH_{A}$ finite-dimensional, $\omega_{AB}\in\cS_{\leq}(\cM_{AB})$, $\eps>0$, and $n\geq2\cdot(1-\eps^{2})$. Then, we have that
\begin{align}
&\frac{1}{n}H_{\min}^{\eps}(A|B)_{\omega^{\ot n}}\geq H(A|B)_{\omega}-\frac{\eta(\eps)}{\sqrt{n}}\ ,\\
&\frac{1}{n}H_{\max}^{\eps}(A|B)_{\omega^{\ot n}}\leq H(A|B)_{\omega}+\frac{\eta(\eps)}{\sqrt{n}}\ ,
\end{align}
where $\eta(\eps)=4\cdot\sqrt{1-2\cdot\log\eps}\cdot(2+\frac{1}{2}\cdot\log|A|)$.
\end{lemma}

For a proof, see~\cite[Theorem 9]{Tomamichel09} for finite-dimensional systems, and~\cite[Proposition 8]{Furrer11} for more general systems.

\begin{definition}[Smooth max-information]\label{def:smooth_Imax}
Let $\omega_{AB}\in\cS_{\leq}(\cM_{AB})$, and $\eps\geq0$. The $\eps$-smooth max-information that $B$ has about $A$ is defined as
\begin{align}
I_{\max}^{\eps}(A:B)_{\omega}=\inf_{\bar{\omega}_{AB}\in\cB^{\eps}(\omega_{AB})}I_{\max}(A:B)_{\bar{\omega}}\ .
\end{align}
\end{definition}

In contrast to the non-smooth case, the smooth max-information is approximately symmetric in its arguments (at least in the finite-dimensional case).

\begin{lemma}\cite[Corollary 4.2.4]{Ciganovic12}\label{lem:ciganovic}
Let $\eps\geq0$, $\eps'>0$, and $\rho_{AB}\in\cS(\cH_{AB})$ with $\cH_{AB}$ finite-dimensional. Then, we have that
\begin{align}
I_{\max}^{\eps+2\eps'}(B:A)_{\rho}\leq I_{\max}^{\eps}(A:B)_{\rho}+\log\left(\frac{2}{\eps'^{2}}+2\right)\ ,
\end{align}
and the same holds for $A$ and $B$ interchanged.
\end{lemma}

The asymptotic equipartition property for the max-information is stated in Lemma \ref{lem:aepimax}.


\subsection{Classical-Quantum Systems}\label{sec:entropy_cq}

Here, we introduce conditional differential von Neumann and conditional differential min- and max-entropy, where the classical system is given by a $\sigma$-finite measure space, and the quantum side information by a von Neumann algebra. As a main result, we prove that these differential entropies can be approximated by their discretized counterparts.

\paragraph{Conditional Differential Min- and Max-Entropy.} For finite classical systems $X$, the conditional min-entropy has the nice feature to provide a direct operational interpretation in terms of the guessing probability~\cite{Koenig09,Berta11_4}.

\begin{proposition}
Let $\cM_{XB}=L^\infty(X)\ot \cM_B$ with $X$ finite, and $\omega_{XB}\in\cS(\cM_{XB})$. Then, we have that
\begin{align}\label{eq:Guessing}
H_{\min}(X|B)_{\omega} =-\log\sup\Big\{\sum_{x}\omega_{B}^{x}(E_{B}^{x}): \; &E_{B}^{x}\in\cM_{B},E_{B}^{x}\geq0,\nonumber\\
&\sum_{x}E_{B}^{x}=\1_{B}\Big\}\ .
\end{align}
\end{proposition}

Furthermore, the conditional max-entropy of classical-quantum states for a finite classical system $X$ can be written as
\begin{align}\label{eq:SecKey}
H_{\max}(X|B)_{\omega}=2 \log\sup\left\{\sum_{x}\sqrt{F(\omega_{B}^{x},\sigma_{B})} :\; \sigma_{B}\in\cS(\cM_{B})\right\}\ .
\end{align}
These quantities admit natural extensions to classical-quantum systems involving continuous classical variables.

\begin{definition}[Conditional differential min- and max-entropy]\label{def:cont_entropy}
Let $\cM_{XB}= L^{\infty}(X)\ot\cM_{B}$, and $\omega_{XB}\in\cS(\cM_{XB})$. The conditional differential min-entropy of $X$ given $B$ is defined as
\begin{align}\label{eq:mincont}
h_{\min}(X|B)_{\omega}=-\log\sup\Big\{\int\omega_{B}^{x}(E_{B}^{x})d\mu(x): \; &E\in L^{\infty}(X)\ot\cM_{B},E\geq0,\nonumber\\
&\int E_{B}^{x}d\mu(x)\leq\1_{B}\Big\}\ ,
\end{align}
and the conditional differential max-entropy of $X$ given $B$ is defined as
\begin{align}\label{eq:max}
h_{\max}(X|B)_{\omega}=2\log\sup\left\{\int\sqrt{F(\omega_{B}^{x},\sigma_{B})} d\mu(x) : \; \sigma_{B}\in\cS(\cM_{B})\right\}\ .
\end{align}
\end{definition}

These quantities are well defined, since the integrands are measurable and positive. In the case of trivial side information~$\cM_{B}=\mathbb C$, they correspond to the differential R\'enyi entropy of order $\infty$ and $1/2$, respectively,
\begin{align}
&h_{\min}(X)_{\omega}=-\log\Vert \omega_X\Vert_{\infty}\label{eq:renyiinfty}\\
&h_{\max}(X)_{\omega}= 2\log\int\sqrt{\omega^x}d\mu(x) = 2\log\Vert \sqrt{\omega_{X}} \Vert_{1}\label{eq:renyi12}\ ,
\end{align}
where $\Vert \cdot \Vert_{p}$ denotes the usual p-norm on $L^p(X)$. Like any differential entropy, the conditional differential min- and max-entropy can be negative. In particular,
\begin{align}
-\infty\leq h_{\min}(X)_{\omega}<\infty,\quad-\infty<h_{\max}(X)_{\omega}\leq\infty\ .
\end{align}
and the same bounds also hold for the conditional versions in~\eqref{eq:mincont} and~\eqref{eq:max}. For a countable $X$ equipped with the counting measure, we retrieve the usual form as in~\eqref{eq:Guessing} and~\eqref{eq:SecKey}, with sums now involving infinitely many terms. For this case, the conditional min- and max-entropy can also be obtained from finite sum approximations, since the terms inside the sums are all positive
\begin{align}
H_{\min}(X|B)_{\omega}&=h_{\min}(X|B)_{\omega}\notag\\
&=-\log\sup_{m}\sup\left\{\sum_{x=1}^{m}\omega_{B}^{x}(E_{B}^{x}): \; E_{B}^{x}\in\cM_{B},E_{B}^{x}\geq0,\sum_{x=1}^{m}E_{B}^{x}\leq\1_{B}\right\}\ ,\label{eq:minapprox}
\end{align}
\begin{align}
H_{\max}(X|B)_{\omega}&=h_{\max}(X|B)_{\omega}\notag\\
&=2 \log\sup_{m}\sup\left\{\sum_{x=1}^{m}\sqrt{F(\omega_{B}^{x},\sigma_{B})} :\; \sigma_{B}\in\cS(\cM_{B})\right\}\label{eq:maxapprox}\ .
\end{align}
From now on, entropies associated with countable classical systems equipped with the counting measure are written with upper case letters. The approximation results can be generalized to measure spaces with an ordered dense sequence of balanced partitions.

\begin{proposition}\label{thm:MinMaxApprox}
Let $\cM_{XB}= L^\infty(X)\ot \cM_B$ with $(X,\Sigma,\mu)$ a measure space with an ordered dense sequence of balanced partitions $\{\cP_{\alpha}\}$, and $\omega_{XB}\in\cS(\cM_{XB})$. Then, we have that
\begin{align}\label{thm,eq2:MinMaxApprox}
h_{\min}(X|B)_\omega=\lim_{\alpha\ra0}\Big(H_{\min}(X_{\cP_\alpha}|B)_\omega + \log\alpha  \Big)\ ,
\end{align}
where $\omega_{X_{\cP_{\alpha}}B}$ is defined as in~\eqref{eq:DiscState}. Furthermore, if there exists an $\alpha_{0}>0$ such that $H_{\max}(X_{\cP_{\alpha_{0}}})_\omega<\infty$,\footnote{There exists $\omega_{X}\in L^{1}(X)$ with $h_{\max}(X)_\omega<\infty$ but $H_{\max}(X_{\cP_{\alpha}})_\omega=\infty$ for all $\alpha>0$. As an example, take $X=\mathbb R$ and $\omega_X$ to be the normalization of the function which is equal to $1$ for $x\in [k,k-1/k^2]$, $k\in\mathbb N$, and $0$ else~\cite{Walter13}. But conversely, $H_{\max}(X_{\cP_{\alpha}})_\omega < \infty$ implies that $h_{\max}(X)_\omega<\infty$ since the relation $h_{\max}(X|B)_\omega\leq H _{\max}(X_{\cP_{\alpha}}|B)_\omega+\log\alpha$ holds for all $\alpha>0$ and $\omega_{XB}\in\cS(\cM_{XB})$.} then
\begin{align}\label{thm,eq1:MinMaxApprox}
h_{\max}(X|B)_\omega=\lim_{\alpha\ra0}\Big(H_{\max}(X_{\cP_\alpha}|B)_\omega+\log\alpha\Big)\ .
\end{align}
\end{proposition}

\begin{proof}
The proof is from the collaboration~\cite{Berta13_2}. We start with the conditional differential min-entropy. Let us fix an $\alpha_0$ and consider $\cP_{\alpha_0}=\{I^{\alpha_0}_l\}_{l\in\Lambda}$, where we can assume that $\Lambda = \{1,2,3,...\} \subset \mathbb N$. For $k\in\Lambda$, we define $C_k=\bigcup_{l=1}^k I_l^{\alpha_0}$, which is compact by the definition of a partition (Definition~\ref{def:partition}). We then write the conditional differential min-entropy as 
\begin{align}
h_{\min}(X|B)_{\omega} =-\log\sup_{k}\sup\Big\{\omega_{XB}(E):\;&E\in L^\infty(C_k) \otimes \cM_{B},E\geq0,\nonumber\\
&\int E_{B}^{x}\,d\mu(x)\leq\1_{B}\Big\}\ .
\end{align}
Since $C_k$ is compact, the set of step functions $\cT^k=\bigcup_{\alpha\leq \alpha_0} \cT^k_\alpha$ with $\cT^k_\alpha$ the step functions corresponding to partitions $\cP^k_{\alpha}$ defined as the restriction of $\cP_\alpha$ onto $C_k$ is $\sigma$-weakly dense in $L^\infty(C_k)$. Because $\omega_{XB}$ is $\sigma$-weakly continuous we get that
\begin{align}\label{eq:min-approx}
h_{\min}(X|B)_{\omega} =-\log\sup_{k}\sup_{\alpha} \Big\{\omega_{XB}(E):E\in \cT^k_\alpha \otimes \cM_{B},E\geq0, \int E_{B}^{x}\,d\mu(x)\leq\1_{B}\Big\}\ ,
\end{align}
where we used that a $\{\cP^k_\alpha\}$ is an ordered family of partitions. By exchanging the two suprema, we find that the right hand side of~\eqref{eq:min-approx} reduces to the supremum of $H_{\min}(X_{\cP_{\alpha}}|B)_{\omega}+\log\alpha$ over $\alpha$, with $\omega_{X_{\cP_{\alpha}B}}$ defined as in~\eqref{eq:DiscState}. Finally, we note that since the expression on the right hand side of~\eqref{eq:min-approx} is monotonic in $\alpha$, the supremum over $\alpha$ in~\eqref{eq:min-approx} can be exchanged by the limit $\alpha \rightarrow 0$. 

For the corresponding approximation arguments of the conditional differential max-entropy, we make use of the fact that exists an $\alpha_{0}\in\mathbb{N}$ such that $H_{\max}(X_{\cP_{\alpha_{0}}})_\omega<\infty$. We start by rewriting the definition of the conditional differential max-entropy \eqref{eq:max}, using an expression for the fidelity~\eqref{eq:alberti}
\begin{align}
h_{\max}(X|B)_{\omega}=2\log\sup\left\{\int \sup_{U\in\pi(\cM_{B})'}|\bra{\xi_{\omega}^{x}}U\ket{\xi_{\sigma}}|d\mu(x) : \; \sigma_{B}\in\cS(\cM_{B})\right\}\ ,
\end{align}
where $\pi$ is some fixed representation of $\cM_{B}$ in which $\omega_{B}^{x}$ and $\sigma_{B}$ admit vector states $\ket{\xi^{x}_{\omega}}$, $\ket{\xi_{\sigma}}$, respectively, and the supremum is taken over all elements $U\in\pi(\cM_{B})'$ with $\|U\|\leq1$. We note that we can always choose $U$ such that $\bra{\xi_{\omega}^{x}}U\ket{\xi_{\sigma}}$ is positive. It follows by the $\sigma$-finiteness of the measure space, together with the theorem of monotone convergence, that we can find a sequence of sets $X^n$ all having finite measure, and
\begin{align}
h_{\max}(X|B)_{\omega}&=2\log\sup_{\sigma_{B}\in\cS(\cM_{B})} \lim_{n\ra\infty} \sup_{U(x)\in\pi(\cM_{B})'}\int_{X^n} \bra{\xi_{\omega}^{x}}U(x)\ket{\xi_{\sigma}}d\mu(x)\ .
\end{align}
For later reasons we note that the sequence $X^n$ can be chosen such that it is compatible with the partitions in the sense that for every $n$ the restriction of $\cP_\alpha$ onto $X^n$ forms again a balanced partition with measure $\alpha$.\footnote{One can take for instance $X^n$ to be generated by finite increasing unions of the sets in a partition $\cP_{\alpha_0}$ for a fixed $\alpha_0$. Then for all $\alpha\leq\alpha_0$ the condition is satisfied.} It then follows from disintegration theory of von Neumann algebras~\cite[Chapter IV.7]{Takesaki01} that the expression involving the third supremum and the integral can again be recognized as a fidelity, more precisely as the square root of $F(\omega_{X^n B}, \mu_{X^n} \ot \sigma_B)$, where $\omega_{X^n B}$ is the state restricted to the subalgebra $L^\infty(X^n) \ot \cM_B \subset L^\infty(X) \ot \cM_B$, and $\mu_{X^n}$ is the Lebesgue measure restricted to the set $X^n$. As the later one has finite measure, $\mu_{X^n}$ can be identified as a positive functional on $L^\infty(X^n) \ot \cM_B$. We now employ similar ideas to the min-entropy case. For an appropriate sequence of partitions $\{\cP_\alpha\}$, we have that the step functions of every fixed partition with support in $X^n$ form a subalgebra $\cN_\alpha$ of $L^\infty(X^n)$, as well as that the set of these subalgebras are weakly dense and have a common identity. Hence, we may apply a result of Alberti~\cite{Alberti83} and find
\begin{align}
F(\omega_{X^n B},\mu_{X^n}\ot\sigma_B)&=\inf_{\alpha\in\mathbb{N}}F(\omega_{X^{n}_{\cP_{\alpha}}B},\mu_{X^{n}_{\cP_{\alpha}}}\ot\sigma_{B})\notag\\
&=\lim_{\alpha\ra\infty}F(\omega_{X^{n}_{\cP_{\alpha}}B},\mu_{X^{n}_{\cP_{\alpha}}}\ot\sigma_{B})
\end{align}
where the restricted states are defined as in~\eqref{eq:DiscState}. This leads to
\begin{align}\label{eq:limits_wrong}
h_{\max}(X|B)_{\omega} &= 2 \log\sup_{\sigma_{B}\in\cS(\cM_{B})}\lim_{n\ra\infty}\lim_{\alpha\ra\infty}\sqrt{F(\omega_{X^n_{\cP_{\alpha}}B}, \mu_{X^n_{\cP_{\alpha}}} \ot \sigma_B)}\ ,
\end{align}
and in order to proceed, we have to interchange the limits. For this, we define
\begin{align}
f_{n}(\sigma,\alpha)=\sqrt{F(\omega_{X^n_{\cP_{\alpha}}B},\mu_{X^n_{\cP_{\alpha}}}\ot\sigma_B)}=\sum_{k\in\Lambda(\alpha,n)}\sqrt{\alpha\cdot F(\omega_{B}^{\cP_{\alpha}^{n},k},\sigma_B)}\ ,
\end{align}
where we have used that $\mu_{X^n}$ restricted to $\cN_{\alpha}$ is just the counting measure multiplied by $\alpha$. Since $f_{n}(\sigma,\alpha)$ is monotonously decreasing in $\alpha$, there exists by assumption an $\alpha_{0}$ such that
\begin{align}
f_{n}(\sigma,\alpha)\leq\sum_{k\in\Lambda(\alpha_{0},n)}\sqrt{\alpha_{0}\cdot F(\omega_{B}^{\cP_{\alpha_{0}}^{n},k},\sigma_B)}\leq\sum_{k\in\Lambda(\alpha_{0},n)}\sqrt{\alpha_{0}\cdot\omega_{B}^{\cP_{\alpha_{0}}^{n},k}(\id)}
\end{align}
is finite in the limit $n\to\infty$. It follows by the Weierstrass' uniform convergence theorem that the sequence $f_{n}(\sigma,\alpha)$ converges uniformly in $\sigma$ and $\alpha$ to the limiting function $f(\sigma,\alpha)=\lim_{n\ra\infty}f_{n}(\sigma,\alpha)$. Hence, the limits in~\eqref{eq:limits_wrong} can be interchanged, and we arrive at
\begin{align}
h_{\max}(X|B)_{\omega}=2\log\sup_{\sigma_{B}\in\cS(\cM_{B})}\lim_{\alpha\ra\infty}f(\sigma,\alpha)=2\log\sup_{\sigma_{B}\in\cS(\cM_{B})}\inf_{\alpha\in\mathbb{N}}f(\sigma,\alpha)\ .
\end{align}
As the last step, we need to interchange the supremum with the infimum. For this, we define
$\bar{f}(\sigma,\lambda)=f(\sigma,\lfloor 1/\lambda\rfloor)$ with $\lfloor\cdot\rfloor$ the floor function, and get
\begin{align}\label{eq:minimax_wrong}
h_{\max}(X|B)_{\omega}=2\log\sup_{\sigma_{B}\in\cS(\cM_{B})}\inf_{\lambda\in[0,1]}\bar{f}(\sigma,\lambda)\ ,
\end{align}
where $\bar{f}(\sigma,0)=\int\sqrt{F(\omega_{B}^{x},\sigma_{B})}d\mu(x)$. We now check the conditions of Sion's minimax theorem (Lemma~\ref{lem:minimax}):
\begin{itemize}
\item The set $[0,1]$ is a compact and convex.
\item The set $\cS(\cM_{B})$ is a convex.
\item The function $\bar{f}(\sigma,\lambda)$ is monotone in $\lambda$, and therefore quasi-convex in $\lambda$.
\item The function $f(\sigma,\alpha)$ is monotonously decreasing in $\alpha$, the floor function $\lfloor\cdot\rfloor$ is upper semi-continuous, and hence the function $\bar{f}(\sigma,\lambda)$ is lower semi-continuous in $\lambda$.
\item Since the fidelity and the square root function are concave, the function $\bar{f}(\sigma,\lambda)$ is concave in $\sigma$.
\item The fidelity is continuous, thus $f_{n}(\sigma,\alpha)$ is continuous in $\sigma$. It then follows by the uniform limit theorem that $f(\sigma,\alpha)$ is continuous in $\sigma$. 
\end{itemize}
Hence, we find
\begin{align}
\sup_{\sigma_{B}\in\cS(\cM_{B})}\inf_{\lambda\in[0,1]}\bar{f}(\sigma,\lambda)=\inf_{\lambda\in[0,1]}\sup_{\sigma_{B}\in\cS(\cM_{B})}\bar{f}(\sigma,\lambda)\ ,
\end{align}
and with~\eqref{eq:minimax_wrong} we conclude
\begin{align}
h_{\max}(X|B)_{\omega}=\inf_{\alpha}\Big(H_{\max}(X_{\cP_{\alpha}}|B)_{\omega}+\log\alpha\Big)\label{eq:maxuncert}\ .
\end{align}
By using that $f(\sigma,\alpha)$ is monotonically increasing in $\alpha$ we arrive at the claim.
\end{proof}

Note that the discretized entropies are regularized by the logarithm of the measure of the partition. This is in accordance with the fact that discretized entropies have no units while the power of the differential entropies $2^{-h_{\min}(X|B)_\omega}$ and $2^{-h_{\max}(X|B)_\omega}$ are in units of $X$.

\paragraph{Conditional Differential Von Neumann Entropy.} In order to motivate our definition of the differential conditional von Neumann entropy, let us first recall the case of finite-dimensional Hilbert spaces and finite classical systems. For a classical-quantum density matrix $\rho_{XB}=\sum_{x}p_{x}\proj{x}_{X}\ot\rho_{B}^{x}$ the conditional von Neumann entropy can be written as
\begin{align}\label{eq:relative_finite}
H(X|B)_{\rho}=H(XB)_{\rho}-H(B)_{\rho}=-\sum_{x}D(p_{x}\rho_{B}^{x}\|\rho_{B})\ .
\end{align}
Now, we can write the conditional differential von Neumann entropy of a classical-quantum state $\omega_{XB}\in\cS(L^{\infty}(X)\ot\cM_{B})$ as
\begin{align}\label{eq:xxx}
h(X|B)_{\omega}=-\int D(\omega_{B}^{x}\|\omega_{B})d\mu(x)\ .
\end{align}
For $\cM_{B}=\mathbb{C}$ and $X=\mathbb R$, this is the differential Shannon entropy. We note that the differential Shannon entropy can take values in $[-\infty,\infty]$ and may be written in the more familiar form~\cite{Petz93}
\begin{align}\label{eq:1}
h(X)=-\int p(x)\log p(x)\,dx\ ,
\end{align}
with $p(x)$ being the density of some probability measure with respect to the Lebesgue measure $dx$.

Even though we are mostly interested in classical-quantum states, we also consider fully quantum states with the first system given by a finite dimensional matrix algebra.

\begin{definition}[Conditional differential von Neumann entropy]\label{def:condneumann}
Let $\cM_{XAB}=\cL^{\infty}(X)\ot\cB(\cH_{A})\ot\cM_{B}$ with $\cH_{A}$ finite-dimensional, and $\omega_{XAB}\in\cS(\cM_{XAB})$. The conditional differential von Neumann entropy of $XA$ given $B$ is defined as
\begin{align}
h(XA|B)_{\omega}=-\int D(\omega_{AB}^{x}\|\tau_A \otimes \omega_{B})\, d\mu(x)\ ,
\end{align}
where $\tau_{A}$ denotes the trace on $\cB(\cH_{A})$. If $X$ is a discrete measure space with counting measure, we also write $h(XA|B)_{\omega}=H(XA|B)_{\omega}$. 
\end{definition}

This representation can be used to derive an approximation result similar to the ones for the conditional differential min- and max-entropy.

\begin{proposition}\cite[Proposition 8]{Berta13_2}\label{prop:vNapprox}
Let $\cM_{XB}=\cL^{\infty}(X)\ot\cM_{B}$ with $(X,\Sigma,\mu)$ a measure space with an ordered dense sequence of balanced partitions $\{\cP_{\alpha}\}$, and $\omega_{XB}\in\cS(\cM_{XB})$. If there exists $\alpha_0>0$ for which $H(X_{\cP_{\alpha_0}} |B)_{\omega}<\infty$, and if $h(X|B)_{\omega}>-\infty$, then
\begin{align}
h(X|B)_{\omega}=\lim_{\alpha\ra0}\big(H(X_{\cP_{\alpha}}|B)_{\omega}+\log\alpha\big)\ .
\end{align}
If furthermore $h(X)_{\omega}<\infty$,\footnote{From $H(X|B)_{\omega}>-\infty$, it follows by the monotonicity of the quantum relative entropy under application of channels (Lemma~\ref{lem:mono}) that $H(X)_{\omega}>-\infty$.} then\footnote{In~\cite{Kuznetsova10}, conditional von Neumann entropy was defined as in~\eqref{eq:kuznetsova} for $\omega_{AB}\in\cS(\cH_{AB})$ with $H(A)_{\omega}<\infty$, and separable Hilbert spaces $\cH_{A},\cH_{B}$.}
\begin{align}\label{eq:kuznetsova}
h(X|B)_{\omega}=h(X)_{\omega}-I(X:B)_{\omega}\ .
\end{align}
\end{proposition}


\subsection{Finite-Dimensional Systems}\label{sec:entropy_finite}

For the rest of this section, all systems (classical and quantum) are finite-dimensional.

\begin{definition}[Conditional R\'enyi 2-entropy]\label{def:h2}
Let $\rho_{AB}\in\cP^{+}(\cH_{AB})$. The conditional R\'enyi 2-entropy of $A$ given $B$ is defined as\footnote{Sometimes $H_{2}(A|B)_{\rho}$ is also known as the conditional collision entropy, see, e.g., \cite{Renner05}.}
\begin{align}
H_2(A|B)_{\rho}=-\log\trace\left[\rho_{AB}(\id_A\ot\rho_B)^{-1/2}\rho_{AB}(\id_A\ot\rho_B)^{-1/2}\right]\ ,
\end{align}
where the inverses denote generalized inverses.\footnote{For $\rho\in\cP^{+}(\cH)$, $\rho^{-1}$ is a generalized inverse of $\rho$ if $\rho\rho^{-1}=\rho^{-1}\rho=\rho^{0}=(\rho^{-1})^{0}$. In the following all inverse are generalized inverses.}
\end{definition}

We have seen in Section~\ref{sec:entropy_general} that the conditional min-entropy of a bipartite quantum state $\rho_{AB}$ can be written as $H_{\min}(A|B)_{\rho}=-\log\big(|A|\cdot F(A|B)_{\rho}\big)$ for
\begin{align}
F(A|B)_{\rho}=\max_{\Lambda_{B\ra A'}}F\big(\Phi_{AA'},\id_A\ot\Lambda_{B\ra A'}(\rho_{AB})\big)\ ,
\end{align}
where $\cH_{A'}\cong\cH_{A}$, $\Phi_{AA'}$ is the maximally entangled state, and the maximum is over all channels $\Lambda_{B\ra A'}:\cB(\cH_{B})\ra\cB(\cH_{A'})$. As it turns out, the conditional R\'enyi 2-entropy allows for a similar operational form. 

\begin{proposition}\label{prop:h2operational}
Let $\rho_{AB}\in\cS(\cH_{AB})$. Then, we have that $H_2(A|B)_{\rho}=-\log\big(|A|\cdot F^{\pg}(A|B)_{\rho}\big)$ for
\begin{align}\label{eq:h2operational}
F^{\pg}(A|B)_{\rho}=F\big(\Phi_{AA'},\id_A \ot \Lambda_{B\ra A'}^{\pg}(\rho_{AB})\big)\ ,
\end{align}
where $\cH_{A'}\cong\cH_{A}$, $\Phi_{AA'}$ is the normalized maximally entangled state, and $\Lambda^{\pg}_{B\rightarrow A'}$ is the pretty good recovery map~\cite{Barnum02}.
\end{proposition}

\begin{proof}
The proof is from the collaboration~\cite{Berta13_3}. This is easily seen by the fact that the pretty good recovery map can be written as
\begin{align}
\Lambda^{\pg}_{B\ra A'}(\cdot)=\frac{1}{|A|}\cdot\cE^{\dagger}_{B\ra A'}\left(\rho_B^{-1/2}(\cdot)\rho_B^{-1/2}\right)\ ,
\end{align}
where $\cE^{\dagger}_{B\rightarrow A'}$ denotes the adjoint of the Choi-Jamilkowski map of $\rho_{AB}$,
\begin{align}
\mathcal{E}_{A\ra B}(\cdot)=|A|\cdot\trace_{A}\big[\left((\cdot)^T \ot \id_{B}\right)\rho_{AB}\big]\ ,
\end{align}
and the transpose $(\cdot)^{T}$ is with respect to the basis used to define $\Phi_{AA'}$.
\end{proof}

The map $\Lambda^{\pg}_{B\ra A'}$ is pretty good in the sense that it is close to optimal for recovering the maximally entangled state, i.e., the following bound holds~\cite{Barnum02} for $\rho_{AB}\in\cS(\cH_{AB})$,
\begin{align}
F^{2}(A|B)_{\rho}\leq F^{\pg}(A|B)_{\rho}\leq F(A|B)_{\rho}\ .
\end{align}

For classical-quantum states $\rho_{KB}=\sum_{k}\proj{k}_{K}\ot\rho_{B}^{k}$, we have $H_{\min}(K|B)_{\rho}=-\log P_{\guess}(K|B)_{\rho}$ with
\begin{align}
P_{\guess}(K|B)_{\rho}=\max_{\{E^{k}_{B}\}}\sum_{k}\trace\left[E_{B}^{k}\rho_{B}^{k}\right]\ ,
\end{align}
where the maximum is over all measurements $\{E_{B}^{k}\}_{k\in K}$ on $B$ (Section~\ref{sec:entropy_cq}). Similarly, Proposition~\ref{prop:h2operational} simplifies for classical-quantum states to the following operational form (originally shown in~\cite{Buhrman08}).

\begin{corollary}\label{cor:pgm}
Let $\rho_{KB}\in\cS(\cH_{KB})$ be classical with respect to the basis $\{\ket{k}\}_{k\in K}$. Then, we have that
\begin{align}
H_{2}(K|B)_{\rho}=-\log P^{\pg}_{\guess}(K|B)_{\rho}\ ,
\end{align}
where $P^{\pg}_{\guess}(K|B)_{\rho}$ denotes the probability of guessing $K$ by performing the pretty good measurement~\cite{Hausladen94}. For $\rho_{KB}=\sum_{k}\proj{k}_{K}\ot\rho_{B}^{k}$, it has measurement operators
\begin{align}
\Pi^{k}_{B}=\rho_{B}^{-1/2}\rho_{B}^{k}\rho_{B}^{-1/2}\ .
\end{align}
\end{corollary}

It is known that the pretty good measurement performs close to optimally, i.e., the following bound holds~\cite{Hausladen94} for $\rho_{KB}\in\cS(\cH_{KB})$ classical on $K$ with respect to the basis $\{\ket{k}\}_{k\in K}$,
\begin{align}
P^{2}_{\guess}(K|B)_{\rho}\leq P^{\pg}_{\guess}(K|B)_{\rho}\leq P_{\guess}(K|B)_{\rho}\ .
\end{align}

For completeness, we mention that the quantum relative entropy, the max- and min-relative entropy, as well as the conditional R\'enyi 2-entropy are all special cases of the following family of R\'enyi type relative entropies (at least for finite-dimensional systems)~\cite{Tomamichel13}.

\begin{definition}[Relative $\alpha$-entropies]
Let $\rho,\sigma\in\cP^{+}(\cH)$, and $\alpha>0$. The relative $\alpha$-entropies are defined as\footnote{The relative $\alpha$-R\'enyi entropies are often defined as $D'_{\alpha}(\rho\|\sigma)=\frac{1}{\alpha-1}\log\trace[\rho^{\alpha}\sigma^{1-\alpha}]$, and we note that this definition is in general not equivalent to~\eqref{eq:alpharenyi}.}
\begin{align}\label{eq:alpharenyi}
D_{\alpha}(\rho\|\sigma)=\frac{\alpha}{\alpha-1}\cdot\log\big\|\sigma^{-\frac{1}{2}}\rho\sigma^{-\frac{1}{2}}\big\|_{\alpha,\sigma}\ ,
\end{align}
with $\|\cdot\|_{\alpha,\sigma}$ as defined in~\eqref{eq:sigma_weighted}. The cases $\alpha=1,\infty$ are defined via the corresponding limit.
\end{definition}

\begin{proposition}\cite{Tomamichel13}
Let $\rho\in\cS(\cH)$, and $\sigma\in\cP^{+}(\cH)$. Then, we have that
\begin{align}
&D_{\max}(\rho\|\sigma)=D_{\infty}(\rho\|\sigma),\;D_{\min}(\rho\|\sigma)=D_{1/2}(\rho\|\sigma),\;D(\rho\|\sigma)=D_{1}(\rho\|\sigma)\ .
\end{align}
Furthermore, for $\rho_{AB}\in\cP^{+}(\cH)$ we have that
\begin{align}
H_{2}(A|B)_{\rho}&=-D_{2}(\rho_{AB}\|\id_{A}\ot\rho_{B})\ .
\end{align}
\end{proposition}

Finally, we also need the following entropies for purely technical reasons (in particular, these entropies will not show up in any of the results discussed in this work).

\begin{definition}[$R$-entropy]
Let $\rho_{A}\in\cP^{+}(\cH_{A})$. The $R$-entropy of $\rho_{A}$ is defined as
\begin{align}
H_{R}(A)_{\rho}=-\sup\{\lambda\in\nR:\rho_{A}\geq2^{\lambda}\cdot\rho_{A}^{0}\}\ .
\end{align}
\end{definition}

\begin{definition}[Conditional zero entropy]\label{def:altmax}
Let $\rho_{AB}\in\cP^{+}(\cH_{AB})$. The conditional zero-entropy of $A$ given $B$ is defined as\footnote{In the literature this quantity is also known as conditional max-entropy~\cite{Renner05,Datta09,Buscemi11}, or alternative conditional max-entropy~\cite{Tomamichel11}.}
\begin{align}
H_{0}(A|B)_{\rho}=\sup_{\sigma_{B}\in\cS(\cH_{B})}\log\trace\left[\rho_{AB}^{0}(\1_{A}\ot\sigma_{B})\right]\ .
\end{align}
\end{definition}

The conditional zero-entropy for quantum-classical states is as follows.

\begin{lemma}\label{lem:h0class}
Let $\rho_{AK}\in\cP^{+}(\cH_{AK})$ be classical on $K$ with respect to the basis $\{\ket{k}\}_{k\in K}$. Then, we have that
\begin{align}
H_{0}(A|K)_{\rho}=\max_{k\in K}H_{0}(A)_{\rho^{k}}\ ,
\end{align}
where $\rho_{AK}=\sum_{k\in K}\rho_{A}^{k}\ot\proj{k}_{K}$.
\end{lemma}

\begin{proof}
The proof is from the collaboration~\cite{Berta11_2}. Inserting the definition of the conditional zero-entropy (Definition~\ref{def:altmax}) we calculate
\begin{align}
H_{0}(A|B)_{\rho}&=\max_{\sigma_{B}\in\cS(\cH_{B})}\log\trace\big[\rho_{AB}^{0}(\1_{A}\ot\sigma_{B})\big]\notag\\
&=\log\max_{\sigma_{B}\in\cS(\cH_{B})}\trace\Big[\Big(\sum_{k}\big(\rho_{A}^{k}\big)^{0}\ot\proj{k}_{B}\Big)\left(\1_{A}\ot\sigma_{B}\right)\Big]\notag\\
&=\log\max_{\sigma_{B}\in\cS(\cH_{B})}\trace\Big[\sigma_{B}\cdot\Big(\sum_{k}\proj{k}_{B}\cdot\trace\big[\big(\rho_{A}^{k}\big)^{0}\big]\Big)\Big]\notag\\
&=\log\Big\|\sum_{k}\proj{k}_{B}\cdot\trace\Big[\big(\rho_{A}^{k}\big)^{0}\Big]\Big\|_{\infty}=\log\max_{k}\trace\Big[\big(\rho_{A}^{k}\big)^{0}\Big]=\max_{k}H_{0}(A)_{\rho^{k}}\ .
\end{align}
\end{proof}

We define a smooth version of the conditional zero-entropy for quantum-classical states.

\begin{definition}[Smooth conditional zero-entropy]
Let $\eps\geq0$, and let $\rho_{AK}\in\cS(\cH_{AK})$ be classical on $K$ with respect to the basis $\{\ket{k}\}_{k\in K}$. The $\eps$-smooth conditional zero-entropy of $A$ given $B$ is defined as
\begin{align}
H_{0}^{\eps}(A|K)_{\rho}=\inf_{\rhob_{AK}\in\cB^{\eps}_{\qc}(\rho_{AK})}H_{0}(A|K)_{\rhob}\ ,
\end{align}
where
\begin{align}
\cB^{\eps}_{\qc}(\rho_{AK})=\big\{\rhob_{AK}\in\cP^{+}(\cH_{AK}):\rhob_{AK}=\sum_{k}\rhob_{A}^{k}\ot\proj{k}_{K},\|\rho_{AK}-\rhob_{AK}\|_{1}\leq\eps\big\}\ .
\end{align}
\end{definition}

An entropic asymptotic equipartition property for the smooth conditional zero-entropy of quantum-classical states is stated in the appendix (Lemma~\ref{lem:h0aep}).


\chapter{Entropic Uncertainty Relations}\label{ch:uncertainty}

This introduction is partly taken from the collaboration~\cite{Berta13_3}.

Heisenberg's uncertainty principle is one of the characteristic features of quantum mechanics. It states that measurements of two non-commuting observables necessarily lead to statistical ignorance of at least one of the outcomes~\cite{Heisenberg27,Robertson29}. Entropic uncertainty relations are an information theoretic way to capture this. They were first derived by Hirschman~\cite{Hirschman57}, and later improved by Beckner~\cite{Beckner75} and Bialynicki-Birula and Mycielski~\cite{Birula75}, but many new results have been achieved in recent years (see the review articles~\cite{Wehner09,Birula10} and references therein). A celebrated result of Maassen and Uffink~\cite{Maassen88} reads as follows. It holds for a measurement chosen at random from two non-degenerate observables $X$ and $Y$ on a finite-dimensional quantum system $\cH_{A}$ that
\begin{align}\label{eq:MU}
H(K|\Theta)_{\rho}=\frac{1}{2}\cdot\big(H(K|\Theta=X)_{\rho}+H(K|\Theta=Y)_{\rho}\big)\geq\frac{1}{2}\cdot\log\frac{1}{c}\ ,
\end{align}
where $H(K|\Theta=X)_{\rho}$ denotes the Shannon entropy of the probability distribution of the outcomes when $X$ is measured on a quantum state $\rho_{A}\in\cS(\cH_{A})$, and
\begin{align}\label{eq:c}
c=\max_{k,k'}|\braket{x_{k}|y_{k'}}|^{2}\ ,
\end{align}
with $\ket{x_{k}}$ and $\ket{y_{k'}}$ the eigenvectors of $X$ and $Y$, respectively.\footnote{There is a priori no reason to prefer the Shannon entropy over other entropies for measuring the uncertainty, and many relations in terms of, e.g., R\'enyi~\cite{Renyi60} or Tsallis~\cite{Tsallis88} entropies have been derived (see the review articles~\cite{Wehner09,Birula10} and references therein).} To see that this is an uncertainty relation note that if one of the two entropies is zero, then~\eqref{eq:MU} tells us that the other is necessarily non-zero, i.e., there is at least some amount of uncertainty. The largest amount of uncertainty is obtained when $|\braket{x_{k}|y_{k'}}| = 1/\sqrt{|A|}$, that is, the two bases $\{\ket{x_{k}}\}_{k=1}^{|A|}$ and $\{\ket{y_{k'}}\}_{k'=1}^{|A|}$ are mutually unbiased~\cite{Ivanovic81}.

One way to think about uncertainty is through the following uncertainty game \cite{Berta10} between two players, Alice ($A$) and Bob ($B$). Before the game commences, Alice and Bob agree on two non-degenerate observables $X$ and $Y$ on a finite-dimensional quantum system $\cH_{A}$. The game proceeds as follows. Bob prepares a quantum state $\rho_{A}\in\cS(\cH_{A})$ of his choosing and sends it to Alice. Alice then carries out one of the two measurements and announces her choice to Bob. Bob's task is to minimize his uncertainty about Alice's measurement outcome. The amount of uncertainty as measured by entropies can be understood as a limit on how well Bob can guess Alice's measurement outcome --- the more difficult it is for Bob to guess the more uncertain Alice's measurement outcomes are. If Bob is not entangled with $A$ but only keeps classical information about the state, such as the description of the density matrix $\rho_{A}$, then~\eqref{eq:MU} bounds Bob's ability to win the uncertainty game~\cite{Hall95,Cerf02,Boileau09}.

However, another central element of quantum mechanics is the possibility of entanglement. In particular, Einstein, Podolsky and Rosen~\cite{Einstein35} observed that if Bob is maximally entangled with $A$, then his uncertainty can be reduced dramatically. Bob can simply measure his half of the maximally entangled state in the same basis as Alice to predict her measurement outcome perfectly, winning the uncertainty game described above. This is precisely the effect observed in~\cite{Einstein35} and shows that the uncertainty relations of the type~\eqref{eq:MU} are inadequate to capture the interplay between entanglement and uncertainty. Fortunately, it is possible to extend the notion of uncertainty relations to take the possibility of entanglement into account, and such relations are known as uncertainty relations with quantum side information. More precisely, we have shown in~\cite{Berta10} that~\eqref{eq:MU} can be extended to
\begin{align}\label{eq:NP}
H(K|B\Theta)_{\rho}=\frac{1}{2}\cdot\big(H(K|B\Theta=X)_{\rho}+H(K|B\Theta=Y)_{\rho}\big)\geq\frac{1}{2}\cdot\big(\log\frac{1}{c}+H(A|B)_{\rho}\big)\ ,
\end{align}
where $H(K|B\Theta=X)_{\rho}$ is the conditional von Neumann entropy of the measurement outcomes of $X$ given the quantum side information $B$,\footnote{More precisely, $H(K|B\Theta=X)_{\rho}$ is the conditional von Neumann entropy of the post-measurement state $\rho_{KB}=\sum_{k}\left(\proj{x_{k}}_{A}\ot\1_{B}\right)\rho_{AB}\left(\proj{x_{k}}_{A}\ot\1_{B}\right)$.} and $H(A|B)_{\rho}$ is the conditional von Neumann entropy of $A$ given $B$ for $\rho_{AB}\in\cS(\cH_{AB})$. If $A$ and $B$ are entangled, then $H(A|B)_{\rho}$ can be negative. Indeed, $H(A|B)_{\rho}=-\log|A|$ when $\rho_{AB}$ is the maximally entangled state, in which case the lower bound in~\eqref{eq:NP} becomes trivial. This shows that uncertainty should not be treated as an absolute, but with respect to the knowledge of an observer who can be quantum as well. The uncertainty relation~\eqref{eq:NP} also shows that little uncertainty, i.e., $H(K|B\Theta)_{\rho}$ small, implies that $H(A|B)_{\rho}$ must be negative, and hence $\rho_{AB}$ is entangled~\cite{Devetak05_3}. As such,~\eqref{eq:NP} is useful for the task of witnessing entanglement~\cite{Prevedel11,Li11}. However, we note that~\eqref{eq:NP} does not tell us if the presence of entanglement really does lead to a significant reduction in uncertainty.

The first part of this chapter (Section~\ref{sec:mub}) is based on~\cite{Berta13_3}, and its aim is to prove that there exists an exact relation between entanglement and uncertainty. More precisely, we show that if we measure $A$ uniformly at random in one of $|A|+1$ possible mutually unbiased bases, the following uncertainty equality holds,
\begin{align}\label{eq:h2_intro}
H_2(K|B\Theta)_{\rho}=\log\big(|A|+1\big)-\log\big(2^{-H_2(A|B)_{\rho}}+1\big)\ ,
\end{align}
where we use the conditional R\'enyi 2-entropy (see Section~\ref{sec:entropy_finite}) to quantify uncertainty and entanglement. With respect to the uncertainty game (now with $|A|+1$ observables from a full set of mutually unbiased bases), the use of the conditional R\'enyi 2-entropy is natural, since the classical-quantum term corresponds to Bob's probability of winning the uncertainty game by using the pretty good measurement (Corollary~\ref{cor:pgm}), and the fully quantum term measures how close Bob can bring $\rho_{AB}$ to a maximally entangled state by performing the pretty good recovery map (Proposition~\ref{prop:h2operational}). Our result~\eqref{eq:h2_intro} reconciles the observations of Einstein, Podolsky and Rosen~\cite{Einstein35}, and Heisenberg~\cite{Heisenberg27} into a single equation.

At the price of again getting inequalities instead of an equality, we then extend~\eqref{eq:h2_intro} to smaller sets of measurements, and deduce uncertainty relations (lower bounds on the uncertainty) as well as certainty relations (upper bounds on the uncertainty) in terms of the conditional R\'enyi 2-entropy, and the smooth conditional min-entropy. Finally, another extension to sets of measurements with simple tensor product structure, also allows us to get uncertainty relations in terms of the conditional von Neumann entropy (Section~\ref{sec:single_qudit}). We briefly mention applications to witnessing entanglement and the noisy storage model in two-party quantum cryptography (cf.~Section~\ref{se:qc}).

The second part of this chapter (Section~\ref{se:two}) is about entropic uncertainty relations with quantum side information in the tripartite setup. It was already realized in~\cite{Boileau09,Berta10} that the monogamy property of entanglement allows for a particularly elegant formulation in this case. It holds for any finite-dimensional tripartite quantum state $\rho_{ABC}\in\cS(\cH_{ABC})$, and two non-degenerate observables $X$ and $Y$ on $A$ that
\begin{align}\label{eq:tri_vN}
H(X|B)_{\rho}+H(Y|C)_{\rho}\geq\log\frac{1}{c}\ ,
\end{align}
where $H(X|B)_{\rho}$ and $H(Y|C)_{\rho}$ are the conditional von Neumann entropy of the measurement outcomes of $X$ and $Y$ given the side information $B$ and $C$, respectively,\footnote{More precisely, $H(X|B)_{\rho}$ is the conditional von Neumann entropy of the post-measurement state $\rho_{XB}=\sum_{k}\left(\proj{x_{k}}_{A}\ot\1_{B}\right)\rho_{AB}\left(\proj{x_{k}}_{A}\ot\1_{B}\right)$, where $\ket{x_{k}}$ denote the eigenvectors of $X$.} and $c$ is as in~\eqref{eq:c}. Furthermore, Tomamichel and Renner~\cite{Tomamichel11_2,Tomamichel12} generalized this to smooth conditional min- and max-entropy. For the same notation as in~\eqref{eq:tri_vN} and $\eps\geq0$, it holds that
\begin{align}\label{eq:tri_minmax}
H_{\min}^{\eps}(X|B)_{\rho}+H_{\max}^{\eps}(Y|C)_{\rho}\geq\log\frac{1}{c}\ .
\end{align}
However, all the previously mentioned results involving quantum side information assume finite-dimensional quantum systems. In contrast to this, the first papers about the uncertainty principle discuss position and momentum measurements (e.g., \cite{Heisenberg27,Hirschman57}), and thus, are for infinite-dimensional quantum systems. This problem was recently addressed by Frank and Lieb in~\cite{Frank12}. They discuss entropic uncertainty relations with quantum side information in terms of the von Neumann entropy, which also apply to continuous position and momentum distributions. Yet, they only consider finite-dimensional quantum side information. The difficulties for defining conditional entropies in an infinite-dimensional setup were already noted by Kuznetsova~\cite{Kuznetsova10}.

In Section~\ref{se:two}, we derive entropic uncertainty relations with quantum side information for infinite-dimensional quantum systems in the setup of von Neumann algebras. In a previous work~\cite{Berta11_4}, we have shown such a relation under the restriction of measurements with a finite outcome range and for smooth conditional min- and max-entropy. Here, we generalize~\eqref{eq:tri_vN} and~\eqref{eq:tri_minmax} to measurements with continuous outcome range. For this purpose, we introduced in Section~\ref{sec:entropy_cq} conditional differential von Neumann and conditional differential min- and max-entropy, where the classical system is described by a measure space, and the quantum side information is modeled by a von Neumann algebra. We proved that these differential entropies can be approximated by their discretized counterparts (Propositions~\ref{prop:vNapprox} and~\ref{thm:MinMaxApprox}), and by employing an elegant proof technique of Coles {\it et al.}~\cite{Coles11,Coles12} this tool then allows us to obtain the uncertainty relations for measurements with a continuous outcome range. As an example, we apply our uncertainty relation to position and momentum measurements (Section~\ref{sec:pos_mom}). In this situation, the aforementioned approximation property can be seen as a coarse-graining of the outcome range with finer and finer intervals.

Tripartite entropic uncertainty relations with quantum side information also found applications in quantum key distribution~\cite{Tomamichel12_2,Tomamichel11_2}. In that respect, we were able to use our infinite-dimensional smooth conditional min- and max-entropy uncertainty relation for finite spacing position and momentum measurements to prove security of a continuous variable quantum key distribution protocol against coherent attacks including finite-size effects~\cite{Furrer12,Furrer12_2}. However, we do not explore this here.


\section{Bipartite Relations}\label{se:several}

The results in this section have been obtained in collaboration with Patrick Coles, Omar Fawzi, and Stephanie Wehner, and have appeared in~\cite{Berta13_3,Berta12_2,Berta11_5}. In this section, all systems (classical and quantum) are finite-dimensional.


\subsection{Mutually Unbiased Bases}\label{sec:mub}

A set of orthonormal bases $\{K_{\theta}\}_{\theta=1}^{n}$ with $K_{\theta}=\{\ket{\theta_{k}}\}_{k=1}^{|A|}$ on some Hilbert space $\cH_{A}$ is said to define mutually unbiased bases if for all elements $\ket{\theta_{k}}\in K_{\theta}$ and $\ket{\theta'_{k'}}\in K_{\theta'}$, $|\braket{\theta_{k}|\theta'_{k'}}|=1/\sqrt{|A|}$. There can be at most $|A|+1$ mutually unbiased bases for $\cH_{A}$, and constructions of full sets of $|A|+1$ mutually unbiased bases are known in prime power dimensions~\cite{Bandyopadhyay02,Wootters89}. We mention that unitaries which generate a full set of mutually unbiased bases can be implemented by quantum circuits of almost linear size~\cite{Fawzi11}. Whenever we talk about a full set of mutually unbiased bases on some Hilbert space $\cH_{A}$, we assume that the set exits. The following is mostly taken from the collaboration~\cite{Berta13_3}.

\paragraph{Main Result.} Here we precisely state our uncertainty equality in terms of the conditional R\'enyi 2-entropy, which shows that there exists an exact relation between entanglement and uncertainty. The proof is based on the fact that a full set of mutually unbiased bases generates a complex projective 2-design~\cite{Klappenecker05}, a property that has already been used to show uncertainty relations without quantum side information~\cite{Ballester07}.

\begin{theorem}\label{thm:h2relation}
Let $\rho_{AB}\in\cS_{\leq}(\cH_{AB})$, and denote the elements of $|A|+1$ mutually unbiased bases on $\cH_{A}$ by $\{\ket{\theta_{k}}\}_{k=1}^{|A|}$. Then, we have that
\begin{align}\label{eq:h2relation}
H_2(K|B\Theta)_{\rho}=\log\big(|A|+1\big)-\log\big(2^{-H_2(A|B)_{\rho}}+1\big)\ , 
\end{align}
where
\begin{align}\label{eq:state_kbtheta}
\rho_{KB\Theta}=\frac{1}{|A|+1}\cdot\sum_{\theta=1}^{|A|+1}\sum_{k=1}^{|A|}\big(\proj{\theta_{k}}\ot\id_{B}\big)\rho_{AB}\big(\proj{\theta_{k}}\ot\id_{B}\big)\ot\proj{\theta}_{\Theta}\ .
\end{align}
\end{theorem}

\begin{proof}
We define $\rhot_{AB}= (\id_{A} \ot\rho_B^{-1/4})\rho_{AB}(\id_{A} \ot\rho_B^{-1/4})$, and rewrite the fully quantum conditional collision entropy as $H_{2}(A|B)_{\rho}=-\log \trace\left[\rhot_{AB}^2\right]$. Similarly, we rewrite the classical-quantum conditional collision entropy as
\begin{align}
H_{2}(K|B\Theta)_{\rho}=-\log\Big(\frac{1}{|A|+1}\cdot\sum_{\theta,k}\trace_{B}\Big[\trace_{A}\big[\rhot_{AB}(\proj{\theta_k}\ot \id_{B})\big]^{2}\Big]\Big)\ .
\end{align}
Now, we introduce the space $\cH_{A'B'}\cong \cH_{AB}$ as well as the state $\rhot_{A'B'} \cong \rhot_{AB}$. We have
\begin{align}
\big(|A|+1\big)\cdot2^{-H_{2}(K|B\Theta)_{\rho}}&=\sum_{\theta,k}\trace_{B}\Big[\trace_{A}\big[(\proj{\theta_k}\ot\id_{B})\rhot_{AB}\big]\trace_{A}\big[(\proj{\theta_k}\ot\id_{B})\rhot_{AB}\big]\Big]\notag\\
&=\trace_{BB'}\Big[\trace_{AA'}\big[\sum_{\theta,k}(\proj{\theta_k}\ot\proj{\theta_k})(\rhot_{AB}\ot\rhot_{A'B'})\big]F_{BB'}\Big]\notag\\
&=\trace_{BB'}\Big[\trace_{AA'}\big[(\id_{AA'}+F_{AA'})(\rhot_{AB}\ot\rhot_{A'B'})\big]F_{BB'}\Big]\ ,\label{eq:h2equality}
\end{align}
where $F_{AA'}=\sum_{t,s}\ket{t}\bra{s}\ot\ket{s}\bra{t}$ is the operator that swaps $A$ and $A'$ (similarly for $F_{BB'}$). Here, the second line uses the swap trick, for operators $M$ and $N$, and swap operator $F$,
\begin{align}\label{eq:swaptrick}
\trace\big[MN\big]=\trace\big[(M\ot N)F\big]\ ,
\end{align}
and the third line invokes that a full set of mutually unbiased bases generates a complex projective 2-design~\cite{Klappenecker05}, that is,
\begin{align}
\sum_{\theta,k} \proj{\theta_{k}}\ot\proj{\theta_{k}}=\id_{AA'}+F_{AA'}\ .
\end{align}
Continuing from~\eqref{eq:h2equality}, we then get
\begin{align}
\big(|A|+1\big)\cdot2^{-H_{2}(K|B\Theta)_{\rho}}&=\trace_{B}\Big[\trace_{A}\big[\rhot_{AB}\big]\trace_{A}\big[\rhot_{AB}\big]\Big]\notag\\
&+\sum_{t,s}\trace_{BB'}\Big[\trace_{AA'}\big[(\ket{t}\bra{s}\ot\ket{s}\bra{t}\ot\id_{BB'})(\rhot_{AB}\ot\rhot_{A'B'})F_{BB'}\big]\Big]\notag\\
&=1+\sum_{t,s}\trace_{B}\Big[\trace_{A}\big[(\ket{t}\bra{s}\ot\id)\rhot_{AB}\big]\trace_{A}\big[(\ket{s}\bra{t}\ot\id)\rhot_{AB}\big]\Big]\notag\\
&=1+\trace\big[\rhot_{AB}^2\big]\ .
\end{align}
\end{proof}

This generalizes a known result without quantum side information~\cite{Ivanovic92,Brukner99}. To gain further intuition about~\eqref{eq:h2relation}, let us first return to the uncertainty game discussed earlier.
Note that in terms of the operational interpretations of the conditional R\'enyi 2-entropy (Proposition~\ref{prop:h2operational} and Corollary~\ref{cor:pgm}), we can rewrite~\eqref{eq:h2relation} as 
\begin{align}\label{eq:mainOperational}
P_{\rm guess}^{\pg}(K|B\Theta)_{\rho}=\frac{1}{|A|+1}\sum_{\theta} P_{\rm guess}^{\pg}(K|B\Theta=\theta)_{\rho}=\frac{|A|\cdot F^{\pg}(A|B)_{\rho}+1}{|A|+1}\ .
\end{align}
In the game, Bob prepares a state $\rho_{AB}$ and sends the $A$ system to Alice. She measures $A$ in one basis chosen uniformly at random from the complete set of $|A|+1$ MUBs, and announces the basis (the index $\theta$) to Bob. Bob's task is to guess Alice's outcome using the pretty good measurement on $B$. Equation~\eqref{eq:mainOperational} says that Bob's ability to win or lose this game is quantitatively connected to the entanglement of $\rho_{AB}$, as measured by $F^{\pg}(A|B)_{\rho}$.

\paragraph{Uncertainty and Certainty Relations.} We may ask if it is necessary to formulate our uncertainty equality~\eqref{eq:h2relation} using $|A|+1$ mutually unbiased bases, can we not use fewer measurements? To answer this, it is instructive to study what kind of relations~\eqref{eq:h2relation} implies. On the one hand, we can deduce regular uncertainty relations, and, e.g., we get a relation in terms of the smooth conditional min-entropy similar to~\cite{Berta12_2,Berta11_5,Fawzi12}.

\begin{corollary}\label{cor:min_UR}
Let $\rho_{AB}\in\cS_{\leq}(\cH_{AB})$, denote the elements of $|A|+1$ mutually unbiased bases on $\cH_{A}$ by $\{\ket{\theta_{k}}\}_{k=1}^{|A|}$, and let $\eps>0$. Then, we have that
\begin{align}
H_{\min}^{\eps}(K|\Theta B)_{\rho}\geq\log\big(|A|+1\big)-\log\big(2^{-H_{\min}(A|B)_{\rho}}+1\big)-1-2\log\frac{1}{\eps}\ ,
\end{align}
where $\rho_{KB\Theta}$ is as in~\eqref{eq:state_kbtheta}.
\end{corollary}

\begin{proof}
The claim follows immediately by the approximate equivalence of the smooth conditional min-entropy and the conditional R\'enyi 2-entropy (Lemma~\ref{lem:h2hmin_equiv}).
\end{proof}

But on the other hand, we also get relations that upper bound the uncertainties of incompatible observables. In the literature these are known as certainty relations~\cite{Sanchez95,Matthews09}, and here we give the first such relations that allow for quantum side information.

\begin{corollary}
Let $\rho_{AB}\in\cS_{\leq}(\cH_{AB})$, denote the elements of $|A|+1$ mutually unbiased bases on $\cH_{A}$ by $\{\ket{\theta_{k}}\}_{k=1}^{|A|}$, and let $\eps>0$. Then, we have that
\begin{align}
H_{\min}(K|\Theta B)_{\rho}\leq\log\big(|A|+1\big)-\log\big(2^{-H_{\min}^{\eps}(A|B)_{\rho}}+\frac{2}{\eps^{2}}\big)+1+2\log\frac{1}{\eps}\ ,
\end{align}
where $\rho_{KB\Theta}$ is as in~\eqref{eq:state_kbtheta}.
\end{corollary}

\begin{proof}
The claim follows immediately by the approximate equivalence of the smooth conditional min-entropy and the conditional R\'enyi 2-entropy (Lemma~\ref{lem:h2hmin_equiv}).
\end{proof}

Now, there is a simple argument that considering fewer than $|A|+1$ measurements implies that only trivial certainty relations can hold. As uncertainty equalities as in~\eqref{eq:h2relation} imply non-trivial certainty relations, such equalities cannot hold for a smaller number of measurements. This is in sharp contrast to uncertainty relations, where non-trivial relations can be obtained for just two measurements. To see this, consider the case where $\rho_A$ is just a one qubit state and we perform measurements in the Pauli $\sigma_X$ and Pauli $\sigma_Z$ eigenbases, respectively. Let $B$ be trivial, i.e., $H_{\min}(K|\Theta B)=H_{\min}(K|\Theta)$. In this case, we are just considering the entropy of the outcome distribution of measuring $\rho_A$ in one of the two bases. Clearly, when $\rho_A$ is an eigenstate of Pauli $\sigma_Y$, then the outcome distribution for both Pauli $\sigma_X$ and Pauli $\sigma_Z$ is uniform and hence $H_{\min}(K|\Theta)=1$, which is the maximum value. The same argument shows that when measuring in fewer than $|A|+1$ mutually unbiased bases, $H_{\min}(K|\Theta)=\log|A|$ which is the maximum value that it can take, and hence only the trivial certainty relation holds. It is thus clear that equalities such as~\eqref{eq:h2relation} can only hold for sets of measurements which are sufficiently rich.

\paragraph{Bounds for Fewer Bases} Even though there does not exist an uncertainty equality for measuring in fewer than $|A|+1$ mutually unbiased bases, we can still give lower and upper bounds for Bob's uncertainty about $1\leq n\leq|A|$ mutually unbiased bases on $A$ in terms of the entanglement between $A$ and $B$. Moreover, these inequalities are tight for all $n$, that is, fixing the set of measurements and the entanglement to be constant, there exist states that achieve the upper and lower bounds. Our relations are again in terms of the conditional collision entropy, and using $P^{\pg}_{\guess}(n)_{\rho}$ as a shorthand to denote Bob's average guessing probability $P^{\pg}_{\guess}(K|B\Theta)_{\rho}$ when Alice does $n$ measurements, we find the following.

\begin{corollary}\cite{Berta13_3}\label{cor:lessbases}
Let $\rho_{AB}\in\cS(\cH_{AB})$, and denote the elements of $|A|+1$ mutually unbiased bases on $\cH_{A}$ by $\{\ket{\theta_{k}}\}_{k=1}^{|A|}$. For $F^{\pg}(A|B)_{\rho}\leq1/|A|$ we have that
\begin{align}\label{eqn8}
\frac{1}{|A|}\leq P^{\pg}_{\guess}(n)_{\rho}\leq\frac{|A|}{n}\cdot F^{\pg}(A|B)_{\rho}+\frac{n-1}{n\cdot|A|}\ .
\end{align}
For $F^{\pg}(A|B)_{\rho}>1/|A|$ we have that
\begin{align}\label{eqn89}
F^{\pg}(A|B)_{\rho}\leq P^{\pg}_{\guess}(n)_{\rho}\leq\frac{n-1}{n}\cdot F^{\pg}(A|B)_{\rho}+\frac{1}{n}\ .
\end{align}
Moreover, these inequalities are tight for all $n$.
\end{corollary}

From Corollary~\ref{cor:lessbases}, we can also see that if Bob can guess two mutually unbiased bases on $A$ well, then he can also guess $|A|+1$ mutually unbiased bases on $A$ fairly well. Conceptually this follows from a two step chain of reasoning: if Bob's uncertainty is low for two mutually unbiased bases, then he must be entangled to Alice, which in turn implies that he must have a low uncertainty for all bases. So entanglement provides the key link, from two mutually unbiased bases to all bases. From the above results, it is straightforward to derive the following quantitative statement of this idea
\begin{align}
P^{\pg}_{\guess}(|A|+1)_{\rho}\geq\frac{|A|\cdot\big(2\cdot P^{\pg}_{\guess}(2)_{\rho}-1\big)+1}{|A|+1}\ .
\end{align}

\paragraph{Witnessing Entanglement.} Following~\cite{Berta10}, Theorem~\ref{thm:h2relation} offers a simple strategy for witnessing entanglement since it connects entanglement to uncertainty, which is experimentally measurable. In particular, Alice and Bob (in their distant labs, receiving $A$ and $B$ respectively) can sample from the source multiple times and communicate their results to gather statistics, say, regarding the $K_{\theta}$ observable on $A$ and the $L_{\theta}$ observable on $B$. Suppose they do this for a set of $n$ mutually unbiased bases $\{K_{\theta}\}_{\theta=1}^n$ on $A$, with Bob measuring in a some set of $n$ bases $\{L_{\theta}\}_{\theta=1}^n$ on $B$, and then they estimate the joint probability distribution for each pair $\{K_{\theta},L_{\theta}\}$. Hence, they can evaluate the classical entropies $H_2(K_{\theta}|L_{\theta})_{\rho}$.

\begin{corollary}\cite{Berta13_3}
Let $\rho_{AB}\in\cS(\cH_{AB})$ be separable between $A$ and $B$, let $\{K_{\theta}\}_{\theta=1}^n$ be a subset (of size $n$) of a complete set of mutually unbiased bases of $\cH_{A}$, and let $\{L_{\theta}\}_{\theta=1}^n$ be a set of $n$ orthonormal bases of $\cH_{B}$. Then, we have that
\begin{align}\label{eq:sepcondition}
\sum_{\theta=1}^n 2^{-H_2(K_{\theta}|L_{\theta})_{\rho}}\leq1+\frac{n-1}{|A|}\ ,
\end{align}
where
\begin{align}
\rho_{K_{\theta}L_{\theta}}=\sum_{p,q}\big(\proj{K_{\theta,p}}\ot\proj{L_{\theta,q}}\big)\rho_{AB}\big(\proj{K_{\theta,p}}\ot\proj{L_{\theta,q}}\big)\ ,
\end{align}
with $K_{\theta}=\{\ket{K_{\theta,p}}\}$, and $L_{\theta}=\{\ket{L_{\theta,q}}\}$.
\end{corollary}

This method offers the flexibility of witnessing entanglement with $2\leq n\leq|A|+1$ observables. We note that for $n=2$, the same strategy based on the uncertainty relation~\eqref{eq:NP} was implemented in~\cite{Prevedel11,Li11}.

\paragraph{Extensions.} We note that the proof of Theorem~\ref{thm:h2relation} is only based on the complex projective 2-design property of a full set of mutually unbiased bases. With this, it is possible to show similar relations for other sets of measurements that generate a complex projective 2-design. As examples, we mention symmetric informationally complete positive operator valued measures (SIC-POVMs), and measurements generated by (approximate) unitary 2-designs. For details we refer to~\cite{Berta13_3,Berta11_5}. Finally, we will use Theorem~\ref{thm:h2relation} in Section~\ref{se:qc} to lift classical randomness extractors to quantum-classical randomness extractors against quantum side information.


\subsection{Single Qudit Measurements}\label{sec:single_qudit}

If $\cH_{A}$ can be decomposed into a tensor product of $d$ dimensional systems (qudits), then we can generalize the results of Section~\ref{sec:mub} to tensor product measurements of full sets of mutually unbiased bases on all the qudit spaces. This is appealing since the measurements then have a very simple tensor product structure. However, this comes at the price of only getting uncertainty inequalities instead of uncertainty equalities.

For the construction we take a full set of mutually unbiased bases in dimension $d$, and represent it by a set of unitary transformations $\{V_{1},\ldots,V_{d+1}\}$ mapping the mutually unbiased bases to some standard basis. For example, for the qubit space ($d=2$), we can choose 
\begin{align}
V_1 = \left( \begin{array}{cc}
1 & 0 \\
0 & 1
\end{array} \right)
\qquad V_2 = \frac{1}{\sqrt{2}} \left( \begin{array}{cc}
1 & 1 \\
1 & -1
\end{array} \right)
\qquad V_3 = \frac{1}{\sqrt{2}} \left( \begin{array}{cc}
1 & i \\
i & -1
\end{array} \right)\ .
\end{align}
We then define the set $\cV_{d,n}$ of unitary transformations on $n$ qudits by
\begin{align}\label{eq:Vdn}
\cV_{d,n}=\big\{V=V_{u_1}\ot\cdots\ot V_{u_n}|u_i\in\left\{1,\dots,d+1\right\}\big\}\ .
\end{align}

\begin{theorem}\label{thm:qudit_UR}
Let $\rho_{AB}\in\cS_{\leq}(\cH_{AB})$ with $|A|=d^{n}$, and let $\cV_{d,n}$ be defined as in~\eqref{eq:Vdn}. Then, we have that
\begin{align}
H_{2}(K|B\Theta)_{\rho}\geq n\cdot\big(\log(d+1)-1\big)+1-\log\big(2^{-H_{2}(A|B)_{\rho}}+1\big)\ ,
\end{align}
where
\begin{align}\label{eq:qudit_state}
\rho_{KB\Theta}=\frac{1}{(d+1)^{n}}\cdot\sum_{V_{\theta}\in\cV_{d,n}}\sum_{k=1}^{d^{n}}\big(V_{\theta}^{\dagger}\proj{k}V_{\theta}\ot\id_{B}\big)\rho_{AB}\big(V_{\theta}^{\dagger}\proj{k}V_{\theta}\ot\id_{B}\big)\ot\proj{\theta}_{\Theta}\ ,
\end{align}
for $\{\ket{k}\}_{k=1}^{|A|}$ an orthonormal basis of $\cH_{A}$.
\end{theorem}

\begin{proof}
As in the proof of Theorem~\ref{thm:h2relation}, we introduce $\rhot_{AB}=(\id_{A}\ot\rho_{B}^{-1/4})\rho_{AB}(\id_{A}\ot\rho_{B}^{-1/4})$, and rewrite $H_{2}(A|B)_{\rho}=-\log\big[\rhot_{AB}^{2}\big]$ as well as
\begin{align}
H_{2}(K|B\Theta)_{\rho}-\log\Big(\frac{1}{(d+1)^{n}}\cdot\sum_{\theta,k}\trace_{B}\Big[\trace_{A}\big[\rhot_{AB}(V_{\theta}^{\dagger}\proj{k}V_{\theta}\ot\id_{B})\big]^{2}\Big]\Big)\ .
\end{align}
For $\cH_{A'B'}\cong\cH_{AB}$ and $\rhot_{A'B'}\cong\rho_{AB}$ we then arrive at
\begin{align}\label{eq:qudit_derivation1}
(d+1)^{n}\cdot2^{-H_{2}(K|B\Theta)_{\rho}}=\trace\Big[\rhot_{AB}^{\ot2}\sum_{\theta,k}\big(V_{\theta}^{\dagger}\proj{k}V_{\theta}\big)^{\ot2}\ot F_{BB'}\Big]\ ,
\end{align}
where $F_{BB'}$ denotes the swap operator. We continue with
\begin{align}\label{eq:qudit_derivation2}
\sum_{\theta,k}\big(V_{\theta}^{\dagger}\proj{k}V_{\theta}\big)^{\ot2}&=\sum_{k_{1},\ldots,k_{n}}\sum_{V_{1},\ldots,V_{n}}\bigotimes_{i}\big(V_{i}^{\dagger}\proj{k_{i}}V_{i}\big)^{\ot2}=\bigotimes_{i}\Big(\sum_{k_{i},V_{i}}V_{i}^{\dagger}\proj{k_{i}}V_{i}\Big)^{\ot2}\ .
\end{align}
As $\{V_{1},\ldots,V_{d+1}\}$  form a full set of mutually unbiased bases in dimension $d$, and with this form a complex projective 2-design~\cite{Klappenecker05}, we get
\begin{align}
\sum_{x=1}^{d}\sum_{V\in\cV_{d,1}}\big(V^{\dagger}\proj{x}V\big)^{\ot2}=2\cdot\Pi^{\sym}\ ,
\end{align}
where $\Pi^{\sym}=\id+F$ denotes the projector onto the symmetric subspace spanned by the vectors $\ket{xx'}+\ket{x'x}$. Furthermore, $(\Pi_{C}^{\sym})^{\ot n}\leq\Pi_{C^{\ot n}}^{\sym}$ for any quantum system $C$, and hence we obtain with~\eqref{eq:qudit_derivation2},
\begin{align}
\sum_{\theta,k}\big(V_{\theta}^{\dagger}\proj{k}V_{\theta}\big)^{\ot2}\leq2^{n-1}\cdot(\id_{AA'}+F_{AA'})\ .
\end{align}
Putting this in~\eqref{eq:qudit_derivation1}, we get
\begin{align}
(d+1)^{n}\cdot2^{-H_{2}(K|B\Theta)_{\rho}}\leq2^{n-1}\cdot\trace\Big[\rhot_{AB}^{\ot2}(\id_{AA'}+F_{AA'})\ot F_{BB'}\Big]\ ,
\end{align}
and by using the same arguments as in the proof of Theorem~\ref{thm:h2relation}, the claim follows.
\end{proof}

Similarly as in Section~\ref{sec:mub}, we deduce an uncertainty relation in terms of the smooth conditional min-entropy.

\begin{corollary}\label{cor:Hmin_qudit}
Let $\rho_{AB}\in\cS_{\leq}(\cH_{AB})$ with $|A|=d^{n}$, let $\cV_{d,n}$ be defined as in~\eqref{eq:Vdn}, and $\eps>0$. Then, we have that
\begin{align}
H_{\min}^{\eps}(K|B\Theta)_{\rho}\geq n\cdot\big(\log(d+1)-1\big)-\log\big(2^{-H_{\min}(A|B)_{\rho}}+1\big)-2\log\frac{1}{\eps}\ ,
\end{align}
where $\rho_{KB\Theta}$ is as in~\eqref{eq:qudit_state}.
\end{corollary}

\begin{proof}
The claim follows immediately by the approximate equivalence of the smooth conditional min-entropy and the conditional R\'enyi 2-entropy (Lemma~\ref{lem:h2hmin_equiv}).
\end{proof}

An uncertainty relation similar to this was used in~\cite{Berta12_2,Berta11_5,Fawzi12} to analyze security in the noisy storage model~\cite{Wehner08,Schaffner08,Wehner10,Koenig12}. We will briefly discuss related ideas in Section~\ref{se:qc} by using Theorem~\ref{thm:qudit_UR} to lift classical randomness extractors to quantum-classical randomness extractors against quantum side information. We also note that in a recent work with Huei Ying Nelly Ng and Stephanie Wehner~\cite{Ng12}, we were able to improve Corollary~\ref{cor:Hmin_qudit} for the case of $n$ qubits and only classical side information. This is crucial for analyzing security in the bounded storage model~\cite{Damgard05,Damgard07,Schaffner07}, and concerning security in the noisy storage model, it is an interesting open question if Corollary~\ref{cor:Hmin_qudit} for the case of quantum side information can be improved similarly.

By the asymptotic equipartition property for the smooth conditional min-entropy (Lemma~\ref{lem:aepminmax}), we also get an uncertainty relation for the conditional von Neumann entropy.

\begin{corollary}\label{cor:vN_mub}
Let $\rho_{AB}\in\cS(\cH_{AB})$ with $|A|=d^{n}$, and let $\cV_{d,n}$ be defined as in~\eqref{eq:Vdn}. Then, we have that
\begin{align}
\frac{1}{(d+1)^{n}}\cdot\sum_{\theta=1}^{(d+1)^{n}}H(K|B)_{\rho^{\theta}}\geq n\cdot\big(\log(d+1)-1\big)+\min\big\{H(A|B)_{\rho},0\big\}\ ,
\end{align}
where
\begin{align}
\rho_{KB}^{\theta}=\sum_{k=1}^{d^{n}}\big(V_{\theta}^{\dagger}\proj{k}V_{\theta}\ot\id_{B}\big)\rho_{AB}\big(V_{\theta}^{\dagger}\proj{k}V_{\theta}\ot\id_{B}\big)\ .
\end{align}
\end{corollary}

\begin{proof}
The starting point is to employ Theorem~\ref{thm:qudit_UR} for an $m$-fold tensor product input system $\cH_{AB}^{\ot m}$. Let $\rhob_{AB}^{m}\in\cB^{\eps}(\rho_{AB}^{\ot m})$ such that $H_{\min}(A|B)_{\rhob^{m}}=H_{\min}^{\eps}(A|B)_{\rho^{\ot m}}$ for some $\eps>0$. Since the conditional R\'enyi 2-entropy is lower bounded by the conditional min-entropy (Lemma~\ref{lem:h2hmin_equiv}), and upper bounded by the conditional von Neumann entropy (Lemma~\ref{lem:h2vN_bounds}), it follows from Theorem~\ref{thm:qudit_UR} that
\begin{align}
H(K|B\Theta)_{\rhob^{m}}&\geq mn\cdot\big(\log(d+1)-1\big)+1-\log\big(2^{-H_{\min}(A|B)_{\rhob^{m}}}+1\big)\ .
\end{align}
By the continuity of the conditional von Neumann entropy (Lemma~\ref{lem:fannes}), the additivity of the conditional von Neumann entropy, and the definition of the state $\rhob_{AB}^{m}$ we arrive at
\begin{align}
m\cdot H(K|B\Theta)_{\rho}&\geq mn\cdot\big(\log(d+1)-1\big)+1-\log\big(2^{-H_{\min}^{\eps}(A|B)_{\rho^{\ot m}}}+1\big)\notag\\
&-8\eps\cdot mn\cdot\log|A|-4\cdot h(2\eps)\notag\\
&\geq mn\cdot\big(\log(d+1)-1\big)+1+\min\big\{H_{\min}^{\eps}(A|B)_{\rho^{\ot m}},0\big\}\notag\\
&-8\eps\cdot mn\cdot\log|A|-4\cdot h(2\eps)\ .
\end{align}
By the asymptotic equipartition theorem for the smooth conditional min-entropy (Lemma~\ref{lem:aepminmax}), we arrive at the claim by letting $m\ra\infty$ and $\eps\ra0$.
\end{proof}

Note that for $n=1$, this gives an uncertainty relation for a full set of mutually unbiased bases in terms of von Neumann entropies
\begin{align}\label{eq:mubvN}
\frac{1}{|A|+1}\cdot\sum_{\theta=1}^{|A|+1}H(K|B)_{\rho^{\theta}}\geq \log\big(|A|+1)\big)-1+\min\big\{H(A|B)_{\rho},0\big\}\ ,
\end{align}
To get a feeling for this bound it is instructive to consider some special cases. It was known that for the Shannon entropy~\cite{Larsen90, Ivanovic92,Sanchez93}
\begin{align}
\frac{1}{|A|+1}\cdot\sum_{\theta=1}^{|A|+1}H(K)_{\rho^{\theta}}\geq\log\big(|A|+1\big)-1\ ,
\end{align}
which is the best known bound for a full set of mutually unbiased bases and general $|A|$. This is clearly a special case of ours~\eqref{eq:mubvN}. However, for $|A|$ even this was improved to~\cite{Sanchez95,Sanchez98}
\begin{align}
\frac{1}{|A|+1}\cdot\sum_{\theta=1}^{|A|+1}H(K)_{\rho^{\theta}}\geq\frac{1}{|A|+1}\cdot\Big(\frac{|A|}{2}\cdot\log\frac{|A|}{2}+\Big(\frac{|A|}{2}+1\Big)\cdot\log\Big(\frac{|A|}{2}+1\Big)\Big)\ .
\end{align}
For $d=2$ the latter gives $2/3$ (which is also optimal), whereas our bound~\eqref{eq:mubvN} gives $\log(3)-1\approx0.585$.


\subsection{Discussion}\label{sec:discussion_several}

There are a number of open question we would like to mention. First of all, is it possible to derive an uncertainty equality for the conditional von Neumann entropy similar to Theorem~\ref{thm:h2relation}? We have shown a lower bound on the uncertainty in terms of the von Neumann entropy in Corollary~\ref{cor:vN_mub}, but how tight is this bound? And can we also find a corresponding upper bound on the uncertainty (certainty relation)?

Concerning the single qudit measurement relation we would like to know if is is possible to improve the bound in Theorem~\ref{thm:qudit_UR}. Furthermore, it would be interesting to see if a similar relation as in Theorem~\ref{thm:qudit_UR} also holds for only two complementary measurements per qubit (instead of the full set of three mutually unbiased bases). This would then also be interesting for the quantum-classical extractor constructions that we discuss in Section~\ref{sec:qudit_extractor}, and in particular for their application in the noisy storage model (Section~\ref{sec:storage}).

\textit{Note added.} The last two questions about single qudit measurement relations were recently settled in~\cite{Dupuis13}.


\section{Tripartite Relations}\label{se:two}

The results in this section have been obtained in collaboration with Matthias Christandl, Fabian Furrer, Volkher Scholz, and Marco Tomamichel, and have appeared in~\cite{Berta13_2,Berta11_4}. Similar results derived with different techniques can also be found in the thesis of Furrer~\cite{Furrer12_2}.


\subsection{General Measurements}

In the following, we derive entropic uncertainty relations with quantum side information for pairs of measurements with an outcome range given by a $\sigma$-finite measure space $(X,\Sigma,\mu)$. Our starting point is a recent proof technique developed for finite measure spaces~\cite{Coles12}, which we will lift by means of the approximation results derived in Section~\ref{sec:entropy_cq}. We start with an uncertainty relation in terms of the conditional min- and max-entropy, and proceed with the conditional von Neumann entropy.

\paragraph{Min- and Max-Entropy.} An entropic uncertainty relation for conditional min- and max-entropy on finite-dimensional systems was proven in~\cite{Tomamichel11_2} and generalized to finite measurements on von Neumann algebras in~\cite{Berta11_4}. We extend this result to measurements $E_{X}\in\Obs(X,\cM_A)$ and $F_{Y}\in\Obs(Y,\cM_A)$ with $(X,\Sigma_X,\mu_X)$ and $(Y,\Sigma_Y,\mu_Y)$ measure spaces with ordered dense sequences of balanced partitions \cite{Berta13_2}. A similar relation under different assumptions has also been shown in the thesis of Furrer~\cite{Furrer12_2} by means of other methods.

\begin{theorem}\label{thm:minmaxcont}
Let $\omega_{ABC}\in\cS(\cM_{ABC})$, and let be $E_{X}\in\Obs(X,\cM_{A})$ and $F_{Y}\in\Obs(Y,\cM_{A})$ with $(X,\Sigma,\mu)$ and $(Y,\Gamma,\nu)$ measure spaces with ordered dense sequences of balanced partitions $\{\cP_{\alpha}\}$ and $\{\cQ_{\beta}\}$, respectively. If for the post-measurement states $\omega_{XBC}=\omega_{ABC}\circ E_{X}$ and $\omega_{YBC}=\omega_{ABC}\circ F_{Y}$, there exists an $\alpha_{0}>0$ such that $H_{\max}(X_{\cP_{\alpha_{0}}})_{\omega}<\infty$, then
\begin{align}
h_{\max}(X|B)_{\omega}+h_{\min}(Y|C)_{\omega}\geq -\log c(E_{X},F_{Y})\ ,
\end{align}
where the overlap of the measurements is quantified by
\begin{align}\label{eq:ComplConstCont}
c(E_{X},F_{Y})=\lim_{\alpha,\beta\ra0}\sup_{I_{k}\in\cP_{\alpha}I_{l}\in\cQ_{\beta}}\frac{\|(E_{k}^{\cP_{\alpha}})^{1/2}\cdot(F_{l}^{\cQ_{\beta}})^{1/2}\|^{2}}{\alpha\cdot\beta}
\end{align}
and the notation is as in Section~\ref{se:qm}. 
\end{theorem}

\begin{proof}
We start by proving the claim for countable measure spaces $(X,\Sigma,\mu)$, $(Y,\Gamma,\nu)$.

We achieve this by first showing an inequality for sub-normalized measurements with a finite number of outcomes, and then use a limit argument to obtain the uncertainty relation for measurements with a countable number of outcomes. We describe sub-normalized measurements $E_X$ and $F_Y$ by a finite collection of positive operators $\{E_x\}_{x\in X}$ and $\{F_y\}_{y\in Y}$, which sum up to $M=\sum_{x}E_{x}\leq\id$ and $\sum_{y}F_{y}\leq\id$.

Let $\cH$ be a Hilbert space such that $\cM_{ABC}\subset\cB(\cH)$ is faithfully embedded, and there exist a purifying vector $\ket{\psi}\in\cH$ for $\omega_{ABC}$, that is, $\omega_{ABC}(\cdot)=\bra{\psi}\cdot\psi\rangle$. We choose a Stinespring dilation~\cite{Stinespring55} for $E_{X}$ of the form
\begin{align}
V:\cH\rightarrow\cH\ot\mathbb{C}^{|X|}\ot\mathbb{C}^{|X'|},\quad V\ket{\psi}=\sum_{x}E_{x}^{1/2}\ket{\psi}\ot\ket{x}\ot\ket{x}\ ,
\end{align}
where $\mathbb{C}^{|X|}$ denotes a $|X|$ dimensional quantum system in which the classical output of the measurements $E_{X}$ is embedded, and $X=X'$. Since $\ket{\psi}\in\cH$ is a purifying vector of $\omega_{ABC}$, we have that $V\ket{\psi}\in\cH\ot\mathbb{C}^{|X|}\ot\mathbb{C}^{|X'|},$ is a purifying vector of $\omega_{XB}=\omega_{AB}\circ E_{X}$. Denoting the commutant of $\cM_{ABC}$ in $\cB(\cH)$ by $\cM_{D}$, we find that the corresponding purifying system is equal to $\cB(\mathbb{C}^{|X'|})\ot\cM_{ACD}$. It then follows from the duality of the conditional min- and max-entropy (Proposition~\ref{prop:duality_minmax}) that
\begin{align}\label{eq:first}
H_{\max}(X|B)_{\omega}=-H_{\min}(X|X'ACD)_{\psi\circ V}\ ,
\end{align}
where $\psi_{ACD}\circ V_{XX'A}(\cdot)=\braket{\psi|V^{\dagger}(\cdot)V\psi}$. Since the conditional min-entropy can be written as a max-relative entropy (Definition~\ref{def:Hmin}), we have that
\begin{align}
-H_{\min}(X|X'ACD)_{\psi\circ V}=\inf_{\sigma}D_{\max}(\psi_{ACD}\circ V_{XX'A}\|\tau_{X}\ot\sigma_{X'ACD})\ ,
\end{align}
where the infimum is over $\sigma_{X'ACD}\in\cS(\cM_{X'ACD})$, and $\tau_{X}$ denotes the trace on $\cB(\mathbb{C}^{|X|})$. We define the completely positive map $\cE:\cB(\cH) \rightarrow \cB(\cH\otimes\mathbb{C}^{|X|}\otimes\mathbb{C}^{|X'|})$ given by $\cE(a)= V a V^*$. The map is sub-unital since $\cE(\id) = VV^*$ and
\begin{align}
\Vert V V^* \Vert = \Vert V^* V\Vert = \Vert\sum_{x}E_x \Vert = \Vert M \Vert \leq 1\ .
\end{align}
Due to the monotonicity of the max-relative entropy under application of sub-unital, completely positive maps (Lemma~\ref{lem:maxmono}), we obtain for fixed $\sigma_{X'ACD}\in\cS(\cM_{X'ACD})$,
\begin{align}\label{eq:uncertdev1}
D_{\max}(\psi_{ACD}\circ V\|\tau_{X}\otimes\sigma_{X'ACD})&\geq D_{\max}((\psi_{ACD}\circ V)\circ\cE \|(\tau_{X}\otimes\sigma_{X'ACD})\circ\cE )\nonumber\\
&=D_{\max}(\omega^V_{ACD}\|\gamma^{\sigma,V}_{ACD})\ ,
\end{align}
where $\omega^V_{ACD}= (\psi_{ACD}\circ V) \circ\cE $ and $\gamma^{\sigma,V}_{ACD}=(\tau_{X}\otimes\sigma_{X'ACD})\circ\cE $. Due to $V^*V= M$ we get
\begin{align}
\omega^V_{ACD}(\cdot)=\psi_{ACD}\circ V\circ V^{*}(\cdot)=\bra{\psi}V^{*}V(\cdot)V^{*}V\psi\rangle=\omega_{ACD}(M(\cdot)M)\ ,
\end{align}
with $\omega_{ACD}$ the state on $\cM_{ACD}$ corresponding to $\ket{\psi}$. Using once more the monotonicity of the max-relative entropy under application of channels (Lemma~\ref{lem:maxmono}), we obtain by first restricting onto the subalgebra $\cM_{AC}$ and then measuring the $A$ system with $F$,
\begin{align}\label{pf,thm:minmaxcont,eq2}
D_{\max}(\omega^V_{ACD}\|\gamma^{\sigma,V}_{ACD})\geq D_{\max}(\omega^V_{AC}\|\gamma^{\sigma,V}_{AC})\geq D_{\max}(\omega^V_{YC}\|\gamma^{\sigma,V}_{YC})\ ,
\end{align}
where $\omega^V_{YC}=\omega^V_{AC}\circ F_Y$ and $\gamma^{\sigma,V}_{YC} = \gamma^{\sigma,V}_{AC} \circ F_Y$. By definition, we have that $\gamma^{\sigma,V}_{YC} (a) = \sum_y \gamma^{\sigma,V}_{AC}(F_y a_y)$ for $a=(a_y)\in \cM_{YC}$. Hence, it holds for all positive $a_y\in\cM_C$ with $y\in Y$ that
\begin{align}\label{eq:boundC}
\gamma^{\sigma,V}_{AC}(F_y a_y)=\tau_X\otimes \sigma_{X'AC}(VF_y  a_y V^*)&=\sum_{x}\sigma^{x,x}_{AC}( \sqrt{E_x}F_y\sqrt{E_x} a_y )\nonumber\\
&\leq\sup_{x,y}\left\|E_{x}^{1/2} F_{y}^{1/2} \right \|^{2} \sigma_C(a_y)\ ,
\end{align} 
where we used that  $a_y$ commutes with $E_x$ and $F_y$, and $\sigma_{X'ACD}=(\sigma^{x,x'}_{ACD})$. Thus, we conclude that
\begin{align}
\gamma^\sigma_{YC}\leq\sup_{x,y}\left\|E_{x}^{1/2} F_{y}^{1/2}\right\|^{2}\cdot\tau_{Y}\otimes\sigma_{C}\ .
\end{align} 
By some elementary properties of the max-relative entropy (Lemma \ref{lem:minmax_elementary1} and Lemma \ref{lem:minmax_elementary2}), it then follows for any $\sigma_{X'ACD}\in\cS(\cM_{X'ACD})$ that
\begin{align}
D_{\max}(\omega^V_{YC}\|\gamma^\sigma_{YC})&\geq D_{\max}(\omega^V_{YC}\|\tau_{Y}\otimes\sigma_{C})-\log\sup_{x,y}\left\|E_{x}^{1/2} F_{y}^{1/2}\right\|^{2}\nonumber\\
&\geq\inf_{\eta}D_{\max}(\omega^V_{YC}\|\tau_{Y}\otimes\eta_{C})-\log\sup_{x,y}\left\|E_{x}^{1/2} F_{y}^{1/2}\right\|^{2}\nonumber\\
&=-H_{\min}(Y|C)_{\omega^V}-\log\sup_{x,y}\left\|E_{x}^{1/2} F_{y}^{1/2}\right\|^{2}\label{eq:uncertdev2}\ ,
\end{align}
where the infimum is over $\eta_{C}\in\cS(\cM_{C})$, and we used again that the conditional min-entropy can be written as a max-relative entropy (Definition~\ref{def:Hmin}). Combining this with all the steps going back to~\eqref{eq:first}, we obtain 
\begin{align} \label{eq:URdiscSubnormal}
H_{\max}(X|B)_{\omega} \geq -H_{\min}(Y|C)_{\omega^V}-\log\sup_{x,y}\left\|E_{x}^{1/2} F_{y}^{1/2}\right\|^{2}\ .
\end{align}
Recall that $\omega^V_{YC}=(\omega^{V,y}_{C})_y$ with $\omega^{V,y}_{C}(\cdot)= \omega(M F_y M \cdot)$, and thus, if $E$ is normalized we obtain the uncertainty relation for measurements with a finite number of outcomes.

Now, we lift the relation to the case of discrete but infinite $X$ and $Y$. We take sequences of increasing finite subsets $X_1 \subset X_2 \subset ... \subset X$ and $Y_1 \subset Y_2 \subset ...\subset  Y$ such that $\bigcup_n X_n = X$ and $\bigcup_n Y_n = Y$. We apply the inequality~\eqref{eq:URdiscSubnormal} derived for sub-normalized measurements to $E_{X_n}=\{E_x\}_{x\in X_n}$ and $F_{Y_m}=\{F_y\}_{y\in Y_m}$. For fixed $n$ and $m$~\eqref{eq:URdiscSubnormal} reads as
\begin{align} 
H_{\max}(X_n|B)_{\omega} \geq -H_{\min}(Y_m|C)_{\omega^n}-\log\sup_{x \in X ,y\in Y}\left\|E_{x}^{1/2} F_{y}^{1/2}\right\|^{2}\ ,
\end{align}
where $\omega_{X_nB} = \omega_{AB} \circ E_{X_n}$ and $\omega^n_{Y_mC} = (\omega^{n,y}_{C})_{y\in Y_m}$ with
\begin{align}
\omega^{n,y}_C(\cdot)=\omega_{AC}({M_n} F_y {M_n}\cdot)\ ,
\end{align}
and $M_n=\sum_{x\in X_n} E_x$. We now take the limit for $n\rightarrow \infty$ on both sides. By using the definition of the conditional max-entropy in~\eqref{eq:maxapprox}, we see that $ H_{\max}(X_n|B)_{\omega}$ converges to $H_{\max}(X|B)_{\omega}$ for $n\rightarrow \infty$. The only term on the right hand side depending on $n$ is the conditional min-entropy of the state $\omega^n_{Y_mC}$, which is given by (see~\eqref{eq:Guessing})
\begin{align}
H_{\min}(Y_m|C)_{\omega^n}=-\log\sup_{G} \sum_{y\in Y_m}\omega_{AC}(M_{n}F_{y}M_{n}G_{y})\ ,
\end{align}
where the supremum is taken over all $G=\{G_{y}\}_{y\in Y_m}$ in $\mathrm{Meas}(Y_m,\cM_{C})$. It holds for every $y\in Y_m$ and $0\leq G_{y}\leq\id$ that 
\begin{align}
|\omega_{AC}(F_{y}G_{y})-\omega_{AC}(M_{n}F_{y}M_{n}G_{y})|&\leq|\omega_{AC}(F_{y}G_{y}(\id-M_n))|\nonumber\\
&+|\omega_{AC}((\id-M_n)F_{y}G_{y}M_n)|\nonumber\\
&\leq2\sqrt{\omega_{AC}((\id-M_n)^2)}\ ,  
\end{align}
where we used the Cauchy-Schwarz inequality. Hence, we have that the $\omega_{C}^{n,y}(\cdot)=\omega_{AC}(M_{n}F_{y}M_{n}\cdot)$ converge uniformly to $\omega^{y}_{C}(\cdot)=\omega_{AC}(F_{y}\cdot)$ on the unit ball of $\cM_C$ for any $y \in Y_m$ (since $M_n$ converges in the $\sigma$-weak topology to $\id$). Because the set $Y_m$ is finite, this also implies that $\omega_{Y_{m}C}^n=(\omega_{C}^{n,y})_{y\in Y_{m}}$ converges uniformly to $\omega_{Y_{m}C}=(\omega_{C}^{y})_{y\in Y_{m}}$ on the unit ball of $\cM_{Y_{m}C}$. Hence, we can interchange the limit for $n \rightarrow \infty$ with the supremum over $\mathrm{Meas}(Y_m,\cM_{A})$ and obtain  
\begin{align}
H_{\max}(X|B)_{\omega}\geq-H_{\min}(Y_m|C)_{\omega}-\log\sup_{x\in X,y\in Y}\left\|E_{x}^{1/2} F_{y}^{1/2}\right\|^{2}\ .
\end{align}
We then take the infimum over all $m\in\mathbb{N}$ which gives (due to the definition of the conditional min-entropy in~\eqref{eq:minapprox}) the uncertainty relation
\begin{align}\label{eq:desired}
H_{\max}(X|B)_{\omega}+H_{\min}(Y|C)_{\omega}\geq-\log\sup_{x\in X,y\in Y}\left\|E_{x}^{1/2} F_{y}^{1/2}\right\|^{2}\ .
\end{align}

To prove the uncertainty relation for measure spaces $(X,\Sigma,\mu)$ and $(Y,\Gamma,\nu)$ with ordered dense sequences of balanced partitions, we now use the uncertainty relation~\eqref{eq:desired} for measurements with a countable number of outcomes. We obtain for any partitions $\cP_\alpha$ and $\cQ_\beta$ the inequality
\begin{align}
\Big(H_{\max}(X_{\cP_{\alpha}}|B)_{\omega}+\log\alpha\Big)&+\Big(H_{\min}(Y_{\cP_{\beta}}|C)_{\omega}+\log\beta\Big)\nonumber\\
&\geq-\log\sup_{I_{k}\in\cP_{\alpha}I_{l}\in\cQ_{\beta}}\frac{\|(E_{k}^{\cP_{\alpha}})^{1/2}\cdot(F_{l}^{\cQ_{\beta}})^{1/2}\|^{2}}{\alpha\cdot\beta}\ ,
\end{align}
where the notation is as in Section~\ref{se:qm}. Taking the limit for $\alpha,\beta \rightarrow 0$ on both sides, we finally obtain the desired uncertainty relation by means of the approximation of the conditional differential min- and max-entropy (Proposition~\ref{thm:MinMaxApprox}).
\end{proof}

\paragraph{Von Neumann Entropy.} The conditional von Neumann entropy can be seen a special case of smooth conditional min- and max-entropy via the fully quantum asymptotic equipartition property (Lemma~\ref{lem:aepminmax}). However, this is only known for finite-dimensional principal systems, and type I factor quantum side information. We could now try to generalize this to the setting of von Neumann algebras, but we do not do this here. Instead, we again employ the proof technique of~\cite{Coles12}, and use the approximation result for the conditional differential von Neumann entropy (Proposition~\ref{prop:vNapprox}).

\begin{theorem}\cite[Theorem 12]{Berta13_2}\label{prop:vN_tri_disc}
Let $\omega_{ABC}\in\cS(\cM_{ABC})$, and let be $E_{X}\in\Obs(X,\cM_{A})$ and $F_{Y}\in\Obs(Y,\cM_{A})$ with $(X,\Sigma,\mu)$ and $(Y,\Gamma,\nu)$ measure spaces with ordered dense sequences of balanced partitions $\{\cP_{\alpha}\}$ and $\{\cQ_{\beta}\}$, respectively. If the post-measurement states $ \omega_{XBC}=\omega_{ABC}\circ E_{X}$ and $\omega_{YBC}=\omega_{ABC}\circ F_{Y}$ satisfy $h(X|B)_{\omega}>-\infty$ and $h(Y|C)_{\omega}>-\infty$, and if there exists $\alpha_0> 0$ for which $H(X_{\cP_{\alpha_0}} |B)_{\omega}<\infty$ as well as $\beta_0>0$ for which $H(Y_{\cQ_{\beta_0}}|C)_{\omega}<\infty$, then
\begin{align}
h(X|B)_{\omega}+h(Y|C)_{\omega}\geq -\log c(E_{X},F_{Y})\ ,
\end{align}
where $c(E_{X},F_{Y})$ is as in~\eqref{eq:ComplConstCont}.
\end{theorem}

The proof is similar to the case of the conditional min- and max-entropy (Theorem~\ref{thm:minmaxcont}). We need to replace all min- and max-entropies by their corresponding von Neumann quantity. The claim then follows by the approximation property of the conditional differential von Neumann entropy (Proposition~\ref{prop:vNapprox}), the self duality of the conditional von Neumann entropy (Proposition~\ref{prop:duality_vN}), an extension of Lemma~\ref{lem:mono} about the monotonicity of the quantum relative entropy under application of channels (to sub-unital maps), and some elementary properties of the quantum relative entropy (Lemmas~\ref{lem:petz1} and~\ref{lem:petz2}).


\subsection{Position and Momentum Measurements}\label{sec:pos_mom}

Entropic uncertainty relations for position and momentum measurements have been heavily studied, see the review article~\cite{Birula10} and references therein. There are basically two types of relations, those for finite spacing position and momentum measurements, and those for continuous position and momentum distributions. Whereas relations for continuous distributions in terms of differential entropies are useful to express the general uncertainty principle in quantum mechanics, finite spacing relations are needed in the sense that any position and momentum measurement in the laboratory always has a finite resolution. These discrete relations then allow to actually test the uncertainty principle, and are also useful for information theoretic and cryptographic applications. Our contribution is to generalize many known relations to the case of quantum side information. We start from finite spacing position and momentum measurements, and then use the limiting results as discussed in Section~\ref{sec:entropy_cq} to go to continuous position and momentum distributions. Similar results can also be found in the thesis of Furrer~\cite{Furrer12_2}.

Let $Q$ and $P$ be position and momentum operators defined via the commutation relation $[Q,P]= i$. The representation space is $\cH=L^2(\mathbb R)$, with $Q$ the multiplication operator and $P$ the differential operator. Both operators possess a spectral decomposition with a positive operator valued measure $E_Q$ and $E_P$ in $\rm {Meas}(\mathbb R,\cB(\cH))$. Let us assume that the precision of the position and momentum measurement are given by intervals of length $\delta q$ and $\delta p$ for the entire range of the spectrum. As we will see, only the spacings $\delta q$ or $\delta p$ are relevant but not the explicit partition into intervals $\{I^q_k\}_{k=1}^\infty$ and $\{I^p_k\}_{k=1}^\infty$. The corresponding measurements are then formed by the operators $Q^k = E_Q(I_k)$ and $P^k = E_P(I_k)$. In the following, we denote the classical systems induced by a position and momentum measurement with precision $\delta p$ and $\delta q$ by $Q(\delta q)$ and $P(\delta p)$. As we know from Theorem~\ref{thm:minmaxcont}, the quantity which enters the entropic uncertainty relation is the overlap of the measurement operators,\footnote{Since we are interested in the case of finite-size intervals, the measure space in question is naturally equipped with the counting measure, and hence the denominator in~\eqref{eq:ComplConstCont} is equal to one.}
\begin{align}\label{eq:DefQPconstDisc}
c(\delta q,\delta p)=\sup_{k,l}\Vert\sqrt{Q^k}\sqrt{P^l}\Vert^{2}=\sup_{k,l}\Vert Q^k P^l Q^k\Vert\ .
\end{align}
This norm can be expressed by~\cite{Slepian64} (see also~\cite{Kiukas10} and references therein)
\begin{align}\label{eq:QPconstDisc}
c(\delta q,\delta p)=\frac{1}{2\pi}\cdot\delta q\delta p\cdot S_{0}^{(1)}\left(1,\frac{\delta q\delta p}{4}\right)^{2}\ ,
\end{align}
where $S_{0}^{(1)}(1,\cdot)$ denotes the 0th radial spheroidal wave function of the first kind. For $\delta q\delta p \rightarrow 0$, it follows that $S_{0}^{(1)}\left(1,\frac{\delta q\delta p}{4}\right)\rightarrow 1$, such that the behavior for small spacing is $c(\delta q,\delta p)\approx\frac{1}{2\pi}\cdot\delta q\delta p$.

\begin{corollary}\label{cor:pqfinite}
Let $\cM_{ABC}=\cB(L^{2}(\mathbb{R}))\ot\cM_{BC}$, $\omega_{ABC}\in\cS(\cM_{ABC})$, and consider position and momentum measurements with spacing $\delta q> 0$ and $\delta p> 0$ on the first system. Then, we have that
\begin{align}
H_{\max}(Q(\delta q)|B)_{\omega}+H_{\min}(P(\delta p) |C)_{\omega}\geq - \log c(\delta q,\delta p)\ ,\label{eq:minmax_pq_disc}
\end{align}
as well as
\begin{align}
H(Q(\delta q)|B)_{\omega}+H(P(\delta p)|C)_{\omega}\geq -\log c(\delta q,\delta p)\ ,\label{eq:neumann_pq_disc}
\end{align}
where $c(\delta q,\delta p)$ is given in~\eqref{eq:QPconstDisc}.
\end{corollary}

This corollary follows directly from the proofs of Theorem~\ref{thm:minmaxcont} and Theorem~\ref{prop:vN_tri_disc} about measurements on von Neumann algebras. Since the statement is invariant under exchanging $Q$ and $P$, the uncertainty relation in~\eqref{eq:minmax_pq_disc} also holds for the conditional min-entropy of $\omega_{Q(\delta q)B}$ and the conditional max-entropy of $\omega_{P(\delta p)C}$. Corollary~\ref{cor:pqfinite} generalizes known results for the Shannon entropy~\cite{Partovi83,Birula84,Rudnicki10,Rudnicki11,Rudnicki12} and for the R\'enyi entropy (for the order pair $\infty-1/2$, cf.~\eqref{eq:renyiinfty} and~\eqref{eq:renyi12})~\cite{Birula06,Rudnicki10,Rudnicki12} to the case of quantum side information.

By applying the formula for the complementary constant~\eqref{eq:ComplConstCont} to position and momentum measurements with partitions of similar fineness, we get
\begin{align}
\lim_{\delta\rightarrow 0} \frac{c(\delta,\delta)}{\delta^2} = \lim_{\delta\rightarrow 0} \frac1{2\pi}\cdot S_{0}^{(1)}\left(1,\frac{\delta^2}{4}\right)^{2}=\frac1{2\pi} \, ,
\end{align}
where we used~\eqref{eq:QPconstDisc}, and that $S_{0}^{(1)}\left(1,\frac{\delta^2}{4}\right)\rightarrow 1$ for $\delta\rightarrow 0$. Hence, we immediately obtain the following corollary.

\begin{corollary}\label{cor:pqinfinite}
Let $\cM_{ABC}=\cB(L^{2}(\mathbb{R}))\ot\cM_{BC}$, $\omega_{ABC}\in\cS(\cM_{ABC})$, and denote the post-measurement states obtained by continuous position and momentum measurements on the first system by $\omega_{QBC}$ and $\omega_{PBC}$. If there exists a finite spacing $\delta q$ such that $H_{\max}(Q(\delta q))_{\omega}<\infty$, then we have that
\begin{align}\label{eq:minmax_pq_cont}
h_{\max}(Q|B)_{\omega}+h_{\min}(P|C)_{\omega}\geq \log 2\pi\ .
\end{align}
If the post-measurement states $\omega_{QBC}$ and $\omega_{PBC}$ satisfy $h(Q|B)_{\omega}>-\infty$ and $h(P|C)_{\omega}>-\infty$, and if there exists a finite spacing $\delta q$ for which $H(Q(\delta q)|B)_{\omega}<\infty$ as well as $\delta p$ such that $H(P(\delta p)|C)_{\omega}<\infty$, then
\begin{align}\label{eq:PQcontURvN}
h(Q|B)_{\omega}+h(P|C)_{\omega} \geq \log 2\pi\ .
\end{align}
\end{corollary}

This generalizes known results for the differential Shannon entropy~\cite{Hirschman57,Beckner75,Birula75}, and for the differential R\'enyi entropy~\cite{Birula06} (for the order pair $\infty-1/2$, cf.~\eqref{eq:renyiinfty} and~\eqref{eq:renyi12}) to the case of quantum side information. The question of the tightness of Corollary~\ref{cor:pqfinite} and Corollary~\ref{cor:pqinfinite} is discussed in~\cite{Berta13_2}.


\subsection{Discussion}

There is one important open question we would like to mention. Is it possible to use our techniques to derive infinite-dimensional bipartite uncertainty relations with quantum side information in terms of the conditional von Neumann entropy (cf.~the related work by Frank and Lieb~\cite{Frank12})? In that respect, we already suggested in~\cite{Berta10} to use uncertainty relations with quantum side information for witnessing entanglement. It might be worthwhile to pursue this idea for continuous variables as well, and ideas in this direction have been developed~\cite{Schneeloch12,Ray13}.


\chapter{Randomness Extractors}\label{ch:extractors}

The starting point for the ideas developed in this chapter is a joint work with Fr\'ed\'eric Dupuis, Renato Renner, and J\"urg Wullschleger~\cite{Berta08,Dupuis10}. The introduction is partly taken from the collaboration~\cite{Berta11_5}. In this chapter, all systems (classical and quantum) are finite-dimensional, although we also make some comments about the stability of randomness extractors against infinite-dimensional quantum side information.

Randomness is an essential resource for information theory, cryptography, and computation~\cite{Vadhan11}. However, most sources of randomness exhibit only weak forms of unpredictability. The goal of randomness extraction is to convert such weak randomness into (almost) uniform random bits. Classically, a weakly random source simply outputs a string $N$ where the amount of randomness is measured in terms of the maximum probability of guessing the value of $N$ ahead of time. That is, it is measured in terms of the min-entropy $H_{\min}(N)=-\log P_{\guess}(N)$. To convert $N$ to perfect randomness, one applies a function Ext that takes $N$, together with a shorter string $D$ of perfect randomness (the seed) to an output string $(M,D)=\mathrm{Ext}(N,D)$. The (catalytic) use of a seed $D$ is thereby necessary to ensure that the extractor works for all sources $N$ about which we know only the min-entropy. Much work has been invested into showing that particular classes of functions have the property that $(M,D)$ is indeed very close to uniform as long as the min-entropy of the source $H_{\min}(N)$ is large enough (see the review articles~\cite{Shaltiel02,Vadhan07}). Yet, for most applications this is not quite enough, and we want an even stronger statement. In particular, imagine that we hold some side information $R$ about $N$ that increases our guessing probability to $P_{\guess}(N|R)$. For example, such side information could come from an earlier application of an extractor to the same source. Intuitively, one would not talk about randomness if, e.g., the output is uniformly distributed, but identical to an earlier output. In a cryptographic setting, side information can also be gathered by an adversary during the course of the protocol. We thus ask that the output is perfectly random even with respect to such side information, i.e., uniform and uncorrelated from $R$. Classically, it is known that extractors are indeed robust against classical side information~\cite{Koenig08}, yielding a uniform output $(M,D)$ whenever the conditional min-entropy $H_{\min}(N|R)=-\log P_{\guess}(N|R)$ is sufficiently high. However, since the underlying world is not classical, $R$ may in fact hold quantum side information about $N$~\cite{Koenig05,Renner05_2}. That this adds substantial difficulty to the problem was emphasized in~\cite{Gavinsky07}, where it was shown that there are situations where using the same extractor gives a uniform output $M$ if $R$ is classical, but becomes more predictable when $R$ is quantum. Positive results were obtained in~\cite{Renner05,Koenig08,Koenig11,Tomamichel11,Szehr11,De12} proving that a wide class of classical extractors are stable against quantum side information. However, we emphasize that a general understanding of stability against quantum side information is still lacking.

In Section~\ref{se:cc}, we introduce and briefly discuss classical randomness extractors. As a result, we find that a special class of extractors, called R\'enyi 2-extractors, are always stable against quantum side information. Apart from that, Section~\ref{se:cc} mainly serves as a preparation for the following sections (Section~\ref{se:qq}-\ref{se:qc}), where we discuss quantum generalizations of classical extractors.

Namely, instead of just analyzing classical extractors in the presence of quantum side information, the concept of extractors can itself be quantized. That is, we start with a quantum state $\rho_{N}$ instead of a classical input, and ask for a maximally mixed quantum state $\frac{\id_{M}}{|M|}$ instead of a uniform distribution as the output. Like in the classical case the strength of a source is measured by the min-entropy $H_{\min}(N)=-\log\lambda_{1}(\rho_{N})$, where $\lambda_{1}(\cdot)$ denotes the largest eigenvalue, and the quantum extractor should output a maximally mixed state as long as the min-entropy of the source is large enough. In addition, we might again demand that the output is not only quantumly fully random, but also uncorrelated from some side information $R$ that the initial state $\rho_{NR}$ was correlated to. Since $N$ and $R$ are both quantum, $\rho_{NR}$ can be entangled, and the strength of the source is then measured by the fully quantum conditional min-entropy $H_{\min}(N|R)=-\log\big(|N|\cdot F(N|R)\big)$ with
\begin{align}
F(N|R)=\max_{\Lambda_{R\ra N'}}F(\Phi_{NN'},\id_{N}\ot\Lambda_{R\ra N'}(\rho_{NR}))\ ,
\end{align}
where $N'$ is a copy of $N$, $\Phi_{NN'}$ the maximally entangled state, and the maximization is over all quantum channels $\Lambda_{R\ra N'}$. The fully quantum conditional min-entropy is a measure for the entanglement between $N$ and $R$, and can be negative for entangled states.

Given this quantum generalization of the (conditional) min-entropy there are different possible notions of quantum extractors, basically because it is not explicit what type of functions we want to consider. We choose to use a definition where the seed system $D$ is still classical, and the resulting quantum extractors (with quantum side information) are then well known in quantum information theory as a consequence of a notion known as decoupling. Decoupling theorems play a central role in quantum coding theory, see~\cite{Dupuis09,Dupuis10} and references therein. In addition, other special cases of quantum extractors include quantum expanders (see~\cite{Ben-Aroya10,Harrow09_2} and references therein), and quantum state randomization (see~\cite{Hayden04,Ambainis04,Dickinson06,Aubrun09} and references therein).

In Section~\ref{se:qq}, we introduce quantum extractors with and without quantum side information, study their properties, and give probabilistic as well as explicit constructions. As mentioned earlier, we find that a special class of quantum extractors, called quantum R\'enyi 2-extractors, are always stable against quantum side information. However, in contrast to the classical case, quantum extractors and their stability in the presence of quantum side information are rather poorly understood. We then discuss a number of open questions.

Finally, we also discuss quantum-classical extractors, which correspond to an intermediate scenario between classical and quantum extractors. We start with a bipartite quantum state $\rho_{NR}$, but then only ask for the creation of perfect classical randomness with respect to the quantum side information $R$. Again, the extractor should work as long as the conditional min-entropy of the source is large enough. The setup is quite specific, but well suited for cryptographic applications where it is often sufficient to extract random classical bits.

In Section~\ref{se:qc}, we define quantum-classical extractors with and without quantum side information, and make similar observations as for classical and quantum extractors. As our main result, we find that the R\'enyi 2-entropic uncertainty relations discussed in Section~\ref{se:several}, allow to lift classical R\'enyi 2-extractors to quantum-classical extractors against quantum side information. We briefly discuss an application to quantum cryptography.

We mention that we will also use classical and quantum extractors with classical and quantum side information in Chapter~\ref{ch:channels} about channel simulations.


\section{Classical to Classical}\label{se:cc}

The results in this section have been obtained in collaboration with Volkher Scholz, and Oleg Szehr, but are currently unpublished. The overview part of this section is inspired by~\cite{Vadhan11,Koenig08,De12}.


\subsection{Min-Entropy Extractors}

Instead of denoting classical systems by $W,X,Y,Z,K$ like in the rest of this work, we use $N$ to denote the classical input system, $M$ to denote the classical output systems, and $D$ to denote the classical seed system. This is in accordance with the literature on classical extractors, see, e.g., \cite{Shaltiel02}. We quickly recapitulate our notation for classical systems as discussed in Section~\ref{sec:classical_systems}. The labels $N,M,D$ are used to specify the subsystem as well as the domain of the classical system, and states on $N$, i.e., probability distributions, are denoted by $P_{N}\in\ell^1(N)$. Furthermore, we denote the set of non-negative distributions on $N$ by $\ell^{+}(N)$, and the uniform distribution on $N$ is denoted by $\upsilon_{N}$. We note that this notation is suitable to be generalized to the quantum setting (Sections~\ref{se:qq} and~\ref{se:qc}).

The definition of a strong (classical) min-entropy extractor is due to Nisan and Zuckerman~\cite{Nissan96}.

\begin{definition}[Strong min-entropy extractor]\label{def:cmin_ext}
Let $M\subset N$ be classical systems, $k\in[0,\log|N|]$, and $\eps>0$. A strong $(k,\eps)$ min-entropy extractor is a set of functions $\{f_{1},\ldots,f_{|D|}\}$ from $N$ to $M$ such that for all $P_{N}\in\ell^1(N)$ with $H_{\min}(N)_{P}\geq k$,
\begin{align}\label{eq:minextractor_class}
\big\|\frac{1}{|D|}\cdot\sum_{i=1}^{|D|}P_{f_{i}(N)}\ot\proj{i}_{D}-\upsilon_{N}\ot\upsilon_{D}\big\|_{1}\leq\eps\ .
\end{align}
The quantity $n=\log|N|$ is called the input size, $m=\log|M|$ the output size, and $d=\log|D|$ the seed size.
\end{definition}

We note that~\eqref{eq:minextractor_class} is equivalent to
\begin{align}\label{eq:minextractor_class2}
\frac{1}{|D|}\cdot\sum_{i=1}^{|D|}\big\|P_{f_{i}(N)}-\upsilon_{N}\big\|_{1}\leq\eps\ .
\end{align}
For defining a weak $(k,\eps)$ min-entropy extractor we just replace~\eqref{eq:minextractor_class} with
\begin{align}\label{eq:weakmin}
\big\|\frac{1}{|D|}\cdot\sum_{i=1}^{|D|}P_{f_{i}(N)}-\upsilon_{N}\big\|_{1}\leq\eps\ .
\end{align}
For a weak extractor the randomness from the seed system $D$ is lost, and the criteria~\eqref{eq:weakmin} are strictly weaker than~\eqref{eq:minextractor_class}/\eqref{eq:minextractor_class2}. An extractor is called permutation based if all the functions $f_{i}:N\ra M$ have the form $f_{i}(\cdot)=\pi_{i}(\cdot)\big|_{M}$ with $\pi_{i}\in S_{|N|}$, the symmetric group on $\{1,2,\ldots,|N|\}$. Permutation based extractors will be of special interest to us concerning quantum generalizations.

It is instructive to consider extractors with domain and range consisting of bit strings, that is, $N=\{0,1\}^{n}$, $M=\{0,1\}^{m}$, $D=\{0,1\}^{d}$ (but the general case is straightforward). We could now optimize the five different parameters $(n,k,m,d,\eps)$. However, typically we are given a fixed $n$, $k$, and $\eps$, and we want to maximize the output length $m$ and minimize the seed length $d$. The following bound by Radhakrishnan and Ta-Shma gives an ultimate limit on $m$ and $d$.

\begin{proposition}\cite{Radhakrishnan00}\label{prop:extractor_converse}
Every strong $(k,\eps)$ min-entropy extractor necessarily has $m\leq k-2\log(1/\eps)+O(1)$, and $d\geq\log(n-k)+2\log(1/\eps)-O(1)$.
\end{proposition}

It turns out that a probabilistic construction using random functions achieves these bounds up to constants.

\begin{proposition}\cite{Sipser88,Radhakrishnan00}\label{prop:random_optimal}
Let $n\in\nN$, $k\in[0,n]$, and $\eps>0$. Then, there exists a strong $(k,\eps)$ min-entropy extractor with $m=k-2\log(1/\eps)-O(1)$, and $d=\log(n-k)+2\log(1/\eps)+O(1)$.
\end{proposition}

Probabilistic constructions are interesting, but for applications we usually need explicit extractors. Starting with Trevisan's breakthrough result~\cite{Trevisan99} there has been a lot of progress in this direction, and there are now many constructions that almost achieve the converse bounds in Proposition~\ref{prop:extractor_converse} (see the review articles~\cite{Shaltiel02,Vadhan07}). Here, we are interested in stability under side information.

\begin{definition}[Strong min-entropy extractor against quantum side information]
For the same premises as in Definition~\ref{def:cmin_ext}, a set of functions $\{f_{1},\ldots,f_{|D|}\}$ from $N$ to $M$ is a strong $(k,\eps)$ min-entropy extractor against quantum side information if for all classical-quantum states $\rho_{NR}\in\cS(\cH_{NR})$ with $H_{\min}(N|R)_{\rho}\geq k$,
\begin{align}
\big\|\frac{1}{|D|}\cdot\sum_{i=1}^{|D|}\rho_{f_{i}(N)R}\ot\proj{i}_{D}-\upsilon_{M}\ot\rho_{R}\ot\upsilon_{D}\big\|_{1}\leq\eps\ .
\end{align}
\end{definition}

Classical side information corresponds to restricting $R$ to be classical with respect to some basis $\{\ket{e}\}_{e\in R}$. The stability against classical side information is immediate, basically by just conditioning on the values of the classical side information.

\begin{proposition}\cite[Proposition 1]{Koenig08}
Every $(k,\eps)$ strong min-entropy extractor is also a $(k+\log(1/\eps),2\eps)$ strong min-entropy extractor against classical side information (of the same output size and the same seed size).
\end{proposition}

But what about quantum side information? It was shown by K\"onig and Terhal that strong one bit output extractors are always stable against quantum side information.

\begin{proposition}\cite[Theorem 1]{Koenig08}
Every $(k,\eps)$ strong min-entropy extractor with one bit output is also a $(k+\log(1/\eps),3\sqrt{\eps})$ strong min-entropy extractor against quantum side information (with one bit output and the same seed size).
\end{proposition}

In general it is known by now that many extractor constructions are completely stable against quantum side information~\cite{Renner05,Koenig08,Koenig11,Tomamichel11,Szehr11} or suffer at most from a decent parameter loss~\cite{De12}. However, Gavinsky {\it et al.}~\cite{Gavinsky07} showed that not all extractors are stable. Moreover, since all known stability results are specifically tailored proofs, there is no general understanding of when a min-entropy extractor is stable against quantum side information. Here, we show that min-entropy extractors based on so-called R\'enyi 2-entropy extractors, are always stable against quantum side information (Section~\ref{sec:class_renyi2}).


\subsection{R\'enyi 2-Extractors}\label{sec:class_renyi2}

\begin{definition}[Strong R\'enyi 2-extractor]\label{def:class_renyi2}
Let $M\subset N$ be classical systems, $k\in[0,\log|N|]$, and $\eps>0$. A strong $(k,\eps)$ R\'enyi 2-extractor is a set of functions $\{f_{1},\ldots,f_{|D|}\}$ from $N$ to $M$ such that for all $P_{N}\in\ell^{+}(N)$ with $H_{2}(N)_{P}\geq k$,\footnote{Strong R\'enyi 2-extractors are often defined with the condition $H_{2}(MD)_{\sigma}\geq\log\big(|M|\cdot|D|\big)-\eps$ (see, e.g., \cite{Vadhan11}), but for small $\eps$ this is approximately the same.}
\begin{align}\label{eq:class_renyi2}
H_{2}(MD)_{Q}\geq\log\left(\frac{|M|\cdot|D|}{\eps+1}\right)\ ,
\end{align}
where
\begin{align}
Q_{MD}=\frac{1}{|D|}\cdot\sum_{i=1}^{|D|}P_{f_{i}(N)}\ot\proj{i}_{D}\ .
\end{align}
\end{definition}

We have the following alternative characterization of R\'enyi 2-extractors.

\begin{proposition}\label{prop:h2_alternative}
A set of functions $\{f_{1},\ldots,f_{|D|}\}$ from $N$ to $M$ is a strong $(k,\eps)$ R\'enyi 2-extractor with seed size $\log|D|$ if and only if we have for the map
\begin{align}
\psi(P_{N})=\frac{1}{|D|}\cdot\sum_{i=1}^{|D|}P_{f_{i}(N)}\ot\proj{i}_{D}
\end{align}
that
\begin{align}
\lambda_{1}\big(\psi^{\dagger}\circ\psi-\tau^{\dagger}\circ\tau\big)\leq2^{k}\cdot\frac{\eps}{|M|\cdot|D|}\ ,
\end{align}
where $\tau(P_{N})=\big(\sum_{j=1}^{|N|}P_{N}^{j}\big)\cdot\big(\upsilon_{M}\ot\upsilon_{D}\big)$. For typical applications, the uniform distribution $\upsilon_{N}$ is the only eigenvector of $\psi^{\dagger}\circ\psi$ with eigenvalue one, and thus the relevant quantity is the second largest eigenvalue $\lambda_{2}(\psi^{\dagger}\circ\psi)$.
\end{proposition}

\begin{proof}
We rewrite~\eqref{eq:class_renyi2} as
\begin{align}
2^{-H_{2}(MD)_{Q}}-\frac{1}{|M|\cdot|D|}\leq\frac{\eps}{|M|\cdot|D|}\ .
\end{align}
If we understand the extractor as a map $\psi:\ell^{+}(N)\ra\ell^{+}(MD)$ with
\begin{align}
\psi(P_{N})=\frac{1}{|D|}\cdot\sum_{i=1}^{|D|}P_{f_{i}(N)}\ot\proj{i}_{D}\ ,
\end{align}
and denote $\tau:\ell^{+}(N)\ra\ell^{+}(MD)$ with $\tau(P_{N})=\big(\sum_{j=1}^{|N|}P_{N}^{j}\big)\cdot\big(\upsilon_{M}\ot\upsilon_{D}\big)$, then
\begin{align}\label{eq:adjoint_map}
\sup_{\substack{P\in\ell^{+}(N)\\\braket{P|P}\leq2^{-k}}}2^{-H_{2}(MD)_{Q}}-\frac{1}{|M|\cdot|D|}&=\sup_{\substack{P\in\ell^{+}(N)\\\braket{P|P}\leq2^{-k}}}\braket{P|(\psi^{\dagger}\circ\psi-\tau^{\dagger}\circ\tau)(P)}\notag\\
&=\frac{1}{2^{k}}\cdot\sup_{\substack{P\in\ell^{+}(N)\\\|P\|_{2}=1}}\braket{P|(\psi^{\dagger}\circ\psi-\tau^{\dagger}\circ\tau)(P)}\notag\\
&=\frac{1}{2^{k}}\cdot\lambda_{1}\big(\psi^{\dagger}\circ\psi-\tau^{\dagger}\circ\tau\big)\ ,
\end{align}
where we have used in the last step that all entries of the matrix $(\psi^{\dagger}\circ\psi-\tau^{\dagger}\circ\tau)$ are non-negative, and applied the Perron–Frobenius theorem (Lemma~\ref{lem:perron_frobenius}).
\end{proof}

Examples for R\'enyi 2-extractors are:\footnote{We refer to~\cite{Vadhan11} for a more extensive discussion and references.}
\begin{itemize}
\item Pairwise independent or 2-universal families of hash functions, as well as almost 2-universal families of hash functions (these are strong extractors).
\item Pairwise independent families of permutations, as well as almost pairwise independent families of permutations (these strong extractors).
\item Constructions based on balanced expander graphs (these are weak extractors).
\end{itemize}

A straightforward calculation shows that we can also state~\eqref{eq:class_renyi2} in terms of the 2-norm as
\begin{align}
\frac{1}{|D|}\cdot\sum_{i=1}^{|D|}\|P_{f_{i}(N)}-\upsilon_{M}\|_{2}^{2}\leq\frac{\eps}{|M|}\ .
\end{align}
Since the R\'enyi 2-entropy is lower bounded by the min-entropy (Lemma~\ref{lem:h2vN_bounds}), and $\|X\|_{1}\leq \sqrt{\rank(X)}\cdot\|X\|_{2}$ (Lemma~\ref{lem:12norm}), it follows that every strong $(k,\eps)$ R\'enyi 2-entropy extractor is also a strong $(k,\sqrt{\eps})$ min-entropy extractor (of the same output size and the same seed size). Hence, the examples mentioned above are also strong min-entropy extractors. They typically give an output size close to optimal (cf.~Proposition~\ref{prop:extractor_converse}), but have the drawback of a long seed size.

\begin{proposition}\label{prop:h2_longseed}
Every strong $(k,\eps)$ R\'enyi 2-extractor with input size $n$, output size $m$, and seed size $d$ necessarily has
\begin{align}
d\geq\min\{n-k,\frac{m}{2}\}+\log\frac{1}{\eps}-1\ .
\end{align} 
\end{proposition}

For a proof see, e.g., \cite{Vadhan11}. We will explicitly discuss a similar proof for the quantum case in Section~\ref{se:qq}. The stability of R\'enyi 2-extractors against quantum side information is defined as follows.

\begin{definition}[Strong R\'enyi 2-extractor against quantum side information]\label{def:renyi_qsi}
For the same premises as in Definition~\ref{def:class_renyi2}, a set of functions $\{f_{1},\ldots,f_{|D|}\}$ from $N$ to $M$ is a strong $(k,\eps)$ R\'enyi 2-extractor against quantum side information if for all classical-quantum states $\rho_{NR}\in\cP^{+}(\cH_{NR})$ with $H_{2}(N|R)_{\rho}\geq k$,
\begin{align}\label{eq:renyi_qsi}
H_{2}(MD|R)_{\sigma}\geq\log\left(\frac{|M|\cdot|D|}{\eps+1}\right)\ ,
\end{align}
where
\begin{align}
\sigma_{MDR}=\frac{1}{|D|}\cdot\sum_{i=1}^{|D|}\rho_{f_{i}(N)R}\ot\proj{i}_{D}\ .
\end{align}
\end{definition}

By the same ideas as in Proposition~\ref{prop:h2_alternative}, the stability against quantum side information is as follows.

\begin{theorem}\label{thm:h2c_stable}
Every strong $(k,\eps)$ R\'enyi 2-extractor is stable against quantum side information.
\end{theorem}

\begin{proof}
We rewrite~\eqref{eq:renyi_qsi} as
\begin{align}\label{eq:rewrite_h2qsi}
2^{-H_{2}(MD|R)_{Q}}-\frac{1}{|M|\cdot|D|}\leq\frac{\eps}{|M|\cdot|D|}\ .
\end{align}
Similarly as in the case without side information, we can understand the extractor as a map $\psi:\cP^{+}(\cH_{N})\ra\cP^{+}(\cH_{MD})$ with
\begin{align}\label{eq:qh2_extractor}
(\psi\ot\cI_{R})(\rho_{NR})=\frac{1}{|D|}\cdot\sum_{i=1}^{|D|}\rho_{f_{i}(N)R}\ot\proj{i}_{D}\ ,
\end{align}
and denote $\tau:\cP^{+}(\cH_{N})\ra\cP^{+}(\cH_{MD})$ with $\tau(\rho_{N})=\trace[\rho_{N}]\cdot\frac{\id_{M}}{|M|}\ot\frac{\id_{D}}{|D|}$. Furthermore, we define
\begin{align}\label{eq:rhoh}
\rhoh_{NR}=(\id_{N}\ot\rho_{R}^{-1/2})\rho_{NR}(\id_{N}\ot\rho_{R}^{-1/2})\ ,
\end{align}
and note that
\begin{align}\label{eq:h2braket}
2^{-H_{2}(MD|R)_{\sigma}}=\braket{(\psi\ot\cI_{R})(\rhoh_{NR})|(\psi\ot\cI_{R})(\rhoh_{NR})}_{(\id\ot\rho)}\ ,
\end{align}
with the $(\id_{MD}\ot\rho_{R})$-weighted Hilbert-Schmidt inner product $\braket{\cdot|\cdot}_{(\id\ot\rho)}$ as defined in~\eqref{eq:sigma_ip}. For the left-hand side of~\eqref{eq:rewrite_h2qsi} we get
\begin{align}
&\sup_{\braket{\rhoh|\rhoh}_{(\id\ot\rho)}\leq2^{-k}}2^{-H_{2}(MD|R)_{\sigma}}-\frac{1}{|M|\cdot|D|}\notag\\
&=\sup_{\braket{\rhoh|\rhoh}_{(\id\ot\rho)}\leq2^{-k}}\braket{\rhoh|((\psi^{\dagger}\circ\psi-\tau^{\dagger}\circ\tau)\ot\cI)(\rhoh)}_{(\id\ot\rho)}\notag\\
&=\frac{1}{2^{k}}\cdot\sup_{\|\rhoh\|_{2,(\id\ot\rho)}=1}\braket{\rhoh|((\psi^{\dagger}\circ\psi-\tau^{\dagger}\circ\tau)\ot\cI)(\rhoh)}_{(\id\ot\rho)}\ ,
\end{align}
where $\rhoh_{NR}$ is as in~\eqref{eq:rhoh} with classical-quantum $\rho_{NR}\in\cP^{+}(\cH_{NR})$, and the adjoints $\psi^{\dagger},\tau^{\dagger}$ are the same as in~\eqref{eq:adjoint_map}, because $\psi,\tau$ only act on $MD$ and the Hilbert-Schmidt inner product is weighted trivially on $MD$. Then, by the same arguments as in~\eqref{eq:adjoint_map}, and since $\psi,\tau$ only act on $MD$ and the 2-norm is weighted trivially on $MD$, we get
\begin{align}
\sup_{\|\rhoh\|_{2,(\id\ot\rho)}=1}\braket{\rhoh|((\psi^{\dagger}\circ\psi-\tau^{\dagger}\circ\tau)\ot\cI)(\rhoh)}_{(\id\ot\rho)}&=\lambda_{1}^{(\id\ot\rho)}\big((\psi^{\dagger}\circ\psi-\tau^{\dagger}\circ\tau)\ot\cI\big)\notag\\
&=\lambda_{1}\big(\psi^{\dagger}\circ\psi-\tau^{\dagger}\circ\tau\big)\ ,
\end{align}
where $\lambda_{1}^{(\id\ot\rho)}(\cdot)$ denotes the largest eigenvalue in $(\id\ot\rho)$-weighted 2-norm, and $\lambda_{1}(\cdot)$ denotes the largest eigenvalue in 2-norm. Together with~\eqref{eq:rewrite_h2qsi}, this is exactly the same as in Proposition~\ref{prop:h2_alternative}, and hence the claim follows.
\end{proof}

Next, we show that R\'enyi 2-extractors also give rise to min-entropy extractors against quantum side information. For this we use a similar calculation as Renner {\it et al.}, who showed (directly) that families of almost 2-universal hash functions~\cite{Renner05,Tomamichel11}, and families of almost pairwise independent permutations~\cite{Szehr11} give rise to strong quantum min-entropy extractors against quantum side information.

\begin{theorem}\label{thm:crenyi2_stability}
Every strong $(k,\eps)$ R\'enyi 2-extractor is also a strong $(k,2\sqrt{\eps})$ min-entropy extractor against quantum side information (of the same output size and the same seed size).
\end{theorem}

\begin{proof}
We think of the extractor as in~\eqref{eq:qh2_extractor} and write for the 1-norm (Lemma~\ref{lem:1norm})
\begin{align}\label{eq:h2_1}
\|((\psi-\tau)&\ot\cI_{R})(\rho_{NR})\|_{1}=2\cdot\max_{0\leq X\leq\id}\trace\Big[\big((\psi-\tau)\ot\cI_{R}\big)(\rho_{NR})X\Big]\notag\\
&=2\cdot\max_{0\leq X\leq\id}\trace\Big[\big((\psi-\tau)\ot\cI_{R}\big)(\rhoh_{NR})(\id_{MD}\ot\rho_{R}^{1/2})X(\id_{MD}\ot\rho_{R}^{1/2})\Big]\notag\\
&=2\cdot\max_{0\leq X\leq\id}\braket{((\psi-\tau)\ot\cI_{R})(\rhoh_{NR})|X}_{(\id\ot\rho)}\ ,
\end{align}
where $\rhoh_{NR}$ is as in~\eqref{eq:rhoh}, and we made use of the $(\id_{MD}\ot\rho_{R})$-weighted Hilbert-Schmidt inner product. By using the $(2,2)$-H\"older inequality for the $(\id_{MD}\ot\rho_{R})$-weighted Hilbert-Schmidt inner product (Lemma~\ref{lem:hoelder}) we get
\begin{align}\label{eq:h2_2}
2\cdot\max_{0\leq X\leq\id}&\braket{((\psi-\tau)\ot\cI_{R})(\rhoh_{NR})|X}_{(\id\ot\rho)}\notag\\
&\leq2\cdot\max_{0\leq X\leq\id}\|X\|_{2,(\id\ot\rho)}\cdot\|((\psi-\tau)\ot\cI_{R})(\rhoh_{NR})\|_{2,(\id\ot\rho)}\ ,
\end{align}
with the $(\id_{MD}\ot\rho_{R})$-weighted 2-norm. We estimate the first term by using the $(1,\infty)$-H\"older inequality for the $(\id_{MD}\ot\rho_{R})$-weighted Hilbert-Schmidt inner product (Lemma~\ref{lem:hoelder})
\begin{align}\label{eq:h2_3}
\|X\|_{2,(\id\ot\rho)}&=\sqrt{\braket{X|X}_{2,(\id\ot\rho)}}\leq\sqrt{\|X\|_{\infty,(\id\ot\rho)}\cdot\|X\|_{1,(\id\ot\rho)}}\notag\\
&=\sqrt{\lambda_{1}(X)\cdot\trace[(\id_{MD}\ot\rho_{R})X]}\leq\sqrt{|M|\cdot|D|}\ .
\end{align}
For the second term a straightforward calculation together with~\eqref{eq:h2braket} shows that
\begin{align}\label{eq:h2_4}
\|((\psi-\tau)\ot\cI_{R})(\rhoh_{NR})\|_{2,(\id\ot\rho)}^{2}&=\braket{(\psi\ot\cI_{R})(\rhoh_{NR})|(\psi\ot\cI_{R})(\rhoh_{NR})}_{2,(\id\ot\rho)}\notag\\
&-2\cdot\braket{(\psi\ot\cI_{R})(\rhoh_{NR})|(\tau\ot\cI_{R})(\rhoh_{NR})}_{2,(\id\ot\rho)}\notag\\
&+\braket{(\tau\ot\cI_{R})(\rhoh_{NR})|(\tau\ot\cI_{R})(\rhoh_{NR})}_{2,(\id\ot\rho)}\notag\\
&=2^{-H_{2}(MD|R)_{\sigma}}-\frac{1}{|M|\cdot|D|}\ ,
\end{align}
with $\sigma_{MDR}=\frac{1}{|D|}\cdot\sum_{i=1}^{|D|}\rho_{f_{i}(N)R}\ot\proj{i}_{D}$. But by assumption $\{f_{1},\dots,f_{N}\}$ is a strong $(k,\eps)$ R\'enyi 2-entropy extractor, and with that also a strong $(k,\eps)$ R\'enyi 2-entropy extractor against quantum side information (Theorem~\ref{thm:crenyi2_stability}). Putting~\eqref{eq:h2_1}-\eqref{eq:h2_4} together and noting that the conditional R\'enyi 2-entropy is lower bounded by the conditional min-entropy (Lemma~\ref{lem:h2vN_bounds}), the claim follows.
\end{proof}

As an example we mention pairwise independent permutations. A family $\cP$ of permutations of a set $X$ is pairwise independent if for all $x_1 \neq x_2$ and $y_1 \neq y_2$, and if $\pi$ is uniformly distributed over $\cP$,
\begin{align}
\mathrm{Pr}\big\{\pi(x_1)= y_1,\pi(x_2)=y_2\big\}=\frac{1}{|X|(|X|-1)}\ .
\end{align}
If $X$ has a field structure, i.e., if $|X|$ is a prime power, it is simple to see that the family $\cP=\{x\mapsto a\cdot x+ b : a \in X^*, b\in X\}$ is pairwise independent.

\begin{proposition}\label{prop:ccperm}
For a family of pairwise independent permutations $\cP$ on the input $N$, we have for $M\subset N$, $P_{N}\in\ell^{+}(N)$,
\begin{align}
\sigma_{MD}=\frac{1}{|\cP|}\cdot\sum_{\pi_{i}\in\cP}P_{\pi_{i}(N)|_{M}}\ot\proj{i}_{D}
\end{align}
that
\begin{align}
H_{2}(MD)_{\sigma}=-\log\left(\frac{1}{|\cP|}\cdot\left(\frac{|N|-|N|/|M|}{|N|-1}\cdot2^{-H_{2}(N)_{P}}+\frac{1}{|M|}\right)\right)\ .
\end{align}
Hence, a family of pairwise independent permutations is in particular a strong $(k,\eps)$ R\'enyi 2-extractor with output size
\begin{align}
\log|M|=m=k-\log\frac{1}{\eps}\ .
\end{align}
By Theorem~\ref{thm:crenyi2_stability}, it is then also a strong $(k,2\sqrt{\eps})$-quantum min-entropy extractor against quantum side information with the same output size (which is optimal up to constants, cf.~Proposition~\ref{prop:extractor_converse}).
\end{proposition}

For a proof see, e.g., \cite[Section 5.2]{Szehr11}. We will explicitly prove a similar statement for the quantum case in Section~\ref{se:qq}.


\subsection{Discussion}

Whereas we only analyze stability against finite-dimensional quantum side information, our initial motivation was to show stability against quantum side information modeled by von Neumann algebras. We believe that the approach presented here also works in the case of von Neumann algebras, and investigations in this direction are in progress. We note that it is important to analyze infinite-dimensional side information for applications in continuous variable quantum information theory (see, e.g., \cite{Furrer12,Furrer12_2}), and that the stability of 2-universal hashing against infinite-dimensional side information has been shown in~\cite{Furrer09,Berta11_4,Furrer12_2}

Since R\'enyi 2-entropy extractors have a long seed (Proposition~\ref{prop:h2_longseed}), we would also like to understand the stability of short seed probabilistic (Proposition~\ref{prop:random_optimal}) or Trevisan based~\cite{Trevisan99} constructions. In fact, the stability of Trevisan based constructions against quantum side information (up to a decent loss in parameters) was recently shown in~\cite{De12}, but a general theory is still missing.  In that respect, we would like to mention the connection of min-entropy extractors to other pseudorandom objects like, e.g., list-decoding and soft-decoding, or expander graphs (see the excellent review article~\cite{Vadhan07,Vadhan11} about a unified theory of pseudorandomness). It is a priori not clear what quantum side information means in this framework, but we believe that the aforementioned connections are worth to be explored. A few ideas in this direction can be found in~Section~\ref{se:extmore}.

Finally, we mention that we will use strong min-entropy extractors against classical and quantum side information in the next chapter about channel simulations (Chapter~\ref{ch:channels}).


\section{Quantum to Quantum}\label{se:qq}

The results in this section have been obtained in collaboration with Omar Fawzi, Volkher Scholz, and Oleg Szehr, but are currently unpublished.


\subsection{Min-Entropy Extractors}\label{sec:qmin-extractors}

To understand our definition of quantum extractors, it is convenient to start with permutation based classical extractors, i.e., a family of permutations acting on the input. This family of permutations should satisfy the following property: for any probability distribution on input bit strings with high min-entropy, applying a typical permutation from the family to the input induces an almost uniform probability distribution on a prefix of the output. We define a quantum to quantum extractor in a similar way by allowing the operations performed to be general unitary transformations and the input to the extractor to be quantum.

\begin{definition}[Strong quantum min-entropy extractor]\label{def:stron_qmin}
Let $\cH_{M}\subset\cH_{N}$, $k\in[0,\log|N|]$, and $\eps>0$. A strong $(k,\eps)$-quantum min-entropy extractor is a set of unitaries $\{U_{N}^{1},\ldots,U_{N}^{|D|}\}$ such that for all $\rho_{N}\in\cS(\cH_{N})$ with $H_{\min}(N)_{\rho}\geq k$,
\begin{align}\label{eq:qqminstrong}
\Big\|\frac{1}{|D|}\cdot\sum_{i=1}^{|D|}\trace_{N\backslash M}\big[U_{N}^{i}\rho_{N}(U_{N}^{i})^{\dagger}\big]\ot\proj{i}_{D}-\frac{\id_{M}}{|M|}\ot\frac{\id_{D}}{|D|}\Big\|_{1}\leq\eps\ .
\end{align}
The quantity $n=\log|N|$ is called the input size, $m=\log|M|$ the output size, and $d=\log|D|$ the seed size.
\end{definition}

The condition~\eqref{eq:qqminstrong} is equivalent to
\begin{align}
\frac{1}{|D|}\cdot\sum_{i=1}^{|D|}\Big\|\trace_{N\backslash M}\big[U_{N}^{i}\rho_{N}(U_{N}^{i})^{\dagger}\big]-\frac{\id_{M}}{|M|}\Big\|_{1}\leq\eps\ .
\end{align}
For defining a weak $(k,\eps)$-quantum min-entropy extractor we just replace~\eqref{eq:qqminstrong} with
\begin{align}
\Big\|\frac{1}{|D|}\cdot\sum_{i=1}^{|D|}\trace_{N\backslash M}\big[U_{N}^{i}\rho_{N}(U_{N}^{i})^{\dagger}\big]-\frac{\id_{M}}{|M|}\Big\|_{1}\leq\eps\ .
\end{align}
We note that the seed $D$ is still classical in this definition. Alternatively, we could also define quantum extractors as general quantum channels from $\cS(\cH_{N})$ to $\cS(\cH_{M})$, and the number of Kraus operators would correspond to the dimension of the quantum seed $D$. For example, the fully depolarizing channel would then correspond to a perfect extractor, independent of the min-entropy of the input. But since the minimal number of Kraus operators of the fully depolarizing channel is equal to the square of the output dimension $|M|$, it would also have quantum seed size $d=2m$. We believe that such a definition for quantum extractors is interesting, but here we restrict ourselves to quantum extractors with classical seed (as in Definition~\ref{def:stron_qmin}).

It is instructive to consider extractors with domain and range consisting of qubit strings, i.e., $\cH_{N}=(\nC^{2})^{\ot n}$ and $\cH_{M}=(\nC^{2})^{\ot m}$, as well as with a binary seed, i.e., $D=\{0,1\}^{d}$. We have the following lower bound on the seed size.

\begin{proposition}\label{prop:coneps}
Every strong $(k,\eps)$-quantum min-entropy extractor with $k\leq n-1$ ($n$ is the output size) necessarily has seed size $d\geq\log(1/\eps)$.
\end{proposition}

\begin{proof}
Let $\cH_{S}\subset\cH_{M}$ with $|S|=|M|/2$, let $\{\ket{t}\}_{t=1}^{|N|/|M|}$ be an orthonormal basis of $\cH_{N/M}$, and consider the state
\begin{align}
\gamma_{N}=\frac{2}{|M|}\cdot\sum_{s\in S}\sum_{t=1}^{|N|/|M|}(U_{N}^{1})^{\dagger}\proj{st}_{N}U_{N}^{1}\ .
\end{align}
Since $H_{\min}(N)_{\sigma}=n-1$, and
\begin{align}
\Big\|\trace_{N/M}\big[U_{N}^{1}\gamma_{N}(U_{N}^{1})^{\dagger}]-\frac{\id_{M}}{|M|}\Big\|_{1}=\Big\|\frac{2}{|M|}\cdot\sum_{s\in S}\proj{s}_{M}-\frac{\id_{M}}{|M|}\Big\|_{1}=1\ ,
\end{align}
the claim follows.
\end{proof}

Note that this lower bound is much weaker than what we get in the classical case, where $d\geq\log(n-k)+2\log(1/\eps)-O(1)$ (Proposition~\ref{prop:extractor_converse}). However, as we will see in Section~\ref{sec:qq_disc} the lower bound in Proposition~\ref{prop:coneps} is nearly tight. Examples for quantum min-entropy extractors in the literature include the following:
\begin{itemize}
\item In~\cite{Dupuis10,Dupuis09,Berta08} we studied so-called decoupling theorems, and in particular we showed that unitary 2-designs (Definition~\ref{def:unitary2design}) are strong quantum min-entropy extractors (Proposition~\ref{prop:unitary2design}). An extension to almost unitary 2-designs is discussed in~\cite{Szehr11_2,Szehr11}.
\item Ben-Aroya {\it et al.}~considered weak quantum min-entropy extractors with the input size equal to the output size~\cite[Definition 5.1]{Ben-Aroya10}, and showed how to use quantum expanders for explicit constructions (see also the related work by Harrow~\cite{Harrow09_2} and references therein).
\item Hayden {\it et al.} studied quantum state randomization, which corresponds to weak $(0,\eps)$-quantum min-entropy extractors with the input size equal to the output size~\cite{Hayden04} (see also the subsequent literature~\cite{Ambainis04,Dickinson06,Aubrun09}).
\end{itemize}
These constructions have many applications in, e.g., quantum information theory, quantum cryptography, and physics.\footnote{We will use strong quantum min-entropy extractors (against quantum side information) in the next chapter about channel simulations (Chapter~\ref{ch:channels}).} Here we are interested in the stability against quantum side information. Note that the fully quantum conditional min-entropy can be negative for entangled states.

\begin{definition}[Strong quantum min-entropy extractor against quantum side information]
Let $\cH_{M}\subset\cH_{N}$, $k\in[-\log|N|,\log|N|]$, and $\eps>0$. A strong $(k,\eps)$-quantum min-entropy extractor against quantum side information is a set of unitaries $\{U_{N}^{1},\ldots,U_{N}^{|D|}\}$ such that for all $\rho_{NR}\in\cS(\cH_{NR})$ with $H_{\min}(N|R)_{\rho}\geq k$,
\begin{align}
\Big\|\frac{1}{|D|}\cdot\sum_{i=1}^{|D|}\trace_{N\backslash M}\big[U_{N}^{i}\rho_{NR}(U_{N}^{i})^{\dagger}\big]\ot\proj{i}_{D}-\frac{\id_{M}}{|M|}\ot\rho_{R}\ot\frac{\id_{D}}{|D|}\Big\|_{1}\leq\eps\ .
\end{align}
\end{definition}

In analogy to the classical case, we can construct quantum min-entropy extractors is by means of quantum R\'enyi 2-entropy extractors. In fact, all constructions (even the probabilistic ones) for quantum min-entropy extractors that are known to be stable against quantum side information are based on R\'enyi 2-entropy extractors.


\subsection{R\'enyi 2-Extractors}

\begin{definition}[Strong quantum R\'enyi 2-extractor]
Let $\cH_{M}\subset\cH_{N}$, $k\in[0,\log|N|]$, and $\eps>0$. A strong $(k,\eps)$-quantum R\'enyi 2-extractor is a set of unitaries $\{U_{N}^{1},\ldots,U_{N}^{|D|}\}$ such that for all $\rho_{N}\in\cP^{+}(\cH_{N})$ with $H_{2}(N)_{\rho}\geq k$,
\begin{align}\label{eq:qq2strong}
H_{2}(MD)_{\sigma}\geq\log\left(\frac{|M|\cdot|D|}{\eps+1}\right)\ ,
\end{align}
where
\begin{align}
\sigma_{MD}=\frac{1}{|D|}\cdot\sum_{i=1}^{|D|}\trace_{N\backslash M}\big[U_{N}^{i}\rho_{N}(U_{N}^{i})^{\dagger}\big]\ot\proj{i}_{D}\ .
\end{align}
\end{definition}

In full analogy to the classical case, we have the following alternative characterization of R\'enyi 2-extractors.

\begin{proposition}\label{prop:h2q_alternative}
A set of functions $\{U_{N}^{1},\ldots,U_{N}^{|D|}\}$ is a strong $(k,\eps)$ R\'enyi 2-extractor with seed size $\log|D|$ if and only if we have for the map
\begin{align}
\psi(\rho_{N})=\frac{1}{|D|}\cdot\sum_{i=1}^{|D|}\trace_{N/M}\big[U_{N}^{i}\rho_{N}(U_{N}^{i})^{\dagger}\big]\ot\proj{i}_{D}
\end{align}
that
\begin{align}
\lambda_{1}\big(\psi^{\dagger}\circ\psi-\tau^{\dagger}\circ\tau\big)\leq2^{k}\cdot\frac{\eps}{|M|\cdot|D|}\ ,
\end{align}
where $\tau(\rho_{N})=\trace[\rho_{N}]\cdot\frac{\id_{M}}{|M|}\ot\frac{\id_{D}}{|D|}$. For typical applications, the completely mixed state $\frac{\id_{N}}{|N|}$ is the only eigenvector of $\psi^{\dagger}\circ\psi$ with eigenvalue one, and thus the relevant quantity is the second largest eigenvalue $\lambda_{2}(\psi^{\dagger}\circ\psi)$.
\end{proposition}

This can be proven in the exact same way as in the classical case (Proposition~\ref{prop:h2_alternative}). An instructive example are unitary 2-designs, which can be seen as a quantum generalization of families of pairwise independent permutations.

\begin{definition}[Unitary 2-design]\label{def:unitary2design}
A set of unitaries $\{U_1, \dots, U_L\}$ acting on $\cH$ is said to be a unitary 2-design if we have for all $M\in\cB(\cH)$ that
\begin{align}
\frac{1}{L}\cdot\sum_{i=1}^L U_i^{\ot 2} M (U_i^{\dagger})^{\ot 2} = \int U^{\ot 2} M (U^{\dagger})^{\ot 2} dU\ ,
\end{align}
where the integration is with respect to the Haar measure on the unitary group.
\end{definition}

Many efficient constructions of unitary 2-designs are known~\cite{Dankert09,Gross07}, and in an $n$-qubit space, such unitaries can typically be computed by circuits of size $O(n^{2})$.

\begin{proposition}\label{prop:unitary2design}
A unitary 2-design is a strong $(k,\eps)$-quantum R\'enyi 2-extractor with output size
\begin{align}\label{eq:unitary2design}
m=\frac{n+k}{2}-\frac{1}{2}\log\frac{1}{\eps}\ ,
\end{align}
where $n$ denotes the input size.
\end{proposition}

Note that this is in contrast to classical case, where we found $m=k-\log\frac{1}{\eps}$ for families of pairwise independent permutations (Proposition~\ref{prop:ccperm}).

\begin{proof}
By the alternative characterization for R\'enyi 2-entropy extractors (Proposition~\ref{prop:h2q_alternative}) the claim follows if we can show that
\begin{align}
\lambda_{1}\big(\psi^{\dagger}\circ\psi-\tau^{\dagger}\circ\tau\big)\leq\frac{|M|}{|D|\cdot|N|}\ ,
\end{align}
where
\begin{align}
\psi(\rho_{N})&=\frac{1}{|D|}\cdot\sum_{i=1}^{|D|}\trace_{N/M}\big[U_{N}^{i}\rho_{N}(U_{N}^{i})^{\dagger}\big]\ot\proj{i}_{D}=\frac{1}{|D|}\cdot\sum_{i=1}^{|D|}\psi_{i}(\rho_{N})\ot\proj{i}_{D}\ .
\end{align}
For any $X_{N},Y_{N}\in\cP^{+}(\cH_{N})$ we calculate
\begin{align}
\braket{X_{N}|(\tau^{\dagger}\circ\tau)(Y_{N})}&=\braket{\tau(X_{N})|\tau(Y_{N})}=\trace\left[\trace\big[X_{N}^{\dagger}\big]\cdot\frac{\id_{MD}}{|M|\cdot|D|}\cdot\trace\big[Y_{N}\big]\cdot\frac{\id_{MD}}{|M|\cdot|D|}\right]\notag\\
&=\frac{1}{|M|\cdot|D|}\cdot\trace\big[X_{N}^{\dagger}\big]\cdot\trace\big[Y_{N}\big]\ .
\end{align}
Furthermore, we calculate
\begin{align}
&\braket{X_{N}|(\psi^{\dagger}\circ\psi)(Y_{N})}=\frac{1}{|D|^{2}}\cdot\sum_{i=1}^{|D|}\braket{X_{N}|(\psi_{i}^{\dagger}\circ\psi_{i})(Y_{N})}\notag\\
&=\frac{1}{|D|^{2}}\cdot\sum_{i=1}^{|D|}\trace\Big[X^{\dagger}_{N}(U_{N}^{i})^{\dagger}\big(\id_{N/M}\ot(\trace_{N/M}[U_{N}^{i}Y_{N}(U_{N}^{i})^{\dagger}])\big)U_{N}^{i}\Big]\notag\\
&=\frac{1}{|D|^{2}}\cdot\sum_{i=1}^{|D|}\trace\Big[\big(\trace_{N/M}[U_{N}^{i}X_{N}^{\dagger}(U_{N}^{i})^{\dagger}]\ot\trace_{N'/M'}[U_{N'}^{i}Y_{N'}(U_{N'}^{i})^{\dagger}]\big)F_{MM'}\Big]\notag\\
&=\frac{1}{|D|}\cdot\trace\Big[\big(X_{N}^{\dagger}\ot Y_{N'}\big)\cdot\frac{1}{|D|}\sum_{i=1}^{|D|}\big((U_{N}^{i})^{\dagger}\ot (U_{N'}^{i})^{\dagger}\big)(F_{MM'}\ot\id_{NN'/MM'})\big(U_{N}^{i}\ot U_{N'}^{i}\big)\Big]\ ,
\end{align}
where we have used that the partial trace commutes with the identity, and the swap trick as in~\eqref{eq:swaptrick} with $F_{MM'}$ the swap operator. Since $\{U_{N}^{1},\dots,U_{N}^{|D|}\}$ is a unitary 2-design we have that (Lemma~\ref{lem:unitary_group})
\begin{align}
&\frac{1}{|D|}\sum_{i=1}^{|D|}\big((U_{N}^{i})^{\dagger}\ot (U_{N'}^{i})^{\dagger}\big)(F_{MM'}\ot\id_{NN'/MM'})\big(U_{N}^{i}\ot U_{N'}^{i}\big)\notag\\
&=\int\big(U_{N}^{\dagger}\ot U_{N'}^{\dagger}\big)(F_{MM'}\ot\id_{NN'/MM'})\big(U_{N}\ot U_{N'}\big)dU\notag\\
&\leq\frac{1}{|M|}\cdot\id_{NN'}+\frac{|M|}{|N|}\cdot F_{NN'}\ ,
\end{align}
and hence we arrive at
\begin{align}
\braket{X_{N}|(\psi^{\dagger}\circ\psi)(Y_{N})}\leq\frac{|M|}{|N|\cdot|D|}\cdot\trace\big[X_{N}^{\dagger}Y_{N}\big]\ .
\end{align}
Finally, we conclude that
\begin{align}
\lambda_{1}\big(\psi^{\dagger}\circ\psi-\tau^{\dagger}\circ\tau\big)&=\sup_{\substack{X_{N}\in\cP^{+}(N)\\\|X_{N}\|_{2}=1}}\braket{X_{N}|(\psi^{\dagger}\circ\psi-\tau^{\dagger}\circ\tau)(X_{N})}\notag\\
&\leq\frac{|M|}{|N|\cdot|D|}\ .
\end{align}
\end{proof}

An extension to almost unitary 2-designs can be found in~\cite{Szehr11_2,Szehr11}. Using~\cite{Harrow09_3}, this shows for instance that random quantum circuits of size $O(n^{2})$ are strong quantum R\'enyi 2-extractors with basically the same parameters as in Proposition~\ref{prop:unitary2design}. Moreover, it is shown in~\cite{Brown12} that even random circuits of size $O(n\cdot\log^{2}n)$ define strong quantum R\'enyi 2-extractors with basically the same parameters as in Proposition~\ref{prop:unitary2design}. 

A straightforward calculation shows that we can also state~\eqref{eq:qq2strong} in terms of the 2-norm as
\begin{align}
\frac{1}{|D|}\cdot\sum_{i=1}^{|D|}\|\trace_{N/M}\big[U_{N}^{i}\rho_{N}(U_{N}^{i})^{\dagger}\big]-\frac{\id_{M}}{|M|}\|_{2}^{2}\leq\frac{\eps}{|M|}\ .
\end{align}
Since, the R\'enyi 2-entropy is lower bounded by the min-entropy (Lemma~\ref{lem:h2vN_bounds}), and $\|X\|_{1}\leq \sqrt{\rank(X)}\cdot\|X\|_{2}$ (Lemma~\ref{lem:12norm}), it follows that every strong $(k,\eps)$-quantum R\'enyi 2-extractor is also a strong $(k,\sqrt{\eps})$-quantum min-entropy extractor (of the same output size and the same seed size). However, like in the classical case, R\'enyi 2-extractors always have a long seed.

\begin{proposition}\label{prop:qq_longseed}
Every strong $(k,\eps)$-quantum R\'enyi 2-extractor with input size $n$, output size $m$, and seed size $d$ necessarily has
\begin{align}
d\geq\min\{n-k,m\}+\log\frac{1}{\eps}-1\ .
\end{align} 
\end{proposition}

\begin{proof}
Let $\cH_{S}\subset\cH_{M}$ with $|S|=\lceil2^{k}\cdot|M|/|N|\rceil$, let $\{\ket{t}\}_{t=1}^{|N|/|M|}$ be an orthonormal basis of $\cH_{N/M}$, and consider the state
\begin{align}
\gamma_{N}=\frac{|M|}{|S|\cdot|N|}\cdot\sum_{s\in S}\sum_{t=1}^{|N|/|M|}(U_{N}^{1})^{\dagger}\proj{st}_{N}U_{N}^{1}\ .
\end{align}
Since $H_{2}(N)_{\gamma}\geq k$, and
\begin{align}
\Big\|\trace_{N/M}\big[U_{N}^{1}\gamma{N}(U_{N}^{1})^{\dagger}\big]-\frac{\id_{M}}{|M|}\Big\|_{2}^{2}&=\Big\|\frac{1}{|S|}\cdot\sum_{s\in S}\proj{s}_{M}-\frac{\id_{M}}{|M|}\Big\|_{2}^{2}\notag\\
&\geq\sum_{s\in S}\left(\frac{1}{|S|}-\frac{1}{|M|}\right)^{2}=|S|\cdot\left(\frac{1}{|S|}-\frac{1}{|M|}\right)^{2}\ ,
\end{align}
we get by assumption that
\begin{align}
\frac{1}{|D|}\cdot|S|\cdot\left(\frac{1}{|S|}-\frac{1}{|M|}\right)^{2}\leq\frac{\eps}{|M|}\ .
\end{align}
Now, if $2^{k}\cdot|M|/|N|\leq1$, then
\begin{align}
d\geq\log\left(\frac{|M|}{\eps}\cdot\left(1-\frac{1}{|M|}\right)^{2}\right)\geq m+\log\frac{1}{\eps}\ .
\end{align}
Otherwise we get
\begin{align}
d&\geq\log\left(\frac{|M|}{\eps}\cdot|S|\cdot\left(\frac{1}{|S|}-\frac{1}{|M|}\right)^{2}\right)=2\cdot\log\big(|M|-|S|\big)-\log|S|-\log|M|+\log\frac{1}{\eps}\notag\\
&\geq\log|M|-\log\left(\frac{2^{k+1}\cdot|M|}{|N|}\right)+\log\frac{1}{\eps}\geq n-k+\log\frac{1}{\eps}-1\ .
\end{align}
\end{proof}

The stability of R\'enyi 2-extractors against quantum side information is defined as follows (note that the fully quantum conditional R\'enyi 2-entropy can be negative for entangled states.)

\begin{definition}[Strong quantum R\'enyi 2-extractor against quantum side information]\label{def:h2_qq}
Let $\cH_{M}\subset\cH_{N}$, $k\in[-\log|N|,\log|N|]$, and $\eps>0$. A strong $(k,\eps)$-quantum R\'enyi 2-extractor against quantum side information is a set of unitaries $\{U_{N}^{1},\ldots,U_{N}^{|D|}\}$ such that for all $\rho_{NR}\in\cP^{+}(\cH_{NR})$ with $H_{2}(N|R)_{\rho}\geq k$,
\begin{align}
H_{2}(MD|R)_{\sigma}\geq\log\left(\frac{|M|\cdot|D|}{\eps+1}\right)\ ,
\end{align}
where
\begin{align}
\sigma_{MDR}=\frac{1}{|D|}\cdot\sum_{i=1}^{|D|}\trace_{N\backslash M}\big[U_{N}^{i}\rho_{NR}(U_{N}^{i})^{\dagger}\big]\ot\proj{i}_{D}\ .
\end{align}
\end{definition}

By using the exact same arguments as in the proof of the stability of classical R\'enyi 2-extractors against quantum side information (Theorem~\ref{thm:h2c_stable}), we find that the condition in Definition~\ref{def:h2_qq} is equivalent to
\begin{align}
\lambda_{1}\big(\psi^{\dagger}\circ\psi-\tau^{\dagger}\circ\tau\big)\leq2^{k}\cdot\frac{\eps}{|M|\cdot|D|}\ ,
\end{align}
Hence, we get the following theorem.

\begin{theorem}
Every strong $(k,\eps)$-quantum R\'enyi 2-extractor is stable against quantum side information.
\end{theorem}

Finally, we show that quantum R\'enyi 2-extractors also give rise to quantum min-entropy extractors against quantum side information. The calculation is inspired by~\cite{Szehr11_2,Szehr11,Dupuis10,Dupuis09,Berta08}, and uses the same argument as in the classical case (Theorem~\ref{thm:crenyi2_stability}).

\begin{theorem}\label{thm:qrenyi2_stability}
Every strong $(k,\eps)$-quantum R\'enyi 2-extractor is also a strong $(k,2\sqrt{\eps})$-quantum min-entropy extractor against quantum side information (of the same output size and the same seed size).
\end{theorem}

This can be proven in the exact same way as in the classical case (Theorem~\ref{thm:crenyi2_stability}). As an example we mention again unitary 2-designs (see Definition~\ref{def:unitary2design}).

\begin{corollary}\label{cor:2design}
A unitary 2-design is a strong $(k,2\sqrt{\eps})$-quantum min-entropy extractor against quantum side information with output size
\begin{align}
m=\frac{n+k}{2}-\frac{1}{2}\log\frac{1}{\eps}\ ,
\end{align}
where $n$ denotes the input size.
\end{corollary}

Note that $k$ can be negative for entangled input states, and that the corresponding classical result for families of pairwise independent permutations reads $m=k-\log\frac{1}{\eps}$ (Proposition~\ref{prop:ccperm}).


\subsection{Discussion}\label{sec:qq_disc}

We only analyzed stability against finite-dimensional quantum side information, but we believe that our approach also works for side information modeled by a von Neumann algebra.

The big challenge is to construct short seed quantum min-entropy extractors with and without quantum side information. An important recent result in this direction is by Fawzi~\cite{Fawzi13}, who gave a probabilistic construction of a strong $(k,\eps)$ quantum min-entropy extractor with output size $m=(n+k)/2-O(\log(1/\eps))$ and seed size $d=2\log(1/\eps)+O(\log\log(1/\eps))$. What is remarkable about this construction is that the seed size does not depend on $n$ (which is in sharp contrast to the classical case). The derivation uses the R\'enyi 2-extractor based construction discussed in Proposition~\ref{prop:unitary2design}, but combines it with measure concentration ideas and $\eps$-net techniques to show that it is actually sufficient to sample a set (of size $2^{d}$) of unitaries from the Haar measure. Since the measure concentration and the $\eps$-net are done in the 1-norm directly (and not in 2-norm), this is not in contradiction to Proposition~\ref{prop:qq_longseed}. Interestingly, this approach fails for several reasons for the case of quantum side information, and it is an open question if the construction is stable against quantum side information. A few ideas concerning short seed quantum min-entropy extractors can be found in Section~\ref{se:extmore}.

Finally, we mention that we will use strong quantum min-entropy extractors against quantum side information in the next chapter about channel simulations (Chapter~\ref{ch:channels}).


\section{Quantum to Classical}\label{se:qc}

The results in this section have been obtained in collaboration with Omar Fawzi, and Stephanie Wehner, and have appeared in~\cite{Berta12_2,Berta11_5}. Many of the results in this section can also be found in the thesis of Fawzi~\cite{Fawzi12}, but presented from a slightly different perspective.

In the context of cryptography, a quantum extractor is often more than we need. In fact, it is usually sufficient to extract random classical bits, which is in general easier to obtain than random qubits. This motivates the definition of quantum-classical extractors (Definition~\ref{def:qc-extractor}), where the difference with a quantum extractor is that the output system $M$ is measured in the computational basis. For that purpose, we define the measurement map for $\cH_{M}\subset\cH_{N}$ as $\cT_{N\ra M}:\cB(\cH_{N})\ra\cB(\cH_{M})$,
\begin{align}\label{eq:meas_map}
\cT_{N\ra M}(\cdot)=\sum_{m,m'}\braket{mm'|(\cdot)mm'}\proj{m}_{M}\ ,
\end{align}
where $\{\ket{mm'}\}$, $\{\ket{m}\}$ are orthonormal bases of $\cH_{N}$, $\cH_{M}$, respectively. A small calculation readily reveals that this map can be understood as tracing out $N/M$, and then measuring the remaining system $M$ in the basis $\{\ket{m}\}$. Using the map~\eqref{eq:meas_map}, we define strong quantum-classical min-entropy extractors against quantum side information.

\begin{definition}[Strong quantum-classical min-entropy extractor against quantum side information]\label{def:qc-extractor}
Let $\cH_{M}\subset\cH_{N}$, $k\in[-\log|N|,\log|N|]$, and $\eps>0$. A strong $(k,\eps)$ quantum-classical min-entropy extractor against quantum side information is a set of unitaries $\{U_{N}^{1},\ldots,U_{N}^{|D|}\}$ such that for all $\rho_{NR}\in\cS(\cH_{NR})$ with $H_{\min}(N|R)_{\rho}\geq k$,
\begin{align}
\Big\|\frac{1}{|D|}\cdot\sum_{i=1}^{|D|}\cT_{N\ra M}\big(U_{N}^{i}\rho_{NR}(U_{N}^{i})^{\dagger}\big)\ot\proj{i}_{D}-\frac{\id_{M}}{|M|}\ot\rho_{R}\ot\frac{\id_{D}}{|D|}\Big\|_{1}\leq\eps\ ,
\end{align}
where the measurement map $\cT_{N\ra M}$ is defined as in~\eqref{eq:meas_map}. The quantity $n=\log|N|$ is called the input size, $m=\log|M|$ the output size, and $d=\log|D|$ the seed size.
\end{definition}

The corresponding definition for quantum-classical R\'enyi 2-extractors against quantum side information is straightforward (cf.~Definitions~\ref{def:renyi_qsi} and~\ref{def:h2_qq}). Observe that Definition~\ref{def:qc-extractor} only allows a specific form of measurements obtained by applying a unitary transformation followed by a measurement in the computational basis of $M$. The reason we use this definition is that we want the output of the extractor to be determined by the input and the choice of the seed. In the quantum setting, a natural way of translating this requirement is by imposing that an adversary holding a system that is maximally entangled with the source can perfectly predict the output. This condition is satisfied by the form of measurements dictated by Definition~\ref{def:qc-extractor}. Allowing general measurements already (implicitly) allows the use of randomness for free.\footnote{Like in the case of fully quantum extractors, alternative definitions could also allow for a quantum seed (cf.~the comments in Section~\ref{sec:qmin-extractors}).} Note also, that in the case where the system $R$ is trivial, a strong $(0,\eps)$ quantum-classical extractor is the same as an $\eps$-metric uncertainty relation as studied by Fawzi {\it et al.}~\cite{Fawzi11,Fawzi12}.

By the monotonicity of the 1-norm under channels, every quantum extractor is also a quantum-classical extractor with the same parameters. For example a unitary 2-design $\{U_{N}^{1},\dots,U_{N}^{|D|}\}$ is then also a strong $(k,\eps)$ quantum-classical min-entropy extractor against quantum side information with output size (Corollary~\ref{cor:2design})
\begin{align}
m=\frac{n+k}{2}-\log\frac{1}{\eps}\ .
\end{align}
But a quantum extractor is stronger than a quantum-classical extractor, since for the latter we only require the output state to be close to uniform after performing a measurement,\footnote{Note that it is possible to obtain a uniform distribution over outcomes even if the state was not maximally mixed, e.g., by measuring the pure state $\proj{0}$ in the Fourier basis.} and thus we might hope to get better parameters for quantum-classical extractors. And indeed, we have shown in~\cite{Dupuis10,Berta11_5} that unitary 2-designs give strong $(k,\eps)$ quantum-classical min-entropy extractor against quantum side information with output size\footnote{For comparison, in the classical case we get for families of pairwise independent permutations an output size $m=k-O(\log(1/\eps))$ (Proposition~\ref{prop:ccperm}).}
\begin{align}
m=\min\big\{n,n+k-2\log\frac{1}{\eps}\big\}\ .
\end{align}
Moreover, it is natural to expect that we can build smaller and simpler sets of unitaries if we are only interested in extracting random classical bits. In fact, in Sections~\ref{sec:mub_extractor}-\ref{sec:qudit_extractor}, we construct simpler sets of unitaries that define quantum-classical extractors. Like in the case of quantum extractors, we will first construct R\'enyi 2-extractors, and then extend them to min-entropy extractors against quantum side information. We mention that this strategy has the drawback of a long seed.

\begin{proposition}\label{prop:qc_longseed}
Every strong $(k,\eps)$ quantum-classical R\'enyi 2-extractor with input size $n$, output size $m$, and seed size $d$ necessarily has
\begin{align}
d\geq\min\{n-k,m\}+\log\frac{1}{\eps}-1\ .
\end{align} 
\end{proposition}

\begin{proof}
The proof employs exactly the same example as in the case of quantum R\'enyi 2-extractors (Proposition~\ref{prop:qq_longseed}).
\end{proof}

In the following, our strategy for constructing quantum-classical R\'enyi 2-extractors is to first measure in an adequate set of bases, and then apply a classical R\'enyi 2-extractor. The set of bases is a good choice if the measurements create high uncertainty in R\'enyi 2-entropy, that is, if the set of bases fulfills a strong (R\'enyi 2-) entropic uncertainty relation. Moreover, if we want to later lift the R\'enyi 2-extractors to min-entropy extractors against quantum side information, we need that the set of bases even fulfills a strong (R\'enyi 2-) entropic uncertainty relation with quantum side information. However, we constructed exactly such entropic uncertainty relations in Section~\ref{se:several}.

With the help of these relations we now give explicit constructions for strong quantum-classical min-entropy extractors by lifting classical R\'enyi 2-extractors (Sections~\ref{sec:mub_extractor} and~\ref{sec:qudit_extractor}). Finally, we mention an application to the noisy storage model in two-party quantum cryptography (Section~\ref{sec:storage}).


\subsection{Mutually Unbiased Bases Extractors}\label{sec:mub_extractor}

For the construction we take a full set of mutually unbiased bases on the input space $\cH_{N}$ and represent it by a set of unitary transformations $\{U_{N}^{\theta}\}_{\theta=1}^{|N|+1}$ mapping the mutually unbiased bases to some standard orthonormal basis $\ket{n}_{n=1}^{|N|}$ of $\cH_{N}$. In addition, we take a set of pairwise independent permutations $\cP=\{P^{j}\}_{j=1}^{|\cP|}$ on $\{1,\ldots,|N|\}$, and note that permutations of basis elements of $\cH_{N}$ can be seen as a unitary transformation on $\cH_{N}$ defined by
\begin{align}\label{eq:perm-basis}
P_{N}^{j}\ket{n}_{N}=\ket{P^{j}(n)}_{N}\ .
\end{align}

\begin{theorem}\label{thm:mub_extractor}
The set of unitaries $\big\{P_{N}^{j}U_{N}^{\theta}:j\in\{1,\ldots,|\cP|\},\theta\in\{1,\ldots,|N|+1\}\big\}$ defines a strong $(k,\eps)$ quantum-classical min-entropy extractor with output size
\begin{align}
m=\min\big\{n,n+k-2\log\frac{1}{\eps}-2\big\}\ ,
\end{align}
and the seed size is
\begin{align}
d=\log|\cP|+\log\big(|N|+1\big)\ .
\end{align}
For example for the set of pairwise permutations described before Proposition~\ref{prop:ccperm}, this shows that a seed size of $d=3n$ is achievable.
\end{theorem}

\begin{proof}
Similarly to the classical case~\eqref{eq:qh2_extractor}, we think of the extractor as a map $\psi:\cB(\cH_{N})\ra\cB(\cH_{M\Theta P})$ with
\begin{align}
&(\psi\ot\cI_{R})(\rho_{NR})\notag\\
&=\frac{1}{|P|}\cdot\frac{1}{|N|+1}\cdot\sum_{j=1}^{|\cP|}\sum_{\theta=1}^{|N|+1}\cT_{N\ra M}\big(P_{N}^{j}U_{N}^{\theta}\rho_{NR}(P_{N}^{j}U_{N}^{\theta})^{\dagger}\big)\ot\proj{\theta}_{\Theta}\ot\proj{j}_{P}\ ,
\end{align}
and denote $\tau:\cB(\cH_{N})\ra\cB(\cH_{M\Theta P})$ with $\tau(\rho_{N})=\trace[\rho_{N}]\cdot\frac{\id_{M}}{|M|}\ot\frac{\id_{\Theta}}{|N|+1}\ot\frac{\id_{P}}{|\cP|}$. Furthermore, we denote
\begin{align}
\sigma_{M\Theta PR}=(\psi\ot\cI_{R})(\rho_{NR})=\frac{1}{|N|+1}\sum_{\theta=1}^{|N|+1}\sigma_{MPR}^{\theta}\ot\proj{\theta}_{\Theta}\ ,
\end{align}
as well as
\begin{align}
\sigma_{N\Theta R}&=\frac{1}{|N|+1}\cdot\sum_{\theta=1}^{|N|+1}\sum_{n=1}^{|N|}\braket{n|U_{N}^{\theta}\rho_{NR}(U_{N}^{\theta})^{\dagger}n}\ot\proj{n}_{N}\ot\proj{\theta}_{\Theta}\notag\\
&=\frac{1}{|N|+1}\cdot\sum_{\theta=1}^{|N|+1}\sigma_{NR}^{\theta}\ot\proj{\theta}_{\Theta}\ ,
\end{align}
where $\{\ket{n}\}_{n=1}^{|N|}$ is the orthonormal basis of $\cH_{N}$ defined by~\eqref{eq:perm-basis}. By the exact same steps as in the classical proof (Theorem~\ref{thm:crenyi2_stability}), we estimate the 1-norm by
\begin{align}\label{eq:1norm_derivation}
\|((\psi-\tau)\ot\cI_{R})(\rho_{NR})\|_{1}\leq2\cdot\sqrt{|M|\cdot\big(|N|+1)\cdot|\cP|\cdot2^{-H_{2}(M\Theta P|R)_{\sigma}}-1}\ ,
\end{align}
and hence it only remains to determine $H_{2}(M\Theta P|R)_{\sigma}$. This, however, we can accomplish by using the fact that a set of pairwise permutations is a R\'enyi 2-extractor (Proposition~\ref{prop:ccperm}), and with this also stable against quantum side information (Theorem~\ref{thm:h2c_stable}), that is,
\begin{align}\label{eq:h2_derivationp}
H_{2}(MP|R)_{\sigma^{\theta}}=-\log\left(\frac{1}{|\cP|}\cdot\left(\frac{|N|-|N|/|M|}{|N|-1}\cdot2^{-H_{2}(N|R)_{\sigma^{\theta}}}+\frac{1}{|M|}\right)\right)\ ,
\end{align}
and by the fact that we have a R\'enyi 2-entropic uncertainty relation with quantum side information for a full set of mutually unbiased bases (Theorem~\ref{thm:h2relation}), that is,
\begin{align}\label{eq:h2_derivationmub}
\frac{1}{|N|+1}\cdot\sum_{\theta=1}^{|N|+1}2^{-H_{2}(N|R)_{\sigma^{\theta}}}=2^{-H_{2}(N|\Theta R)_{\sigma}}=\frac{2^{-H_{2}(N|R)_{\rho}}+1}{|N|+1}\ .
\end{align}
We calculate
\begin{align}
2^{-H_{2}(M\Theta P|R)_{\sigma}}&=\frac{1}{\big(|N|+1\big)^{2}}\cdot\sum_{\theta=1}^{|N|+1}2^{-H_{2}(MP|R)_{\sigma^{\theta}}}\notag\\
&=\frac{1}{\big(|N|+1\big)^{2}}\cdot\sum_{\theta=1}^{|N|+1}\left(\frac{1}{|\cP|}\cdot\left(\frac{|N|-|N|/|M|}{|N|-1}\cdot2^{-H_{2}(N|R)_{\sigma^{\theta}}}+\frac{1}{|M|}\right)\right)\notag\\
&=\frac{1}{|N|+1}\left(\frac{1}{|\cP|}\cdot\left(\frac{|N|-|N|/|M|}{|N|-1}\cdot\frac{2^{-H_{2}(N|R)_{\rho}}+1}{|N|+1}+\frac{1}{|M|}\right)\right)\ ,
\end{align}
where we have used~\eqref{eq:h2_derivationp} and~\eqref{eq:h2_derivationmub}. Going back to~\eqref{eq:1norm_derivation} we get
\begin{align}
&\|((\psi-\tau)\ot\cI_{R})(\rho_{NR})\|_{1}\notag\\
&\leq2\cdot\sqrt{\frac{|M|\cdot|N|-|N|}{|N|^{2}-1}\cdot2^{-H_{2}(N|R)_{\rho}}+\frac{1}{|N|+1}+\frac{|M|\cdot|N|-|N|}{|N|^{2}-1}-1}\notag\\
&\leq2\cdot\sqrt{\frac{|M|}{|N|+1}\cdot2^{-H_{2}(N|R)_{\rho}}}\leq2\cdot\sqrt{\frac{|M|}{|N|+1}\cdot2^{-H_{\min}(N|R)_{\rho}}}\ ,
\end{align}
where we have used that the conditional R\'enyi 2-entropy is lower bounded by the conditional min-entropy (Lemma~\ref{lem:h2hmin_equiv}). Now, we set
\begin{align}
\eps=2\cdot\sqrt{\frac{|M|}{|N|+1}\cdot2^{-H_{\min}(N|R)_{\rho}}}\ ,
\end{align}
and the claim follows since by definition $k=H_{\min}(N|R)_{\rho}$.
\end{proof}


\subsection{Qudit-Wise Extractors}\label{sec:qudit_extractor}

For input spaces $\cH_{N}$ that can be decomposed into a tensor product of $f$ dimensional systems (qudits), we construct in this section even simpler unitaries that define quantum-classical extractors. They are composed of unitaries $U$ acting on single qudits followed by permutations $P$ of the computational basis elements. Note that this means that the measurements defined by these unitaries can be implemented with current technology. As the measurement $\cT$ commutes with the permutations $P$, we can first apply $U$, then measure in the computational basis, and finally apply the permutation to the classical outcome of the measurement. In addition to the computational efficiency, the fact that the unitaries act on single qudits is often a desirable property for the design of cryptographic protocols. In particular, the application to the noisy storage model that we will discuss in Section~\ref{sec:storage} does make use of this fact. For the sake of simplicity we restrict in the following to the case $f=2$ (qubits), but the general case is straightforward (for details see~\cite{Berta11_5}). 

For the construction we take a full set of mutually unbiased bases in dimension $2$, and represent it by a set of unitary transformations $\{U^{1},U^{2},U^{3}\}$ mapping the mutually unbiased bases to some standard basis. For example, we can choose 
\begin{align}
U^1 = \left( \begin{array}{cc}
1 & 0 \\
0 & 1
\end{array} \right)
\qquad U^2 = \frac{1}{\sqrt{2}} \left( \begin{array}{cc}
1 & 1 \\
1 & -1
\end{array} \right)
\qquad U^3 = \frac{1}{\sqrt{2}} \left( \begin{array}{cc}
1 & i \\
i & -1
\end{array} \right)\ .
\end{align}
We then define the set $\cU_{2,n}$ of unitary transformations on $n$ qubits by
\begin{align}
\cU_{2,n}=\big\{U_{N}=U^{u_1}\ot\cdots\ot U^{u_n}|u_i\in\left\{1,2,3\right\}\big\}\ .
\end{align}
Note that this is in complete analogy to the section about single qudit measurement uncertainty relations (Section~\ref{sec:single_qudit}). In addition, we take a set of pairwise independent permutations $\cP$ (as in Section~\ref{sec:mub_extractor}).

\begin{theorem}\label{thm:qudit_extractor}
Let $\cH_{N}\cong(\nC^{2})^{\ot n}$. Then, the set of unitaries $\{P_{N}U_{N}:P_{N}\in\cP,U_{N}\in\cU_{2,n}\}$ defines a strong $(k,\eps)$ quantum-classical min-entropy extractor with output size
\begin{align}
m=n\cdot(\log3-1)+\min\{0,k\}-O(\log(1/\eps))-O(1)\ ,
\end{align}
and the seed size is $d=3^{n}\cdot|\cP|$. For example for the set of pairwise permutations described before Proposition~\ref{prop:ccperm}, this shows that a seed size of $d=n\cdot(2+\log3)$ is achievable.
\end{theorem}

\begin{proof}
By employing the single qubit R\'enyi 2-entropic uncertainty relation with quantum side information (Theorem~\ref{thm:qudit_UR}), and the fact that a set of pairwise permutations is a R\'enyi 2-extractor (Proposition~\ref{prop:ccperm}), the proof follows exactly as in Theorem~\ref{thm:mub_extractor}. For details we refer to~\cite{Berta11_5}.
\end{proof}


\subsection{Application: Noisy Storage Model}\label{sec:storage}

The concept of quantum-classical extractors has an application in quantum cryptography. Namely, it allows to relate the security of cryptographic protocols in the noisy-storage model~\cite{Wehner08,Schaffner08,Wehner10,Koenig12} to the strong converse quantum capacity.\footnote{For the definitions of strong converse capacities and regular capacities, see Section~\ref{se:app}.} The appeal of the noisy storage model is that it allows to solve any cryptographic problem involving two mutually distrustful parties, such as bit commitment, oblivious transfer~\cite{Koenig12} or secure identification~\cite{Damgard07,Bouman12}. This is impossible without imposing any assumptions, such as a noisy quantum memory, on the adversary~\cite{Lo97,Mayers96,Lo98,Lo97_2,Mayers97}. Proposed protocols can be implemented with any hardware suitable for quantum key distribution.

Here, we only give a brief overview of the noisy-storage model (as illustrated in Figure~\ref{fig:noisyModel}), and informally state our result in order to advertise the concept of quantum-classical extractors. Our precise result and all technical arguments can be found in~\cite{Berta11_5,Fawzi12}.

The central assumption of the noisy-storage model is that the adversary can only store quantum information in a memory described by a particular channel $\cF:\cB(\cH_{in}) \rightarrow \cB(\cH_{out})$. In practice, the use of the memory device is enforced by introducing waiting times $\Delta t$ into the protocol. This is the only restriction imposed on the adversary who is otherwise all-powerful. In particular, he can store an unlimited amount of classical information, and all his actions are instantaneous. This includes any computations, communications, measurements and state preparation that may be necessary to perform an error-correcting encoding and decoding before and after using his noisy memory device. In~\cite{Koenig12}, a natural link was formed between security in the noisy-storage model, and the information carrying capacity of the storage channel $\cF$. Of particular interest were memory assumptions that scale with the number $m$ of qubits transmitted during the protocol.\footnote{In turn, this tells us how many qubits need to be send in order to achieve security against an attacker with a certain amount of storage.} That is, the channel is of the form $\cF = \cE^{\ot \nu \cdot m}$, where $\nu$ is referred to as the storage rate. It was shown that any two-party cryptographic problem can in principle\footnote{That is, by transmitting a sufficiently large number $m$ of qubits.} be implemented securely if~\cite{Koenig12}
\begin{align}\label{eq:storage}
C({\cE})\cdot\nu<\frac{1}{2}\ ,
\end{align}
where $C(\cE)$ denotes the strong converse classical capacity of the channel $\cE$ (which is known to equal the classical capacity for certain classes of channels~\cite{Koenig09_2}). For the special case of $\cE = \cI_2$, i.e., the one qubit identity channel, the condition simplifies to $\nu < \frac{1}{2}$. This case is also known as bounded-storage~\cite{Damgard05,Damgard07,Schaffner07}. For protocols involving qubits in a simple BB84 like scheme this is the best bound known.
\begin{figure}[ht]\label{fig:noisyModel}
\begin{center}
\includegraphics[width=0.8\linewidth]{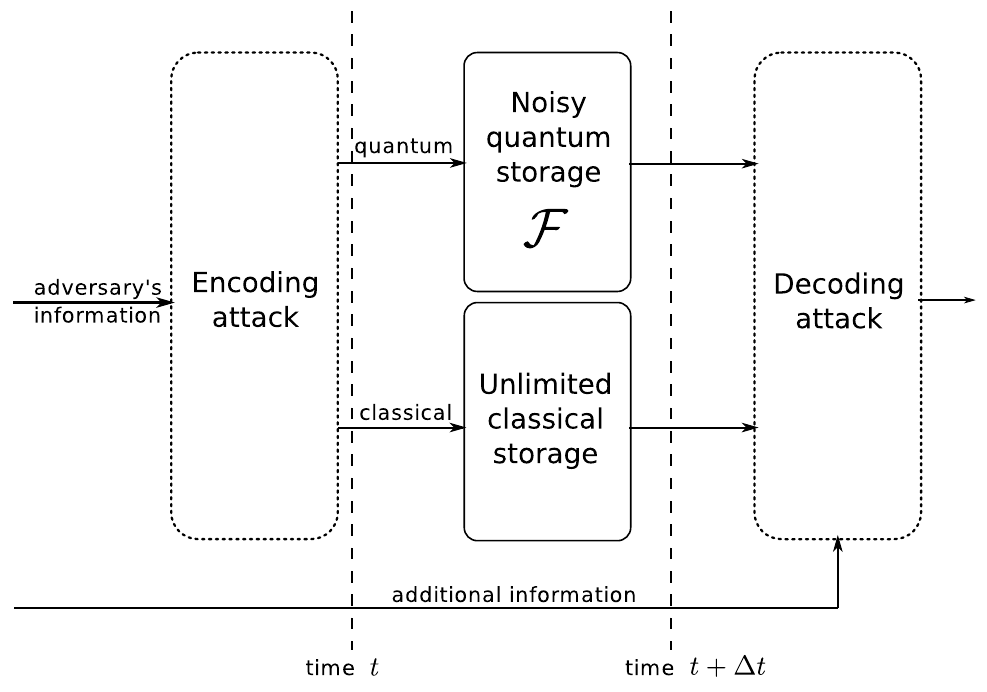}
\end{center}
\caption{Noisy storage model.}{Noisy-storage assumption: during waiting times $\Delta t$, the adversary can only use his noisy memory device to store quantum information. However, he is otherwise all powerful, and storage of classical information is free.}
\end{figure}
When considering storing quantum information exchanged during the protocol, it may come as a surprise that the classical capacity should be relevant. Indeed, looking at Figure~\ref{fig:noisyModel} it becomes clear that a much more natural quantity would be the quantum capacity of $\cE$.

This is exactly where the concept of quantum-classical extractors comes into play. By the changing the encoding from a qubit BB84 scheme to a qubit six state scheme, and using the qubit-wise quantum-classical extractor of Proposition~\ref{thm:qudit_extractor}, we can derive a statement similar to~\eqref{eq:storage}, but with the strong converse classical capacity replaced by the strong converse quantum capacity. This then extends the parameter regime where security of all existing protocols can be proven. Even though there is in general no closed expression for the strong converse quantum capacity, we calculate in Section~\ref{se:cost} a few examples for security rates by means of the entanglement cost of quantum channels (which is an upper bound on the strong converse quantum capacity).

It is an interesting open question if a similar argument also works for the qubit BB84 scheme. For this, we would need a single qubit R\'enyi 2-entropic uncertainty relation for only two complementary measurements per qubit (cf.~Section~\ref{sec:discussion_several}).

\textit{Note added.} The question about the qubit BB84 scheme was recently settled in~\cite{Dupuis13}.


\subsection{Discussion}\label{sec:qc_discussion}

In addition to the explicit constructions discussed above (Sections~\ref{sec:mub_extractor}-\ref{sec:qudit_extractor}), we were able to use operator Chernoff bound methods (see, e.g., \cite{Ahlswede02,Tropp12}) to give a probabilistic construction for a R\'enyi 2-entropy extractor with output size
\begin{align}
m=\min\{n,n+k-O(\log(1/\eps))-O(1)\}\ ,
\end{align}
and seed size
\begin{align}
d=m+\log n+O(\log(1/\eps))+O(1)\ .
\end{align}
This is almost matching the lower bound for the seed size of R\'enyi 2-extractors $d\geq\min\{n-k,m\}+\log(1/\eps)-1$ (Proposition~\ref{prop:qc_longseed}). For details, we refer to~\cite{Berta11_5}.

However, the challenge is to directly construct short seed quantum-classical min-entropy extractors with and without quantum side information. Here, the probabilistic construction for strong quantum min-entropy extractors (without quantum side information) of \cite{Fawzi13} also works for strong quantum-classical extractors and gives an output size $m=\min\{n,n+k-O(\log(1/\eps))$ for a seed size $d=2\log(1/\eps)+O(\log\log(1/\eps))$ and general $k\in[-n,n]$ and $\eps>0$. Concerning explicit constructions, there is a result by Fawzi {\it et al.}~\cite{Fawzi11,Fawzi12} for the case $k=0$. In particular, they can construct strong $(0,\eps)$ quantum-classical min-entropy extractors with output size $m=n-O(\log\log n)-O(\log(1/\eps))$, and seed size $d=O(\log n\cdot\log\log(n/\eps))$. Interestingly, they work with the fidelity instead of the 1-norm, and can thereby avoid the use of the 2-norm (which would account for a long seed). But their techniques do not directly translate to the case of non-vanishing min-entropy and quantum side information.

Finally, we mention that quantum-classical extractors against quantum side information can also be used to derive entropic uncertainty relations with quantum side information (rather then employing such uncertainty relations to derive quantum-classical extractors). For details we refer to~\cite{Berta11_5}, but note that the seed size of the quantum-classical extractor translates into the number of measurements in the entropic uncertainty relation. Since it is in general difficult to obtain strong uncertainty relations for a small set of measurements (except for the special case of two), finding quantum-classical extractors with a small seed size is also worth pursuing from the point of view of uncertainty relations. For a similar discussion in the case without quantum side information see~\cite{Fawzi11,Fawzi12}.


\section{More about Extractors}\label{se:extmore}

The results in this section have been obtained in collaboration with Omar Fawzi, and Volkher Scholz, but are currently unpublished.

\begin{sidewaystable}[ht]\label{tab:extractors}
\centering
\begin{tabular}{l l l}
& Classical & Classical with QSI \\
\hline
Seed LB & $\log(n-k)+2\log(1/\eps)$~\cite{Radhakrishnan00} & $\log(n-k)+2\log(1/\eps)$~\cite{Radhakrishnan00}\\
Seed UB & $\log(n-k)+2\log(1/\eps)$ NE~\cite{Sipser88,Radhakrishnan00} & $\log^{2}(n/\eps)\cdot\log(k-4\log(1/\eps))$~\cite{De12} \\
& $\log n +O(\log(k/\eps))$~\cite{Guruswami09} & \\
Output UB & $k-2\log(1/\eps)$~\cite{Radhakrishnan00} & $k-2\log(1/\eps)$~\cite{Radhakrishnan00}\\
Output LB & $k-2\log(1/\eps)$~\cite{Sipser88,Impagliazzo89} & $k-2\log(1/\eps)$~\cite{Renner05}\\
\hline
& & \\
& Quantum & Quantum with QSI\\
\hline
Seed LB & $\log(1/\eps)$ [Prop~\ref{prop:coneps}] & $\log(1/\eps)$ [Prop~\ref{prop:coneps}]\\
Seed UB & $2\log(1/\eps)+O(\log\log(1/\eps))$ NE~\cite{Fawzi13}  & $5n$ [Cor~\ref{cor:2design}] with~\cite{Chau05} \\
& $5n$ [Cor~\ref{cor:2design}] with~\cite{Chau05} & \\
Output UB & $(n+H_{\min}^{\eps'}(N))/2+O(\log(1/\eps))$~\cite{Dupuis10} & $(n+H_{\min}^{\eps'}(N|R))/2+O(\log(1/\eps))$~\cite{Dupuis10}\\
Output LB & $(n+k)/2-\log(1/\eps)$ [Cor~\ref{cor:2design}] & $(n+k)/2-\log(1/\eps)$ [Cor~\ref{cor:2design}]\\
\hline
& & \\
& Quantum-classical & Quantum-classical with QSI\\
\hline
Seed LB & $\log(1/\eps)$~\cite{Berta11_5} & $\log(1/\eps)$~\cite{Berta11_5}\\
Seed UB & $2\log(1/\eps)+O(\log\log(1/\eps))$ NE~\cite{Fawzi13} & $\min\{n,n+k\}+\log n$ NE~\cite{Berta11_5}\\
& $3n$ [Thm~\ref{thm:mub_extractor}] & $3n$ [Thm~\ref{thm:mub_extractor}]\\
Output UB & $n$ (trivial) & $\min\{n,n+H_{\min}^{\eps'}(N|R)\}$~\cite{Berta11_5}\\
Output LB & $\min\{n,n+k-2\log(1/\eps)\}$ [Thm~\ref{thm:mub_extractor}] & $\min\{n,n+k-2\log(1/\eps)\}$ [Thm~\ref{thm:mub_extractor}]\\
\hline
\end{tabular}
\caption{Classical, quantum, and quantum-classical randomness extractors.}{Some known lower bounds (LB) and upper bounds (UB) on the seed size and output size in terms of (qu)bits for different kinds of strong $(k,\eps)$ min-entropy extractors. $n$ refers to the number of input (qu)bits, and $k$ to the min-entropy of the input. Note that for quantum and quantum-classical extractors with quantum side information, $k$ can be as small as $-n$. Additive absolute constants are omitted. The symbol NE denotes non-explicit constructions, QSI stands for quantum side information, and the smoothing parameter $\eps'$ is of order $\mathrm{poly}(\eps)$.}
\end{sidewaystable}

The big open question of this chapter is to find short seed probabilistic and explicit constructions for quantum (or quantum-classical) min-entropy extractors with and without quantum side information (see Table~\ref{tab:extractors} for an overview of existing constructions). The only lower bound on the seed size that we have (Proposition~\ref{prop:coneps}) together with the probabilistic construction of Fawzi~\cite{Fawzi13} for the case without side information thereby suggest that short seed means as small as $d=O(\log(1/\eps))$. However, as we have seen, it is not sufficient to use R\'enyi 2-extractors for such constructions (Proposition~\ref{prop:qq_longseed}). So how can we construct quantum min-entropy extractors directly? Having a look at the derivation of min-entropy extractors from R\'enyi 2-extractors (Theorem~\ref{thm:crenyi2_stability} and Theorem~\ref{thm:qrenyi2_stability}), a natural approach could be the following. We write the extractor as a map $\psi:\cP^{+}(\cH_{N})\ra\cP^{+}(\cH_{MD})$ and denote $\tau:\cP^{+}(\cH_{N})\ra\cP^{+}(\cH_{MD})$ with $\tau(\cdot)=\trace[\rho_{N}]\cdot\frac{\id_{M}}{|M|}\ot\frac{\id_{D}}{|D|}$. Here, $\psi$ could take the same form as in the section about quantum extractors (Section~\ref{se:qq}), that is,
\begin{align}
\psi(\rho_{N})=\frac{1}{|D|}\cdot\sum_{i=1}^{|D|}\trace_{N/M}\big[U_{N}^{i}\rho_{N}(U_{N}^{i})^{\dagger}\big]\ot\proj{i}_{D}\ ,
\end{align}
but $\psi$ could also be of a more general form (see the comments about quantum seed extractors at the beginning of Section~\ref{sec:qmin-extractors}). We estimate the 1-norm using Lemma~\ref{lem:1norm},
\begin{align}\label{eq:min_directly}
\|((\psi-\tau)&\ot\cI_{R})(\rho_{NR})\|_{1}=2\cdot\max_{0\leq X\leq\id}\trace\Big[\big((\psi-\tau)\ot\cI_{R}\big)(\rho_{NR})X\Big]\notag\\
&=2\cdot\max_{0\leq X\leq\id}\trace\Big[\big((\psi^{\dagger}-\tau^{\dagger}\big)\ot\cI_{R})(X)\rho_{NR}\Big]\notag\\
&=2\cdot\max_{0\leq X\leq\id}\trace\Big[(\id_{MD}\ot\rho_{R}^{1/2})\rhoh_{NR}(\id_{MD}\ot\rho_{R}^{1/2})\big((\psi^{\dagger}-\tau^{\dagger})\ot\cI_{R}\big)(X)\Big]\notag\\
&=2\cdot\max_{0\leq X\leq\id}\braket{\rhoh_{NR}|((\psi^{\dagger}-\tau^{\dagger})\ot\cI_{R})(X)}_{(\id\ot\sigma)}\notag\\
&\leq2\cdot\max_{0\leq X\leq\id}\|\rhoh_{NR}\|_{\infty,(\id\ot\sigma)}\cdot\|((\psi^{\dagger}-\tau^{\dagger})\ot\cI_{R})(X)\|_{1,(\id\ot\sigma)}\notag\\
&=2\cdot2^{-H_{\min}(N|R)_{\rho}}\cdot\max_{0\leq X\leq\id}\|((\psi^{\dagger}-\tau^{\dagger})\ot\cI_{R})(X)\|_{1,(\id\ot\sigma)}\ ,
\end{align}
where $\sigma_{R}\in\cS(\cH_{R})$, $\rhoh_{NR}$ as in~\eqref{eq:rhoh}, and we have used the $(1,\infty)$-H\"older inequality for the $(\id_{MD}\ot\sigma_{R})$-weighted Hilbert-Schmidt inner product (Lemma~\ref{lem:hoelder}). Now, the real challenge is to bound the second term in~\eqref{eq:min_directly} for given extractor maps $\psi$.

The above derivation is for quantum extractors, but similarly we could analyze quantum-classical extractors as well. Moreover, this approach also seems to be promising for understanding the stability of classical extractors in the presence of quantum side information. It is interesting that if the side information $R$ is trivial, and $M,N$ classical, then this approach corresponds to the connection of extractors to soft-decoding or list-decoding. We note that this connection gave a lot of insight for both research fields, randomness extractors and coding theory, and for references we point to the review article~\cite{Vadhan07,Vadhan11} about the unified classical theory of pseudorandomness. We believe that similar ideas are possible in the quantum case, and this might lead to connections with quantum error correction, and in particular quantum list-decoding~\cite{Leung08_2}. In that spirit we might ask more generally, is it possible to develop a unified quantum theory of pseudorandomness? Notable results in that direction are (balanced) quantum expanders (see, e.g., \cite{Harrow09_2,Ben-Aroya10}), and non-commutative graph theory (see, e.g., \cite{Duan10,Duan11}).


\chapter{Channel Simulations}\label{ch:channels}

In this chapter, all systems (classical and quantum) are finite-dimensional.

The quantification of the information theoretic power of channels is one of the most fundamental problems in information theory, and of particular interest is the study of a channel's capacity for information transmission. This quantity corresponds to the number of bits $m$ that can be sent reliably when using the channel $n$ times with optimal encoding and decoding operations. In 1948, Shannon derived his famous noisy channel coding theorem~\cite{Shannon48}. It shows that the asymptotic iid capacity $C$ of a classical channel $\cG$ is given by the maximum, over the input $X$, of the mutual information between the input $X$ and the output $\cG(X)$, that is,
\begin{align}
C(\cG)=\max_{X}I(X:\cG(X))\ .
\end{align}
Shannon also showed that the capacity does not increase if we allow to use shared randomness between the sender and the receiver.

A different way of understanding the problem of capacities is to think more broadly in terms of channel simulations. For example, the process of sending $m$ bits reliably using $n$ uses of a channel $\cG$ can be understood as a simulation of $m$ perfect, noise-free, channels using $n$ instances of $\cG$. The capacity of the channel $\cG$ is then simply the rate $m/n$ at which such a simulation is possible in the limit of large $n$. However, we can also turn the problem upside down and ask: what is the optimal rate at which a perfect channel can simulate a noisy one?

In 2001 Bennett {\it et al.}~\cite{Bennett02} (see also~\cite{Winter02,Bennett09}) proved the classical reverse Shannon theorem which states that, given free shared randomness between the sender and the receiver, every channel can be simulated using an amount of classical communication equal to the capacity of the channel. This is particularly interesting because it implies that in the presence of free shared randomness, the capacity of a channel $\cG$ to simulate another channel $\cJ$ is given by the ratio of their plain capacities
\begin{align}\label{eq:CR2}
C_{R}(\cG,\cJ)=\frac{C(\cG)}{C(\cJ)}\ ,
\end{align}
and hence only a single parameter suffices to characterize classical channels. We discuss the classical reverse Shannon theorem in Section~\ref{se:shannon}.

Unlike classical channels, quantum channels have various distinct capacities, depending on the kind of information that is sent (e.g., classical or quantum) or on the kind of assistance that is allowed (e.g., free entanglement or free classical communication). Important examples of quantum channel capacities include the entanglement assisted quantum capacity $Q_{E}$~\cite{Bennett02}, and the classical communication assisted quantum capacities $Q_{\ra}$, $Q_{\leftarrow}$ and $Q_{\leftrightarrow}$ (with arrows indicating the direction of the assisting communication)~\cite{Devetak05,Shor02,Lloyd97}.

In~\cite{Bennett02,Bennett99} Bennett {\it et al.}~argue that the entanglement assisted quantum capacity $Q_{E}$ of a quantum channel $\cE_{A\ra B}$ is the natural quantum generalization of the classical capacity of a classical channel. They show that the entanglement assisted quantum capacity is given by the quantum mutual information
\begin{align}\label{eq:QE}
Q_{E}(\cE)=\frac{1}{2}\cdot\max_{\rho}I(B:R)_{(\cE\ot\cI)(\rho)}\ ,
\end{align}
where the maximum is over all input distributions $\rho_{A}$, $\rho_{AR}$ is a purification of $\rho_{A}$, and $\cI_{R}$ is the identity channel on the purifying system. Motivated by this, they conjectured the quantum reverse Shannon theorem in~\cite{Bennett02}. Subsequently Bennett {\it et al.}~proved the theorem in~\cite{Bennett09}, and we were able to give an alternative proof in~\cite{Berta11,Berta11_3}. The theorem states that any quantum channel can be simulated by an unlimited amount of shared entanglement and an amount of quantum communication equal to the channel's entanglement assisted quantum capacity. So if entanglement is free, we can conclude in complete analogy with the classical case, that the capacity of a quantum channel $\cE$ to simulate another quantum channel $\cF$ is given by
\begin{align}\label{eq:QE2}
Q_{E}(\cE,\cF)=\frac{Q_{E}(\cE)}{Q_{E}(\cF)}\ ,
\end{align}
and hence only a single parameter suffices to characterize quantum channels.

Free entanglement in quantum information theory is usually given in the form of maximally entangled states. But for the quantum reverse Shannon theorem it surprisingly turned out that maximally entangled states are not the appropriate resource for general input sources. More precisely, even if we have arbitrarily many maximally entangled states as an entanglement resource, the quantum communication rate in~\eqref{eq:QE} is not achievable~\cite{Bennett09}. This is because of an issue known as entanglement spread, which arises from the fact that entanglement cannot be conditionally discarded without using communication~\cite{Harrow09}. If we change the entanglement resource from maximally entangled states to embezzling states~\cite{vanDam03} however, the problem of entanglement spread can be overcome and~\eqref{eq:QE} becomes achievable. A $\delta$-ebit embezzling state is a bipartite state $\mu_{AB}$ with the feature that the transformation $\mu_{AB}\mapsto\mu_{AB}\ot\phi_{A'B'}$, where $\phi_{A'B'}$ denotes an ebit (maximally entangled state of Schmidt-rank $2$), can be accomplished up to an error $\delta$ with local operations. Remarkably, $\delta$-ebit embezzling states exist for all $\delta>0$~\cite{vanDam03}. We discuss the quantum reverse Shannon theorem in Section~\ref{se:qshannon}.

In~\cite{Berta13}, we studied a scenario that lies in between the classical and the quantum reverse Shannon theorems. We analyzed the simulation cost for quantum-classical channels, i.e., quantum measurements. Our main finding is that a measurement $\cM_{A\ra X}$ can be simulated by an amount of classical communication equal to the quantum mutual information of the measurement, 
\begin{align}\label{eq:MC}
I(\cM)=\max_{\rho}I(X:R)_{(\cM\ot\cI)(\rho)}\ ,
\end{align}
if sufficient shared randomness is available. Comparing this with the fully quantum reverse Shannon theorem, we conclude that if the channel has a classical output, the entanglement (embezzling states) assistance is not necessary, and shared randomness is sufficient. Hence, our result is a generalization of the classical reverse Shannon theorem to quantum-classical channels. Following the measurement compression ideas developed by Winter~\cite{Winter04}, we identify~\eqref{eq:MC} as the information gain of the measurement, independent of the input state on which it is performed. We call this universal measurement compression and discuss it in Section~\ref{se:qshannon}.

In addition to studying special classes of channels, it is also natural to ask how the simulation cost changes in the presence of other free resources. In~\cite{Berta11_2,Berta12}, we considered the simulation of a noisy quantum channel $\cE$ by a noise-free channel in the presence of free classical communication. It turns out not to matter whether there is free classical forward, backward, or even two-way communication --- the capacity is identical in all scenarios. The problem can therefore be understood as the reverse problem for all the classical communication assisted quantum capacities $Q_{\ra}$, $Q_{\leftarrow}$ and $Q_{\leftrightarrow}$. Note that by quantum teleportation~\cite{Bennett93}, the perfect quantum channel can equivalently be replaced with perfect entanglement, and thus the question can be summarized as: at what rate is entanglement, in the form of ebits, needed in order to asymptotically simulate a quantum channel $\cE_{A\ra B}$, when classical communication is given for free? We call this the entanglement cost $E_C$ of a quantum channel and proved the formula~\cite{Berta11_2,Berta12},
\begin{align}\label{eq:EC}
E_{C}(\cE)=\lim_{n\ra\infty}\frac{1}{n}\cdot\max_{\rho^{n}}E_{F}\left(\left(\cE^{\ot n}_{A\ra B}\ot\cI_{R}\right)\left(\rho_{AR}^{n}\right)\right)\ ,
\end{align}
where the maximization is over all purifications $\rho_{AR}^{n}$ of input states to the $n$-fold tensor product quantum channel $\cE^{\ot n}_{A\ra B}$, and $\cI_{R}$ stands for the identity channel on the purifying system. The entanglement of formation $E_{F}$ is computed between the purifying system and the channel output; it is defined as~\cite{Bennett96}
\begin{align}
E_{F}(\rho_{AB})=\inf_{\{p_{i},\rho^{i}\}}\left\{\sum_{i}p_{i}H(A)_{\rho^{i}}\right\}\ ,
\end{align}
where the infimum ranges over all pure state decompositions $\rho_{AB}=\sum_{i}p_{i}\rho^{i}_{AB}$, and $H$ denotes the von Neumann entropy. Note that~\eqref{eq:EC} involves a regularization, and is therefore not a single-letter formula. Even if we would know that we can restrict the maximization to non-entangled input states,~\eqref{eq:EC} would still not reduce to such a formula, due to Hasting's counterexample to the additivity of the entanglement of formation~\cite{Hastings09,Shor04}. Furthermore, $E_{C}$ is generally larger than $Q_\leftrightarrow$,\footnote{The same applies to $Q_\ra$ and $Q_\leftarrow$ since both are smaller or equal to $Q_\leftrightarrow$.} in fact more strikingly, there exist so-called bound entangled channels $\cE$. These are for instance entangling positive partial transpose channels~\cite{Horodecki96,Peres96,Yang05} for which $E_C(\cE)$ is strictly greater than zero, but $Q_\leftrightarrow(\cE)=0$. This fact highlights an important difference to the case of free entanglement, where the corresponding rates are equal due to the quantum reverse Shannon theorem. In particular, when
\begin{align}\label{eq:ECQ2}
E_{C}(\cE)>Q_\leftrightarrow(\cE)\ ,
\end{align}
the concatenated protocol which first simulates $\cE$ from a noiseless channel and then the noiseless channel from $\cE$ will result in a net loss. We discuss the entanglement cost of quantum channels in Section~\ref{se:meas}.

At this point we should mention that channel simulations come in two variants, feedback simulations and non-feedback simulations, and that this distinction becomes important as soon as all the necessary resources are quantified. For channels with classical output, a feedback simulation is one in which the sender and receiver are both required to obtain the channel's output, whereas a non-feedback simulation only provides the receiver with the output. For channels with quantum output it is impossible to give the sender and the receiver the output, due to the no-cloning theorem~\cite{Wootters82,Dieks82}. Following~\cite{Bennett09}, we define a quantum feedback channel as an isometry in which the part of the output that does not go to the receiver is kept by Alice (instead of going to the environment).\footnote{At first, classical and quantum feedback simulations might look like rather different notions. However, classical and quantum feedback behave similarly with respect to the resources needed for achieving the channel simulations.} We note that various optimal trade-offs for achieving feedback and non-feedback channel simulations are nicely outlined in~\cite[Figure 2]{Bennett09}. In this thesis we discuss the following channel simulations.
\begin{itemize}
\item Classical reverse Shannon theorem (Section~\ref{se:shannon}): we give the optimal rate region consisting of the rates of shared randomness and classical communication that are both necessary and sufficient for the existence of classical feedback and non-feedback channel simulations.
\item Quantum reverse Shannon theorem (Section~\ref{se:qshannon}): we characterize the optimal amount of quantum communication that is both necessary and sufficient for the existence of quantum feedback and non-feedback channel simulations in the presence of free entanglement.
\item Universal measurement compression (Section~\ref{se:meas}): we characterize the optimal rate region consisting of the rates of shared randomness and classical communication that are both necessary and sufficient for the existence of feedback and non-feedback measurement simulations.
\item Entanglement cost of quantum channels (Section~\ref{se:cost}): we characterize the optimal amount of entanglement that is both necessary and sufficient for the existence of quantum feedback and non-feedback channel simulations in the presence of free classical communication.
\end{itemize}

At first, it might seem that channel simulations are mainly interesting for purely information theoretic reasons, since statements like~\eqref{eq:CR2},~\eqref{eq:QE2}, or~\eqref{eq:ECQ2} allow for insights into the structure of classical and quantum channels. But in addition to this, channel simulations also give upper bounds on so-called strong converse capacities. A strong converse capacity is the minimal rate above which any attempt to send information necessarily has exponentially small fidelity (potentially in the presence of some fixed free resources).\footnote{The strong converse capacity is greater than or equal to the standard capacity, which is defined as the minimal rate above which the fidelity does not approach one.} The concept of a strong converse capacity is appealing since it really gives a sharp threshold for information transmission. The classical reverse Shannon theorem shows that the classical capacity is a strong capacity~\cite{Bennett02} (a fact first proven by Wolfowitz~\cite{Wolfowitz64}), and the quantum reverse Shannon theorem shows that the entanglement assisted quantum capacity is a strong capacity~\cite{Bennett09}. However, determining other strong converse capacities for quantum channels turns out to be an elusive problem and only partial results are known~\cite{Koenig09_2,Winter99,Ogawa99,Bennett09,Morgan13,Sharma13}. In~\cite{Berta11_2,Berta12}, we showed that the entanglement cost of quantum channels gives an upper bound on the strong converse two-way classical communication assisted quantum capacity,\footnote{Of course, this is then also an upper bound on the strong converse forward classical communication assisted, backward communication assisted, and plain quantum capacity.} and we discuss this in Section~\ref{se:app}. Strong converse capacities also have applications in quantum cryptography and we mention this in Section~\ref{se:app} (see also Section~\ref{se:qc}). Last but not least, channel simulations lead to protocols for classical and quantum rate distortion coding (lossy data compression)~\cite{Winter02,Datta12,Wilde12_3,Datta13}. This, however, we do not explore in this work.

The proofs of our channel simulation results are based on the use of classical and quantum randomness extractors against classical and quantum side information (as discussed in Chapter~\ref{ch:extractors}), as well as the post-selection technique for quantum channels. The technique was introduced in~\cite{Christandl09} and is a tool to estimate the closeness of channels that act symmetrically on an $n$-partite system, in the metric induced by the diamond norm (Definition~\ref{def:diamond}). The definition of this norm involves a maximization over all possible inputs to the joint mapping consisting of the channel tensored with an identity map on an outside system. The post-selection technique allows to drop this maximization. In fact, it suffices to consider a single de Finetti type input state, i.e., a state which consist of $n$ identical and independent copies of an (unknown) state on a single subsystem. The technique was originally applied in quantum cryptography to show that security of discrete-variable quantum key distribution against a restricted type of attacks, called collective attacks, already implies security against coherent attacks~\cite{Christandl09}.

For a detailed description of the proof ideas, we point to Sections~\ref{se:shannon}-\ref{se:cost}.


\section{Classical Reverse Shannon Theorem}\label{se:shannon}

The ideas in this section have been obtained in collaboration with Matthias Christandl, Joseph Renes, and Renato Renner, but are unpublished. The classical reverse Shannon theorem states that every classical channel can be simulated using an amount of classical communication equal to the channel's capacity, if sufficient shared randomness is available. The theorem was first proven by Bennett {\it et al.}~\cite{Bennett02}, and later improved to get the optimal shared randomness consumption rate~\cite{Winter02,Cuff08,Bennett09}. Here, we sketch these results in order to motivate the quantum generalizations that we present afterwards (Sections~\ref{se:qshannon}-\ref{se:cost}). As we will see, however, the classical reverse Shannon theorem is a strict special case of universal measurement compression (as discussed in Section~\ref{se:meas}), and therefore, we do not present any proofs in this section.

The classical reverse Shannon theorem is as follows (taken from~\cite{Bennett09}).

\begin{theorem}\label{thm:rst}
Consider a bipartite system with classical parties Alice and Bob. Let $\cG$ be an iid classical channel with input $Y$ (a random variable) at Alice's side, and output $X=\cG(Y)$ at Bob's side. Then, $\cG$ can be asymptotically simulated for all input distributions $P$ if and only if classical communication rate $C$ and shared randomness rate $S$ lie in the rate region given by the union of the following regions,
\begin{align}
C&\geq\max_{P}I(W:Y)\ ,\\
C+S&\geq\max_{P}I(W:XY)\ ,
\end{align}
where the union is with respect to all random variables $W$ such that $X\,-\,W\,-\,Y$ form a Markov chain. Furthermore, let $\cG_{F}$ denote the feedback version of $\cG$, which gives Alice a copy of Bob's output $X=\cG(Y)$ as well. Then, $\cG_{F}$ can be asymptotically simulated for all input distributions $P$ if and only if the classical communication rate $C$ and the shared randomness rate $S$ lie in the following rate region\footnote{Note that the two maxima in~\eqref{eq:mainrsta} and~\eqref{eq:mainrstb} can be achieved for different inputs.}
\begin{align}
C&\geq\max_{P}I(X:Y)\label{eq:mainrsta}\ ,\\
C+S&\geq\max_{P}H(X)\label{eq:mainrstb}\ .
\end{align}
\end{theorem}

We can prove the achievability part of the classical reverse Shannon theorem as presented above by using an approach based on strong classical min-entropy extractors against classical side information (as discussed in Section~\ref{se:cc}) and a classicalized version of the post-selection technique (Proposition~\ref{prop:postselect}). The basic idea is that Alice locally simulates the channel, and then uses an optimal amount of classical communication and shared randomness to bring the channel's output to Bob (state splitting). Essentially, this is done by first employing a classical extractor in order to extract all the randomness from the correlations between the input and the output of the channel. This randomness then does not need to get sent to Bob, but can just be created using shared randomness. Finally, the post-selection technique shows that this can be done universally, i.e., that this works asymptotically for any possible input state and rates as given in Theorem~\ref{thm:rst}.

In order to pursue this approach for the quantum generalizations of the classical reverse Shannon theorem, we need strong classical min-entropy extractors against quantum side information (as discussed in Section~\ref{se:cc}) for universal measurement compression in Section~\ref{se:meas}, as well as strong quantum min-entropy extractors against quantum side information (as discussed in Section~\ref{se:qq}) for the quantum reverse Shannon theorem in Section~\ref{se:qshannon} and the entanglement cost of quantum channels in Section~\ref{se:cost}. In the rest of this chapter we will discuss all the necessary technical steps in detail.


\section{Quantum Reverse Shannon Theorem}\label{se:qshannon}

The results in this section have been obtained in collaboration with Matthias Christandl and Renato Renner, and have appeared in~\cite{Berta11,Berta11_3}. The quantum reverse Shannon theorem was conjectured by Bennett {\it et al.}~in~\cite{Bennett02}, and then first proven by Bennett {\it et al.}~in~\cite{Bennett09}. The theorem states that any quantum channel can be simulated by an unlimited amount of shared entanglement and an amount of quantum communication equal to the channel's entanglement assisted quantum capacity.


\subsection{Proof Ideas}\label{sec:qshannon_ideas}

Here we present a proof of the quantum reverse Shannon theorem based on one-shot information theory. In quantum information theory one usually makes the assumption that the resources are iid and is interested in asymptotic rates. In this case many operational quantities can be expressed in terms of a few information measures, which are usually based on the von Neumann entropy. In contrast to this, one-shot information theory applies to arbitrary (structureless) resources. For example, in the context of source coding, it is possible to analyze scenarios where only finitely many, possibly correlated messaged are encoded. For this the smooth entropy formalism was introduced by Renner {\it et al.}~\cite{Renner05,Renner04,Renner05_2}. Smooth entropy measures have properties similar to the ones of the von Neumann entropy and like in the iid case many operational interpretations are known (see~\cite{Renner05,Tomamichel12} and references therein).

For our proof of the quantum reverse Shannon theorem we work in this smooth entropy formalism and use a one-shot version for quantum state merging and its dual quantum state splitting as well as the post-selection technique for quantum channels (see Appendix~\ref{ap:postselect}). As in the original proof of the quantum reverse Shannon theorem~\cite{Bennett09} we need embezzling states~\cite{vanDam03}.
  
Quantum state merging was introduced by Horodecki {\it et al.}~in~\cite{Horodecki05,Horodecki07}. It has since become an important tool in quantum information processing and was subsequently reformulated in~\cite{Abeyesinghe09}, where it is called the mother protocol. Quantum state merging corresponds to the quantum generalization of classical Slepian and Wolf coding~\cite{Slepian73}. For its description, we consider a sender system, traditionally called Alice, a receiver system, Bob, as well as a reference system $R$. In quantum state merging, Alice, Bob, and the reference are initially in a joint pure state $\rho_{ABR}$ and we ask how much of a given resource, such as classical or quantum communication or entanglement, is needed in order to move the $A$-part of $\rho_{ABR}$ from Alice to Bob. The dual of this, called quantum state splitting, addresses the problem of how much of a given resource, such as classical or quantum communication or entanglement, is needed in order to transfer the $A'$-part of a pure state $\rho_{AA'R}$, where part $AA'$ is initially with Alice, from Alice to Bob.

Our proof of the quantum reverse Shannon theorem is based on the following idea. Let $\cE_{A\ra B}$ be a quantum channel that takes inputs $\rho_{A}$ on Alice's side and outputs $\cE_{A\ra B}(\rho_{A})=\rho_{B}$ on Bob's side. To find a way to simulate this quantum channel, it is useful to think of $\cE_{A\ra B}$ as
\begin{align}
\cE_{A\ra B}(\rho_{A})=\trace_{E}\left[U_{A\ra BE}(\rho_{A})\right]\ ,
\end{align}
where $E$ is an additional register (sometimes called the environment), and $U_{A\ra BE}$ is some isometry from $A$ to $BE$. This is the Stinespring dilation~\cite{Stinespring55}. Now the idea is to first simulate the isometry $U_{A\ra BE}$ locally at Alice's side, resulting in $\rho_{BE}=U_{A\ra BE}(\rho_{A})$, and in a second step use quantum state splitting to do an optimal state transfer of the $B$-part to Bob's side, such that he holds $\rho_{B}=\cE_{A\ra B}(\rho_{A})$ in the end. This simulates the channel $\cE_{A\ra B}$. To prove the quantum reverse Shannon theorem, it is then sufficient to show that the quantum communication rate of the quantum state splitting protocol is $Q_{E}(\cE)$.

We realize this idea in two steps. Firstly, we propose a new version of quantum state splitting (since the known protocols are not good enough to achieve a classical communication rate of $Q_{E}(\cE)$), which is based on one-shot quantum state merging~\cite{Berta08}. For the analysis we require a strong quantum min-entropy randomness extractor against quantum side information (as discussed in Section~\ref{se:qq}). Secondly, we use the post-selection technique to show that our protocol for quantum state splitting is sufficient to asymptotically simulate the channel $\cE_{A\ra B}$ with a quantum communication rate of $Q_{E}(\cE)$. This then completes the proof of the quantum reverse Shannon theorem.


\subsection{Quantum State Merging and State Splitting}\label{sec:qshannon_splitting}

The main goal of this section is to prove that there exists a tight one-shot quantum state splitting protocol (Theorem~\ref{thm:qssemb}), that is optimal in terms of its quantum communication cost (Theorem~\ref{thm:qssconv}). The protocol is obtained by inverting a one-shot quantum state merging protocol. The main technical ingredient for the construction of these protocols are strong quantum min-entropy extractors against quantum side information. Since we only care about the optimal output size of the extractor and not about its seed size, we just use a unitary 2-design (Defintion~\ref{def:unitary2design}). Corollary~\ref{cor:2design} about the extractor properties of unitary 2-designs can then be rephrased as follows.

\begin{corollary}\label{cor:decoupling}
Let $\eps>0$, $\rho_{AR}\in\cS_{\leq}(\cH_{AR})$, and $\{U_{A}^{j}\}_{j\in J}$ be a unitary 2-design. Furthermore, consider a decomposition of the system $A$ into two subsystems $A_{1}$ and $A_{2}$, and define
\begin{align}
\sigma_{A_{1}R}^{j}=\trace_{A_{2}}\left[(U_{A}^{j}\ot\1_{R})\rho_{AR} (U_{A}^{j}\ot\1_{R})^{\dagger}\right]\ .
\end{align}
If we choose
\begin{align}\label{eq:decoupling}
\log |A_{1}|\leq\frac{\log|A|+H_{\min}(A|R)_\rho}{2}-\log\frac{1}{\eps}\ ,
\end{align}
then we have that
\begin{align}
\frac{1}{|J|}\cdot\sum_{j\in J}\left\|\sigma_{A_{1}R}^{j}-\frac{\1_{A_{1}}}{|A_{1}|}\ot\rho_{R}\right\|_{1}\leq\eps\ .
\end{align}
\end{corollary}

An excellent introduction about strong quantum extractors against quantum side information in quantum coding theory (from an asymptotic iid perspective) can be found in~\cite{Hayden11}. Quantum state merging, quantum state splitting, and other related quantum information processing primitives are discussed in detail in~\cite{Horodecki05,Horodecki07,Abeyesinghe09,Oppenheim08,Berta08}. Note that we are interested not only in asymptotic rates, but in (tight) one-shot protocols. This is reflected by the following definitions.

\begin{definition}[Quantum state merging]\label{def:qsm}
Consider a bipartite system with parties Alice and Bob. Let $\eps>0$, and $\rho_{ABR}\in\cV_{\leq}(\cH_{ABR})$, where Alice controls $A$, Bob $B$ and $R$ is a reference system. A quantum protocol $\cE$ is called $\eps$-error quantum state merging of $\rho_{ABR}$ if it consists of applying local operations at Alice's side, sending $q$ qubits from Alice to Bob, applying local operations at Bob's side, and outputs a state $\omega_{B'BRA_{1}B_{1}}=(\cE\ot\cI_{R})(\rho_{ABR})$ with
\begin{align}
\omega_{B'BRA_{1}B_{1}}\approx_{\eps}(\cI_{A\ra B'}\ot\cI_{BR})(\rho_{ABR})\ot\Phi^{L}_{A_{1}B_{1}}\ ,
\end{align}
where $\Phi^{L}_{A_{1}B_{1}}$ is a maximally entangled state of Schmidt-rank $L$. The quantity $q$ is called quantum communication cost, and $e=\lfloor\log L\rfloor$ entanglement gain.
\end{definition}

Quantum state merging is also called fully quantum Slepian Wolf or mother protocol~\cite{Abeyesinghe09}.

\begin{lemma}\label{lem:qsm}
Let $\eps>0$, and $\rho_{ABR}\in\cV_{\leq}(\cH_{ABR})$. Then, there exists an $\eps$-error quantum state merging protocol for $\rho_{ABR}$ with quantum communication cost
\begin{align}
q=\left\lceil\frac{1}{2}\cdot\big(H_{0}(A)_{\rho}-H_{\min}(A|R)_{\rho}\big)+2\log\frac{1}{\eps}\right\rceil\ ,
\end{align}
and entanglement gain
\begin{align}
e=\left\lfloor\frac{1}{2}\cdot\big(H_{0}(A)_{\rho}+H_{\min}(A|R)_{\rho}\big)-2\log\frac{1}{\eps}\right\rfloor\ .
\end{align}
\end{lemma}

\begin{proof}
The intuition is as follows (see~Figure~\ref{fig:qsmss}). First Alice applies a unitary $U_{A\ra A_{1}A_{2}}$. After this she sends $A_{2}$ to Bob who then performs a local isometry $V_{A_{2}B\ra B'BB_{1}}$. We choose $U_{A\ra A_{1}A_{2}}$ such that it randomizes $A_{1}$ and decouples it from the reference $R$. After sending the $A_{2}$-part to Bob, the state on $A_{1}R$ is given by $\frac{\1_{A_{1}}}{|A_{1}|}\ot\rho_{R}$, and Bob holds a purification of this. But $\frac{\1_{A_{1}}}{|A_{1}|}\ot\rho_{R}$ is the reduced state of $\rho_{B'BR}\ot\Phi^{L}_{A_{1}B_{1}}$, and since all purifications are equal up to local isometries, there exists an isometry $V_{A_{2}B\ra B'BB_{1}}$ on Bob's side that transforms the state into $\rho_{B'BR}\ot\Phi^{L}_{A_{1}B_{1}}$.

More formally, let $A=A_{1}A_{2}$ with
\begin{align}
\log|A_{2}|=\left\lceil\frac{1}{2}\cdot\big(\log|A|-H_{\min}(A|R)_{\rho}\big)+2\log\frac{1}{\eps}\right\rceil\ .
\end{align}
According to our results about strong quantum min-entropy extractors against quantum side information in Section~\ref{se:qq} (cf.~Corollary~\ref{cor:decoupling}), there exists a unitary $U_{A\ra A_{1}A_{2}}$ such that we have for
\begin{align}
\sigma_{A_{1}A_{2}BR}=U_{A\ra A_{1}A_{2}}(\rho_{ABR})
\end{align}
that
\begin{align}
\left\|\sigma_{A_{1}R}-\frac{\1_{A_{1}}}{|A_{1}|}\ot\rho_{R}\right\|_{1}\leq\eps^{2}\ .
\end{align}
By an upper bound of the purified distance in terms of the trace distance (Lemma~\ref{lem:pdbounds}) this implies $\sigma_{A_{1}R}\approx_{\eps}\frac{\1_{A_{1}}}{|A_{1}|}\ot\rho_{R}$. We apply this unitary $U_{A\ra A_{1}A_{2}}$, and then send $A_{2}$ to Bob; therefore
\begin{align}
q=\left\lceil\frac{1}{2}\cdot\big(\log|A|-H_{\min}(A|R)_{\rho}\big)+2\log\frac{1}{\eps}\right\rceil\ .
\end{align}
Uhlmann's theorem~\cite{Uhlmann76,Jozsa94} tells us that there exists an isometry $V_{A_{2}B\ra B'BB_{1}}$ with
\begin{align}
P\big(\sigma_{A_{1}R},\frac{\1_{A_{1}}}{|A_{1}|}\ot\rho_{R}\big)=P\big(V_{A_{2}B\ra B'BB_{1}}(\sigma_{A_{1}A_{2}BR}),\Phi^{L}_{A_{1}B_{1}}\ot\rho_{B'BR}\big)\ .
\end{align}
Hence, the entanglement gain is given by
\begin{align}
e=\left\lfloor\frac{1}{2}\cdot\big(\log|A|+H_{\min}(A|R)_{\rho}\big)-2\log\frac{1}{\eps}\right\rfloor\ .
\end{align}
Now if $\rho_{A}$ has full rank this is already what we want. In general $\log\trace\left[\rho_{A}^{0}\right]=\log|\hat{A}|\leq\log|A|$. But in this case we can restrict the Hilbert space $\cH_{A}$ to the subspace $\cH_{\hat{A}}$, on which $\rho_{A}$ has full rank.
\end{proof}

\begin{definition}[Quantum state splitting with maximally entangled states]\label{def:qss}
Consider a bipartite system with parties Alice and Bob. Let $\eps>0$, and $\rho_{AA'R}\in\cV_{\leq}(\cH_{AA'R})$, where Alice controls $AA'$ and $R$ is a reference system. Furthermore, let $\Phi^{K}_{A_{1}B_{1}}$ be a maximally entangled state of Schmidt-rank $K$ between Alice and Bob. A quantum protocol $\cE$ is called $\eps$-error quantum state splitting of $\rho_{AA'R}$ with maximally entangled states if it consists of applying local operations at Alice's side, sending $q$ qubits from Alice to Bob, applying local operations at Bob's side, and outputs a state $\omega_{ABR}=(\cE\ot\cI_{R})(\rho_{AA'R}\ot\Phi^{K}_{A_{1}B_{1}})$ with
\begin{align}
\omega_{ABR}\approx_{\eps}(\cI_{A'\ra B}\ot\cI_{AR})(\rho_{AA'R})\ .
\end{align}
The quantity $q$ is called quantum communication cost, and $e=\lfloor\log K\rfloor$ entanglement cost.
\end{definition}

This is also called the fully quantum reverse Shannon protocol~\cite{Abeyesinghe09}. Quantum state splitting with maximally entangled states is dual to quantum state merging in the sense that every quantum state merging protocol already defines a protocol for quantum state splitting with maximally entangled states, and vice versa.

\begin{lemma}\label{lem:qss}
Let $\eps>0$, and $\rho_{AA'R}\in\cV_{\leq}(\cH_{AA'R})$. Then, there exists an $\eps$-error quantum state splitting protocol with maximally entangled states for $\rho_{AA'R}$ with quantum communication cost
\begin{align}\label{eq:qqss}
q=\left\lceil\frac{1}{2}\left(H_{0}(A')_{\rho}-H_{\min}(A'|R)_{\rho}\right)+2\log\frac{1}{\eps}\right\rceil\ ,
\end{align}
and entanglement cost
\begin{align}
e=\left\lfloor\frac{1}{2}(H_{0}(A')_{\rho}+H_{\min}(A'|R)_{\rho})-2\log\frac{1}{\eps}\right\rfloor\ .
\end{align}
\end{lemma}

\begin{proof}
The quantum state splitting protocol with maximally entangled states is defined by running the quantum state merging protocol of Lemma~\ref{lem:qsm} backwards (see Figure~\ref{fig:qsmss}). The claim then follows from Lemma~\ref{lem:qsm}.
\end{proof}

\begin{figure}\label{fig:qsmss}
\begin{center}
\includegraphics[width=1.0\linewidth]{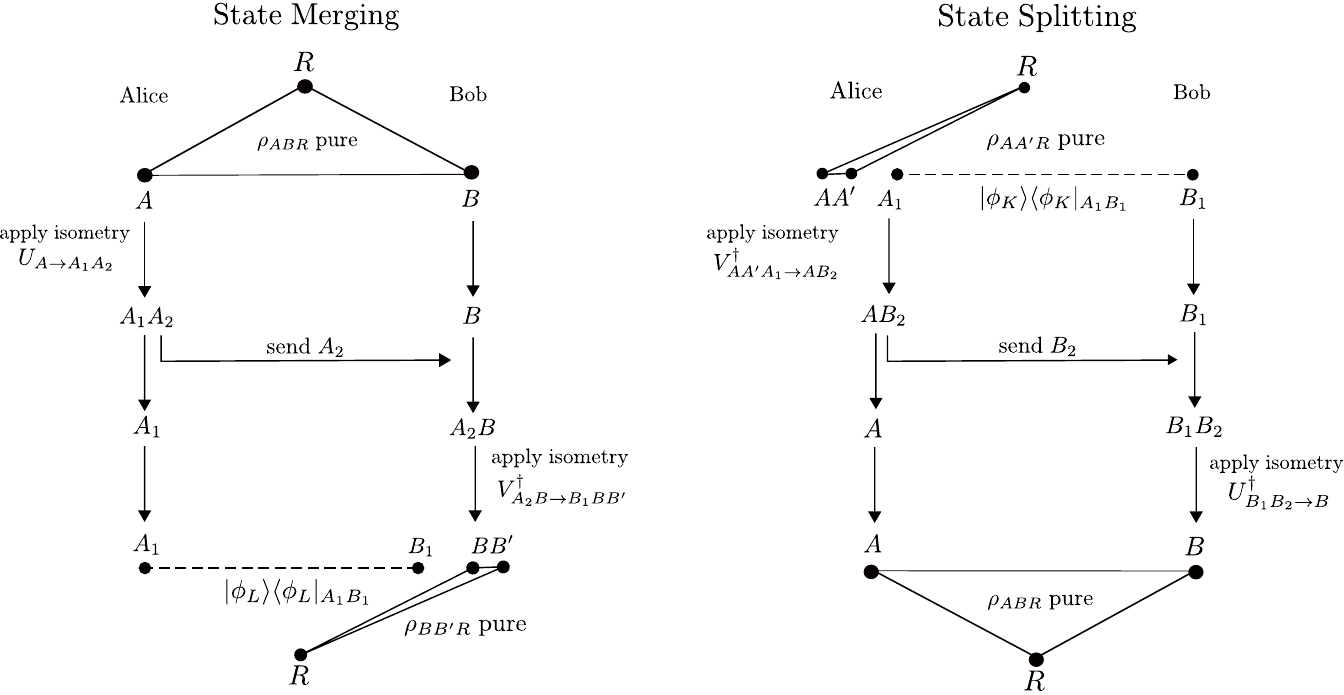}
\end{center}
\caption{Quantum state merging and quantum state splitting.}{From the quantum state merging protocol in Lemma~\ref{lem:qsm}, we already get a protocol for quantum state splitting with maximally entangled states. All we have to do is to time-reverse the protocol and exchange the roles of Alice and Bob.}
\end{figure}

In order to obtain a quantum state splitting protocol that is tight in terms of its quantum communication cost,\footnote{To see that~\eqref{eq:qqss} can be very far off, consider the following classical example. Let $R$ be trivial, and let the probability distribution on $A'$ be $1/2$ at the first position and a flat distribution at all the other $m-1$ positions. Then, we have $H_{0}(A')_{\rho}-H_{\min}(A')_{\rho}\ra\infty$ for $m\ra\infty$, even though there is actually no communication needed. This example is also stable under smoothing.} we need to replace the maximally entangled states by embezzling states~\cite{vanDam03}.

\begin{definition}\label{def:emb}
Let $\delta>0$. $\mu_{AB}\in\cS(\cH_{AB})$ is called $\delta$-ebit embezzling state if there exist isometries $X_{A\ra AA'}$ and $X_{B\ra BB'}$ such that
 \begin{align}
(X_{A\ra AA'}\ot X_{B\ra BB'})(\mu_{AB})\approx_{\delta}\mu_{AB}\ot\Phi_{A'B'}\ ,
\end{align}
where $\Phi_{A'B'}\in\cS(\cH_{A'B'})$ denotes an ebit (maximally entangled state of Schmidt rank $2$).
\end{definition}

\begin{proposition}\cite{vanDam03}
$\delta$-ebit embezzling states exist for all $\delta>0$.
\end{proposition}

We would like to highlight two interesting examples. For the first example consider the state $\proj{\mu_{m}}_{A^{m}B^{m}}\in\cV(\cH_{A^{m}B^{m}})$ defined by
\begin{align}
\ket{\mu_{m}}_{A^{m}B^{m}}=C\cdot\sum_{j=0}^{m-1}\ket{\varphi}_{AB}^{\ot j}\ot\ket{\Phi}_{AB}^{\ot(m-j)}\ ,
\end{align}
where $\proj{\varphi}_{AB}\in\cV(\cH_{AB})$, $\proj{\Phi}_{AB}\in\cV(\cH_{AB})$ denotes an ebit, and $C$ is such that $\proj{\mu_{m}}_{A^{m}B^{m}}$ is normalized. Note that applying the cyclic shift operator $U_{A_{0}A^{m}}$ that sends $A_{i}\ra A_{i+1}$ at Alice's side (modulo $m+1$), and the corresponding cyclic shift operator $U_{B_{0}B^{m}}$ at Bob's side, maps $\ket{\varphi}_{A_{0}B_{0}}\ot\ket{\mu_{m}}_{A^{m}B^{m}}$ to $\ket{\Phi}_{A_{0}B_{0}}\ot\ket{\mu_{m}}_{A^{m}B^{m}}$ up to an accuracy of $\sqrt{\frac{2}{m}}$~\cite{Leung08}. For the choice $\ket{\varphi}_{AB}=\ket{\varphi}_{A}\ot\ket{\varphi}_{B}$, $\proj{\mu_{m}}_{A^{m}B^{m}}$ is a  $\sqrt{\frac{2}{m}}$-ebit embezzling state with the isometries $X_{A^{m}\ra A_{0}A^{m}}=U_{A_{0}A^{m}}\ket{\varphi}_{A_{0}}$ and $X_{B^{m}\ra B_{0}B^{m}}=U_{B_{0}B^{m}}\ket{\varphi}_{B_{0}}$.

The second example is the state $\proj{\tilde{\mu}^{m}}_{AB}\in\cV(\cH_{AB})$ defined by
\begin{align}
\ket{\tilde{\mu}_{m}}_{AB}=\Big(\sum_{j=1}^{2^{m}}\frac{1}{j}\Big)^{-1/2}\cdot\sum_{j=1}^{2^{m}}\frac{1}{\sqrt{j}}\cdot\ket{j}_{A}\ot\ket{j}_{B}\ .
\end{align}
It is a $\sqrt{\frac{2}{m}}$-ebit embezzling state~\cite{vanDam03}.\footnote{The state $\proj{\tilde{\mu}_{m}}_{AB}$ is even a $(\sqrt{\frac{2}{m}},r)$-universal embezzling state~\cite{vanDam03}. That is, for any $\proj{\varsigma}_{A'B'}\in\cV(\cH_{A'B'})$ of Schmidt-rank at most $r$, there exist isometries $X_{A\ra AA'}$ and $X_{B\ra BB'}$ such that
 \begin{align}
(X_{A\ra AA'}\ot X_{B\ra BB'})(\proj{\tilde{\mu}_{m}}_{AB})\approx_{\sqrt{\frac{2}{m}}}\proj{\tilde{\mu}_{m}}_{AB}\ot\proj{\varsigma}_{A'B'}\ .
\end{align}}

\begin{remark}\label{rmk:emb}
By using $\delta$-ebit embezzling states multiple times, it is possible to create maximally entangled states of higher dimension. More precisely, for every $\delta$-ebit embezzling state $\mu_{AB}\in\cS(\cH_{AB})$ there exist isometries $X_{A\ra AA'}$ and $X_{B\ra BB'}$ such that
 \begin{align}
(X_{A\ra AA'}\ot X_{B\ra BB'})(\mu_{AB})\approx_{\delta\cdot\log L}\mu_{AB}\ot\Phi^{L}_{A'B'}\ ,
\end{align}
where $\Phi^{L}_{A'B'}\in\cV(\cH_{A'B'})$ denotes a maximally entangled state of Schmidt-rank $L$ (with $L$ being a power of $2$).
\end{remark}

We are now ready to define quantum state splitting with embezzling states.

\begin{definition}[Quantum state splitting with embezzling states]\label{def:qssemb}
Consider a bipartite system with parties Alice and Bob. Let $\eps>0$, $\delta>0$, and $\rho_{AA'R}\in\cV_{\leq}(\cH_{AA'R})$, where Alice controls $AA'$ and $R$ is a reference system. A quantum protocol $\cE$ is called $\eps$-error quantum state splitting of $\rho_{AA'R}$ with a $\delta$-ebit embezzling state if it consists of applying local operations at Alice's side, sending $q$ qubits from Alice to Bob, applying local operations at Bob's side, using a $\delta$-ebit embezzling state $\mu_{A_{\emb}B_{\emb}}$, and outputs a state $\omega_{ABR}=(\cE\ot\cI_{R})(\rho_{AA'R}\ot\mu_{A_{\emb}B_{\emb}})$ with
\begin{align}
\omega_{ABR}\approx_{\eps}(\cI_{A'\ra B}\ot\cI_{AR})(\rho_{AA'R})\ .
\end{align}
The quantity $q$ is called quantum communication cost.
\end{definition}

The following theorem about the achievability of quantum state splitting with embezzling states (Theorem~\ref{thm:qssemb}) is the main result of this section. In Section~\ref{sec:qshannon_main} we use this theorem to prove the quantum reverse Shannon theorem.

\begin{theorem}\label{thm:qssemb}
Let $\eps>0$, $\eps'\geq0$, $\delta>0$, and $\rho_{AA'R}\in\cV_{\leq}(\cH_{AA'R})$. Then, there exists an $(\eps+\eps'+\delta\cdot\log|A'|+|A'|^{-1/2})$-error\footnote{The error term $|A'|^{-1/2}$ can be made arbitrarily small by enlarging the Hilbert space $\cH_{A'}$. Of course this increases the error term $\delta\cdot\log|A'|$, but this can again be compensated with decreasing $\delta$. Enlarging the Hilbert space $\cH_{A'}$ also increases the quantum communication cost~\eqref{eq:qssemb}, but only slightly.} quantum state splitting protocol for $\rho_{AA'R}$ with a $\delta$-ebit embezzling state for a quantum communication cost
\begin{align}\label{eq:qssemb}
q\leq\frac{1}{2}I^{\eps'}_{\max}(A':R)_{\rho}+2\log\frac{1}{\eps}+4+\log\log|A'|\ .
\end{align}
\end{theorem}

\begin{proof}
The idea for the protocol is as follows (cf.~Figure~\ref{fig:qssemb}). First, we disregard the eigenvalues of $\rho_{A'}$ that are smaller then $|A|^{-2}$. This introduces an error $\alpha=|A|^{-1/2}$, but because of the monotonicity of the purified distance (Lemma~\ref{lem:pdmono}), the error at the end of the protocol is still upper bounded by the same $\alpha$. As a next step we let Alice perform a coherent measurement $\cW$ with roughly $2\cdot\log|A|$ measurement outcomes in the eigenbasis of $\rho_{A'}$. That is, the state after the measurement is of the form $\omega_{AA'RI_{A}}=\proj{\omega}_{AA'RI_{A}}$ with
\begin{align}
\ket{\omega}_{AA'RI_{A}}=\sum_{i\in I}\sqrt{p_{i}}\cdot\ket{\rho^{i}}_{AA'R}\ot\ket{i}_{I_{A}}\ .
\end{align}
Here the index $i$ indicates which measurement outcome occurred, $p_{i}$ denotes its probability and $\rho^{i}_{AA'R}=\proj{\rho^{i}}_{AA'R}$ the corresponding post-measurement state. Then, conditioned on the index $i$, we use the quantum state splitting protocol with maximally entangled states from Lemma~\ref{lem:qss} for each state $\rho^{i}_{AA'R}$ and denote the corresponding quantum communication cost and entanglement cost by $q_{i}$ and $e_{i}$, respectively. The total amount of quantum communication we need for this is given by $\max_{i}q_{i}$ plus the amount needed to send the register $I_{A}$ (which is of order $\log\log|A|$). In addition, since the different branches of the protocol use different amounts of entanglement, we need to provide a superposition of different (namely $e_{i}$ sized) maximally entangled states. We do this by using embezzling states.\footnote{Note that it is not possible to get such a superposition starting from any amount of maximally entangled states only using local operations. This problem is known as entanglement spread and is discussed in~\cite{Harrow09}.} As the last step, we undo the initial coherent measurement $\cW$. This completes the quantum state splitting protocol with embezzling states for $\rho_{AA'R}$. All that remains to do is to bring the expression for the quantum communication cost in the right form. In the following, we describe the proof in detail.

Let $Q=\lceil2\cdot\log|A'|-1\rceil$, $I=\{0,1,\ldots,Q,(Q+1)\}$, and let $\{P_{A'}^{i}\}_{i\in I}$ be a collection of projectors on $\cH_{A'}$ defined as follows: $P_{A'}^{Q+1}$ projects on the eigenvalues of $\rho_{A'}$ in $[2^{-2\log|A'|},0]$, $P_{A'}^{Q}$ projects on the eigenvalues of $\rho_{A'}$ in $[2^{-Q},2^{-2\log|A'|}]$, and for $i=0,1,\dots,(Q-1)$, $P_{A'}^{i}$ projects on the eigenvalues of $\rho_{A'}$ in $[2^{-i},2^{-(i+1)}]$. Furthermore, let $p_{i}=\trace\left[P_{A'}^{i}\rho_{A'}\right]$, $\rho_{AA'R}^{i}=\proj{\rho^{i}}_{AA'R}$ with $\ket{\rho^{i}}_{AA'R}=p_{i}^{-1/2}\cdot P_{A'}^{i}\ket{\rho}_{AA'R}$ and define the state $\rhob_{AA'R}=\proj{\rhob}_{AA'R}$ with
\begin{align}
\ket{\rhob}_{AA'R}=\Upsilon^{-1/2}\cdot\sum_{i=0}^{Q}\sqrt{p_{i}}\cdot\ket{\rho^{i}}_{AA'R}\ ,
\end{align}
where $\Upsilon=\sum_{i=0}^{Q}p_{i}$. We have
\begin{align}\label{eq:barstate}
\rhob_{AA'R}\approx_{|A'|^{-1/2}}\rho_{AA'R}\ ,
\end{align}
as can be seen as follows. We have
\begin{align}
P(\rhob_{AA'R},\rho_{AA'R})&=\sqrt{1-F^{2}(\rhob_{AA'R},\rho_{AA'R})}=\sqrt{1-|\ip{\rhob,\rho}_{AA'R}|^{2}}\notag\\
&=\sqrt{1-\sum_{i=0}^{Q}p_{i}}=\sqrt{p_{Q+1}}\ ,
\end{align}
but because at most $|A'|$ eigenvalues of $\rho_{A'}$ can lie in $[2^{-2\log|A'|},0]$, each one smaller or equal to $2^{-2\log|A'|}$, we obtain $p_{Q+1}\leq|A'|\cdot2^{-2\log|A'|}=|A'|^{-1}$, and hence~\eqref{eq:barstate} follows.

\begin{figure}\label{fig:qssemb}
\begin{center}
\includegraphics[width=1.0\linewidth]{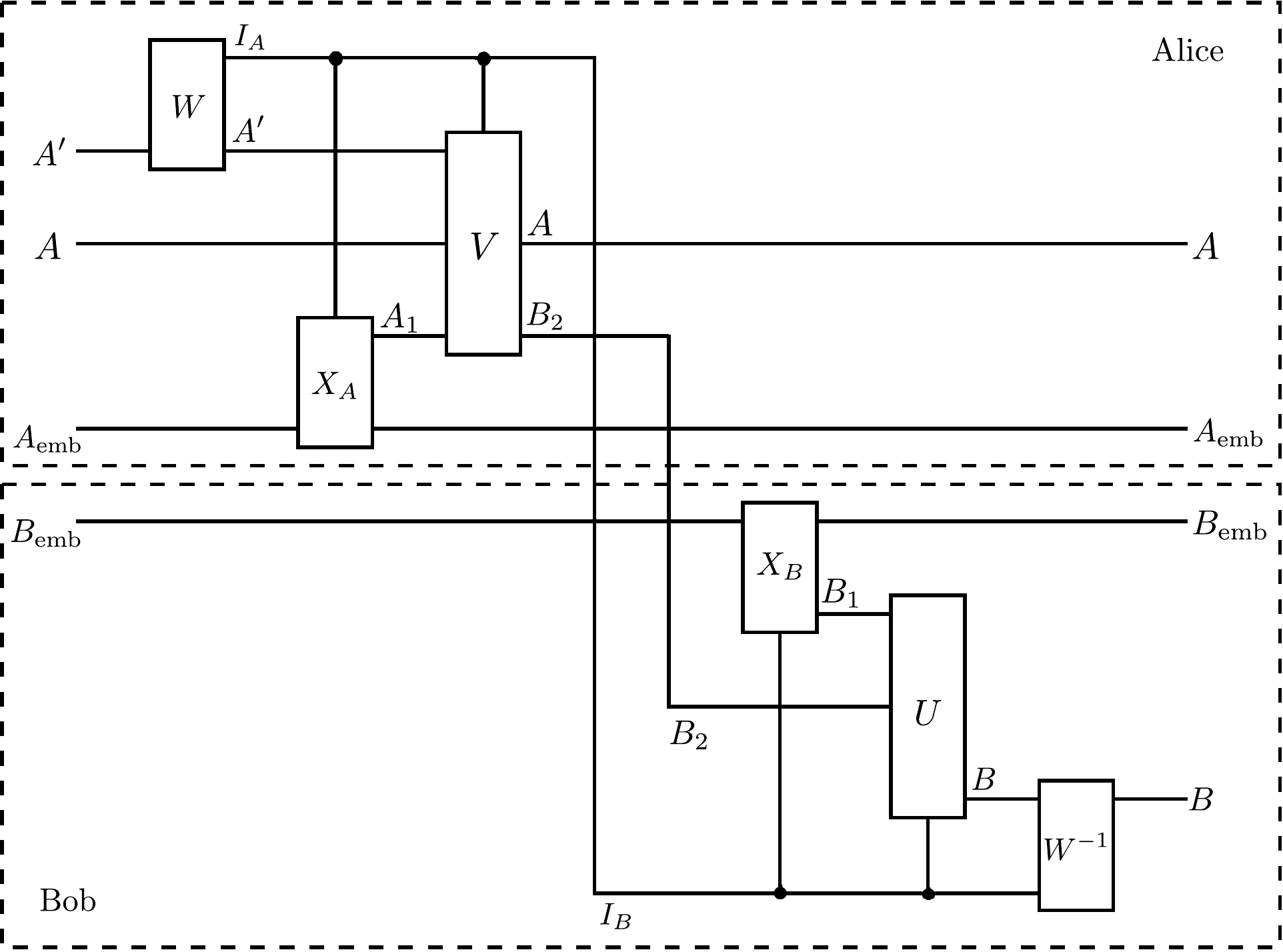}
\end{center}
\caption{Quantum state splitting with embezzling states.}{A schematic description of our protocol for quantum state splitting with embezzling states in the language of the quantum circuit model~\cite{Deutsch89,Nielsen00}. See the text for definitions and a precise description.}
\end{figure}

We proceed by defining the operations that we need for the quantum state splitting protocol with embezzling states for $\rhob_{AA'R}$ (cf.~Figure~\ref{fig:qssemb}). Define the isometry
\begin{align}\label{eq:cohmeas}
W_{A'\ra A'I_{A}}=\sum_{i\in I}P_{A'}^{i}\ot\ket{i}_{I_{A}}\ ,
\end{align}
where the vectors $\ket{i}_{A}$ are mutually orthogonal, and $I_{A}$ is at Alice's side. We want to use the $\eps$-error quantum state splitting protocol with maximally entangled states from Lemma~\ref{lem:qss} for each $\rho^{i}_{AA'R}$. For $i=0,1,\ldots,Q$ this protocol has a quantum communication cost
\begin{align}
q_{i}=\left\lceil\frac{1}{2}\cdot\big(H_{0}(A')_{\rho^{i}}-H_{\min}(A'|R)_{\rho^{i}}\big)+2\log\frac{1}{\eps}\right\rceil\ ,
\end{align}
and an entanglement cost
\begin{align}\label{eq:ecostbranch}
e_{i}=\left\lfloor\frac{1}{2}\cdot\big(H_{0}(A')_{\rho^{i}}+H_{\min}(A'|R)_{\rho^{i}}\big)-2\log\frac{1}{\eps}\right\rfloor\ .
\end{align}
For $A_{1}$ on Alice's side, $B_{1}$ on Bob's side and $A_{1}^{i}$, $B_{1}^{i}$ $2^{e_{i}}$-dimensional subspaces of $A_{1}$, $B_{1}$ respectively, the quantum state splitting protocol from Lemma~\ref{lem:qss} has the following form: apply the isometry $V^{i}_{AA'A_{1}^{i}\ra AB_{2}^{i}}$ on Alice's side, send $B_{2}^{i}$ from Alice to Bob (for a quantum communication cost of $q_{i}$), and then apply the isometry $U^{i}_{B_{1}^{i}B_{2}^{i}\ra B}$ on Bob's side. As a next ingredient to the protocol, we define the isometries that supply the maximally entangled states of size $e_{i}$. For $i=0,1,\ldots,Q$, let $X_{A_{\emb}\ra A_{\emb}A_{1}^{i}}^{i}$ and $X_{B_{\emb}\ra B_{\emb}B_{1}^{i}}^{i}$ be the isometries at Alice's and Bob's side respectively, that embezzle, with accuracy $\delta\cdot e_{i}$, a maximally entangled state of dimension $e_{i}$ out of the embezzling state and put it in $A_{1}^{i}B_{1}^{i}$.

We are now ready to put the isometries together and give the protocol for quantum state splitting with embezzling states for $\rhob_{AA'R}$ (cf.~Figure~\ref{fig:qssemb}). Alice applies the isometry $W_{A'\ra A'I_{A}}$ followed by the isometry
\begin{align}
X_{A_{\emb}I_{A}\ra A_{\emb}A_{1}I_{A}}=\sum_{i=1}^{Q}X_{A_{\emb}\ra A_{\emb}A_{1}^{i}}^{i}\ot\proj{i}_{I_{A}}\ ,
\end{align}
and the isometry 
\begin{align}
V_{AA'A_{1}I_{A}\ra AB_{2}I_{A}}=\sum_{i=0}^{Q}V_{AA'A_{1}^{i}\ra AB_{2}^{i}}^{i}\ot\proj{i}_{I_{A}}\ .
\end{align}
Afterwards she sends $I_{A}$ and $B_{2}$, that is,
\begin{align}
q=\max_{i}\left\lceil\frac{1}{2}\cdot\big(H_{0}(A')_{\rho^{i}}-H_{\min}(A'|R)_{\rho^{i}}\big)+2\cdot\log\frac{1}{\eps}\right\rceil+\log\left\lceil2\cdot\log|A'|\right\rceil
\end{align}
qubits to Bob (where we rename $I_{A}$ to $I_{B}$). Then, Bob applies the isometry
\begin{align}
X_{B_{\emb I_{B}}\ra B_{\emb}B_{1}I_{B}}=\sum_{i=1}^{Q}X_{B_{\emb}\ra B_{\emb}B_{1}^{i}}^{i}\ot\proj{i}_{I_{B}}\ ,
\end{align}
followed by the isometry
\begin{align}\label{eq:isodecoding}
U_{B_{1}B_{2}I_{B}\ra BI_{B}}=\sum_{i=0}^{Q}U^{i}_{B_{1}^{i}B_{2}^{i}\ra B}\ot\proj{i}_{I_{B}}\ .
\end{align}

Next we analyze how the resulting state looks. By the definition of embezzling states (Definition~\ref{def:emb} and Remark~\ref{rmk:emb}), the monotonicity of the purified distance (Lemma~\ref{lem:pdmono}), and the triangle inequality for the purified distance, we obtain a state $\sigma_{ABRI_{B}}=\ket{\sigma}\bra{\sigma}_{ABRI_{B}}$ where
\begin{align}
\ket{\sigma}_{ABRI_{B}}=\Upsilon^{-1/2}\cdot\sum_{i=0}^{Q}\sqrt{p_{i}}\cdot\ket{\rhot^{i}}_{ABR}\ot\ket{i}_{I_{B}}\ ,
\end{align}
$\proj{\rhot^{i}}_{ABR}=\rhot^{i}_{ABR}\approx_{\eps+\delta\cdot e_{i}}\rho^{i}_{ABR}$, and $\rho^{i}_{ABR}=(\cI_{A'\ra B}\ot\cI_{AR})(\rho^{i}_{AA'R})$ for $i=0,1,\ldots,Q$. The state $\sigma_{ABRI_{B}}$ is close to the state $\omega_{ABRI_{B}}=\ket{\omega}\bra{\omega}_{ABRI_{B}}$ with
\begin{align}
\ket{\omega}_{ABRI_{B}}=\Upsilon^{-1/2}\cdot\sum_{i=0}^{Q}\sqrt{p_{i}}\cdot\ket{\rho^{i}}_{ABR}\ot\ket{i}_{I_{B}}\ ,
\end{align}
as can be seen as follows. Because we can assume without lost of generality that all $\ip{\rhot^{i},\rho^{i}}$ are real and nonnegative,\footnote{This can be done be multiplying the isometries $U_{B_{1}^{i}B_{2}^{i}\ra B}^{i}$ in~\eqref{eq:isodecoding} with appropriately chosen phase factors.} we obtain
\begin{align}
P(\sigma_{ABRI_{B}},\omega_{ABRI_{B}})&=\sqrt{1-\left|\frac{1}{\Upsilon}\cdot\sum_{i=0}^{Q}p_{i}\ip{\rhot^{i},\rho^{i}}_{ABR}\right|^{2}}\notag\\
&=\sqrt{1-\left(\frac{1}{\Upsilon}\cdot\sum_{i=0}^{Q}p_{i}\sqrt{1-P^{2}\left(\rhot^{i}_{ABR},\rho^{i}_{ABR}\right)}\right)^{2}}\notag\\
&\leq\sqrt{1-\left(\frac{1}{\Upsilon}\cdot\sum_{i=0}^{Q}p_{i}\sqrt{1-(\eps+\delta\cdot\max_{i}e_{i})^{2}}\right)^{2}}\notag\\
&=\eps+\delta\cdot\max_{i}e_{i}\leq\eps+\delta\cdot\log|A'|\ ,
\end{align}
where the last inequality follows from~\eqref{eq:ecostbranch}. To decode the state $\sigma_{ABRI_{B}}$ to a state that is $(\eps+\delta\cdot\log|A'|)$-close to $\rhob_{ABR}$, we define the isometry $W_{B\ra BI_{B}}$ analogously to $W_{A'\ra A'I_{A}}$ in~\eqref{eq:cohmeas}. Because all isometries are injective, we can define an inverse of $W$ on the image of $W$ (which we denote by $\mathrm{Im}(W)$). The inverse is again an isometry and we denote it by $W^{-1}_{\mathrm{Im}(W)\ra B}$. The last step of the protocol is then to apply the channel to the state $\sigma_{ABRI_{B}}$, that first does a measurement on $BI_{B}$ to decide whether $\sigma_{BI_{B}}\in\mathrm{Im}(W)$ or not and then, if $\sigma_{BI_{B}}\in\mathrm{Im}(W)$, applies the isometry $W^{-1}_{\mathrm{Im}(W)\ra B}$ and otherwise maps the state to $\proj{0}_{B}$.

By the monotonicity of the purified distance (Lemma~\ref{lem:pdmono}) we finally get a state that is $(\eps+\delta\cdot\log|A'|)$-close to $\rhob_{ABR}$. Hence, we showed the existence of an $(\eps+\delta\cdot\log|A'|)$-error quantum state splitting protocol with embezzling states for $\rhob_{AA'R}$ with quantum communication cost
\begin{align}\label{eq:ecost}
q=\max_{i}\left\lceil\frac{1}{2}\cdot\big(H_{0}(A')_{\rho^{i}}-H_{\min}(A'|R)_{\rho^{i}}\big)+2\log\frac{1}{\eps}\right\rceil+\log\left\lceil2\cdot\log|A'|\right\rceil\ ,
\end{align}
where $i\in\{0,1,\ldots,Q\}$. But by the monotonicity of the purified distance (Lemma~\ref{lem:pdmono}), \eqref{eq:barstate} and the triangle inequality for the purified distance, this implies the existence of an $\left(\eps+\delta\cdot\log|A'|+|A|^{-1/2}\right)$-error quantum state splitting protocol with embezzling states for $\rho_{AA'R}$ with a quantum communication cost as in~\eqref{eq:ecost}.

We now proceed with simplifying the expression for the quantum communication cost~\eqref{eq:ecost}. We have $H_{0}(A')_{\rho^{i}}\leq H_{\min}(A')_{\rho^{i}}+1$ for $i=0,1,\ldots,Q$ as can be seen as follows. We have
\begin{align}
2^{-(i+1)}\leq\lambda_{\min}(p_{i}\cdot\rho_{A'}^{i})\leq\frac{1}{p_{i}\cdot\rk\left(\rho_{A'}^{i}\right)}\leq\lambda_{1}\left(p_{i}\cdot\rho_{A'}^{i}\right)\leq2^{-i}\ ,
\end{align}
where $\lambda_{\min}(\rho_{A'}^{i})$ denotes the smallest non-zero eigenvalue of $\rho_{A'}^{i}$. Thus,
\begin{align}
\rk\left(p_{i}\cdot\rho_{A'}^{i}\right)\leq2^{i+1}=2^{i}\cdot2\leq\lambda_{1}\left(p_{i}\cdot\rho_{A'}^{i}\right)^{-1}\cdot2\ ,
\end{align}
and this is equivalent to the claim. Hence, we get an $(\eps+\delta\cdot\log|A'|+|A'|^{-1/2})$-error quantum state splitting protocol with embezzling states for $\rho_{AA'R}$ with quantum communication cost
\begin{align}
q&=\max_{i}\left\lceil\frac{1}{2}\left(H_{\min}(A')_{\rho^{i}}-H_{\min}(A'|R)_{\rho^{i}}+1\right)+2\log\frac{1}{\eps}\right\rceil+\log\left\lceil2\log|A'|\right\rceil\ .
\end{align}
Using a lower bound for the max-information in terms of min-entropies (Lemma~\ref{lem:maxbounds}), and the behavior of the max-information under projective measurements (Corollary \ref{cor:maxproj}) we can simplify this to
\begin{align}
q&\leq\left\lceil\max_{i}\frac{1}{2}\cdot I_{\max}(A':R)_{\rho^{i}}+2\log\frac{1}{\eps}+\frac{1}{2}\right\rceil+\log\left\lceil2\log|A'|\right\rceil\notag\\
&\leq\left\lceil\frac{1}{2}\cdot I_{\max}(A':R)_{\rho}+2\log\frac{1}{\eps}+\frac{1}{2}\right\rceil+\log\left\lceil2\log|A'|\right\rceil\ .
\end{align}
It is then easily seen that
\begin{align}
q\leq\frac{1}{2}\cdot I_{\max}(A':R)_{\rho}+2\log\frac{1}{\eps}+4+\log\log|A'|\ .
\end{align}
As the last step, we transform the max-information term in the formula for the quantum communication cost into a smooth max-information. Namely, we can reduce the quantum communication cost if we do not apply the protocol as described above to the state $\rho_{AA'R}$, but pretend that we have another (possibly sub-normalized) state $\hat{\rho}_{AA'R}$ that is $\eps'$-close to $\rho_{AA'R}$ and then apply the protocol for $\hat{\rho}_{AA'R}$. By the monotonicity of the purified distance (Lemma~\ref{lem:pdmono}), the additional error term that we get from this is upper bounded by $\eps'$ and by the triangle inequality for the purified distance this results in an accuracy of $\eps+\eps'+\delta\cdot\log|A'|+|A'|^{-1/2}$. But if we minimize $q$ over all $\hat{\rho}_{AA'R}$ that are $\eps'$-close to $\rho_{AA'R}$, we can reduce the quantum communication cost to
\begin{align}\label{final:cost}
q\leq\frac{1}{2}\cdot I_{\max}^{\eps'}(A':R)_{\rho}+2\log\frac{1}{\eps}+4+\log\log|A'|\ .
\end{align}
This shows the existence of an $(\eps+\eps'+\delta\cdot\log|A'|+|A'|^{-1/2})$-error quantum state splitting protocol with embezzling states for $\rho_{AA'R}$ for a quantum communication cost as in~\eqref{final:cost}.
\end{proof}

The following theorem shows that the quantum communication cost in Theorem~\ref{thm:qssemb} is optimal up to small additive terms.

\begin{theorem}\label{thm:qssconv}
Let $\eps\geq0$, $\eps'>0$, and $\rho_{AA'R}\in\cV_{\leq}(\cH_{AA'R})$. Then, the quantum communication cost for any $\eps$-error quantum state splitting protocol for $\rho_{AA'R}$ is lower bounded by\footnote{We suppress the mentioning of any entanglement resource, since the statement holds independently of it.}
\begin{align}
q\geq\frac{1}{2}\cdot I^{\eps+\eps'}_{\max}(A':R)_{\rho}-\frac{1}{2}\cdot\log\left(\frac{8}{\eps'^{2}}+2\right)\ .
\end{align}
\end{theorem}

\begin{proof}
We have a look at the correlations between Bob and the reference by analyzing the max-information that Bob has about the reference. At the beginning of any protocol, there is no register at Bob's side and therefore the max-information that Bob has about the reference is zero. Since back communication is not allowed, we can assume that the protocol for quantum state splitting has the following form: applying local operations at Alice's side, sending qubits from Alice to Bob and then applying local operations at Bob's side. Local operations at Alice's side have no influence on the max-information that Bob has about the reference. By sending $q$ qubits from Alice to Bob, the max-information that Bob has about the reference can increase, but at most by $2q$ (Lemma~\ref{lem:maxdbound}). By applying local operations at Bob's side the max-information that Bob has about the reference can only decrease (Lemma~\ref{lem:imaxmono}). So the max-information that Bob has about the reference is upper bounded by $2q$. Therefore, any state $\omega_{BR}$ at the end of a quantum state splitting protocol must satisfy $I_{\max}(R:B)_{\omega}\leq2q$. But we also need $\omega_{BR}\approx_{\eps}\rho_{BR}=(\cI_{A'\ra B}\ot\cI_{R})(\rho_{A'R})$ by the definition of $\eps$-error quantum state splitting (Definition~\ref{def:qssemb}). Using the definition of the smooth max-information, we get
\begin{align}
q\geq \frac{1}{2}\cdot I_{\max}^{\eps}(R:A')_{\rho}\ .
\end{align}
Since the smooth max-information is approximately symmetric (Lemma~\ref{lem:ciganovic}), the claim follows.
\end{proof}


\subsection{Main Theorem}\label{sec:qshannon_main}

Here we present the main result of this section, a proof of the quantum reverse Shannon theorem. The intuition is as follows. Let $\cE_{A\ra B}$ be a quantum channel with
\begin{align}
\cE_{A\ra B}:\quad&\cB(\cH_{A})\ra\cB(\cH_{B})\notag\\
& \rho_{A}\mapsto\cE_{A\ra B}(\rho_{A})=\rho_{B}\ ,
\end{align}
where we want to think of subsystem $A$ being at Alice's side, and subsystem $B$ being at Bob's side. The quantum reverse Shannon theorem states that if Alice and Bob share embezzling states, they can asymptotically simulate $\cE_{A\ra B}$ only using local operations at Alice's side, local operations at Bob's side, and a quantum communication rate (from Alice to Bob) of
\begin{align}
Q_{E}=\frac{1}{2}\cdot\max_{\rho}I(B:R)_{(\cE\ot\cI)(\rho)}\ ,
\end{align}
where $\rho_{AR}$ is a purification of $\rho_{A}$. Using Stinespring's dilation~\cite{Stinespring55}, we can think of $\cE_{A\ra B}$ as
\begin{align}\label{eq:dilationidea}
\cE_{A\ra B}(\rho_{A})=\trace_{E}\left[U_{A\ra BE}(\rho_{A})\right]\ ,
\end{align}
where $E$ is an additional register with $|E|\leq|A||B|$, and $U_{A\ra BE}$ some isometry. The idea of our proof is to first simulate the quantum channel locally at Alice's side, resulting in $\rho_{BC}=U_{A\ra BE}(\rho_{A})$, and then use quantum state splitting with embezzling states (Theorem~\ref{thm:qssemb}) to do an optimal state transfer of the $B$-part to Bob's side, such that he holds $\rho_{B}=\cE_{A\ra B}(\rho_{A})$ in the end. More formally, we make the following definitions. We start with the non-feedback case, but we will see afterwards that the quantum communication cost in the feedback case is actually the same (Section~\ref{sec:qshannon_discussion}).

\begin{definition}[One-shot non-feedback reverse Shannon simulation]\label{def:oneqrst}
Consider a bipartite system with parties Alice and Bob. Let $\eps\geq0$ and $\cE:\cB(\cH_{A})\ra\cB(\cH_{B})$ be a channel, where Alice controls $\cH_{A}$ and Bob $\cH_{B}$. A quantum protocol $\cP$ is a one-shot non-feedback reverse Shannon simulation for $\cE$ with error $\eps$ if it consists of applying local operations at Alice's side, local operations at Bob's side, sending $q$ qubits from Alice to Bob, using a $\delta$-ebit embezzling state for some $\delta>0$, and
\begin{align}\label{eq:oneqrst}
\|\cP-\cE\|_{\Diamond}\leq\eps\ .
\end{align}
The quantity $q$ is called quantum communication cost of the one-shot reverse Shannon simulation.
\end{definition}

\begin{definition}[Asymptotic non-feedback reverse Shannon simulation]\label{def:qrstasym}
Let $\cE:\cB(\cH_{A})\ra\cB(\cH_{B})$ be a channel. An asymptotic non-feedback reverse Shannon simulation for $\cE$ is a sequence of one-shot non-feedback reverse Shannon protocols $\cP^{n}$ for $\cE^{\ot n}$ with quantum communication cost $q_{n}$, and error $\eps_{n}$, such that $\lim_{n\ra\ift}\eps_{n}=0$. The quantum quantum cost of this simulation is $q=\limsup_{n\ra\ift}\frac{q_{n}}{n}$.
\end{definition}

A precise statement of the quantum reverse Shannon theorem is now as follows.

\begin{theorem}\label{thm:qrst}
Let $\cE_{A\ra B}:\cB(\cH_{A})\ra\cB(\cH_{B})$ be a channel. Then the minimal quantum communication cost $Q_{\qrst}$ of asymptotic non-feedback reverse Shannon simulations for $\cE_{A\ra B}$ is equal to the entanglement assisted quantum capacity $Q_{E}$ of $\cE_{A\ra B}$. That is,
\begin{align}
Q_{\qrst}=\frac{1}{2}\cdot\max_{\rho}I(B:R)_{(\cE\ot\cI)(\rho)}\ ,
\end{align}
where $\rho_{AR}\in\cV(\cH_{AR})$ is a purification of the input state $\rho_{A}\in\cS(\cH_{A})$.\footnote{Since all purifications give the same amount of entropy, we do not need to specify which one we use.}
\end{theorem}

\begin{proof}
First, note that $Q_{\qrst}\geq Q_{E}$ by the entanglement assisted quantum capacity theorem~\cite{Bennett02}.\footnote{Assume that $Q_{\qrst}\leq Q_{E}-\delta$ for some $\delta>0$ and start with the identity channel $\cI_{A\ra B}$. Then we could use $Q_{\qrst}\leq Q_{E}-\delta$ together with the entanglement assisted quantum capacity theorem to asymptotically simulate the perfect quantum identity channel at a rate $\frac{Q_{E}}{Q_{E}-\delta}>1$; a contradiction to Holevo's theorem~\cite{Holevo98,Nielsen00}.} Hence, it remains to show that $Q_{\qrst}\leq Q_{E}$.

We start by making some general statements about the structure of the proof, and then dive into the technical arguments. Because the quantum reverse Shannon theorem makes an asymptotic statement, we have to make our considerations for a general $n\in\mathbb{N}$. Thus the goal is to show the existence of a one-shot reverse Shannon simulation $\cP^{n}_{A\ra B}$ for $\cE_{A\ra B}^{\ot n}$ that is arbitrarily close to $\cE_{A\ra B}^{\ot n}$ for $n\ra\ift$, has a quantum communication rate of $Q_{E}$ and works for any input. We do this by using quantum state splitting with embezzling states (Theorem~\ref{thm:qssemb}), and the post-selection technique for quantum channels (Proposition~\ref{prop:postselect}).

Any hypothetical map $\cP^{n}_{A\ra B}$ (that we may want to use for the simulation of $\cE_{A\ra B}^{\ot n}$), can be made to act symmetrically on the $n$-partite input system $\cH_{A}^{\ot n}$ by inserting a symmetrization step. This works as follows. First Alice and Bob generate some shared randomness by generating maximally entangled states from the embezzling states and measuring their part in the same computational basis (for $n$ large, $O(n\log n)$ maximally entangled states are needed). Then, before the original map $\cP_{A\ra B}^{n}$ starts, Alice applies a random permutation $\pi$ on the input system chosen according to the shared randomness. Afterwards they run the map $\cP_{A\ra B}^{n}$ and then, in the end, Bob undoes the permutation by applying $\pi^{-1}$ on the output system. From this we obtain a permutation invariant version of $\cP^{n}_{A\ra B}$. Since the maximally entangled states can only be created with finite precision, the shared randomness, and therefore the permutation invariance, is not perfect. However, as we will argue at the end, this imperfection can be made arbitrarily small, and can therefore be neglected.

Note that the simulation will need embezzling states $\mu_{A_{\emb}B_{\emb}}$, and maximally entangled states $\Phi^{m}_{A_{\ebit}B_{\ebit}}$ (to assure the permutation invariance). But since the input on these registers is fixed, we are allowed to think of the simulation as a map $\cP_{A\ra B}^{n}$, see Figure~\ref{fig:emb}.

\begin{figure}[ht]\label{fig:emb}
\begin{center}
\includegraphics[width=0.7\linewidth]{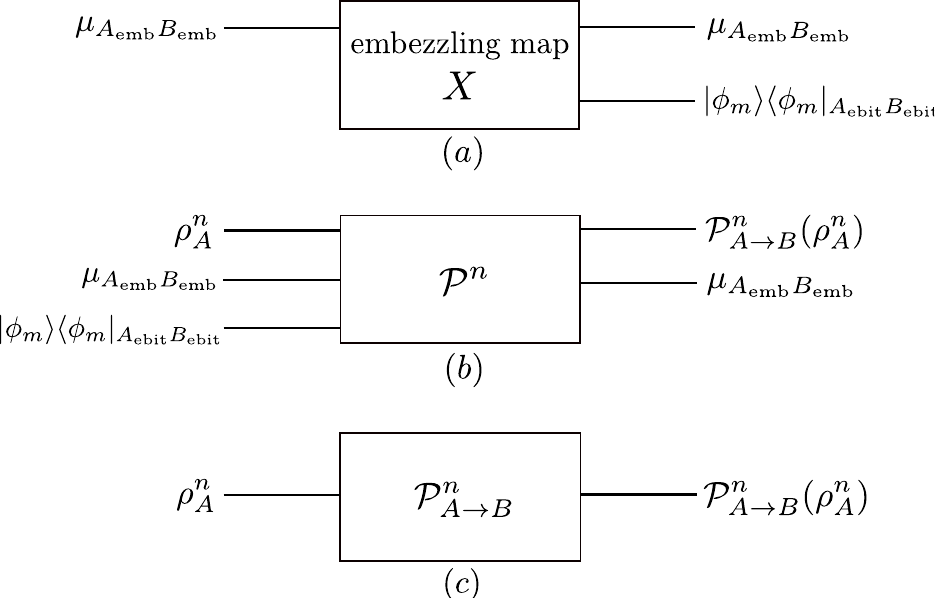}
\end{center}
\caption{Embezzling map.}{(a) $X$ is the map that embezzles $m$ maximally entangled states $\Phi^{m}_{A_{\ebit}B_{\ebit}}$ out of $\mu_{A_{\emb}B_{\emb}}$. These maximally entangled states are then used in the protocol. (b) The whole map that should simulate $\cE_{A\ra B}^{\ot n}$ takes $\rho_{A}^{n}\ot\mu_{A_{\emb}B_{\emb}}\ot\Phi^{m}_{A_{\ebit}B_{\ebit}}$ with $\rho_{A}^{n}\in\cS(\cH_{A}^{\ot n})$ as an input. But since this input is constant on all registers except for $A$, we can think of the map as in (c), namely as a channel $\cP^{n}_{A\ra B}$ which takes only the input $\rho_{A}^{n}$.}
\end{figure}

Let $\beta>0$. Our aim is to show the existence of a map $\cP^{n}_{A\ra B}$, that consists of applying local operations at Alice's side, local operation at Bob's side, sending qubits from Alice to Bob at a rate of $C_{E}$, and such that
\begin{align}\label{end}
\|\cE_{A\ra B}^{\ot n}-\cP_{A\ra B}^{n}\|_{\Diamond}\leq\beta\ .
\end{align}
Because we assume that the map $\cP^{n}_{A\ra B}$ is permutation invariant, we are allowed to use the post-selection technique (Proposition~\ref{prop:postselect}). Thus~\eqref{end} relaxes to
\begin{align}\label{eq:qrstpost}
\left\|\left((\cE_{A\ra B}^{\ot n}-\cP^{n}_{A\ra B})\ot\cI_{RR'})(\zeta^{n}_{ARR'}\right)\right\|_{1}\leq\beta(n+1)^{-(|A|^{2}-1)}\ ,
\end{align}
where $\zeta^{n}_{ARR'}$ is a purification of $\zeta^{n}_{AR}=\int\omega_{AR}^{\ot n}\;d(\omega_{AR})$, $\omega_{AR}\in\cV(\cH_{AR})$, and $d(.)$ is the measure on the normalized pure states on $\cH_{AR}$ induced by the Haar measure on the unitary group acting on $\cH_{AR}$, normalized to $\int d(.)=1$.

To show~\eqref{eq:qrstpost}, we consider a local simulation of the channel $\cE_{A\ra B}^{\ot n}$ at Alice's side (using Stinespring's dilation as in~\eqref{eq:dilationidea}) followed by quantum state splitting with embezzling states. Applied to the de Finetti type input state $\zeta^{n}_{ARR'}$, we obtain the state
\begin{align}
\zeta_{BCRR'}^{n}=U_{A\ra BC}^{n}(\zeta^{n}_{ARR'})\ .
\end{align}
As described above, this map can be made permutation invariant (cf.~Figure~\ref{fig:circuit}).\\

\begin{figure}[ht]\label{fig:circuit}
\begin{center}
\includegraphics[width=1.0\linewidth]{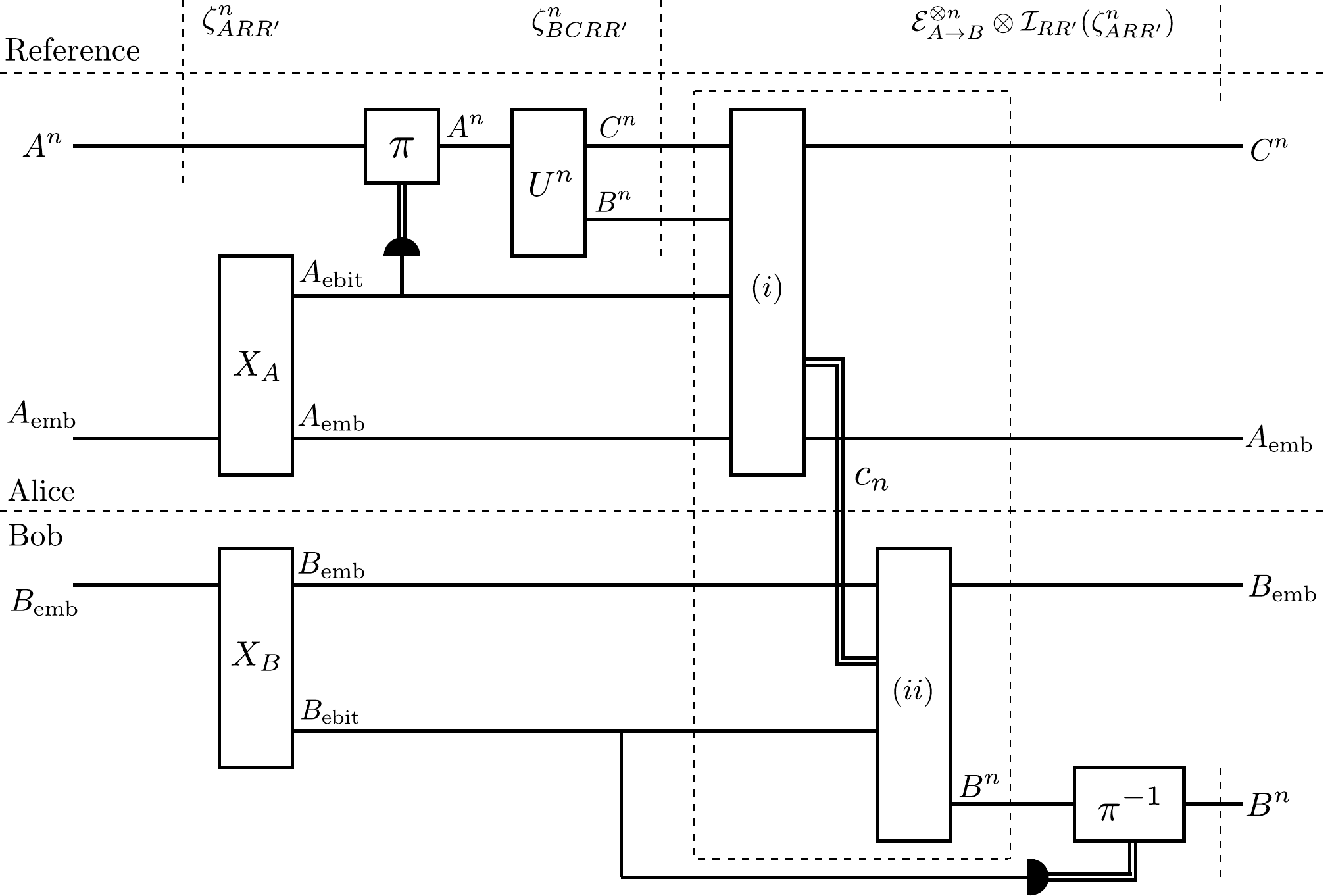}
\end{center}
\caption{Circuit quantum reverse Shannon theorem.}{A schematic description of the protocol that is used to prove the quantum reverse Shannon theorem. The channel simulation is done for the de Finetti type input state $\zeta^{n}_{ARR'}$. Because our simulation is permutation invariant, the post-selection technique (Proposition~\ref{prop:postselect}) shows that this is also sufficient for all input states. The whole simulation is called $\cP^{n}$ in the text. (i) and (ii) denote the subroutine of quantum state splitting with embezzling states; with local operations on Alice's and Bob's side and a quantum communication rate of $q_{n}$.}
\end{figure}

Now we use this map as $\cP^{n}_{A\ra B}$ in~\eqref{eq:qrstpost}. We obtain from the achievability of quantum state splitting with embezzling states (Theorem~\ref{thm:qssemb}) that
\begin{align}
P\left((\cE_{A\ra B}^{\ot n}\ot\cI_{RR'})(\zeta^{n}_{ARR'}),(\cP^{n}_{A\ra B}\ot\cI_{RR'})(\zeta^{n}_{ARR'})\right)\leq\eps+\eps'+\delta n\cdot\log|B|+|B|^{-n/2}\ ,
\end{align}
for a quantum communication cost
\begin{align}\label{eq:qcost}
q_{n}\leq\frac{1}{2}\cdot I_{\max}^{\eps'}(B:RR')_{(\cE^{\ot n}\ot\cI)(\zeta^{n})}+2\log\frac{1}{\eps}+4+\log n+\log\log|B|\ ,
\end{align}
where the last two terms on the right come from the fact that $\log\log |X_{B}|^{n}=\log n+\log\log|X_{B}|$. Because the trace distance is upper bounded by two times the purified distance (Lemma~\ref{lem:pdbounds}), this implies
\begin{align}
\left\|\left((\cE_{A\ra B}^{\ot n}-\cP^{n}_{A\ra B})\ot\cI_{RR'})(\zeta^{n}_{ARR'}\right)\right\|_{1}\leq2(\eps+\eps'+\delta n\cdot\log|B|+|B|^{-n/2})\ ,
\end{align}
and by choosing $\eps=\eps'$, and $\delta=\frac{\eps'}{n\cdot\log|B|}$ we obtain
\begin{align}
\|((\cE_{A\ra B}^{\ot n}-\cP^{n}_{A\ra B})\ot\cI_{RR'})(\zeta^{n}_{ARR'})\|_{1}\leq6\eps'+2\cdot|B|^{-n/2}\ .
\end{align}
Furthermore, we choose $\eps'=\frac{1}{6}\cdot\beta(n+1)^{-(|A|^{2}-1)}-\frac{1}{3}|B|^{-n/2}$ (for large enough $n$), and hence
\begin{align}
\left\|\left((\cE_{A\ra B}^{\ot n}-\cP^{n}_{A\ra B})\ot\cI_{RR'})(\zeta^{n}_{ARR'}\right)\right\|_{1}\leq\beta(n+1)^{-(|A|^{2}-1)}\ .
\end{align}
This is~\eqref{eq:qrstpost}, and by the post-selection technique (Proposition~\ref{prop:postselect}) this implies~\eqref{end}.

It thus remains to show that the quantum communication rate of the resulting map is upper bounded by $Q_{E}$. Set
\begin{align}
\chi=2\log\frac{1}{\eps'}+4+\log n+\log\log|B|\ ,
\end{align}
and it follows from~\eqref{eq:qcost} and below that the quantum communication cost of $\cP^{n}_{A\ra B}$ is quantified by
\begin{align}
q_{n}\leq\frac{1}{2}\cdot I_{\max}^{\eps'}(B:RR')_{(\cE^{\ot n}\ot\cI)(\zeta^{n})}+\chi\ .
\end{align}
By the upper bound in Lemma~\ref{lem:maxdbound}, and the fact that we can assume $|R'|\leq(n+1)^{|A|^{2}-1}$ (Proposition~\ref{prop:postselect}), we get
\begin{align}
q_{n}&\leq\frac{1}{2}\cdot I_{\max}^{\eps'}(B:R)_{(\cE^{\ot n}\ot\cI)(\zeta^{n})}+\log|R'|+\chi\notag\\
&\leq\frac{1}{2}\cdot I_{\max}^{\eps'}(B:R)_{(\cE^{\ot n}\ot\cI)(\zeta^{n})}+\log\left((n+1)^{|A|^{2}-1}\right)+\chi\ .
\end{align}
By a corollary of Carath\'eodory's theorem (Corollary~\ref{cor:cara}), we can write
\begin{align}
\zeta_{AR}^{n}=\sum_{i}p_{i}(\rho^{i}_{AR})^{\ot n}\ ,
\end{align}
where $\rho^{i}_{AR}\in\cV(\cH_{AR})$, $i\in\{1,2,\ldots,(n+1)^{2|A||R|-2}\}$, and $\{p_{i}\}$ a probability distribution. Using a quasi-convexity property of the smooth max-information (Lemma~\ref{lem:imaxqconvex}) we then obtain
\begin{align}
q_{n} &\leq\frac{1}{2}\cdot I_{\max}^{\eps'}(B:R)_{(\cE^{\ot n}\ot\cI)(\sum_{i}p_{i}(\rho^{i})^{\ot n})}+\log\left((n+1)^{|A|^{2}-1}\right)+\chi\notag\\
&\leq\frac{1}{2}\cdot\max_{i}I_{\max}^{\eps'}(B:R)_{\left[(\cE\ot\cI)(\rho^{i})\right]^{\ot n}}+\log\left((n+1)^{|A||R|-1}\right)+\log\left((n+1)^{|A|^{2}-1}\right)+\chi\notag\\
&\leq\frac{1}{2}\cdot\max_{\rho}I_{\max}^{\eps'}(B:R)_{\left[(\cE\ot\cI)(\rho)\right]^{\ot n}}+\log\left((n+1)^{|A||R|-1}\right)+\log\left((n+1)^{|A|^{2}-1}\right)+\chi\ ,
\end{align}
where the last maximum ranges over all $\rho_{AR}\in\cV(\cH_{AR})$. From the asymptotic equipartition property for the smooth max-information (Lemma~\ref{lem:aepimax}) we obtain
\begin{align}
q_{n}\leq&\;n\cdot\frac{1}{2}\cdot\max_{\rho}I(B:R)_{(\cE\ot\cI)(\rho)}+\sqrt{n}\cdot\frac{\xi(\eps')}{2}-\log\frac{\eps'^{2}}{24}+\log\left((n+1)^{|A||R|-1}\right)\notag\\
&+\log\left((n+1)^{|A|^{2}-1}\right)+\chi\ ,
\end{align}
where $\xi(\eps')=8\cdot\sqrt{13-4\cdot\log\eps'}\cdot(2+\frac{1}{2}\cdot\log|A|)$. Since $\eps'=\frac{1}{6}\cdot\beta(n+1)^{-(|A|^{2}-1)}-\frac{1}{3}|B|^{-n/2}$, the quantum communication rate is then upper bounded by
\begin{align}
q & =\limsup_{\beta\ra0}\limsup_{n\ra\ift}\frac{q_{n}}{n}\leq\frac{1}{2}\cdot\max_{\rho}I(B:R)_{(\cE\ot\cI)(\rho)}\ .
\end{align}
Thus it only remains to justify why it is sufficient that the maximally entangled states, which we used to make the protocol permutation invariant, only have finite precision. For this, it is useful to think of the channel $\cP^{n}_{A\ra B}$ that we constructed above, as in Figure~\ref{fig:emb} (b). Let $\eps''>0$ and assume that the entanglement is $\eps''$-close to the perfect input state $\mu_{A_{\emb}B_{\emb}}\ot\Phi^{m}_{A_{\ebit}B_{\ebit}}$. The purified distance is monotone (Lemma~\ref{lem:pdmono}), and hence the corresponding imperfect output state is $\eps''$-close to the state obtained under the assumption of perfect permutation invariance. Since $\eps'''$ can be made arbitrarily small (Definition~\ref{def:emb}), the channel based on the imperfect entanglement does the job.
\end{proof}


\subsection{Discussion}\label{sec:qshannon_discussion}

\paragraph{Quantum Feedback Simulation.} Our main result (Theorem~\ref{thm:qrst}) concerns the case of a non-feedback quantum channel simulation. But for the corresponding feedback version, we can just modify the Definitions~\ref{def:oneqrst} and~\ref{def:qrstasym} by exchanging the channel $\cE_{A\ra B}$ in~\eqref{eq:oneqrst} with its Stinespring dilation $U_{A\ra BE}$ (where the register $E$ is at Alice's side). It is then obvious from our proof strategy (e.g., this can be seen from~\eqref{eq:dilationidea}), that Theorem~\ref{thm:qrst} also holds for the feedback case. Hence, we have the following corollary.

\begin{corollary}
Let $\cE_{A\ra B}:\cB(\cH_{A})\ra\cB(\cH_{B})$ be a channel. Then, the minimal quantum communication cost of asymptotic feedback reverse Shannon simulations for $\cE_{A\ra B}$ is equal to the entanglement assisted classical capacity of $\cE_{A\ra B}$.
\end{corollary}

\paragraph{Classical Communication.} The entanglement assisted classical capacity $C_{E}$ and the entanglement assisted quantum capacity $Q_{E}$ of quantum channels are equivalent by the means of quantum teleportation~\cite{Bennett93} and superdense coding~\cite{Bennett92}: $C_{E}=2\cdot Q_{E}$. By the same reasoning we also get the equivalence for the corresponding reverse capacities: $C_{\qrst}=2\cdot Q_{\qrst}$.

\paragraph{Without Embezzling States.} A natural follow up question is to ask for the rate trade-off between the different resources needed (classical communication, quantum communication, shared randomness, entanglement) to achieve channel simulations. We mention that some quantum channel simulation results with free entanglement in the form of maximally entangled states are discussed in~\cite{Bennett09}. A quantum communication rate of $Q_{E}$ is then not achievable, but additional communication is needed in order to deal with the problem of entanglement spread~\cite{Harrow09}.


\section{Universal Measurement Compression}\label{se:meas}

The results in this section have been obtained in collaboration with Joseph Renes and Mark Wilde, and have appeared in~\cite{Berta13}. We show that quantum-classical channels, i.e., quantum measurements, can be simulated by an amount of classical communication equal to the quantum mutual information of the measurement, if sufficient shared randomness is available. Our result is a generalization of the classical reverse Shannon theorem to quantum-classical channels. In this sense, it is as a quantum reverse Shannon theorem for quantum-classical channels, but with the entanglement assistance and quantum communication replaced by shared randomness and classical communication, respectively. However, our result is also a generalization of Winter's measurement compression theorem for fixed independent and identically distributed inputs~\cite{Winter04} to arbitrary inputs. In this spirit, we want to argue that our channel simulation result identifies the quantum mutual information of a quantum measurement as the information gained by performing it, independent of the input state on which it is performed.


\subsection{Identifying the Information Gain of Quantum Measurements}\label{sec:gain}

Measurement is an integral part of quantum theory. It is the means by which we gather information about a quantum system. Although the classical notion of a measurement is rather straightforward, the quantum notion of measurement has been the subject of much thought and debate~\cite{Einstein49}. One interpretation is that the act of measurement on a quantum system causes it to abruptly jump or collapse into one of several possible states with some probability, an evolution seemingly different from the smooth, unitary transitions resulting from Schr\"odinger's wave equation. Some have advocated for a measurement postulate in quantum theory~\cite{Dirac82}, while others have advocated that our understanding of quantum measurement should follow from other postulates~\cite{Zurek09}.

In spite of the aforementioned difficulties in understanding and interpreting quantum measurement, there is a precise question that one can formulate concerning it: how much information is gained by performing a given quantum measurement? This question has a rather long history, which to our knowledge begins with the work of Groenewold~\cite{Groenewold71}. In 1971, Groenewold argued on intuitive grounds for the following entropy reduction to quantify the information gained by performing a quantum measurement
\begin{align}\label{eq:groenewold}
H(\rho)-\sum_x p_x H(\rho_x)\ ,
\end{align} 
where $\rho$ is the initial state before the measurement occurs, and $\{ p_x , \rho_x \}$ is the post-measurement ensemble induced by the measurement. The intuition behind this measure is that it quantifies the reduction in uncertainty after performing a quantum measurement on a quantum system in state $\rho$, and its form is certainly reminiscent of a Holevo-like quantity~\cite{Holevo73}, although the classical data in the above Groenewold quantity appears at the output of the process rather than at the input as in the case of the Holevo quantity. Groenewold left open the question of whether this quantity is non-negative for all measurements, and Lindblad proved that non-negativity holds whenever the measurement is of the von Neumann-L\"uders kind (projecting onto an eigenspace of an observable)~\cite{Lindblad72}. Ozawa then settled the matter by proving that the above quantity is non-negative if and only if the post-measurement states are of the form
\begin{align}\label{eq:special-measurement}
\rho_x =\frac{M_x \rho M_x^\dag}{\trace[M_x^\dag M_x \rho]}\ ,
\end{align}
for some operators $\{M_x\}$ such that $\sum_x  M_x^\dag M_x = \1$~\cite{Ozawa86}. Such measurements are termed efficient, and differ from general measurements as the latter may have several operators $M_{x,s}$ corresponding to the result $x$~\cite{Fuchs01}.

The fact that the quantity in~\eqref{eq:groenewold} can become negative for some quantum measurements excludes it from being a generally appealing measure of information gain. To remedy this situation, Buscemi {\it et al.}~later advocated for the following measure to characterize the information gain of a quantum measurement $\cM$ when acting upon a particular state $\rho$~\cite{Buscemi08,Luo10,Shirokov11,Wilde12_2},
\begin{align}\label{eq:winter-info-gain}
I(X:R)_{(\cM\ot\cI)(\rho)}=H(X)_{\cM(\rho)}+H(R)_{\rho}-H(XR)_{(\cM\ot\cI)(\rho)}\ ,
\end{align}
the quantum mutual information of the state
\begin{align}\label{eq:info-gain-state}
(\cM\ot\cI_{R})(\rho_{AR})=\sum_x \vert x\rangle \langle x\vert_X \ot\trace_A\big[(\cM_x \ot\cI_{R})(\rho_{AR})\big]\ ,
\end{align}
where $\cM=\{\cM_x\}$, $X$ is a classical register containing the outcome of the measurement, and $\rho_{AR}$ is a purification of the initial state $\rho$ on system $A$ to a purifying system $R$. The advantages of the measure of information gain in~\eqref{eq:winter-info-gain} are as follows:
\begin{itemize}
\item It is non-negative.
\item It reduces to Groenewold's quantity in~\eqref{eq:groenewold}
for the special case of measurements of the form in~\eqref{eq:special-measurement}~\cite{Buscemi08}.
\item It characterizes the trade-off between information and disturbance in quantum measurements~\cite{Buscemi08}. 
\item It has an operational interpretation in Winter's measurement compression protocol as the optimal rate at which a measurement gathers information~\cite{Winter04}.
\end{itemize}
This last advantage is the most compelling one from the perspective of quantum information theory; one cannot really justify a measure as an information measure unless it corresponds to a meaningful information processing task. Indeed, when reading the first few paragraphs of Groenewold's paper~\cite{Groenewold71}, it becomes evident that his original motivation was information theoretic in nature, and with this in mind, Winter's measure in~\eqref{eq:winter-info-gain} is clearly the one Groenewold was seeking after all.

In spite of the above arguments in favor of the information measure in~\eqref{eq:winter-info-gain} as a measure of information gain, it is still lacking in one aspect: it is dependent on the state on which the quantum measurement $\cM$ acts in addition to the measurement itself. A final requirement that one should impose for a measure of information gain by a measurement is that it should depend only on the measurement itself. A simple way to remedy this problem is to maximize the quantity in~\eqref{eq:winter-info-gain} over all possible input states, leading to the following characterization of information gain
\begin{align}\label{eq:max-info-gain}
I(\cM)=\max_{\rho}I(X:R)_{(\cM\ot\cI)(\rho)}\ ,
\end{align}
for $(\cM\ot\cI_{R})(\rho_{AR})$ as in~\eqref{eq:info-gain-state}. The quantity above has already been identified and studied by previous authors as an important information quantity, being labeled as the purification capacity of a measurement~\cite{Jacobs03,Jacobs06} or the information capacity of a quantum observable~\cite{Holevo12}. The above quantity also admits an operational interpretation as the entanglement-assisted capacity of a quantum measurement for transmitting classical information~\cite{Bennett02,Bennett99,Holevo02,Holevo12}, though it is our opinion that this particular operational interpretation is not sufficiently compelling such that we should associate the measure in~\eqref{eq:max-info-gain} with the notion of information gain. Here, the aim is to address this issue by providing a compelling operational interpretation of the measure in~\eqref{eq:max-info-gain}, and our contribution is to show that $I(\cM)$ is the optimal rate at which a measurement gains information when many identical instances of it act on an arbitrary input state.

In more detail, let $\cH_{A}$ denote the input for a given measurement $\cM$. We suppose that a third party prepares an arbitrary quantum state on $\cH_{A}^{\ot n}$, where $n$ is a large positive number. A sender and receiver can then exploit some amount of shared random bits and classical communication to simulate the action of $n$ instances of the measurement $\cM$ (denoted by $\cM^{\ot n}$) on the chosen input state, in such a way that it becomes physically impossible for the third party, to whom the receiver passes along the measurement outcomes, to distinguish between the simulation and the ideal measurement $\cM^{\ot n}$ as $n$ becomes large (the third party can even keep the purifying system of a purification of the chosen input state in order to help the distinguishing task). This information-theoretic task is known as channel simulation (depicted in Figure~\ref{fig:measurement-compression}). By design, the information gained by the measurement is that relayed by the classical communication. Following~\cite{Winter04}, we call this task universal measurement compression. We prove that the optimal rate of classical communication is equal to $I(\cM)$, if sufficient shared randomness is available. In our opinion, this result establishes~\eqref{eq:max-info-gain} as the information-theoretic measure of information gain of a quantum measurement.

\begin{figure}[ptb]\label{fig:measurement-compression}
\begin{center}
\includegraphics[width=1.0\linewidth]{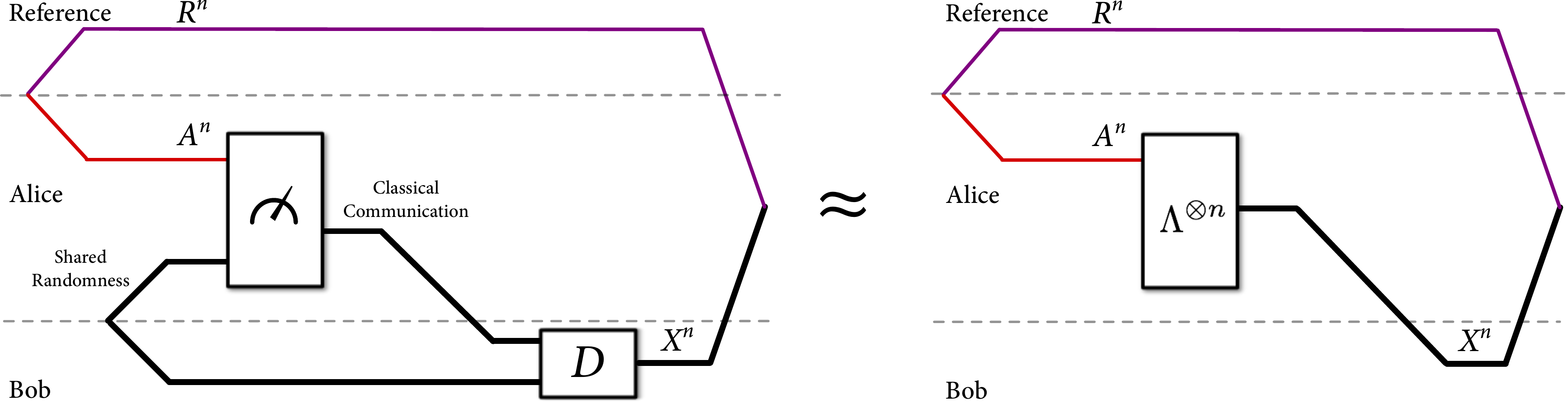}
\end{center}
\caption{Universal measurement compression.}{Simulation (left) of the measurement $\Lambda^{\ot n}$ (right). In the simulation, Alice uses shared randomness to perform a new measurement, whose result she communicates to Bob, such that Bob can recover the actual measurement output $X^n$ using the message and the shared randomness. If the simulation scheme works for any input, we can associate the amount of communication with the information gained by the measurement.}
\end{figure}

We also characterize the optimal rate region consisting of the rates of shared randomness and classical communication that are both necessary and sufficient for the existence of a measurement simulation, for the feedback and the non-feedback case. Note that if sufficient shared randomness is available, and we are only interested in quantifying the rate of classical communication, then there is no advantage of a non-feedback simulation over a feedback one, the optimal rate of classical communication is given by~\eqref{eq:max-info-gain}.


\subsection{Proof Ideas}\label{sec:meas_ideas}

The proof technique is the same as in our proofs of the classical and quantum reverse Shannon theorems (Sections~\ref{se:shannon} and~\ref{se:qshannon}). In fact, we can think of our approach here as a classicalized or dephased version of our proof of the quantum reverse Shannon theorem. In particular, the proof is also based on one-shot information theory and uses the smooth entropy formalism. In addition, we make use of the idea of classically coherent states. We say that a pure state $\proj{\psi}_{X_A X_B R}\in \cV(\cH_{X_A X_B R})$ is classically coherent with respect to systems $X_A X_B$ if there is an orthonormal basis $\{\ket{x}\}$ such that $ \ket{\psi}$ can be written in the form
\begin{equation}\label{eq:clacoh}
\ket{\psi}_{X_A X_B R} = \sum_x \sqrt{p_x} \,\ket{x x}_{X_A X_B} \ot \ket{\psi_x}_R\ ,
\end{equation}
for some probability distribution $\{p_x\}$, and pure states $\proj{\psi_x}_R\in\cV(\cH_{R})$. Harrow realized the importance of classically coherent states for quantum communication tasks~\cite{Harrow04}, while~\cite{Dupuis12,Szehr11} recently exploited this notion in devising a decoupling approach to the Holevo-Schumacher-Westmoreland coding theorem~\cite{Holevo98,Schumacher97} that is useful for our purposes here.

We begin by establishing a protocol known as classically coherent state merging, which is a variation of the well-known state merging protocol~\cite{Horodecki05,Horodecki07} specialized to classically coherent states. For the analysis we require a strong classical min-entropy randomness extractor against quantum side information (as discussed in Section~\ref{se:cc}). We then show how time-reversing this protocol and exchanging the roles of Alice and Bob leads to a protocol known as classically coherent state splitting. It suffices for our purposes for this protocol to use shared randomness and classical communication rather than entanglement and quantum communication, respectively. Generalizing this last protocol then leads to a one-shot channel simulation which is essentially optimal when acting on a single copy of a known state. Finally, we exploit the post-selection technique for quantum channels (see Appendix~\ref{ap:postselect}), and the aforementioned state splitting protocol to show that it suffices to simulate many instances of a measurement on a purification of a particular de Finetti quantum input state in order to guarantee that the simulation is asymptotically perfect when acting on an arbitrary quantum state. For the non-feedback case we additionally need the idea of randomness recycling~\cite{Bennett09}.


\subsection{Classically Coherent State Merging and State Splitting}\label{sec:meas_splitting}

We first establish one-shot protocols for state merging and state splitting of classically coherent quantum states. The main technical ingredient for the construction of these protocols are permutation based strong classical min-entropy extractors against quantum side information. Since we only care about the optimal output size of the extractor and not about its seed size, we will just use a family of pairwise independent permutations. Proposition~\ref{prop:ccperm} about the extractor properties of families of pairwise independent permutations can then be rephrased as follows.

\begin{corollary}\label{cor:ccperm}
Let $\eps>0$, $\rho_{XR}\in\cS_{\leq}(\cH_{XR})$ be classical on $X$ with respect to a basis $\{\ket{x}\}_{x\in X}$, and $\{P_{X}^{j}\}_{j\in J}$ be a family of pairwise independent permutations. Furthermore, consider a decomposition of the system $X$ into two subsystems $X_{1}$ and $X_{2}$, and define
\begin{align}
\sigma_{X_{1}R}^{j}=\trace_{X_{2}}\left[(P_{X}^{j}\ot\1_{R})\rho_{XR} (P_{X}^{j}\ot\1_{R})^{\dagger}\right]\ .
\end{align}
If we choose
\begin{align}\label{eq:ccperm}
\log |X_{1}|\leq H_{\min}(X|R)_{\rho}-2\log\frac{1}{\eps}\ ,
\end{align}
then we have that
\begin{align}
\frac{1}{|J|}\cdot\sum_{j\in J}\left\|\sigma_{X_{1}R}^{j}-\frac{\1_{X_{1}}}{|X_{1}|}\ot\rho_{R}\right\|_{1}\leq\eps\ .
\end{align}
\end{corollary}

We mention that classical randomness extractors against quantum side information already proved useful in quantum coding theory~\cite{Dupuis12}.

\begin{definition}[State Merging for classically coherent states]
Consider a bipartite system with parties Alice and Bob. Let $\eps>0$, and $\rho_{X_{A}X_{B}BR}\in\cV_{\leq}(\cH_{X_{A}X_{B}BR})$ be classically coherent on $X_{A}X_{B}$ with respect to the basis $\{\ket{x}\}_{x\in X_{A}X_{B}}$, where Alice controls $X_{A}$, Bob $X_{B}B$, and $R$ is a reference system. A quantum protocol $\cE$ is called an $\eps$-error state merging of $\rho_{X_{A}X_{B}BR}$ if it consists of applying local operations at Alice's side, sending $q$ qubits from Alice to Bob, local operations at Bob's side, and it outputs a state $\omega_{X_{B'}X_{B}BRX_{A_{1}}B_{1}}=(\cE\ot\cI_{R})(\rho_{X_{A}X_{B}BR})$ such that
\begin{align}
\omega_{X_{B'}X_{B}BRX_{A_{1}}B_{1}}\approx_{\eps}\cI_{X_{A}\rightarrow X_{B'}}(\rho_{X_{A}X_{B}BR})\ot\Phi^{E}_{X_{A_{1}}B_{1}}\ ,
\end{align}
where $\Phi^{E}_{X_{A_{1}}B_{1}}$ is a maximally entangled state of Schmidt rank $E$. The quantity $q$ is called quantum communication cost, and $e=\lfloor\log E\rfloor$ entanglement gain.
\end{definition}

\begin{lemma}\label{lem:ccmerging}
Let $\eps>0$, and $\rho_{X_{A}X_{B}BR}\in\cV_{\leq}(\cH_{X_{A}X_{B}BR})$ be classically coherent on $X_{A}X_{B}$ with respect to the basis $\{\ket{x}\}_{x\in X_{A}X_{B}}$. Then, there exists an $\eps$-error state merging protocol for $\rho_{X_{A}X_{B}BR}$ with quantum communication cost
\begin{align}
q=\left\lceil H_{0}(X_{A})_{\rho}-H_{\min}(X_{A}|R)_{\rho}+4\log\frac{1}{\eps}\right\rceil\ ,
\end{align}
and entanglement gain
\begin{align}
e=\left\lfloor H_{\min}(X_{A}|R)_{\rho}-4\log\frac{1}{\eps}\right\rfloor\ .
\end{align}
\end{lemma}

\begin{figure}[ptb]\label{fig:state-merging-classically-coherent}
\begin{center}
\includegraphics[width=0.6\linewidth]{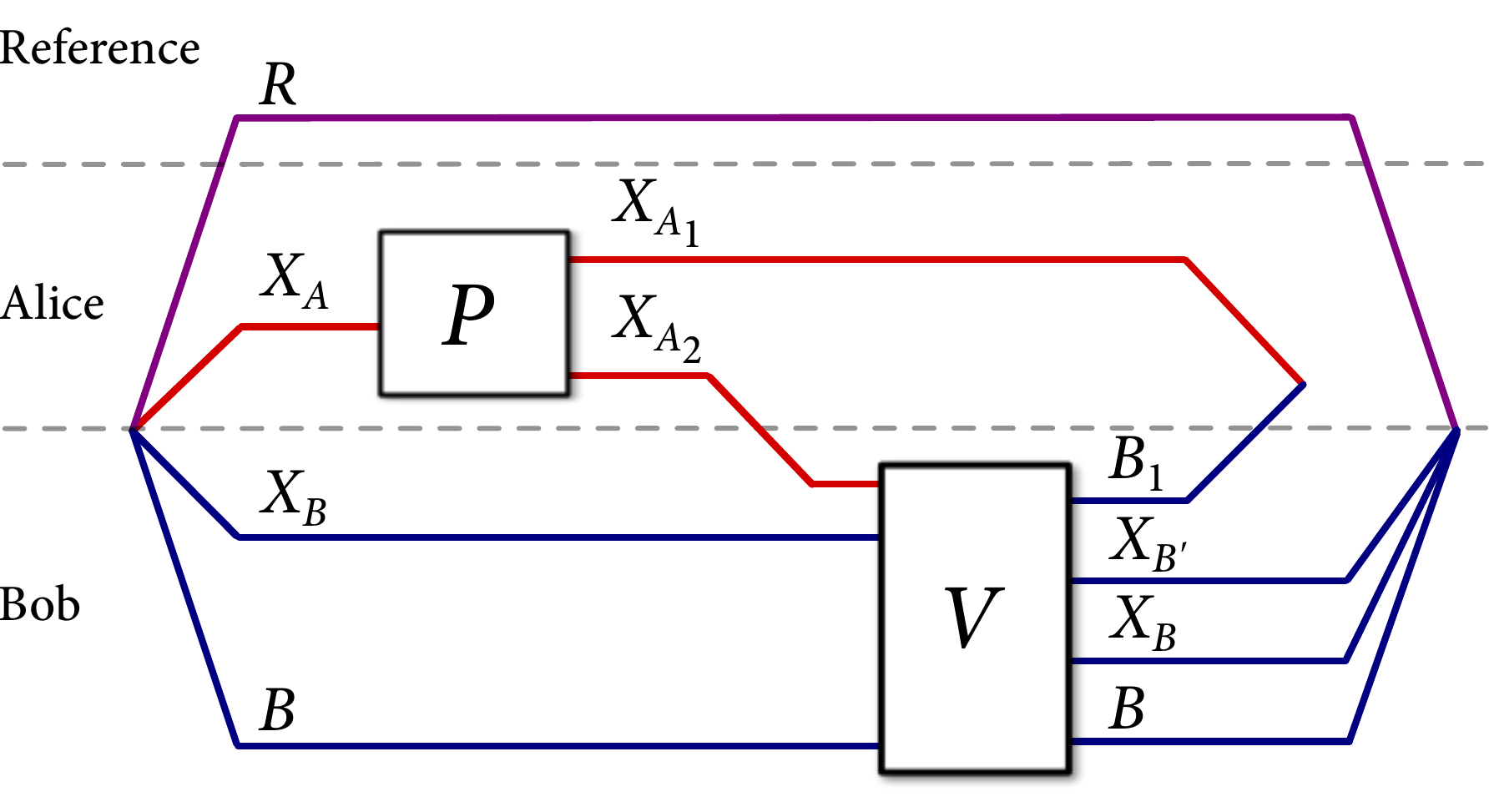}
\end{center}
\caption{State merging of classically coherent states.}{The protocol from the proof of Lemma~\ref{lem:ccmerging} for state merging of a classically coherent state on systems $R X_{A}X_{B}B$. The operation $P$ is a permutation of states in the orthonormal basis $\{ \ket{x} \}$ of $X_{A}$, and it also splits $X_{A}$ into two subsystems. The operation $V$ is an isometry guaranteed by Uhlmann's theorem to complete the merging task, while also generating entanglement between Alice and Bob.}
\end{figure}

\begin{proof}
The intuition is as follows. First Alice applies a permutation $P_{X_{A}\rightarrow X_{A_{1}}X_{A_{2}}}$ in the basis $\{\ket{x}\}_{x\in X_{A}}$ that also splits the output into two subsystems $X_{A_1}$ and $X_{A_2}$. Then, she sends $X_{A_{2}}$ to Bob, who finally performs a local isometry $V_{X_{A_{2}}X_{B}B\rightarrow X_{B'}X_{B}BB_{1}}$. After Alice applies the permutation, the state on $X_{A_{1}}R$ is approximately given by $\frac{\1_{X_{A_{1}}}}{|X_{A_{1}}|}\ot\rho_{R}$ and Bob holds a purification of this. But $\frac{\1_{X_{A_{1}}}}{|X_{A_{1}}|}\ot\rho_{R}$ is the reduced state of $\rho_{X_{B'}X_{B}BR}\ot\Phi^{E}_{X_{A_{1}}B_{1}}$, and since all purifications are equivalent up to local isometries, there exists an isometry $V_{X_{A_{2}}X_{B}B\rightarrow X_{B'}X_{B}BB_{1}}$ on Bob's side that transforms the state into $\rho_{X_{B'}X_{B}BR}\ot\Phi^{E}_{X_{A_{1}}B_{1}}$. Figure~\ref{fig:state-merging-classically-coherent} depicts this protocol.

More formally, let $X_{A}=X_{A_{1}}X_{A_{2}}$ with
\begin{align}
\log|X_{A_{2}}|=\left\lceil\log|X_{A}|-H_{\min}(X_{A}|R)_{\rho}+4\log\frac{1}{\eps}\right\rceil\ .
\end{align}
According to our results about permutation based strong classical min-entropy extractors against quantum side information in Section~\ref{se:cc} (cf.~Corollary~\ref{cor:ccperm}), there exists a permutation $P_{X_{A}\rightarrow X_{A_{1}}X_{A_{2}}}$ such that we have for
\begin{align}
\sigma_{X_{A_{1}}X_{A_{2}}BR}=P_{X_{A}\rightarrow X_{A_{1}}X_{A_{2}}}(\rho_{X_{A}X_{B}BR})\ ,
\end{align}
that
\begin{align}\label{eq:random}
\left\|\sigma_{X_{A_{1}}R}-\frac{\1_{X_{A_{1}}}}{|X_{A_{1}}|}\ot\rho_{R}\right\|_{1}\leq\eps^{2}\ .
\end{align}
By an upper bound of the purified distance in terms of the trace distance (Lemma~\ref{lem:pdbounds}), this implies $\sigma_{X_{A_{1}}R}\approx_{\eps}\frac{\1_{X_{A_{1}}}}{|X_{A_{1}}|}\ot\rho_{R}$. Alice applies this permutation $P_{X_{A}\rightarrow X_{A_{1}}X_{A_{2}}}$ and then sends $X_{A_{2}}$ to Bob; therefore
\begin{align}
q=\left\lceil\log|X_{A}|-H_{\min}(X_{A}|R)_{\rho}+4\log\frac{1}{\eps}\right\rceil\ .
\end{align}
By Uhlmann's theorem~\cite{Uhlmann76,Jozsa94} there exists an isometry $V_{X_{A_{2}}X_{B}B\rightarrow X_{B'}X_{B}BB_{1}}$ such that
\begin{align}\label{eq:uhlmann}
P\Big(\sigma_{X_{A_{1}}R},\frac{\1_{X_{A_{1}}}}{|X_{A_{1}}|}\ot\rho_{R}\Big)=P\Big(&V_{X_{A_{2}}X_{B}B\rightarrow X_{B'}X_{B}BB_{1}}(\sigma_{X_{A_{1}}X_{A_{2}}X_{B}BR}),\notag\\
&\Phi^{E}_{X_{A_{1}}B_{1}}\ot\rho_{X_{B'}X_{B}BR}\Big)\ .
\end{align}
Hence, the entanglement gain is given by
\begin{align}
e=\left\lfloor H_{\min}(X_{A}|R)_{\rho}-4\log\frac{1}{\eps}\right\rfloor\ .
\end{align}
Now, if $\rho_{X_{A}}$ has full rank, this is already what we want. In general $\log\trace\left[\rho_{X_{A}}^{0}\right]=\log|X_{\hat{A}}|\leq\log|X_{A}|$. But in this case we can restrict $X_{A}$ to the subspace $X_{\hat{A}}$ on which $\rho_{X_{A}}$ has full rank, i.e., those $x$ for which $p_x\neq 0$.
\end{proof}

\begin{definition}[State Splitting for classically coherent states]\label{def:ccsplitting}
Consider a bipartite scenario with parties Alice and Bob. Let $\eps>0$, and $\rho_{AX_{A}X_{A'}R}\in\cV_{\leq}(\cH_{AX_{A}X_{A'}R})$ be classically coherent on $X_{A}X_{A'}$ with respect to the basis $\{\ket{x}\}_{x\in X_{A}X_{A'}}$, where Alice controls $AX_{A}X_{A'}$, and $R$ is a reference system. Furthermore let $\Phi^{E}_{A_{1}B_{1}}$ be a maximally entangled state of Schmidt rank $E$ shared between Alice and Bob. A quantum protocol $\cE$ is called an $\eps$-error state splitting of $\rho_{AX_{A}X_{A'}R}$ if it consists of applying local operations at Alice's side, sending $q$ qubits from Alice to Bob, local operations at Bob's side, and it outputs a state $\omega_{AX_{A}X_{B}R}=(\cE\ot\cI_{R})(\rho_{AX_{A}X_{A'}R}\ot\Phi^{E}_{A_{1}B_{1}})$ such that
\begin{align}
\omega_{AX_{A}X_{B}R}\approx_{\eps}\cI_{X_{A'}\rightarrow X_{B}}(\rho_{AX_{A}X_{A'}R})\ .
\end{align}
The quantity $q$ is called the quantum communication cost, and $e=\lfloor\log E\rfloor$ the entanglement cost.
\end{definition}

\begin{lemma}\label{lem:ccsplitting}
Let $\eps>0$, and $\rho_{AX_{A}X_{A'}R}\in\cV_{\leq}(\cH_{AX_{A}X_{A'}R})$ be classically coherent on $X_{A}X_{A'}$ with respect to the basis $\{\ket{x}\}_{x\in X_{A}X_{A'}}$. Then, there exists an $\eps$-error state splitting protocol for $\rho_{AX_{A}X_{A'}R}$ with quantum communication cost
\begin{align}
q=\left\lceil H_{0}(X_{A'})_{\rho}-H_{\min}(X_{A'}|R)_{\rho}+4\log\frac{1}{\eps}\right\rceil\ ,
\end{align}
and entanglement cost
\begin{align}
e=\left\lfloor H_{\min}(X_{A'}|R)_{\rho}-4\log\frac{1}{\eps}\right\rfloor\ .
\end{align}
\end{lemma}

\begin{proof}
We get the desired state splitting protocol by time-reversing the state merging protocol of Lemma~\ref{lem:ccmerging} and interchanging the roles of Alice and Bob. Figure~\ref{fig:state-splitting-classically-coherent}(a) depicts the state splitting protocol for classically coherent states. More precisely, we first define an isometry $V_{X_{A'_{2}}X_{A}A\rightarrow X_{A'}X_{A}AA_{1}}$, analogously to $V_{X_{A_{2}}X_{B}B\rightarrow X_{B'}X_{B}BB_{1}}$ of~\eqref{eq:uhlmann} in the state merging protocol. Because all isometries are injective, we can define an inverse of $V$ acting on the image of $V$ (which we denote by $\mathrm{Im}(V)$). The inverse is again an isometry and we denote it by $V^{-1}_{\mathrm{Im}(V)\rightarrow X_{A'_{2}}X_{A}A}$. The protocol starts by measuring the $AX_AX_{A'}A_1$ systems to decide whether $\rho_{AX_{A}X_{A'}}\ot\Phi^{E}_{A_{1}}\in\mathrm{Im}(V)$ or not. If so, the protocol proceeds by applying the isometry $V^{-1}_{\mathrm{Im}(W)\rightarrow X_{A'_{2}}X_{A}A}$, but otherwise the state is discarded and replaced with $\proj{0}_{X_{A'_{2}}X_{A}A}$. This step is necessary because the output of merging is not exactly $\rho_{AX_AX_{A'}R}$. 
 The next step is to send $X_{A'_{2}}$ to Bob, who then applies the permutation $P^{-1}_{X_{A'_{2}}B_{1}\rightarrow X_{B}}$ defined analogously to $P_{X_{A}\rightarrow X_{A_{1}}X_{A_{2}}}$ in~\eqref{eq:random}. By the monotonicity of the purified distance, we get a state that is $\eps$-close to $\cI_{X_{A'}\rightarrow X_{B}}(\rho_{AX_{A}X_{A'}R})$.
 \end{proof}

\begin{figure}[ptb]\label{fig:state-splitting-classically-coherent}
\begin{center}
\includegraphics[width=1.0\linewidth]{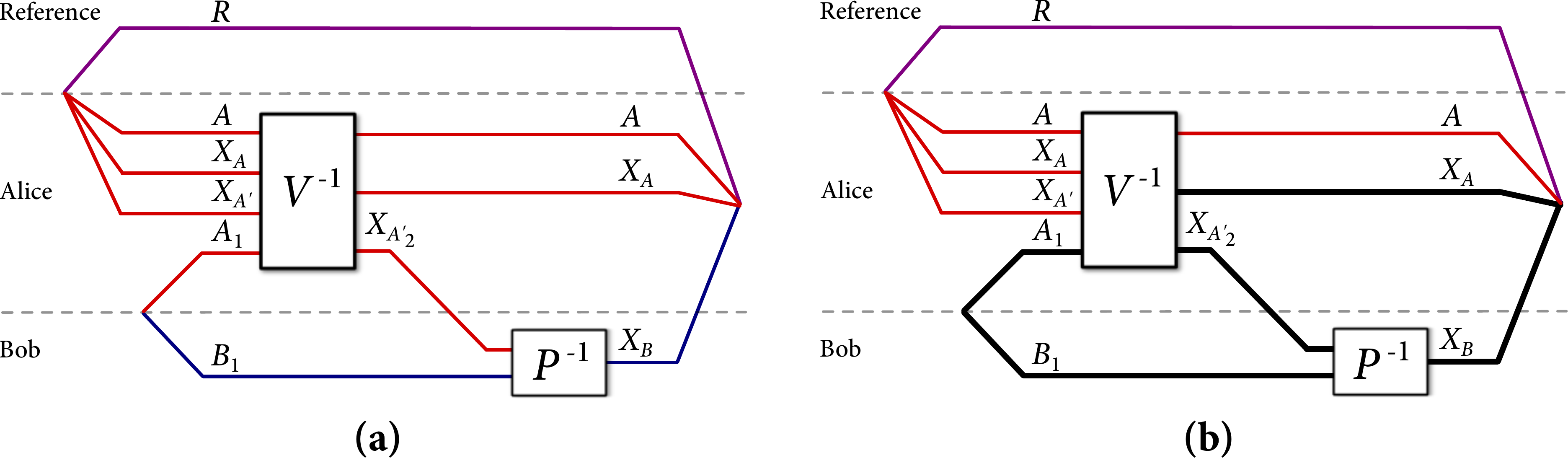}
\end{center}
\caption{State splitting of classically coherent states.}{\textbf{(a)} A simple protocol for state splitting obtained by time-reversing the state merging protocol of Lemma~\ref{lem:ccmerging} and interchanging the roles of Alice and Bob. \textbf{(b)} If it is not necessary to maintain the quantum coherence of the $X$ systems (if they can be dephased to classical registers), then the state splitting protocol can exploit shared randomness and classical communication instead of entanglement and quantum communication, respectively.}
\end{figure}

If we are not concerned with the coherence of the registers $X_{A}$ and $X_{B}$ shared between Alice and Bob, then the protocol given above (Lemma~\ref{lem:ccsplitting}) also works if the entanglement assistance and the quantum communication are replaced by the same amount of shared randomness assistance and classical communication, respectively. More precisely, we define the following.

\begin{definition}[Classical state splitting of classically coherent states]\label{def:csplitting}
Consider a bipartite system with parties Alice and Bob. Let $\eps>0$, and $\rho_{AX_{A}X_{A'}R}\in\cV_{\leq}(\cH_{AX_{A}X_{A'}R})$ be classically coherent on $X_{A}X_{A'}$ with respect to the basis $\{\ket{x}\}_{x\in X_{A}X_{A'}}$, where Alice controls $AX_{A}X_{A'}$, and $R$ is a reference system. Furthermore let $\overline{\Phi}^{S}_{X_{A_{1}}X_{B_{1}}}$ denote $S$ bits of shared randomness shared between Alice and Bob. A quantum protocol $\cE$ is called an $\eps$-error classical state splitting of $\rho_{AX_{A}X_{A'}R}$ if it consists of applying local operations at Alice's side, sending $c$ bits from Alice to Bob, local operations at Bob's side, and it outputs a state $\omega_{AX_{A}X_{B}R}=(\cE\ot\cI_{R})(\rho_{AX_{A}X_{A'}R}\ot\overline{\Phi}^{S}_{X_{A_{1}}X_{B_{1}}})$ such that
\begin{align}
\omega_{AX_{A}X_{B}R}\approx_{\eps}\sum_{x}\bra{x}\rho_{AX_{A}X_{A'}R}\ket{x}_{X_{A'}}\ot\proj{x}_{X_{B}}\ .
\end{align}
The quantity $c$ is called the classical communication cost, and $s=\lfloor\log S\rfloor$ shared randomness cost.
\end{definition}

Using the achievability of state splitting of classically coherent states (Lemma~\ref{lem:ccsplitting}) we get the following.

\begin{corollary}\label{cor:csplitting}
Let $\eps>0$, and $\rho_{AX_{A}X_{A'}R}\in\cV_{\leq}(\cH_{AX_{A}X_{A'}R})$ be classically coherent on $X_{A}X_{A'}$ with respect to the basis $\{\ket{x}\}_{x\in X_{A}X_{A'}}$. Then, there exists a classical $\eps$-error state splitting protocol for $\rho_{AX_{A}X_{A'}R}$ with classical communication cost
\begin{align}
c=\left\lceil H_{0}(X_{A'})_{\rho}-H_{\min}(X_{A'}|R)_{\rho}+4\log\frac{1}{\eps}\right\rceil\ ,
\end{align}
and shared randomness cost
\begin{align}
s=\left\lfloor H_{\min}(X_{A'}|R)_{\rho}-4\log\frac{1}{\eps}\right\rfloor\ .
\end{align}
\end{corollary}

\begin{proof}
Note that it is sufficient to find a protocol for state splitting of classically coherent states (as in Definition~\ref{def:ccsplitting}) that only works up to random phase flips on the $X_{B}$ register. These random phase flips then commute with the action of the permutation that takes systems $B_1$ and $X_{A'_2}$ to $X_B$. Thus, if we use the protocol for state splitting of classically coherent states described before (Lemma~\ref{lem:ccsplitting}), random phase flips on $X_{B}$ are the same as random phase flips on $X_{A'_{2}}B_{1}$ before the permutation $P^{-1}_{X_{A'_{2}}B_{1}\rightarrow X_{B}}$ is applied. Since random phase flips on $B_{1}$ just transform the maximally entangled state $\Phi_{A_{1}B_{1}}$ to shared randomness $\overline{\Phi}_{X_{A_{1}}X_{B_{1}}}$ of the same size (with the relabeling of $A_{1}B_{1}$ to $X_{A_{1}}X_{B_{1}}$), and they dephase the quantum system $X_{A'_{2}}$ to a classical system, the protocol of Lemma~\ref{lem:ccsplitting} also works for classical state splitting of classically coherent states.
\end{proof}

Note that the above idea is similar to how Hsieh {\it et al.}~recovered the Holevo-Schumacher-Westmoreland coding theorem for classical communication from a protocol for entanglement-assisted classical communication~\cite{Hsieh08}, simply by dephasing shared entanglement to common randomness and replacing random unitaries with random permutations. 

However, the classical communication cost of this protocol is not good enough (for the general one-shot case considered here). To improve this, we use the same idea as in our proof of the quantum reverse Shannon theorem (Section~\ref{se:qshannon}). However, there is no entanglement spread problem here because shared randomness can just be conditionally diluted (in contrast to entanglement). The following lemma is the crucial ingredient for the proof of our main result in this section: universal measurement compression (Theorem~\ref{thm:measmain}).

\begin{theorem}\label{thm:splitting}
Let $\eps,\eps'>0$, and $\rho_{AX_{A}X_{A'}R}\in\cV(\cH_{AX_{A}X_{A'}R})$ be classically coherent on $X_{A}X_{A'}$ with respect to the basis $\{\ket{x}\}_{x\in X_{A}X_{A'}}$. Then, there exists a classical $(\eps+\eps'+\sqrt{12\eps'}+|X_{A'}|^{-1})$-error state splitting protocol for $\rho_{AX_{A}X_{A'}R}$ with
\begin{align}\label{eq:msplitting}
&c\leq I_{\max}^{\eps'}(X_{A'}:R)_{\rho}+4\log\frac{1}{\eps}+4+\log\log|X_{A'}|\ ,\\
&c+s\leq H_{\max}^{\eps'}(X_{A'})_{\rho}+2\log\frac{1}{\eps'}+\log\log|X_{A'}|\ ,
\end{align}
where $c$ denotes the classical communication cost, and $s$ the shared randomness cost.
\end{theorem}

\begin{proof}
The idea for the protocol is as follows. Let $\rho_{AX_{A}X_{A'}R}=\proj{\rho}_{AX_{A}X_{A'}R}$ with
\begin{align}
\ket{\rho}_{AX_{A}X_{A'}R}=\sum_{x}\sqrt{p_{x}}\cdot\ket{xx}_{X_{A}X_{A'}}\ot\ket{\rho^{x}}_{AR}\ .
\end{align}
First, in our proof, we disregard all the $x$ with $p_{x}\leq|X_{A'}|^{-2}$. This introduces an error $|X_{A'}|^{-1}$, but the error at the end of the protocol is still upper bounded by $|X_{A'}|^{-1}$ due to the monotonicity of the purified distance. As the next step, we let Alice perform a measurement $W_{X_{A'}\rightarrow X_{A'}Y_{A}}$ with roughly $2\cdot\log|X_{A'}|$ measurement outcomes in the
basis $\{\ket{x}\}_{x\in X_{A'}}$. That is, the state after the measurement is of the form
\begin{align}
\omega_{AX_{A}X_{A'}RY_{A}}=\sum_{y}q_{y}\cdot\rho^{y}_{AX_{A}X_{A'}R}\ot\proj{y}_{Y_{A}}\ ,
\end{align}
where the index $y$ indicates which measurement outcome occurs, $q_{y}$ denotes its probability, and $\rho^{y}_{AX_{A}X_{A'}R}$ is the corresponding post-measurement state. Then, conditioned on the index $y$, we use the classical state splitting protocol for classically coherent states from Lemma~\ref{cor:csplitting} for each state $\rho^{y}_{AX_{A}X_{A'}R}$, and denote the corresponding classical communication cost and shared randomness cost by $c_{y}$ and $s_{y}$, respectively. The total amount of classical communication we need for this is given by $\max_{y}c_{y}$, plus the amount needed to send the register $Y_{A}$ (which is of order $\log\log|X_{A'}|$).
The sum cost is given by $\max_{y}c_{y} + s_{y}$ (along with the amount for sending $Y_{A}$). This completes the description of the classical state splitting protocol for $\rho_{AX_{A}X_{A'}R}$. All that remains to do is to bring the expression for the classical communication cost and the sum cost into the right form. In the following, we describe the proof in detail.

Let $Q=\lceil2\cdot\log|X_{A'}|-1\rceil$, $Y=\{0,1,\ldots,Q,(Q+1)\}$ and let $\{T_{X_{A'}}^{y}\}_{y\in Y}$ be a collection of projectors on $X_{A'}$ defined as
\begin{align}
T_{X_{A'}}^{Q+1}=\sum_{\substack{x\\0\leq p_{x}\leq2^{-2\log|X_{A'}|}}}\proj{x}_{X_{A'}}\ ,\qquad T_{X_{A'}}^{Q}=\sum_{\substack{x\\ 2^{-2\log|X_{A'}|} \leq p_{x}\leq2^{-Q} }}\proj{x}_{X_{A'}}\ ,
\end{align}
and for $y=0,1,\dots,(Q-1)$ as
\begin{align}
T_{X_{A'}}^{y}=\sum_{\substack{x\\2^{-(y+1)}\leq p_{x}\leq2^{-y}}}\proj{x}_{X_{A'}}\ .
\end{align}
These define a measurement
\begin{align}\label{eq:preproc}
W_{X_{A'}\rightarrow X_{A'}Y_{A}}(\cdot)=\sum_{y\in Y}T_{X_{A'}}^{y}(\cdot)T_{X_{A'}}^{y}\ot\proj{y}_{Y_{A}}\ ,
\end{align}
where the vectors $\ket{y}_{Y_{A}}$ form an orthonormal basis, and $Y_{A}$ is at Alice's side. Furthermore let
\begin{align}
q_{y} & =\trace\left[T_{X_{A'}}^{y}\rho_{X_{A'}}\right]\ ,\notag\\
\rho_{AX_{A}X_{A'}R}^{y} & =q_{y}^{-1}\cdot T_{X_{A'}}^{y}\rho_{AX_{A}X_{A'}R}T_{X_{A'}}^{y}\ ,
\end{align}
and define the sub-normalized state
\begin{align}
\bar{\rho}_{AX_{A}X_{A'}R}=\sum_{y=0}^{Q}q_{y}\cdot\rho^{y}_{AX_{A}X_{A'}R}\ .
\end{align}
We have
\begin{align}\label{eq:cutoff}
P(\bar{\rho}_{AX_{A}X_{A'}R},\rho_{AX_{A}X_{A'}R})&=\sqrt{1-F^{2}(\bar{\rho}_{AX_{A}X_{A'}R},\rho_{AX_{A}X_{A'}R})}\notag\\
& \leq \sqrt{1-\big(\sum_{y=0}^{Q}q_{y}\big)^{2}}\leq q_{Q+1}\leq|X_{A'}|\cdot2^{-2\log|X_{A'}|}=|X_{A'}|^{-1}\ .
\end{align}
We proceed by defining the operations that we need for the classical state splitting protocol for $\bar{\rho}_{AX_{A}X_{A'}R}$. We want to use the $\eps$-error classical state splitting protocol from Corollary~\ref{cor:csplitting} for each $\rho^{y}_{AX_{A}X_{A'}R}$. For $y=0,1,\ldots,Q$ this protocol has a classical communication cost
\begin{align}
c_{y}\leq H_{0}(X_{A'})_{\rho^{y}}-H_{\min}(X_{A'}|R)_{\rho^{y}}+4\log\frac{1}{\eps}+1\ ,
\end{align}
and sum cost
\begin{align}\label{eq:sumcost}
c_{y}+s_{y}\leq H_{0}(X_{A'})_{\rho^{y}}\ ,
\end{align}
where $s_{y}$ denotes the shared randomness cost.

For $X_{A_{1}}$ on Alice's side, $X_{B_{1}}$ on Bob's side, and $X_{A_{1}^{y}}$, $X_{B_{1}^{y}}$ $2^{s_{y}}$-dimensional subspaces of $X_{A_{1}}$, $X_{B_{1}}$ respectively, the classical state splitting protocol from Corollary~\ref{cor:csplitting} has basically the following form: apply some isometry $V_{AX_{A'}X_{A}X_{A_{1}^{y}}\rightarrow AX_{(A'_{2})^{y}}X_{A}}$ on Alice's side, send $X_{(A'_{2})^{y}}$ from Alice to Bob (relabel it to $X_{B_{2}^{y}}$), and then apply some isometry $U_{X_{B_{1}^{y}}X_{B_{2}^{y}}\rightarrow B}$ on Bob's side. As the next ingredient, we define the operations that supply the shared randomness of size $s_{y}$. For $y=0,1,\ldots,Q$, let $S_{X_{A_{1}^{y}}}$ and $S_{X_{B_{1}^{y}}}$ be the local operations at Alice's and Bob's side respectively, that put shared randomness of size $s_{y}$ on $X_{A_{1}^{y}}X_{B_{1}^{y}}$.

\begin{figure}[ptb]\label{fig:state-splitting-classically-coherent-max-info}
\begin{center}
\includegraphics[width=0.8\linewidth]{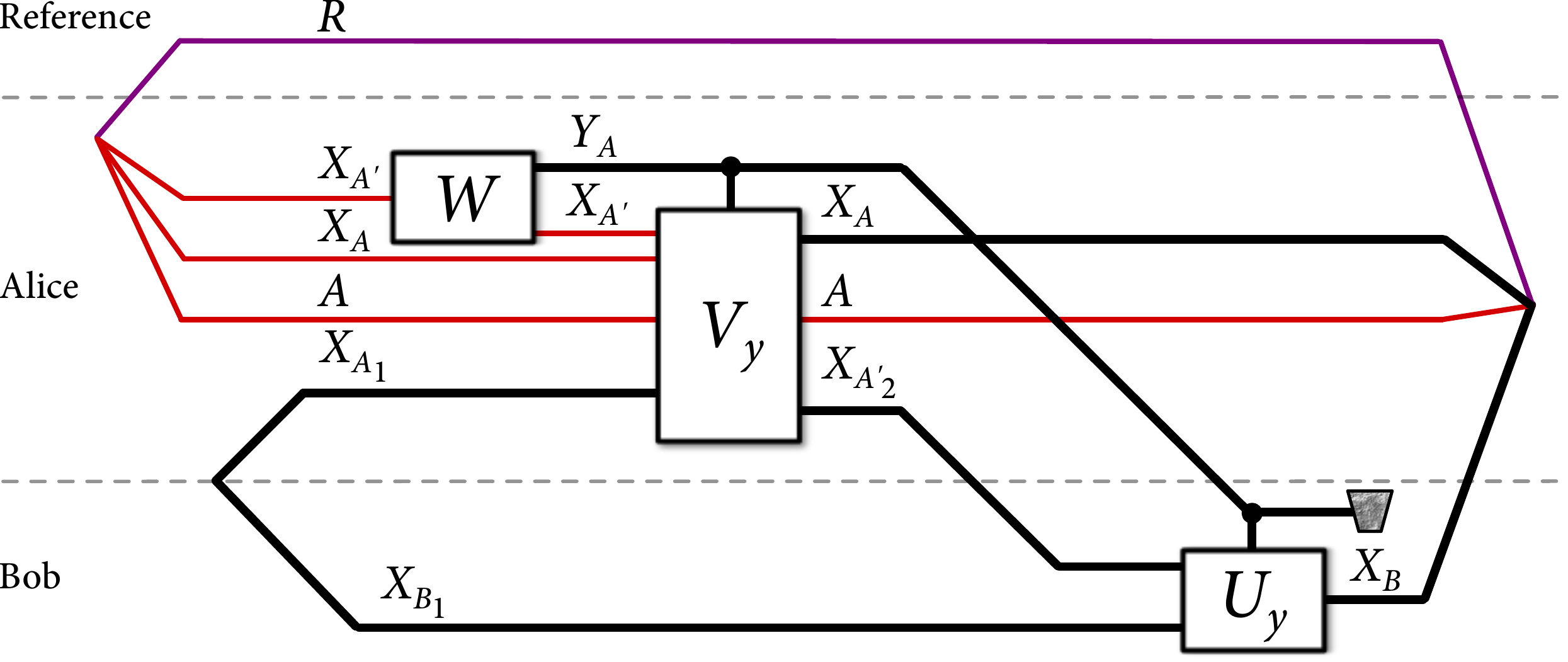}
\end{center}
\caption{Optimal classically coherent state splitting.}{Our final one-shot protocol for state splitting that achieves the smooth max-information rate of Theorem~\ref{thm:splitting}.}
\end{figure}

We are now ready to put the steps together and give the protocol for classical state splitting of $\bar{\rho}_{AX_{A}X_{A'}R}$ (depicted in Figure~\ref{fig:state-splitting-classically-coherent-max-info}). Alice applies the measurement $W_{X_{A'}\rightarrow X_{A'}Y_{A}}$ from~\eqref{eq:preproc} followed by
\begin{align}
S_{A_{1}Y_{A}}=\sum_{y=1}^{Q}S_{X_{A_{1}^{y}}}\ot\proj{y}_{Y_{A}}\ ,
\end{align}
and the isometry 
\begin{align}
V_{AX_{A}X_{A'}X_{A_{1}}Y_{A}\rightarrow AX_{A'_{2}}X_{A}Y_{A}}=\sum_{y=0}^{Q}V_{AX_{A}X_{A'}X_{A_{1}^{y}}\rightarrow AX_{(A'_{2})^{y}}X_{A}}\ot\proj{y}_{Y_{A}}\ .
\end{align}
Afterwards she sends $X_{A'_{2}}$ and $Y_{A}$, that is
\begin{align}\label{eq:ccost}
c\leq\max_{y}\left[H_{0}(X_{A'})_{\rho^{y}}-H_{\min}(X_{A'}|R)_{\rho^{y}}\right]+4\log\frac{1}{\eps}+1+\log\left\lceil2\cdot\log|X_{A'}|\right\rceil
\end{align}
bits to Bob (and we now rename $X_{A'_{2}}$ to $X_{B_{2}}$ and $Y_{A}$ to $Y_{B}$). Then Bob applies
\begin{align}
S_{B_{1}Y_{B}}=\sum_{y=1}^{Q}S_{B_{1}^{y}}\ot\proj{y}_{Y_{B}}\ ,
\end{align}
followed by the isometry
\begin{align}\label{eq:decoding}
U_{X_{B_{1}}X_{B_{2}}Y_{B}\rightarrow X_{B}Y_{B}}=\sum_{y=0}^{Q}U_{X_{B_{1}^{y}}X_{B_{2}^{y}}\rightarrow X_{B}}\ot\proj{y}_{Y_{B}}\ .
\end{align}
We obtain a sub-normalized state
\begin{align}
\sigma_{AX_{A}X_{B}RY_{B}}=\sum_{y=0}^{Q}q_{y}\cdot\tilde{\rho}^{y}_{AX_{A}X_{B}R}\ot\proj{y}_{Y_{B}}\ ,
\end{align}
with $\tilde{\rho}^{y}_{AX_{A}X_{B}R}\approx_{\eps}\cI_{X_{A'}\rightarrow X_{B}}(\rho^{y}_{AX_{A}X_{A'}R})$ for $y=0,1,\ldots,Q$. By the (quasi) convexity of the purified distance in its arguments (Lemma~\ref{lem:pdconvex}), and the monotonicity of the purified distance (Lemma~\ref{lem:pdmono}), we have
\begin{align}
\sigma_{AX_{A}X_{B}R}\approx_{\eps}\cI_{X_{A'}\rightarrow X_{B}}(\bar{\rho}_{AX_{A}X_{A'}R})\ .
\end{align}
Hence, we have shown the existence of an $\eps$-error classical state splitting protocol for $\bar{\rho}_{AX_{A}X_{A'}R}$ with classical communication cost as in~\eqref{eq:ccost}. But by the monotonicity of the purified distance, and the triangle inequality for the purified distance, this implies the existence of an $\left(\eps+|X_{A'}|^{-1}\right)$-error classical state splitting protocol for $\rho_{AX_{A}X_{A'}R}$, with the same classical communication cost as in~\eqref{eq:ccost}.

We now proceed by simplifying~\eqref{eq:ccost}. We have $H_{0}(X_{A'})_{\rho^{y}}\leq H_{\min}(X_{A'})_{\rho^{y}}+1$ for $y=0,1,\ldots,Q$ as can be seen as follows,
\begin{align}
2^{-(y+1)}\leq\lambda_{\min}(q_y \cdot \rho_{X_{A'}}^{y})\leq\rank^{-1}\left(q_y \cdot  \rho_{X_{A'}}^{y}\right)\leq\lambda_{1}\left(q_y \cdot  \rho_{X_{A'}}^{y}\right)\leq2^{-y}\ ,
\end{align}
where $\lambda_{\min}(\rho_{X_{A'}}^{y})$ denotes the smallest non-zero eigenvalue of $\rho_{X_{A'}}^{y}$. Thus,
\begin{align}\label{eq:minmaxequiv}
\rank\left(q_y \cdot  \rho_{X_{A'}}^{y}\right)\leq2^{y+1}=2^{y}\cdot2\leq\lambda_{1}\left(q_y\cdot\rho_{X_{A'}}^{y}\right)^{-1}\cdot2\ ,
\end{align}
and this is equivalent to the claim. Hence, we get an $(\eps+|X_{A'}|^{-1})$-error classical state splitting protocol for $\rho_{AX_{A}X_{A'}R}$ with classical communication cost
\begin{align}
c&\leq\max_{y}\big(H_{\min}(X_{A'})_{\rho^{y}}-H_{\min}(X_{A'}|R)_{\rho^{y}}\big)+4\log\frac{1}{\eps}+2+\log\left\lceil2\cdot\log|X_{A'}|\right\rceil\notag\\
&\leq\max_{y}\big(H_{\min}(X_{A'})_{\rho^{y}}-H_{\min}(X_{A'}|R)_{\rho^{y}}\big)+4\log\frac{1}{\eps}+4+\log\log|X_{A'}|\ .
\end{align}
Using a lower bound for the max-information in terms of min-entropies (Lemma~\ref{lem:maxbounds}), and the behavior of the max-information under projective measurements (Corollary \ref{cor:maxproj}) this simplifies to
\begin{align}
c&\leq\max_{y}I_{\max}(X_{A'}:R)_{\rho^{y}}+4\log\frac{1}{\eps}+4+\log\log|X_{A'}|\notag\\
&\leq I_{\max}(X_{A'}:R)_{\rho}+4\log\frac{1}{\eps}+4+\log\log|X_{A'}|\ .\label{eq:ccost2}
\end{align}
Furthermore, it easily seen from~\eqref{eq:sumcost} that
\begin{align}\label{eq:sumcost2}
c+s&\leq H_{0}(X_{A'})_{\rho}+2+\log\log|X_{A'}|\ .
\end{align}

As the last step, we reduce the classical communication and shared randomness cost by smoothing the max-information and the R\'enyi zero-entropy in~\eqref{eq:ccost2} and~\eqref{eq:sumcost2}, respectively. For that, we do not apply the protocol as described above to the state $\rho_{AX_{A}X_{A'}R}$, but pretend that we have another classically coherent (sub-normalized) state $\bar{\rho}_{AX_{A}X_{A'}R}$ that is $(\sqrt{12\eps'}+\eps')$-close to $\rho_{AX_{A}X_{A'}R}$, and then apply the protocol for $\bar{\rho}_{AX_{A}X_{A'}R}$. By the monotonicity of the purified distance (Lemma~\ref{lem:pdmono}) the additional error term from this is upper bounded by $\sqrt{12\eps'}+\eps'$, and by the triangle inequality for the purified distance this results in a total accuracy of $\eps+\eps'+\sqrt{12\eps'}+|X_{A'}|^{-1}$. We now proceed by defining $\bar{\rho}_{AX_{A}X_{A'}R}$. Let $\tilde{\rho}_{X_{A'}R}\in\cB^{\eps'}(\rho_{X_{A'}R})$ such that
\begin{align}\label{eq:finalsmooth}
I_{\max}^{\eps'}(X_{A'}:R)_{\rho}=I_{\max}(X_{A'}:R)_{\tilde{\rho}}\ .
\end{align}
Furthermore, since the R\'enyi zero-entropy can be smoothed by applying a projection and is equivalent to the smooth max-entropy (Lemma~\ref{lem:hrhmaxequiv}), there exists $\Pi_{X_{A'}}\in\cP(X_{A'})$ with $\Pi_{X_{A'}}\leq\1_{X_{A'}}$ such that
\begin{align}\label{eq:simsmoothing}
H_{\max}^{2\eps'}(X_{A'})_{\tilde{\rho}}+2\log\frac{1}{2\eps'}\geq H_{0}(X_{A'})_{\bar{\rho}}\ ,
\end{align}
with $\bar{\rho}_{X_{A'}}=\Pi_{X_{A'}}\tilde{\rho}_{X_{A'}}\Pi_{X_{A'}}\in\cB^{\sqrt{12\eps'}}(\tilde{\rho}_{X_{A'}})$ classical with respect to the basis $\{\ket{x}\}_{x\in X_{A'}}$. By the properties of the purified distance~\cite[Chapter 3]{Tomamichel12}, there exists a purification $\bar{\rho}_{AX_{A}X_{A'}R}\in\cB^{\sqrt{12\eps'}+\eps'}(\rho_{AX_{A}X_{A'}R})$ that is coherently classical on $X_{A}X_{A'}$ with respect to the basis $\{\ket{x}\}_{x\in X_{A}X_{A'}}$. Applying the protocol for this state $\bar{\rho}_{AX_{A}X_{A'}R}$, the classical communication cost~\eqref{eq:ccost2} becomes by the monotonicity of the max-information under projections (Lemma~\ref{lem:nonnegative}) and~\eqref{eq:finalsmooth},
\begin{align}
c\leq I_{\max}^{\eps'}(X_{A'}:R)_{\rho}+4\log\frac{1}{\eps}+4+\log\log|X_{A'}|\ ,
\end{align}
and by~\eqref{eq:simsmoothing} the sum cost~\eqref{eq:sumcost2} becomes
\begin{align}
c+s\leq H_{\max}^{\eps'}(X_{A'})_{\rho}+2\log\frac{1}{\eps'}+\log\log|X_{A'}|\ .
\end{align}
\end{proof}

For completeness we also state a converse for the classical communication cost of classical state splitting of classically coherent states.

\begin{theorem}\label{thm:ccconverse}
Let $\eps\geq0$, $\eps'>0$, and $\rho_{AX_{A}X_{A'}R}\in\cV_{\leq}(\cH_{AX_{A}X_{A'}R})$ be classically coherent on $X_{A}X_{A'}$ with respect to the basis $\{\ket{x}\}_{x\in X_{A}X_{A'}}$. Then, the classical communication cost for any $\eps$-error classical state splitting protocol for $\rho_{AX_{A}X_{A'}R}$ is lower bounded by\footnote{We do not mention the required shared randomness (or entanglement) resource, since the statement holds independently of it.}
\begin{align}
c\geq I^{\eps+\eps'}_{\max}(X_{A'}:R)_{\rho}-\log\left(\frac{8}{\eps'^{2}}+2\right)\ .
\end{align}
\end{theorem}

\begin{proof}
We have a look at the correlations between Bob and the reference by analyzing the max-information that Bob has about the reference. At the beginning of any protocol, there is no register at Bob's side correlated with the reference and therefore the max-information that Bob has about the reference is zero. Since back communication is not allowed, we can assume that the protocol for state splitting has the following form: applying local operations at Alice's side, sending bits from Alice to Bob and then applying local operations at Bob's side. Local operations at Alice's side have no influence on the max-information that Bob has about the reference. By sending $c$ bits from Alice to Bob, the max-information that Bob has about the reference can increase, but at most by $c$ (Corollary~\ref{cor:maxdbound}). By applying local operations at Bob's side the max-information that Bob has about the reference can only decrease (Lemma~\ref{lem:imaxmono}). So the max-information that Bob has about the reference is upper bounded by $c$. Therefore, any state $\omega_{X_{B}R}$ at the end of a state splitting protocol must satisfy $I_{\max}(R:X_{B})_{\omega}\leq c$. But we also need $\omega_{X_{B}R}\approx_{\eps}\rho_{X_{B}R}=\cI_{X_{A'}\rightarrow X_{B}}(\rho_{X_{A'}R})$ by the definition of $\eps$-error state splitting (Definition~\ref{def:csplitting}). Using the definition of the smooth max-information, and that the smooth max-information is approximately symmetric (Lemma~\ref{lem:ciganovic}), we obtain the bound in the statement of the theorem.
\end{proof}


\subsection{Main Theorem}\label{sec:meas_main}

In this section, we establish our main result: universal measurement compression. Theorem~\ref{thm:measmain} characterizes the trade-off between shared randomness and classical communication required for the feedback case, and Theorem~\ref{thm:measnonmain} characterizes the trade-off between shared randomness and classical communication required for the non-feedback case.

\begin{definition}[One-shot measurement compression]\label{def:meassimulation}
Consider a bipartite system with parties Alice and Bob. Let $\delta\geq0$, and $\cM:\cB(\cH_{A})\rightarrow\cB(\cH_{X})$ be a quantum-classical channel, with quantum input $A$ and classical output $X$. A quantum protocol $\cP$ is a one-shot feedback / non-feedback measurement compression for $\cM$ with error $\delta$ if it consists of using $s$ bits of shared randomness, applying local operations at Alice's side, sending $c$ classical bits from Alice to Bob, applying local operations at Bob's side, and for the non-feedback case\footnote{If we state the task of measurement compression as being that a verifier who is given the reference system and classical output should not be able to distinguish the true channel from the simulation (as we do in Section~\ref{sec:gain}), then we should also demand then the common randomness and classical communication is private from the verifier.}
\begin{align}
\|\cP-\cM\|_{\Diamond}\leq\delta\ ,
\end{align}
 or for the feedback case
\begin{align}
\|\cP-\overline{\Delta} \circ \cM\|_{\Diamond}\leq\delta\ ,
\end{align}
where $\overline{\Delta}:\cB(\cH_{X})\rightarrow\cB(\cH_{X_A})\ot\cB(\cH_{X_B})$ is a classical copying map,
\begin{align}
\overline{\Delta}(\sigma)=\sum_x \bra{x} \sigma \ket{x} \,\, \ket{x}\bra{x}_{X_A} \ot \ket{x}\bra{x}_{X_B},
\end{align}
ensuring that both Alice and Bob obtain the measurement outcome. The quantity $c$ is called the classical communication cost, and $s$ is the shared randomness cost.
\end{definition}

\begin{definition}[Universal measurement compression]\label{def:measmain}
Let $\cM:\cB(\cH_{A})\rightarrow\cB(\cH_{X})$ be a quantum-classical channel. An asymptotic feedback/non-feedback measurement compression for $\cM$ is a sequence of one-shot feedback / non-feedback measurement compressions $\cP^{n}$ for $\cM^{\ot n}$ with error $\delta_{n}$, such that $\lim_{n\rightarrow\infty}\delta_{n}=0$. The classical communication rate is $\limsup_{n\rightarrow\infty}\frac{\log c_{n}}{n}$ and the shared randomness rate is $\limsup_{n\rightarrow\infty}\frac{\log s_{n}}{n}$ (where $c_{n}$ and $r_{n}$ denote the corresponding costs for the one-shot feedback / non-feedback measurement compressions).
\end{definition}

\paragraph{Feedback Measurement Compression.} Even though the feedback simulation can be understood as a special case of the non-feedback simulation (as seen later), we start with the feedback simulation because its proof is more transparent and less involved.

\begin{theorem}\label{thm:measmain}
Let $\cM:\cB(\cH_{A})\rightarrow\cB(\cH_{X})$ be a quantum-classical channel. Then, there exist asymptotic feedback measurement compressions for $\cM$ if and only if the classical communication rate $C$ and shared randomness rate $S$ lie in the following rate region\footnote{Note that the two maxima in~\eqref{eq:maina} and~\eqref{eq:mainb} can be achieved for different states.}
\begin{align}
C&\geq\max_{\rho}I(X:R)_{(\cM\ot\cI)(\rho)}\label{eq:maina}\ ,\\
C+S&\geq\max_{\rho}H(X)_{\cM(\rho)}\label{eq:mainb}\ ,
\end{align}
where $\rho_{AR}\in\cV(\cH_{AR})$ is a purification of the input state $\rho_{A}\in\cS(\cH_{A})$. Or equivalently, for a given shared randomness rate $S$, the optimal rate of classical communication is equal to
\begin{align}\label{eq:mainc}
C(S)=\max\left\{\max_{\rho} I(X:R)_{(\cM\ot\cI)(\rho)},\,\,\max_{\rho} H(X)_{\cM(\rho)}-S\right\}\ .
\end{align}
In particular, when sufficient shared randomness is available, the rate of classical communication is given by
\begin{align}
C(\infty)=\max_{\rho} I(X:R)_{(\cM\ot\cI)(\rho)}\ .
\end{align}
\end{theorem}

\begin{proof}
We first show that the right-hand side of~\eqref{eq:maina} is a lower bound for the classical communication rate, and that~\eqref{eq:mainb} is a lower bound for the sum rate (Proposition~\ref{prop:measnonconverse}). Then we show that these lower bounds can be achieved (Proposition~\ref{prop:achiev}). The general rate trade-off in~\eqref{eq:maina}-\eqref{eq:mainb} and \eqref{eq:mainc} immediately follows, since any fraction of the shared randomness used can always be thought of as being created by classical communication.
\end{proof}

\begin{proposition}\label{prop:measnonconverse}
Let $\cM:\cB(\cH_{A})\rightarrow\cB(\cH_{X})$ be a quantum-classical channel. Then, we have for any asymptotic measurement compression for $\cM$ that
\begin{align}
C & \geq \max_{\rho}I(X:R)_{(\cM\ot\cI)(\rho)}\ ,\\
C + S & \geq \max_{\rho}H(X)_{\cM(\rho)}\ ,
\end{align}
where $\rho_{AR}\in\cV(AR)$ is a purification of the input state $\rho_{A}\in\cS(A)$. 
\end{proposition}

\begin{proof}
This follows from the converse for the case of a fixed iid source~\cite[Theorem 8]{Winter04}, since the asymptotic measurement compressions must in particular work for any fixed IID input state $\rho^{\ot n}_{A}$ (for $n\rightarrow\infty$).\footnote{To see this explicitly cf.~\cite[Section 2.4]{Wilde12_2}.}
\end{proof}

\begin{proposition}\label{prop:achiev}
Let $\cM:\cB(\cH_{A})\rightarrow\cB(\cH_{X})$ be a quantum to classical channel. Then, there exist asymptotic feedback measurement compressions for $\cM$ with\begin{align}
C&\leq \max_{\rho}I(X:R)_{(\cM\ot\cI)(\rho)}\label{eq:achiev}\ ,\\
C+S&\leq \max_{\rho}H(X)_{\cM(\rho)}\ \label{eq:achiev2},
\end{align}
where $\rho_{AR}\in\cV(AR)$ is a purification of the input state $\rho_{A}\in\cS(A)$. 
\end{proposition}

\begin{proof}
We show the existence of a sequence of one-shot feedback measurement compressions $\cP^{n}$ for $\cM^{\ot n}$ with asymptotically vanishing error $\delta_{n}$, a classical communication rate $\frac{c_{n}}{n}$ as in~\eqref{eq:achiev}, and a shared randomness rate $\frac{s_{n}}{n}$ such that the sum rate becomes as in~\eqref{eq:achiev2}. Without loss of generality, we choose $\cP^{n}$ to be permutation covariant.\footnote{By the same argument as in the proof of the quantum reverse Shannon theorem (Section~\ref{sec:qshannon_main}), every protocol can be made permutation covariant. To start with, Alice applies a random permutation $\pi$ on the input system chosen according to some shared randomness. This is then followed by the original protocol (which might not yet be permutation covariant), and Bob who undoes the permutation by applying $\pi^{-1}$ on the output system. The shared randomness cost of this procedure can be kept sub-linear in $n$ by using randomness recycling as discussed in~\cite[Section IV. D]{Bennett09}.} The post-selection technique for quantum channels (Proposition~\ref{prop:postselect}) then applies and upper bounds the error by
\begin{align}\label{eq:mainstep}
\delta_{n}&=\|\cM_{A\rightarrow X_{A}X_{B}}^{\ot n}-\cP^{n}_{A^n\rightarrow X^{n}_{A}X^{n}_{B}}\|_{\Diamond}\notag\\
&\leq(n+1)^{|A|^{2}-1}\cdot\|((\cM_{A\rightarrow X_{A}X_{B}}^{\ot n}-\cP^{n}_{A^n\rightarrow X^n_{A}X^n_{B}})\ot\cI_{R}^{\ot n}\ot\cI_{R'}(\zeta^{n}_{ARR'})\|_{1}\ ,
\end{align}
where $\zeta^{n}_{ARR'}$ is a purification of the de Finetti state $\zeta^{n}_{AR}=\int\psi_{AR}^{\ot n}\,\, d(\psi_{AR})$ with $\psi_{AR}\in\cV(\cH_{AR})$, $A\cong R$ and $d(\cdot)$ the measure on the normalized pure states on $\cH_{AR}$ induced by the Haar measure on the unitary group acting on $\cH_{AR}$, normalized to $\int d(\cdot)=1$. Hence, it is sufficient to consider simulating the measurement on a purification of the de Finetti state
\begin{align}\label{eq:measstructure}
\omega_{X_{A}X_{B}RR'}^{n}=\big(\cM_{A\rightarrow X_{A}X_{B}}^{\ot n}\ot\cI_{R}^{\ot n}\ot\cI_{R'}\big)\left(\zeta^{n}_{ARR'}\right)\ ,
\end{align}
up to an error $o\left((n+1)^{1-|A|^{2}}\right)$ in trace distance, for an asymptotic simulation cost smaller than in~\eqref{eq:achiev} and~\eqref{eq:achiev2}. For this, we consider a local Stinespring dilation $U_{A\rightarrow EX_{A}X_{A'}}$ of the measurement $\cM_{A\rightarrow X_{A}X_{A'}}$ at Alice's side, followed by classical state splitting of the resulting classically coherent state (Theorem~\ref{thm:splitting}). As mentioned above, this map can be made permutation invariant. Let $U_{A^n\rightarrow E^n X^n_{A} X^n_{A'}} = U^{\ot n}_{A\rightarrow EX_{A}X_{A'}}$ and
\begin{align}\label{eq:splitstate}
\omega_{E^n X^n_{A}X^n_{A'}R^nR'}=U_{A^n \rightarrow E^n X^n_{A} X^n_{A'}}(\zeta^{n}_{ARR'})\ .
\end{align}
For fixed $\eps_{n}>0$, Theorem~\ref{thm:splitting} then assures that the classical state splitting protocol outputs a state which is
\begin{align}\label{eq:distance}
4\cdot\eps_{n}+4\sqrt{2\eps_{n}}+2\cdot|X_{A'}|^{-n}
\end{align}
close to~\eqref{eq:measstructure} in trace distance,\footnote{The trace distance is upper bounded by two times the purified distance (Lemma~\ref{lem:pdbounds}).} for a classical communication cost
\begin{align}\label{eq:deltacost}
c_{n}\leq I_{\max}^{\eps_{n}}(X_{A'}:RR')_{\omega}+4\log\frac{1}{\eps_{n}}+4+\log\log|X_{A'}|+\log n\ ,
\end{align}
and a sum cost
\begin{align}\label{eq:deltacost2}
c_{n}+s_{n}\leq H_{\max}^{\eps_{n}}(X_{A'})_{\omega}+2\log\frac{1}{\eps_{n}}+\log\log|X_{A'}|+\log n\ ,
\end{align}
where the last two terms on the right in each of the above expressions come from the fact that $\log\log |X_{A'}|^n = \log\log|X_{A'}|+\log n$. We now analyze the asymptotic behavior of~\eqref{eq:deltacost} and~\eqref{eq:deltacost2}. By a dimension upper bound for the smooth max-information (Lemma~\ref{lem:maxdbound}), and the fact that we can assume $|R'|\leq(n+1)^{|A|^{2}-1}$ (Proposition~\ref{prop:postselect}), we get
\begin{align}\label{eq:beforecara}
c_{n}\leq I_{\max}^{\eps_{n}}(X_{A'}:R)_{\omega}+2\cdot\log\left((n+1)^{|A|^{2}-1}\right)+4\log\frac{1}{\eps_{n}}+4+\log\log|X_{A'}|+\log n\ .
\end{align}
By a corollary of Carath\'eodory's theorem (Corollary~\ref{cor:cara}), we write
\begin{align}\label{eq:caramain}
\zeta_{AR}^{n}=\sum_{i\in I}p_{i}\cdot(\sigma^{i}_{AR})^{\ot n}\ ,
\end{align}
where $\sigma^{i}_{AR}\in\cV(\cH_{AR})$, $I=\{1,2,\ldots,(n+1)^{2|A||R|-2}\}$, and $\{p_{i}\}_{i\in I}$ a probability distribution. Using a quasi-convexity property of the smooth max-information (Lemma \ref{lem:imaxqconvex}), and for
\begin{align}
\chi=2\cdot\log\left((n+1)^{|A|^{2}-1}\right)+4\log\frac{1}{\eps_{n}}+4+\log\log|X_{A'}|+\log n\ ,
\end{align}
we obtain
\begin{align}
c_{n} &\leq I_{\max}^{\eps_{n}}(X_{A'}:R)_{(\cM^{\ot n}\ot\cI)(\sum_{i}p_{i}(\sigma^{i})^{\ot n})}+\chi\notag\\
&\leq\max_{i\in I}I_{\max}^{\eps_{n}}(X_{A'}:R)_{\left[(\cM\ot\cI)(\sigma^{i})\right]^{\ot n}}+\log\left((n+1)^{2|A||R|-2}\right)+\chi\notag\\
&\leq\max_{\rho}I_{\max}^{\eps_{n}}(X_{A'}:R)_{\left[(\cM\ot\cI)(\rho)\right]^{\ot n}}+\log\left((n+1)^{2|A||R|-2}\right)+\chi\ ,
\end{align}
where the last maximum ranges over all $\rho_{AR}\in\cV(\cH_{AR})$. From the asymptotic equipartition property for the smooth max-information (Lemma~\ref{lem:aepimax}) we obtain
\begin{align}\label{eq:beforefinal}
c_{n}\leq n\cdot\max_{\rho}I(X_{A'}:R)_{(\cM\ot\cI)(\rho)}+\sqrt{n}\cdot\xi(\eps_{n})-2\log\frac{\eps_{n}^{2}}{24}+\log\left((n+1)^{2|A||R|-2}\right)+\chi\ ,
\end{align}
where $\xi(\eps_{n})=8\sqrt{13-4\cdot\log\eps_{n}}\cdot(2+\frac{1}{2}\cdot\log|A|)$. By choosing
\begin{align}\label{eq:epsn}
\eps_{n}=(n+1)^{4(1-|A|^{2})}\ ,
\end{align}
we get an asymptotic classical communication cost of
\begin{align}
c=\limsup_{n\rightarrow\infty}\frac{c_{n}}{n}\leq\max_{\rho}I(X_{A'}:R)_{(\cM\ot\cI)(\rho)}\ ,
\end{align}
for a vanishing asymptotic error~\eqref{eq:mainstep}, \eqref{eq:distance}, \eqref{eq:epsn},
\begin{align}
\limsup_{n\rightarrow\infty}\delta_{n}\leq\limsup_{n\rightarrow\infty}\Big(&\big(4\cdot(n+1)^{4(1-|A|^{2})}+4\sqrt{2}\cdot(n+1)^{2(1-|A|^{2})}+2\cdot|X_{A'}|^{-n}\big)\notag\\
&\cdot(n+1)^{|A|^{2}-1}\Big)=0\ .
\end{align}
Furthermore, we estimate the asymptotic behavior of the sum cost~\eqref{eq:deltacost2} by using~\eqref{eq:caramain} and a quasi-convexity property of the smooth max entropy (Lemma~\ref{lem:maxquasicon}). For $\chi'=2\cdot\log\frac{1}{\eps_{n}}+\log\log|X_{A'}|+\log n$ we get
\begin{align}
c_{n}+s_{n}&\leq\max_{i}H_{\max}^{\eps_{n}}(X_{A'})_{\cM(\sigma^{i})^{\ot n}}+\log\left((n+1)^{2|A||R|-2}\right)+\chi'\notag\\
&\leq\max_{\rho}H_{\max}^{\eps_{n}}(X_{A'})_{\cM(\rho)^{\ot n}}+\log\left((n+1)^{2|A||R|-2}\right)+\chi'\ ,
\end{align}
where $\rho_{A}\in\cS(\cH_{A})$. By the asymptotic equipartition property for the smooth max-entropy (Lemma~\ref{lem:aepminmax}), we arrive at
\begin{align}
c_{n}+s_{n}&\leq\max_{\rho}H^{\eps_{n}/2}_{\max}(X_{A'})_{\cM(\rho)^{\ot n}}+2\log\frac{8}{\eps_{n}^{2}}+\log\left((n+1)^{2|A||R|-2}\right)+\chi'\notag\\
&\leq n\cdot\max_{\rho}H(X_{A'})_{\cM(\rho)}+\sqrt{n}\cdot4\sqrt{1-2\log\frac{\eps_{n}}{2}}\cdot(2+\frac{1}{2}\cdot\log|X_{A'}|) \notag\\
& \,\,\,\, +2\log\frac{8}{\eps_{n}^{2}}+\log\left((n+1)^{2|A||R|-2}\right)+\chi'\ ,
\end{align}
where $\rho_{A}\in\cS(\cH_{A})$. By employing~\eqref{eq:epsn}, we get for the asymptotic limit 
\begin{align}
c+s=\limsup_{n\rightarrow\infty}\frac{1}{n}(c_{n}+s_{n})\leq\max_{\rho}H(X_{A'})_{\cM(\rho)}\ , \label{eq:total-cost-bound}
\end{align}
where $\rho_{A}\in\cS(\cH_{A})$.
\end{proof}

\paragraph{Non-Feedback Measurement Compression.} The non-feedback case is as follows.

\begin{theorem}\label{thm:measnonmain}
Let $\cM:\cB(\cH_{A})\rightarrow\cB(\cH_{X})$ be a quantum-classical channel. Then, there exist asymptotic non-feedback measurement compressions for $\cM$ if and only if the classical communication rate $C$ and shared randomness rate $S$ lie in the rate region given by the union of the following regions
\begin{align}
C&\geq\max_{\rho}I(W:R)_{\beta}\label{eq:measnonmainC}\ ,\\
C+S&\geq\max_{\rho}I(W:XR)_{\beta}\label{eq:measnonmainCS}\ ,
\end{align}
where the state $\beta_{WXR}$ has the form,
\begin{align}\label{eq:betaoptimization}
\beta_{WXR}=\sum_{w,x}q_{x|w}\cdot\proj{w}_{W}\ot\proj{x}_{X}\ot\trace_{A}\big[(\cN_{w}\ot\cI)(\rho_{AR})\big]\ ,
\end{align}
$\rho_{AR}\in\cV(\cH_{AR})$ is a purification of the input state $\rho_{A}\in\cS(\cH_{A})$, and the union is with respect to all decompositions of the measurement $\cM$ in terms of internal measurements $\cN=\{\cN_w\}$ and conditional post-processing distributions $q_{x|w}$, that is,
\begin{align}\label{eq:mdecomp}
\cN:\sum_{x,w}q_{x|w}\cdot\cN_{w}=\cM\ .
\end{align}
Or equivalently, for a given shared randomness rate $S$, the optimal rate of classical communication is equal to
\begin{align}\label{eq:tragenonf}
C(S)=\min_{\cN\, : \, \sum_{x,w} q_{x|w} \cN_w = \cM}
\max\left\{\max_{\rho} I(W:R)_{\beta},\,\,\max_{\rho} I(W:XR)_{\beta}-S\right\}\ .
\end{align}
\end{theorem}

By the data processing inequality for the mutual information $I(W:R)_{\beta}\geq I(X:R)_{\cM(\rho)}$, and hence, the classical communication cost can only increase compared to a feedback simulation (Theorem~\ref{thm:measmain}). But if the savings in common randomness consumption are larger than the increase in classical communication, then there is an advantage to performing a non-feedback simulation. It follows from the considerations in~\cite{Martens90,Wilde12_2} that the rate trade-offs~\eqref{eq:mainc} and~\eqref{eq:tragenonf} become identical if and only if the elements of the measurement to simulate are all rank-one operators.

\begin{proof}
We see from the converse for the case of a fixed iid source~\cite[Theorem 9]{Wilde12_2}, that the right-hand side of~\eqref{eq:measnonmainC} is a lower bound for the classical communication rate, and that~\eqref{eq:measnonmainCS} is a lower bound for the sum rate. This is because the asymptotic non-feedback measurement compression must in particular work for any fixed iid input state $\rho_{A}^{\ot n}$ (for $n\ra\infty$). As the next step we show that these lower bounds can be achieved. The general rate trade-off in~\eqref{eq:measnonmainC}-\eqref{eq:measnonmainCS} and~\eqref{eq:tragenonf} then immediately follows, since any fraction of the shared randomness used can always be thought of as being created by classical communication.

The idea for the achievability is as follows. Given a particular decomposition $\cN:\sum_{x,w}q_{x|w}\cdot\cN_{w}=\cM$ as stated above, Alice and Bob just use a feedback measurement compression protocol (as in the proof of Theorem~\ref{thm:measmain}) to simulate the measurement $\cN=\{\cN_{w}\}$. This is followed by a local simulation of the classical map $q_{x|w}$ at no cost at Bob's side. Finally, Alice and Bob can use randomness recycling to extract $H_{\min}(W|RX)_{\beta}$ bits of shared randomness back~\cite{Bennett09}. In the one-shot case, this leads to a classical communication cost of $I_{\max}(W:R)_{\beta}$, and a sum cost $I_{\max}(W:RX)_{\beta}$. For technical reasons, we smooth the states using typical projectors (see Appendix~\ref{app:typical} for background on typical projectors) and arrive at the rates given in the statement of the theorem.

Let $\{q_{x|w},\,\cN_{\omega}\}$ be a fixed decomposition of $\cM$. As in the feedback case (Theorem~\ref{thm:measmain}) we employ the post-selection technique (Proposition~\ref{prop:postselect}) to upper bound the error for one-shot non-feedback compressions $\cP^{n}$ for $\cM^{\ot n}$ by
\begin{align}\label{eq:postmeasnonf}
\delta_{n}&=\|\cM_{A\rightarrow X_{B}}^{\ot n}-\cP^{n}_{A\rightarrow X_{B}}\|_{\Diamond}\notag\\
&\leq(n+1)^{|A|^{2}-1}\cdot\|((\cM_{A\rightarrow X_{B}}^{\ot n}-\cP^{n}_{A\rightarrow X_{B}})\ot\cI_{R}^{\ot n}\ot\cI_{R'}(\zeta^{n}_{ARR'})\|_{1}\ ,
\end{align}
where $\zeta^{n}_{ARR'}$ is a purification of the de Finetti state $\zeta^{n}_{AR}=\int\psi_{AR}^{\ot n}\,\, d(\psi_{AR})$ with $\psi_{AR}\in\cV(\cH_{AR})$, $A\cong R$ and $d(\cdot)$ the measure on the normalized pure states on $\cH_{AR}$ induced by the Haar measure on the unitary group acting on $\cH_{AR}$, normalized to $\int d(\cdot)=1$. Hence, it is sufficient to consider simulating the measurement $\cM^{\ot n}$ on a purification of the de Finetti state
\begin{align}
\omega_{X_{B}RR'}^{n}=\big(\cM_{A\rightarrow X_{B}}^{\ot n}\ot\cI_{R}^{\ot n}\ot\cI_{R'}\big)\left(\zeta^{n}_{ARR'}\right)\ ,
\end{align}
up to an error $o\left((n+1)^{1-|A|^{2}}\right)$ in trace distance, for an asymptotic simulation cost smaller than in~\eqref{eq:measnonmainC} and~\eqref{eq:measnonmainCS}. For this, the idea is to consider a local Stinespring dilation $V_{A\ra EW_{A}W_{A'}}$ of the measurement $\cN_{A\ra W_{A}W_{A'}}$ at Alice's side, followed by classical state splitting of the resulting classically coherent state (along Theorem~\ref{thm:splitting}). Let $V^{n}_{A\ra EW_{A}W_{A'}}=V^{\ot n}_{A\ra EW_{A}W_{A'}}$ and
\begin{align}\label{eq:definettimeasnonf}
\omega_{EW_{A}W_{A'}RR'}^{n}=V^{n}_{A\ra EW_{A}W_{A'}}\left(\zeta^{n}_{ARR'}\right)\ .
\end{align}

However, Alice and Bob will not execute the protocol with respect to the state $\omega_{EW_{A}W_{A^{\prime}}RR^{\prime}}^{n}$ directly, but they will do so with respect to another pure, sub-normalized state $\bar{\gamma}_{EW_{A}W_{A^{\prime}}RR^{\prime}}^{n}$ that is also coherently classical on $W_{A}W_{A^{\prime}}$ with respect to the basis $\{\ket{w}\}_{w\in W_{A}}$, and such that
\begin{align}\label{eq:gammastate}
\|\bar{\gamma}^{n}_{EW_{A}W_{A'}RR'}-\omega_{EW_{A}W_{A'}RR'}^{n}\|_{1}\leq\eps_{n}\ ,
\end{align}
for some $\eps_{n}>0$. By a corollary of Carath\'eodory's theorem (Corollary~\ref{cor:cara}), we write
\begin{align}\label{eq:carameasnonf}
\zeta_{AR}^{n}=\sum_{i\in I}p_{i}\cdot(\sigma^{i}_{AR})^{\ot n}\ ,
\end{align}
where $\sigma^{i}_{AR}\in\cV(\cH_{AR})$, $I=\{1,2,\ldots,(n+1)^{2|A||R|-2}\}$, and $\{p_{i}\}_{i\in I}$ a probability distribution. From this, we define
\begin{align}
\gamma_{EW_{A}W_{A^{\prime}}R}^{i,n}=\big((V_{A\rightarrow EW_{A}W_{A^{\prime}}}\ot\cI_{R})(\sigma_{AR}^{i})\big)^{\ot n}\ ,
\end{align}
as well as its reduction as a classical-quantum state $\gamma_{W_{A}R}^{i,n}$ on the systems $W_{A^{\prime}}^{n}R^{n}$
\begin{align}
\gamma_{W_{A}R}^{i,n}=\sum_{w^{n}}p_{W^{n}|i}(w^{n}|i)\,\proj{w^n}_{W_{A^{\prime}}^{n}}\ot\gamma_{R^{n}}^{i,w^{n}}\ ,
\end{align}
for some distribution $p_{W^{n}|i}(w^{n}|i)$. On this state, we act with typical projectors to flatten its spectrum, defining the projected state $\bar{\gamma}_{W_{A}R}^{i,n}$ as follows
\begin{align}\label{eq:typicalmeasnonf}
\bar{\gamma}_{W_{A}R}^{i,n}=\sum_{w^{n}}p_{W^{n}|i}(w^{n}|i)\,\Pi_{\delta}^{W^{n}|i}\,\proj{w^n}_{W_{A^{\prime}}^{n}}\,\Pi_{\delta}^{W^{n}|i}\ot\Pi_{\gamma^{i},\delta}^{n}\,\Pi_{\gamma^{i,w^{n}},\delta}^{n}\,\gamma_{R^{n}}^{i,w^{n}}\,\Pi_{\gamma^{i,w^{n}},\delta}^{n}\,\Pi_{\gamma^{i},\delta}^{n}\ ,
\end{align}
where $\Pi_{\delta}^{W^{n}|i}$ is a typical projector corresponding to the distribution $p_{W^{n}|i}(w^{n}|i)$, $\Pi_{\gamma^{i,w^{n}},\delta}^{n}$ is a conditionally typical projector corresponding to the conditional state $\gamma^{i,w^{n}}$ on the system $R^{n}$, and $\Pi_{\gamma^{i},\delta}^{n}$ is a typical projector corresponding to the state $\gamma_{R}^{i,n}$ (see Appendix~\ref{app:typical} for details of typical projectors). It follows from the properties of typical projectors that the projected state $\bar{\gamma}_{W_{A}R}^{i,n}$ becomes arbitrarily close in trace distance to the original state $\gamma_{W_{A}R}^{i,n}$,
\begin{align}
\Vert\gamma_{W_{A}R}^{i,n}-\bar{\gamma}_{W_{A}R}^{i,n}\Vert_{1}\leq\frac{\varepsilon_{n}^{2}}{4}\ ,
\end{align}
for some $\eps_{n}>0$, and sufficiently large $n$. The equivalence of the trace distance and the purified distance (Lemma~\ref{lem:pdbounds}) together with Uhlmann's theorem then imply the existence of some sub-normalized pure state $\bar{\gamma}_{EW_{A}W_{A^{\prime}}R}^{i,n}$ such that
\begin{align}
P(\gamma_{EW_{A}W_{A^{\prime}}R}^{i,n},\bar{\gamma}_{EW_{A}W_{A^{\prime}}R}^{i,n})\leq\frac{\varepsilon_{n}}{2}\ .
\end{align}
Hence, we get by~\eqref{eq:carameasnonf} and~\eqref{eq:definettimeasnonf} that
\begin{align}
\bar{\gamma}_{EW_{A}W_{A^{\prime}}R}^{n}=\sum_{i\in I}p_{i}\cdot\bar{\gamma}_{EW_{A}W_{A^{\prime}}R}^{i,n}
\end{align}
is $\eps_{n}/2$-close to $\omega_{EW_{A}W_{A^{\prime}}R}^{n}$ in purified distance. By features of the purified distance~\cite[Chapter 3]{Tomamichel12}, and the equivalence of the trace distance and the purified distance (Lemma~\ref{lem:pdbounds}), we then get that there exists an extension $\bar{\gamma}_{EW_{A}W_{A^{\prime}}RR^{\prime}}^{n}$ of $\bar {\gamma}_{EW_{A}W_{A^{\prime}}R}^{n}$ with the desired properties such that~\eqref{eq:gammastate} holds.

Alice and Bob will now act with a classical state splitting protocol for $W_{A}W_{A^{\prime}}$ with respect to the classically coherent state
$\bar{\gamma}_{EW_{A}W_{A^{\prime}}RR^{\prime}}^{n}$. However, we do not directly use our result about classical state splitting (Theorem~\ref{thm:splitting}), but go for a non-smooth version that is implicit in the proof of Theorem~\ref{thm:splitting}. It follows from~\eqref{eq:ccost2} that for an $(\eps_{n}+|W_{A'}|^{-n})$-error (in purified distance) classical state splitting protocol for $W_{A}W_{A'}$, a classical communication cost
\begin{align}\label{eq:ccostmeasnonf}
c_{n}\leq\max_{y}I_{\max}(W_{A'}:RR')_{\bar{\gamma}^{n,y}}+4\log\frac{1}{\eps_{n}}+4+\log\log|W_{A'}|+\log n
\end{align}
is achievable, and it follows from~\eqref{eq:sumcost}/\eqref{eq:sumcost2} that the sum cost becomes
\begin{align}
c_{n}+s_{n}\leq\max_{y}H_{0}(W_{A'})_{\bar{\gamma}^{n,y}}+2+\log\log|W_{A'}|+\log n\ ,
\end{align}
where the measurement outcomes $y$ are with respect to the pre-processing measurement defined in~\eqref{eq:preproc}. This provides Bob with the measurement outcomes of $\cN$ for the fixed de Finetti type input state $\zeta_{ARR'}$, and a total error of $(3\cdot\eps_{n}+2\cdot|W_{A'}|^{-n})$ in trace distance.\footnote{The trace distance is upper bounded by two times the purified distance (Lemma~\ref{lem:pdbounds}).} A local simulation of the classical map $q_{x^{n}|w^{n}}$ at no cost at Bob's side then provides Bob with the measurement outcomes of $\cM$ as desired
(again for the fixed de Finetti type input state $\zeta_{ARR^{\prime}}$ and the same error). However, the sum cost of this non-feedback measurement simulation can be reduced by invoking an additional randomness recycling step. We do this by applying, conditioned on $y$, a strong classical min-entropy extractor on $W$ against the (quantum) side information $XRR'$ (Corollary~\ref{cor:ccperm}), and this lowers the sum cost to
\begin{align}\label{eq:sumcostmeasnonf}
c_{n}+s_{n}&\leq\max_{y}\big(H_{0}(W_{A'})_{\bar{\gamma}^{n,y}}-H_{\min}(W_{A'}|RR'X_{A'})_{\bar{\gamma}^{n,y}}\big)\notag\\
&+4\log\frac{1}{\eps_{n}}+2+\log\log|W_{A'}|\ ,
\end{align}
for an additional error $\eps_{n}$ in trace distance, leading to a total error of
\begin{align}\label{eq:totalerrormeasnonf}
(4\cdot\eps_{n}+2\cdot|W_{A'}|^{-n})
\end{align}
in trace distance. The min-entropy extractor is performed with respect to the following typical projected state, in order to increase the amount of randomness that can be extracted
\begin{multline}
\bar{\gamma}_{XW_{A}R}^{i,n}=\sum_{w^{n},x^{n}}q\left(  x^{n}|w^{n}\right)p_{W^{n}|i}(w^{n}|i)\,\Pi_{\delta}^{X^{n}|W^{n},i}\ \proj{x^n}_{X^{n}}%
\ \Pi_{\delta}^{X^{n}|W^{n},i}\ot\\
\Pi_{\delta}^{W^{n}|i}\,\proj{w^n}_{W_{A^{\prime}}^{n}}\,\Pi_{\delta}^{W^{n}|i}\ot\Pi_{\gamma^{i},\delta}^{n}\,\Pi_{\gamma^{i,w^{n}},\delta
}^{n}\,\gamma_{R^{n}}^{i,w^{n}}\,\Pi_{\gamma^{i,w^{n}},\delta}^{n}\,\Pi_{\gamma^{i},\delta}^{n}\ .
\end{multline}
In the rest of the proof we bring the classical communication cost~\eqref{eq:ccostmeasnonf} and the sum cost~\eqref{eq:sumcostmeasnonf} into the right form, and show that the asymptotic error for the measurement simulation~\eqref{eq:postmeasnonf} becomes zero. By the behavior of the max-information under projective measurements (Corollary~\ref{cor:maxproj}), a dimension upper bound for the max-information (Lemma~\ref{lem:maxdbound}), the fact that we can assume $|R'|\leq(n+1)^{|A|^{2}-1}$ (Proposition~\ref{prop:postselect}), and a quasi-convexity property of the max-information (Lemma~\ref{lem:imaxqconvex}), we get
\begin{align}
c_{n}\leq\max_{i\in I}I_{\max}(W_{A'}:R)_{\bar{\gamma}^{i,n}}+\chi\ ,
\end{align}
where
\begin{align}
\chi&=2\cdot\log\left((n+1)^{|A|^{2}-1}\right)+\log\left((n+1)^{2|A||R|-2}\right)+4\log\frac{1}{\eps_{n}}+4\notag\\
&+\log\log|W_{A'}|+\log n\ .
\end{align}
By an upper bound on the max-information (Lemma~\ref{lem:maxbounds}), and a lower bound on the conditional min-entropy (Lemma~\ref{lem:minlower}), this can be estimated to be
\begin{align}
c_{n}\leq\max_{i\in I}\big(H_{R}(W_{A'})_{\bar{\gamma}^{i,n}}-H_{\min}(W_{A'}R)_{\bar{\gamma}^{i,n}}+H_{0}(R)_{\bar{\gamma}^{i,n}}\big)+\chi\ .
\end{align}
By~\eqref{eq:typicalmeasnonf}, as well as the definition of typical projectors (Appendix~\ref{app:typical}), we get
\begin{align}
c_{n} &  \leq n\cdot\max_{i\in I}\big(H(W_{A^{\prime}})_{\gamma^{i}}-H(W_{A^{\prime}}R)_{\gamma^{i}}+H(R)_{\gamma^{i}}\big)+5nc\delta+\chi\notag\\
&  \leq n\cdot\max_{\rho}\big(H(W_{A^{\prime}})_{\cN(\rho)}-H(W_{A^{\prime}}R)_{(\cN\ot\cI)(\rho)}+H(R)_{\rho}\big)+5nc\delta+\chi\ ,
\end{align}
where $\rho_{AR}\in\cV(\cH_{AR})$. 

By choosing
\begin{align}
\eps_{n}=(n+1)^{4(1-|A|^{2})}\ ,
\end{align}
we finally get an asymptotic classical communication cost of
\begin{align}
c=\limsup_{n\ra\infty}\frac{c_{n}}{n}\leq\max_{\rho}I(W_{A'}:R)_{\beta}\ ,
\end{align}
where $\rho_{AR}\in\cV(\cH_{AR})$, $\beta_{W_{A'}R}$ is as in~\eqref{eq:betaoptimization}, and a vanishing asymptotic error~\eqref{eq:postmeasnonf}, \eqref{eq:totalerrormeasnonf},
\begin{align}
\limsup_{n\ra\infty}\delta_{n}\leq\limsup_{n\ra\infty}\Big((n+1)^{|A|^{2}-1}\cdot\big(4\cdot(n+1)^{4(1-|A|^{2})}+2\cdot|X_{A'}|^{-n}\big)\Big)=0\ .
\end{align}
For the sum cost~\eqref{eq:sumcostmeasnonf} we get by the definition of the measurement in~\eqref{eq:preproc} with outcomes $y$, and a line of argument as in~\eqref{eq:minmaxequiv} that
\begin{align}
c_{n}+s_{n}&\leq\max_{y}\big(H_{\min}(W_{A'})_{\bar{\gamma}^{n,y}}-H_{\min}(W_{A'}|RX_{A'})_{\bar{\gamma}^{n,y}}\big)\notag\\
&+4\cdot\log\frac{1}{\eps_{n}}+2+\log\log|W_{A'}|+\log n\notag\\
&\leq\max_{y}I_{\max}(W_{A'}:RX_{A'})_{\bar{\gamma}^{n,y}}\notag\\
&+2\cdot\log\left((n+1)^{|A|^{2}-1}\right)+4\cdot\log\frac{1}{\eps_{n}}+2+\log\log|W_{A'}|+\log n\ ,
\end{align}
where we used a lower bound on the max-information (Lemma \ref{lem:maxbounds}), as well as a dimension upper bound for the max-information (Lemma~\ref{lem:maxdbound}), and the fact that $|R'|\leq(n+1)^{|A|^{2}-1}$ (Proposition~\ref{prop:postselect}). Using similar arguments (see Appendix~\ref{app:typical}) as in the estimation of the classical communication cost we arrive at
\begin{align}
c+s=\limsup_{n\ra\infty}\frac{c_{n}+s_{n}}{n}\leq\max_{\rho}I(W_{A'}:RX_{A'})_{\beta}\ ,
\end{align}
where $\rho_{AR}\in\cV(\cH_{AR})$, and $\beta_{W_{A'}RX_{A'}}$ is as in~\eqref{eq:betaoptimization}. By minimizing over all decompositions of the measurement $\cM$ as in~\eqref{eq:mdecomp}, the claim follows.
\end{proof}


\subsection{Discussion}\label{sec:meas_discussion}

\paragraph{Fixed IID Source.} If we are only interested in simulating the measurement on a fixed iid source, we can directly apply the one-shot protocol from Theorem~\ref{thm:splitting} to the case of a fixed iid source, and then simply invoke the asymptotic equipartition property for the smooth max-information and the smooth max-entropy. In this way, we provide an alternative proof for the original measurement compression results by Winter~\cite{Winter04}, that avoids the use of typical projectors and the operator Chernoff bounds.

\paragraph{Classical Reverse Shannon Theorem.} By applying our simulation results (Theorem~\ref{thm:measmain} and Theorem~\ref{thm:measnonmain}) to fully classical channels, we see that the classical reverse Shannon theorem (Section~\ref{se:shannon}) is a strict specialization of universal measurement compression.
 
\paragraph{Universal Instrument Compression.} In Winter's original paper on measurement compression~\cite{Winter04}, additional arguments were required to establish a quantum instrument compression protocol. For an instrument compression protocol, Alice and Bob receive the classical outcomes of the measurement while Alice obtains the post-measurement states (see~\cite[Section V]{Winter04}). We note that our protocols here already function as universal instrument compression protocols due to our use of the classical state splitting protocol as a coding primitive. 

\paragraph{Applications.} There are a number of open questions to consider going forward from here. Given that there are applications of information gain or entropy reduction in thermodynamics~\cite{Jabobs09}, and quantum feedback control~\cite{Doherty01}, it would be interesting to explore whether the quantity in~\eqref{eq:max-info-gain} has some application in these domains. Also, Buscemi {\it et al.}~showed that the static measure of information gain in~\eqref{eq:winter-info-gain} plays a role in quantifying the trade-off between information extraction and disturbance~\cite{Buscemi08}, and it would be interesting to determine if there is a role in this setting for the information quantity in~\eqref{eq:max-info-gain}.


\section{Entanglement Cost of Quantum Channels}\label{se:cost}

The results in this section have been obtained in collaboration with Fernando Brand\~ao, Matthias Christandl, and Stephanie Wehner, and have appeared in~\cite{Berta12,Berta11_2}. The entanglement cost of a quantum channel is the minimal rate at which entanglement (between sender and receiver) is needed in order to simulate many copies of a quantum channel in the presence of free classical communication. In this section we show how to express this quantity as a regularized optimization of the entanglement formation over states that can be generated between sender and receiver. Our formula is the channel analogue of a well-known formula for the entanglement cost of quantum states in terms of the entanglement of formation~\cite{Hayden01,Bennett96}; and shares a similar relation to the recently shattered hope for additivity~\cite{Hastings09,Shor04}.


\subsection{Proof Ideas}\label{sec:cost_ideas}

The proof of our main result~\eqref{eq:EC} is similar in spirit to the proofs of the channel simulation results in Sections~\ref{se:shannon} to~\ref{se:meas}. In particular, it is based on one-shot information theory and uses the smooth entropy formalism. In order to prove the direct part of~\eqref{eq:EC}, we need to show the existence of a channel simulation for $\cE^{\ot n}$, whose asymptotic rate of entanglement consumption is upper bounded by $E_{C}(\cE)$. That is, we need to construct a quantum protocol that is close to $\cE^{\ot n}$ in the diamond norm, and that uses local operations and classical communication as well as ebits at a rate of at most $E_{C}(\cE)$. We employ the post-selection technique for quantum channels (see Appendix~\ref{ap:postselect}), and with this, it becomes sufficient to find a channel that creates the purification of a special de Finetti input state, and to quantify how much entanglement this map consumes. Since the state is a purification of a de Finetti state (and not a de Finetti state itself) it does not have iid structure. In order to deal with this fact we employ the one-shot entanglement cost for quantum states $E_{C}^{(1)}(\rho_{AB},\eps)$, which quantifies how much entanglement is needed in order to create one single copy of a bipartite quantum state $\rho_{AB}$ using local operations and classical communication~\cite{Buscemi11,Hayashi06}. For the analysis it is possible to use a strong quantum min-entropy extractor against quantum side information. The resulting entanglement cost of the channel simulation is then upper bounded by an expression similar to~\eqref{eq:EC}, but with the maximization over input states and the minimization in the definition of the entanglement of formation interchanged. Finally, in order to arrive at~\eqref{eq:EC}, we discretize the set of Kraus decompositions of $\cE$ and apply Sion's minimax theorem~\cite{Sion58} to swap the minimization and the maximization. The proof of the converse follows a standard argument applied to the one-shot entanglement cost.


\subsection{Entanglement Cost of Quantum States}\label{sec:cost_states}

It is the main of result of~\cite{Buscemi11} to quantify how much entanglement is needed in order to create a single copy of a bipartite state $\rho_{AB}$~\cite{Hayashi06}, a scenario previously studied in the asymptotic iid setting~\cite{Hayden01,Bennett96}.

\begin{definition}[One-shot entanglement dilution]\label{def:mainstate}
Consider a bipartite system with parties Alice and Bob, where Alice controls a system $\cH_{A}$ and Bob $\cH_{B}$. Let $\eps\geq0$, $\Phi_{\bar{A}\bar{B}}$ be a maximally entangled state between Alice and Bob, and $\rho_{AB}\in\cS(\cH_{AB})$. An $\eps$-faithful one-shot entanglement dilution protocol for $\rho_{AB}$ is a local operation and classical communication (LOCC) operation $\Lambda$ between Alice and Bob with $\bar{A}\rightarrow A$ at Alice's side and $\bar{B}\rightarrow B$ at Bob's side, such that
\begin{align}
\Lambda(\Phi_{\bar{A}\bar{B}})\approx_{\eps}\rho_{AB}\ .
\end{align}
If $\Phi_{\bar{A}\bar{B}}$ has Schmidt rank $L$, $\log L$ is the dilution cost of the one-shot entanglement dilution protocol.
\end{definition}

\begin{definition}[One-shot entanglement cost]\label{def:state_cost}
Let $\eps\geq0$, and $\rho_{AB}\in\cS(\cH_{AB})$. The minimal dilution cost of all $\eps$-faithful one-shot entanglement dilution protocols for $\rho_{AB}$ is called $\eps$-faithful one-shot entanglement cost of $\rho_{AB}$, and is denoted by $E_{C}^{(1)}(\rho_{AB},\eps)$.
\end{definition}

\begin{proposition}\cite[Theorem 1]{Buscemi11}\label{prop:ecstates}
Let $\eps\geq0$, and $\rho_{AB}\in\cS(\cH_{AB})$. Then, we have that
\begin{align}\label{eq:mainstate}
\min_{\{p_{k},\rho^{k}\}}H_{0}^{2\sqrt{\eps}}(A|K)_{\rho}\leq E_{C}^{(1)}(\rho_{AB},\eps)\leq\min_{\{p_{k},\rho^{k}\}}H_{0}^{\eps/2}(A|K)_{\rho}\ ,
\end{align}
where the minima range over all pure states decompositions $\rho_{AB}=\sum_{k}p_{k}\rho_{AB}^{k}$, and $\rho_{AR}=\sum_{k}p_{k}\rho^{k}_{A}\ot\proj{k}_{K}$.
\end{proposition}

We sketch a possible idea for the achievability. For any pure state decomposition $\rho_{AB}=\sum_{k}p_{k}\rho_{AB}^{k}$ Alice can locally create the classical-quantum state $\rho_{ABR}=\sum_{k}p_{k}\rho_{AB}^{k}\ot\proj{k}_{K}$, and then, conditioned on the index $i$, bring the $B$-part of the pure states $\rho_{AB}^{k}$ to Bob. For this Alice can, e.g., use quantum state merging based on strong quantum min-entropy randomness extractors against quantum side information (cf.~Lemma~\ref{def:qsm}). Minimizing over all pure state decompositions, a straightforward analysis shows that the resulting entanglement cost in the presence of free classical communication is bounded as in Proposition~\ref{prop:ecstates}. Following~\cite{Buscemi11} we identify $\min_{\{p_{k},\rho^{k}\}}H_{0}^{\eps}(A|K)_{\rho}$ as the quantity representing the one-shot entanglement cost $E_{C}^{(1)}(\rho_{AB},\eps)$.\footnote{Since the smoothing parameters in~\eqref{eq:mainstate} are different, the upper and lower bounds in~\eqref{eq:mainstate} may differ as well. However, for our purpose the bounds turn out to be sufficient.}

\begin{remark}\label{rmk:datta}
The bounds given in~\eqref{eq:mainstate} also hold if we only allow one-way classical communication (forward or backward).
\end{remark}


\subsection{Main Theorem}\label{sec:cost_main}

We are now in the position to define the entanglement cost of quantum channels and prove our formula (Theorem~\ref{thm:mainEC}). We start with the non-feedback case, but we will see afterwards that the feedback case is actually an easy special case of the non-feedback case (Section~\ref{sec:cost_discussion}).

\begin{definition}[One-shot non-feedback entanglement simulation]\label{def:simulation}
Consider a bipartite system with parties Alice and Bob. Let $\eps\geq0$, $\Phi_{\bar{A}\bar{B}}$ be a maximally entangled state between Alice and Bob, and $\cE:\cB(\cH_{A})\rightarrow\cB(\cH_{B})$ be a channel, where Alice controls $\cH_{A}$ and Bob $\cH_{B}$. A one-shot non-feedback entanglement simulation for $\cE$ with error $\eps$ is a quantum protocol
\begin{align}
\cP:\quad&\cB(\cH_{A})\rightarrow\cB(\cH_{B})\notag\\
&\rho_{A}\qquad\mapsto\Lambda(\rho_{A}\ot\Phi_{\bar{A}\bar{B}})\ ,
\end{align}
where $\Lambda$ is a LOCC operation between Alice and Bob with $A\bar{A}\rightarrow0$ (no output) at Alice's side and $\bar{B}\rightarrow B$ at Bob's side, as well as
\begin{align}\label{eq:diamondEC}
\|\cP-\cE\|_{\Diamond}\leq\eps\ .
\end{align}
If $\Phi_{\bar{A}\bar{B}}$ has Schmidt rank $L$, $\log L$ is the entanglement cost of the one-shot channel simulation.
\end{definition}

By the definition of the diamond norm (Definition~\ref{def:diamond}), this assures that for any possible input state, the output of the channel simulation $\cP$ can only distinguished with small probability from the corresponding output of $\cE$. 

\begin{definition}[Asymptotic non-feedback entanglement simulation]\label{def:main}
Let $\cE:\cB(\cH_{A})\rightarrow\cB(\cH_{B})$ be a channel. An asymptotic non-feedback entanglement simulation for $\cE$ is a sequence of one-shot non-feedback entanglement simulations $\cP^{n}$ for $\cE^{\ot n}$ with error $\eps_{n}$, such that $\lim_{n\rightarrow\infty}\eps_{n}=0$. The entanglement cost of the simulation is $\limsup_{n\rightarrow\infty}\frac{\log L_{n}}{n}$ (where $L_{n}$ denotes the entanglement cost of $\cP^{n}$).
\end{definition}

\begin{theorem}\label{thm:mainEC}
Let $\cE_{A\rightarrow B}:\cB(\cH_{A})\rightarrow\cB(\cH_{B})$ be a channel. The minimal entanglement cost of asymptotic non-feedback entanglement simulations for $\cE_{A\rightarrow B}$ is given by
\begin{align}\label{eq:main}
E_{C}(\cE_{A\rightarrow B})=\lim_{n\rightarrow\infty}\frac{1}{n}\cdot\max_{\rho^{n}_{AR}}E_{F}\left(\left(\cE^{\ot n}_{A\rightarrow B}\ot\cI_{R}\right)\left(\rho^{n}_{AR}\right)\right)\ ,
\end{align}
where $\rho^{n}_{AR}\in\cV(\cH_{A}^{\ot n}\ot\cH_{R}^{\ot n})$, and $\cH_{R}\cong\cH_{A}$.
\end{theorem}

\begin{proof}
We first show that the right-hand side of~\eqref{eq:main} can be achieved (Proposition~\ref{prop:ECfinal}), and thereafter that it is also a lower bound (Proposition~\ref{prop:ECconverse}).
\end{proof}

The proof proceeds in three steps leading to Proposition~\ref{prop:ECfinal}. The basic idea is as follows. Given a quantum channel $\cE$, we need to show the existence of a sequence of one-shot channel simulations with asymptotically vanishing error, and an entanglement cost upper bounded by the right-hand side of~\eqref{eq:main}. The crucial step is that by the post-selection technique for quantum channels (Proposition~\ref{prop:postselect}), it is sufficient to come up with a quantum protocol (which consists of using maximally entangled states, local operations, and classical communication) that works for the purification of one special de Finetti input state. This reduces the problem to the one-shot entanglement cost of quantum states (Proposition~\ref{prop:ecstates}). Thus all that remains to do is to estimate the one-shot entanglement cost of the state given by $\cE^{\ot n}$ applied to the purification of the special de Finetti input state (in the limit $n\rightarrow\infty$).

\begin{lemma}\label{lem:main_2}
Let $\cE_{A\rightarrow B}:\cB(\cH_{A})\rightarrow\cB(\cH_{B})$ be a channel. Then, we have that
\begin{align}\label{eq:firststep}
E_{C}(\cE_{A\rightarrow B})\leq\inf_{\{M^{k}_{A\rightarrow B}\}}\sup_{\rho_{A}\in\cS(\cH_{A})}\sum_{k}p_{k}\cdot H(B)_{\rho^{k}}\ ,
\end{align}
where the infimum is over all Kraus decompositions $\{M^{k}_{A\rightarrow B}\}$ of $\cE_{A\rightarrow B}$,
\begin{align}
\rho^{k}_{B}=\frac{1}{p_{k}}\cdot M^{k}_{A\rightarrow B}\rho_{A}\left.M^{k}_{A\rightarrow B}\right.^{\dagger}\ ,
\end{align}
and $p_{k}=\trace\left[M^{k}_{A\rightarrow B}\rho_{A}\left.M^{k}_{A\rightarrow B}\right.^{\dagger}\right]$.
\end{lemma}

\begin{proof}
We construct a sequence of one-shot channel simulations $\cP^{n}$ with asymptotically vanishing error $\eps_{n}$, and an entanglement cost $\frac{\log L_{n}}{n}$ as in~\eqref{eq:firststep}. Without loss of generality we choose $\cP^{n}$ to be permutation-covariant.\footnote{The arguments for this are exactly the same as in Section~\ref{sec:meas_main}, except that here Alice and Bob first need to create shared randomness by using classical communication.} The post-selection technique (Proposition~\ref{prop:postselect}) applies to permutation-covariant quantum channels, and upper bounds the error
\begin{align}\label{eq:selection}
\eps_{n}=\left\|\cE_{A\rightarrow B}^{\ot n}-\cP^{n}_{A\rightarrow B}\right\|_{\Diamond}
\end{align} by
\begin{align}\label{eq:selection2}
\eps_{n}\leq(n+1)^{|A|^{2}-1}\cdot\left\|\left(\left(\cE_{A\rightarrow B}^{\ot n}-\cP^{n}_{A\rightarrow B}\right)\ot\cI_{R}^{\ot n}\ot\cI_{R'}\right)\left(\zeta^{n}_{ARR'}\right)\right\|_{1}\ ,
\end{align}
where $\zeta^{n}_{ARR'}$ is a purification of the de Finetti state $\zeta^{n}_{AR}=\int\rho_{AR}^{\ot n}\;d(\rho_{AR})$ with $\rho_{AR}\in\cV(\cH_{A}\ot\cH_{R})$, $\cH_{R}\cong\cH_{A}$, and $d(\cdot)$ the measure on the normalized pure states on $\cH_{A}\ot\cH_{R}$ induced by the Haar measure on the unitary group acting on $\cH_{A}\ot\cH_{R}$, normalized to $\int d(\cdot)=1$. Hence, it is sufficient that the channel simulation $\cP^{n}$ creates the state
\begin{align}\label{eq:structure}
\omega_{BRR'}^{n}=\left(\cE_{A\rightarrow B}^{\ot n}\ot\cI_{R}^{\ot n}\ot\cI_{R'}\right)\left(\zeta^{n}_{ARR'}\right)\ ,
\end{align}
up to an error $o\big((n+1)^{1-|A|^{2}}\big)$ in trace distance, for an entanglement cost smaller than~\eqref{eq:firststep}.\\

Let $\delta_{n}\geq0$. Since $\zeta^{n}_{ARR'}$ is pure, the $\delta_{n}$-faithful one-shot entanglement cost for $\omega_{BRR'}^{n}$ in the bipartition $B|RR'$ is bounded by (Proposition~\ref{prop:ecstates})
\begin{align}\label{eq:crucial}
E_{C}^{(1)}(\omega_{BRR'}^{n},\delta_{n})\leq\inf_{\{M_{A\rightarrow B}^{n,k}\}}H_{0}^{\delta_{n}/2}(B|K)_{\omega^{n}}\ ,
\end{align}
where the infimum ranges over all Kraus decomposition $\{M_{A\rightarrow B}^{n,k}\}$ of $\cE_{A\rightarrow B}^{\ot n}$, and
\begin{align}
\omega_{BRR'K}=\sum_{k}M_{A\rightarrow B}^{n,k}\zeta^{n}_{ARR'}{M_{A\rightarrow B}^{n,k}}^{\dagger}\ot\proj{k}_{K}\ .
\end{align}
Using a corollary of Carath\'{e}odory's theorem (Corollary~\ref{cor:cara}), we know that 
\begin{align}
\zeta^{n}_{AR}=\int\rho_{AR}^{\ot n}\;d(\rho_{AR})=\sum_{j=1}^{N}q_{j}\cdot{\rho^{j}}^{\ot n}_{AR}\ ,
\end{align}
with $\rho_{AR}^{j}\in\cV(\cH_{A}\ot\cH_{R})$, $N=(n+1)^{2\left(|A|^{2}-1\right)}$, and $\{q_{j}\}_{j=1}^{N}$ a probability distribution. This allows us to write
\begin{align}
\omega_{BK}^{n}=\sum_{j=1}^{N}q_{j}\cdot\sum_{k}M_{A\rightarrow B}^{n,k}{\rho^{j}}^{\ot n}_{A}{M_{A\rightarrow B}^{n,k}}^{\dagger}\ot\proj{k}_{K}\ .
\end{align}
One particular choice for a Kraus decomposition $\{M_{A\rightarrow B}^{n,k}\}$ of $\cE_{A\rightarrow B}^{\ot n}$ in~\eqref{eq:crucial} is obtained by choosing a Kraus decomposition $\{M_{A\rightarrow B}^{k}\}$ for $\cE_{A\rightarrow B}$, and by using it for every tensor product factor. Thus, we find
\begin{align}
E_{C}^{(1)}(\omega_{BRR'}^{n},\delta_{n})\leq\inf_{\{M_{A\rightarrow B}^{k}\}}H_{0}^{\delta_{n}/2}(B|K)_{\omega^{n}}\ ,
\end{align}
where the infimum ranges over all Kraus decompositions $\{M_{A\rightarrow B}^{k}\}$ of $\cE_{A\rightarrow B}$, and $\omega_{BK}^{n}=\sum_{j=1}^{N}q_{j}\cdot{\omega^{j}}^{\ot n}_{BK}$ with
\begin{align}
\omega^{j}_{BK}=\sum_{k}M_{A\rightarrow B}^{k}\rho_{A}^{j}\left.M_{A\rightarrow B}^{k}\right.^{\dagger}\ot\proj{k}_{K}\ .
\end{align}
By a property of the smooth conditional R\'enyi zero-entropy (Lemma~\ref{lem:h0qconvex}) this implies
\begin{align}
E_{C}^{(1)}(\omega_{BRR'}^{n},\delta_{n})\leq\inf_{\{M_{A\rightarrow B}^{k}\}}\max_{j}H_{0}^{\delta_{n}/2}(B|K)_{{\omega^{j}}^{\ot n}}+2\left(|A|^{2}-1\right)\cdot\log(n+1)\ .
\end{align}
Using the asymptotic equipartition property for the smooth conditional R\'enyi zero-entropy (Lemma~\ref{lem:h0aep}) we arrive at
\begin{align}\label{eq:concluding}
E_{C}^{(1)}(\omega_{BRR'}^{n},\delta_{n})\leq&\;n\cdot\left(\inf_{\{M_{A\rightarrow B}^{k}\}}\max_{j}H(B|K)_{\omega^{j}}\right)\notag\\
&+\sqrt{n}\cdot\log\left(|B|+3\right)\cdot\sqrt{\log\frac{4}{\delta_{n}^{2}}}+2\left(|A|^{2}-1\right)\cdot\log(n+1)\ .
\end{align}
In summary, there exists a $\delta_{n}$-faithful (measured in purified distance) one-shot entanglement dilution protocol $\cP^{n}$ for the state $\omega_{BRR'}^{n}$ for an entanglement cost as given in~\eqref{eq:concluding}.

Now, we choose $\delta_{n}=\frac{1}{2}(n+1)^{2\left(1-|A|^{2}\right)}$, and the entanglement cost becomes
\begin{align}
E_{C}^{(1)}(\omega_{BRR'}^{n},\frac{1}{2}(n+1)^{2\left(1-|A|^{2}\right)})\leq\;&n\cdot\left(\inf_{\{M_{A\rightarrow B}^{k}\}}\max_{j}H(B|K)_{\omega^{j}}\right)\notag\\
&+2\left(|A|^{2}-1\right)\cdot\log(n+1)\notag\\
&+\sqrt{n}\cdot\log\left(|B|+3\right)\cdot2\sqrt{1+\log(n+1)\cdot\left(|A|^{2}-1\right)}\ .
\end{align}
By the equivalence of the purified distance and the trace distance (Lemma~\ref{lem:pdbounds}), the error measured in the trace distance is then upper bounded by $(n+1)^{2\left(1-|A|^{2}\right)}$. This together with~\eqref{eq:selection} implies that there exists a sequence of one-shot channel simulations $\cP^{n}$ for $\cE^{\ot n}$ with error
\begin{align}
\lim_{n\rightarrow\infty}\eps_{n}=\lim_{n\rightarrow\infty}\|\cE_{A\rightarrow B}^{\ot n}-\cP^{n}_{A\rightarrow B}\|_{\Diamond}\leq\lim_{n\rightarrow\infty}(n+1)^{1-|A|^{2}}=0\ ,
\end{align}
where the entanglement cost of this channel simulation is bounded by
\begin{align}
\inf_{\{M_{A\rightarrow B}^{k}\}}\max_{j}H(B|K)_{\omega^{j}}\leq\inf_{\{M_{A\rightarrow B}^{k}\}}\sup_{\rho_{A}\in\cS(\cH_{A})}\sum_{k}p_{k}\cdot H(B)_{\rho^{k}}\ ,
\end{align}
and the infimum ranges over all Kraus decompositions $\{M_{A\rightarrow B}^{k}\}$ of $\cE_{A\rightarrow B}$, $\rho_{B}^{k}=\frac{1}{p_{k}}M_{A\rightarrow B}^{k}\rho_{A}{M_{A\rightarrow B}^{k}}^{\dagger}$, and $p_{k}=\trace\left[M_{A\rightarrow B}^{k}\rho_{A}{M_{A\rightarrow B}^{k}}^{\dagger}\right]$.
\end{proof}

\begin{lemma}\label{lem:nonblocking}
Let $\cE_{A\rightarrow B}:\cB(\cH_{A})\rightarrow\cB(\cH_{B})$ be a channel. Then, we have that
\begin{align}
E_{C}(\cE_{A\rightarrow B})\leq\max_{\rho_{AR}}E_{F}\left(\left(\cE_{A\rightarrow B}\ot\cI_{R}\right)\left(\rho_{AR}\right)\right)=E_{F}(\cE_{A\rightarrow B})\ ,
\end{align}
where $\rho_{AR}\in\cV(\cH_{A}\ot\cH_{R})$, and $\cH_{A}\cong\cH_{R}$. We call the expression $E_{F}(\cE_{A\rightarrow B})$ on the right-hand side the entanglement of formation of the quantum channel $\cE_{A\rightarrow B}$.
\end{lemma}

\begin{proof}
Note that the entanglement of formation involves a minimization, and therefore the only thing to do to go from Lemma~\ref{lem:main_2} to Lemma~\ref{lem:nonblocking}, is to interchange the infimum with the supremum. We will do this by using Sion's minimax theorem (Lemma~\ref{lem:minimax}). To start with, we want to discretize the set of Kraus decompositions $\{M_{k}\}$ of $\cE$ with at most $\chi$ Kraus operators. For this we note that every such Kraus decomposition $\{M_{k}\}$ can be seen as a vector $v_{\chi}\in\mathbb{C}^{\chi\cdot|A||B|}$, by just writing all Kraus operators one after another in a vector.\footnote{Kraus decompositions with less than $\chi$ Kraus operators can just be filled up with zeros.} Furthermore, we have $\sum_{k}M_{k}^{\dagger}M_{k}=\1_{B}$ and therefore $v_{\chi}\in\cN_{\chi}=\{w\in\mathbb{C}^{\chi\cdot|A||B|}\mid\|w\|_{2}=\sqrt{|B|}\}$.\footnote{For this note that $\|v_{\chi}\|_{2}=\|\sum_{k}M_{k}^{\dagger}M_{k}\|_{2}$, where the norm on the left-hand side denotes the euclidean vector norm, and the norm on the right-hand side denotes the Hilbert-Schmidt matrix norm.} We now discretize the set $\cT_{\chi}\subset\cN_{\chi}$ of all $v_{\chi}$ that correspond to a Kraus decomposition $\{M_{k}\}$ of $\cE$ with at most $\chi$ Kraus operators, using a lemma about $\eps$-nets (Lemma~\ref{lem:net}). The lemma states that there exists a set $\cT_{\chi,\eps}\subset\cT_{\chi}$ with $|\cT_{\chi,\eps}|\leq\left(\frac{2\sqrt{|B|}}{\eps}+1\right)^{2\chi\cdot|A||B|}$, such that for every $v_{\chi}\in\cT_{\chi}$, there exists a $v_{\chi,\eps}\in\cT_{\chi,\eps}$ with $\|v_{\chi}-v_{\chi,\eps}\|_{2}\leq\eps$.

As the next step we consider the set $\Gamma_{\chi,\eps}$ of probability distributions $\{q_{j}\}_{j=1}^{N}$ over $\cT_{\chi,\eps}$, and note for every such probability distribution, there exists a corresponding Kraus decomposition $\{\sqrt{q_{j}}\cdot M_{j,k}\}_{j,k=1}^{N,\chi}$ of $\cE$. Restricting the infimum in~\eqref{eq:firststep} to $\Gamma_{\chi,\eps}$, we find
\begin{align}\label{eq:secondstep}
E_{C}(\cE_{A\rightarrow B})\leq\inf_{\Gamma_{\chi,\eps}}\sup_{\rho}\sum_{j}q_{j}\cdot\sum_{k}p_{j,k}\cdot H(B)_{\rho^{j,k}}\ ,
\end{align}
where $\rho^{j,k}=\frac{1}{p_{j,k}}M_{j,k}\rho M_{j,k}^{\dagger}$, and $p_{j,k}=\trace\left[M_{j,k}\rho M_{j,k}^{\dagger}\right]$.\\

In order to apply the minimax theorem (Lemma~\ref{lem:minimax}) to~\eqref{eq:secondstep}, we now check all the conditions of Lemma~\ref{lem:minimax}:
\begin{itemize}
\item$\cS(\cH_{A})$ is a compact, convex set.
\item In order to see that $\sum_{j}q_{j}\cdot\sum_{k}p_{j,k}\cdot H(B)_{\rho^{j,k}}$ is concave in $\rho_{A}$, we consider $\rho_{A}=r^{(1)}\cdot\rho_{A}^{1}+r^{(2)}\cdot\rho_{A}^{2}$ with $\rho_{A}^{1},\rho_{A}^{2}\in\cS(\cH_{A})$, and $r^{(1)}+r^{(2)}=1$. We define $\tilde{r}^{(1)}_{j,k}=\frac{r^{(1)}\cdot p^{(1)}_{j,k}}{p_{j,k}}$, $\tilde{r}^{(2)}_{j,k}=\frac{r^{(2)}\cdot p^{(2)}_{j,k}}{p_{j,k}}$ with $p^{(1)}_{j,k}=\trace\left[M_{j,k}\rho^{1}_{A}M_{j,k}^{\dagger}\right]$, $p^{(2)}_{j,k}=\trace\left[M_{j,k}\rho^{2}_{A}M_{j,k}^{\dagger}\right]$. Since $\tilde{r}^{(1)}_{j,k}+\tilde{r}^{(2)}_{j,k}=1$, and since von Neumann entropy is concave, we have 
\begin{align}
H(B)_{\rho^{j,k}}\geq\tilde{r}^{(1)}_{j,k}\cdot H(B)_{\rho^{1,j,k}}+\tilde{r}^{(2)}_{j,k}\cdot H(B)_{\rho^{2,j,k}}\ .
\end{align}
for $\rho_{B}^{1,j,k}=\frac{M_{j,k}\rho^{1}_{A}M_{j,k}^{\dagger}}{p^{(1)}_{j,k}}$, $\rho_{B}^{2,j,k}=\frac{M_{j,k}\rho^{2}_{A}M_{j,k}^{\dagger}}{p^{(2)}_{j,k}}$. By multiplying this with $q_{j}\cdot p_{j,k}$ and taking the sum over all $j,k$ we conclude
\begin{align}
\sum_{j}q_{j}\cdot\sum_{k}p_{j,k}\cdot H(B)_{\rho^{j,k}}\geq&\;r^{(1)}\cdot\sum_{j}q_{j}\cdot\sum_{k}p^{(1)}_{j,k}\cdot H(B)_{\rho^{1,j,k}}\notag\\
&+r^{(2)}\cdot\sum_{j}q_{j}\cdot\sum_{k}p^{(2)}_{j,k}\cdot H(B)_{\rho^{2,j,k}}\ .
\end{align}
\item The function $\sum_{j}q_{j}\cdot\sum_{k}p_{j,k}\cdot H(B)_{\rho^{j,k}}$ is also continuous in $\rho_{A}$, since for any $\rho_{A}^{1},\rho_{A}^{2}\in\cS(\cH_{A})$ with $\|\rho^{1}_{A}-\rho^{2}_{A}\|_{1}\leq\delta$ for some $\delta>0$, it follows from the monotonicity of the trace norm and the continuity of the conditional von Neumann entropy (Lemma~\ref{lem:fannes}) that
\begin{align}
|\sum_{j}q_{j}\cdot\sum_{k}p^{(1)}_{j,k}\cdot H(B)_{\rho^{1,j,k}}-\sum_{j}q_{j}\cdot\sum_{k}p^{(2)}_{j,k}\cdot H(B)_{\rho^{2,j,k}}|\leq4\delta\log|B|+2h(\delta)\ ,
\end{align}
where $h(\delta)=-\delta\log\delta-(1-\delta)\log(1-\delta)$ denotes the binary Shannon entropy.
\item $\Gamma_{\chi,\eps}$ is a compact, convex set.
\item Moreover, $\sum_{j}q_{j}\cdot\sum_{k}p_{j,k}\cdot H(B)_{\rho^{j,k}}$ is linear in $\{q_{j}\}$, and therefore in particular convex and continuous.
\end{itemize}
Finally, applying the minimax theorem (Lemma~\ref{lem:minimax}) to the right-hand side of \eqref{eq:secondstep}, we find
\begin{align}\label{eq:christandl1}
E_{C}(\cE_{A\rightarrow B})\leq\sup_{\rho}\inf_{\Gamma_{\chi,\eps}}\sum_{j}q_{j}\cdot\sum_{k}p_{j,k}\cdot H(B)_{\rho^{j,k}}\ .
\end{align}
Since the function is concave, the infimum is taken on an extreme point, and hence
\begin{align}\label{eq:christandl2}
\inf_{\Gamma_{\chi,\eps}}\sum_{j}q_{j}\cdot\sum_{k}p_{j,k}\cdot H(B)_{\rho^{j,k}}=\inf_{\{M_{k}\}}\sum_{k}p_{k}\cdot H(B)_{\rho^{k}}\ ,
\end{align}
where the second infimum ranges over all Kraus decompositions $\{M_{k}\}\cong v_{\chi,\eps}\in\cT_{\chi,\eps}$ of $\cE$.

Now, let $0<\eps\leq\frac{1}{2\chi|B|}$. As the next step we show that for every Kraus decomposition $\{M_{k}\}\cong v_{\chi}\in\cT$ of $\cE$, there exists a Kraus decomposition $\{M_{k,\eps}\}\cong v_{\chi,\eps}\in\cT_{\chi,\eps}$ of $\cE$, such that
\begin{align}\label{eq:finalprop}
|\sum_{k}p_{k,\eps}\cdot H(B)_{\rho^{k,\eps}}-\sum_{k}p_{k}\cdot H(B)_{\rho^{k}}|\leq8\eps\chi|B|\log|B|+2h(2\eps\chi|B|)\ ,
\end{align}
where $\rho^{k,\eps}=\frac{1}{p_{k,\eps}}M_{k,\eps}\rho M_{k,\eps}^{\dagger}$, and $p_{k,\eps}=\trace\left[M_{k,\eps}\rho M_{k,\eps}^{\dagger}\right]$. To see this, we rewrite the left-hand side of~\eqref{eq:finalprop} to
\begin{align}\label{eq:ideas}
|\sum_{k}p_{k,\eps}\cdot H(B)_{\rho^{k,\eps}}-\sum_{k}p_{k}\cdot H(B)_{\rho^{k}}|=|H(B|K)_{\rho^{k,\eps}}-H(B|K)_{\rho^{k}}|\ ,
\end{align}
where $\rho^{k,\eps}_{BR}=\sum_{k}p_{k,\eps}\cdot\rho^{k,\eps}_{B}\ot\proj{k}_{K}$, and $\rho^{k}_{BR}=\sum_{k}p_{k}\cdot\rho^{k}_{B}\ot\proj{k}_{K}$. In order to bound~\eqref{eq:ideas} we will use the continuity of the conditional von Neumann entropy (Lemma~\ref{lem:fannes}). Note that we have
\begin{align}
\Big\|\sum_{k}p_{k,\eps}\cdot\rho^{k,\eps}_{B}\ot\proj{k}_{R}-\sum_{k}p_{k}\cdot\rho^{k}_{B}\ot\proj{k}_{R}\Big\|_{1}=\sum_{k}\Big\|M_{k,\eps}\rho M_{k,\eps}^{\dagger}-M_{k}\rho M_{k}^{\dagger}\Big\|_{1}\ .
\end{align}
By the triangle inequality for the trace norm, the equivalence of the trace norm and the Hilbert-Schmidt norm (Lemma~\ref{lem:12norm}), and the sub-multiplicativity of the Hilbert-Schmidt norm (Lemma~\ref{lem:sub}), we find
\begin{align}
&\sum_{k}\big\|M_{k,\eps}\rho M_{k,\eps}^{\dagger}-M_{k}\rho M_{k}^{\dagger}\big\|_{1}\notag\\
&\leq\sum_{k}\big\|M_{k,\eps}\rho\big(M_{k,\eps}^{\dagger}-M_{k}^{\dagger}\big)\big\|_{1}+\big\|\big(M_{k,\eps}-M_{k}\big)\rho M_{k}^{\dagger}\big\|_{1}\notag\\
&\leq\sqrt{|B|}\cdot\Big(\sum_{k}\big\|M_{k,\eps}\rho\big(M_{k,\eps}^{\dagger}-M_{k}^{\dagger}\big)\big\|_{2}+\big\|\big(M_{k,\eps}-M_{k}\big)\rho M_{k}^{\dagger}\big\|_{2}\Big)\notag\\
&\leq\sqrt{|B|}\cdot\Big(\sum_{k}\|M_{k,\eps}\|_{2}\cdot\|\rho\|_{2}\cdot\big\|M_{k,\eps}^{\dagger}-M_{k}^{\dagger}\big\|_{2}+\big\|M_{k,\eps}-M_{k}\big\|_{2}\cdot\|\rho\|_{2}\cdot\|M_{k,\eps}^{\dagger}\|_{2}\Big)\notag\\
&\leq\sqrt{|B|}\cdot\left(\chi\sqrt{|B|}\eps+\eps\chi\sqrt{|B|}\right)=2\eps\chi|B|\ .\label{eq:blocklast}
\end{align}
Now, \eqref{eq:finalprop} follows by the continuity of the conditional von Neumann entropy (Lemma \ref{lem:fannes}) applied to~\eqref{eq:ideas} together with~\eqref{eq:blocklast}. Thus, we find using~\eqref{eq:christandl1} and~\eqref{eq:christandl2} that
\begin{align}
E_{C}(\cE_{A\rightarrow B})\leq\sup_{\rho}\inf_{\{M_{k}\}}\sum_{k}p_{k}\cdot H(B)_{\rho^{k}}+8\eps\chi|B|\log|B|+2h(2\eps\chi|B|)\ ,
\end{align}
where the infimum ranges over all Kraus decompositions $\{M_{k}\}\cong v_{\chi}\in\cT$ of $\cE$. Finally, note that
\begin{align}
\inf_{\{M_{k}\}}\sum_{k}p_{k}\cdot H(B)_{\rho^{k}}=E_{F}\Big(\sum_{k}M^{k}_{A\rightarrow B}\rho_{AR}{M^{k}_{A\rightarrow B}}^{\dagger}\Big)\ ,
\end{align}
where the infimum ranges over all Kraus decompositions $\{M_{k}\}$ of $\cE$, $\rho_{AR}\in\cV(\cH_{A}\ot\cH_{R})$, and $\cH_{R}\cong\cH_{A}$. This infimum is actually taken for a decomposition of size at most $|A|^{2}|B|^{2}$ (Lemma~\ref{lem:uhlman}). Thus, for $\chi=|A|^{2}|B|^{2}$ and $\eps\rightarrow0$ we find
\begin{align}
E_{C}(\cE_{A\rightarrow B})\leq\sup_{\rho_{AR}}E_{F}\left(\left(\cE_{A\rightarrow B}\ot\cI_{R}\right)\left(\rho_{AR}\right)\right)\ ,
\end{align}
where $\rho_{AR}\in\cV(\cH_{A}\ot\cH_{R})$, and $\cH_{R}\cong\cH_{A}$. Since the entanglement of formation is continuous (Lemma~\ref{lem:nielsen}) and $\cS(\cH_{A})$ is compact, the supremum can be turned into a maximum.
\end{proof}

\begin{proposition}\label{prop:ECfinal}
Let $\cE_{A\rightarrow B}:\cB(\cH_{A})\rightarrow\cB(\cH_{B})$ be a channel. Then, we have that
\begin{align}
E_{C}(\cE_{A\rightarrow B})\leq\lim_{n\rightarrow\infty}\frac{1}{n}\max_{\rho^{n}_{AR}}E_{F}\left(\left(\cE^{\ot n}_{A\rightarrow B}\ot\cI_{R}\right)\left(\rho^{n}_{AR}\right)\right)\ ,
\label{form}
\end{align}
where $\rho^{n}_{AR}\in\cV(\cH_{A}^{\ot n}\ot\cH_{R}^{\ot n})$, and $\cH_{R}\cong\cH_{A}$.
\end{proposition}

\begin{proof}
This follows from standard blocking arguments as in~\cite{Barnum98}. Namely, by applying the non-regularized achievability (Lemma~\ref{lem:nonblocking}) to the quantum channel $\cE_{A\rightarrow B}^{\ot n}$ for some $n>1$, we find
\begin{align}
E_{C}(\cE_{A\rightarrow B}^{\ot n})\leq\frac{1}{n}\cdot\max_{\rho_{AR}^{n}}E_{F}\left(\left(\cE_{A\rightarrow B}^{\ot n}\ot\cI_{R}\right)\left(\rho_{AR}^{n}\right)\right)\ ,
\end{align}
where $\rho^{n}_{AR}\in\cV(\cH_{A}^{\ot n}\ot\cH_{R}^{\ot n})$, and $\cH_{R}\cong\cH_{A}$. Since $n\cdot E_{C}(\cE_{A\rightarrow B})\leq E_{C}(\cE_{A\rightarrow B}^{\ot n})$,\footnote{This is immediate since a channel simulation for $\cE_{A\rightarrow B}^{\ot n}$ is a channel simulation for $n$ copies of $\cE_{A\rightarrow B}$.} we obtain the claim for $n\rightarrow\infty$.
\end{proof}

The idea of the proof of the converse is that any channel simulation for $\cE_{A\rightarrow B}$ is able to produce any state of the form $\left(\cE^{\ot n}_{A\rightarrow B}\ot\cI_{R}^{\ot n}\right)\left(\rho^{n}_{AR}\right)$. The converse for the one-shot entanglement cost for quantum states (Proposition~\ref{prop:ecstates}), however, can be used to derive a lower bound on the entanglement needed to produce such states.

\begin{proposition}\label{prop:ECconverse}
Let $\cE_{A\rightarrow B}:\cB(\cH_{A})\rightarrow\cB(\cH_{B})$ be a channel. Then, we have that
\begin{align}
E_{C}(\cE_{A\rightarrow B})\geq\lim_{n\rightarrow\infty}\frac{1}{n}\cdot\max_{\rho^{n}_{AR}}E_{F}\left(\left(\cE^{\ot n}_{A\rightarrow B}\ot\cI_{R}\right)\left(\rho^{n}_{AR}\right)\right)\ ,
\end{align}
where $\rho^{n}_{AR}\in\cV(\cH_{A}^{\ot n}\ot\cH_{R}^{\ot n})$, and $\cH_{R}\cong\cH_{A}$.
\end{proposition}

\begin{proof}
By the definition of an $\eps$-faithful one-shot channel simulation $\cP^{n}$ for $\cE^{\ot n}$ (Definition~\ref{def:simulation}), we have that
\begin{align}
\left\|\cP^{n}-\cE^{\ot n}\right\|_{\Diamond}\leq\eps\ .
\end{align}
This implies in particular that
\begin{align}
\max_{\rho^{n}_{AR}}\left\|\left(\left(\cP_{A\rightarrow B}^{n}-\cE^{\ot n}_{A\rightarrow B}\right)\ot\cI_{R}\right)\left(\rho^{n}_{AR}\right)\right\|_{1}\leq\eps\ ,
\end{align}
where $\rho^{n}_{AR}\in\cV(\cH_{A}^{\ot n}\ot\cH_{R}^{\ot n})$, and $\cH_{R}\cong\cH_{A}$. Hence, every $\eps$-faithful one-shot channel simulation $\cP^{n}$ for $\cE^{\ot n}$ needs to be able to produce any state of the form $\left(\cE^{\ot n}_{A\rightarrow B}\ot\cI_{R}\right)\left(\rho^{n}_{AR}\right)$ up to an error $\eps$ (measured in trace distance). But by the definition of the one-shot entanglement cost for quantum states (Definition~\ref{def:state_cost}), the entanglement that is needed for this, is given by
\begin{align}
\max_{\rho^{n}_{AR}}E_{C}^{(1)}\left(\left(\cE^{\ot n}_{A\rightarrow B}\ot\cI_{R}\right)\left(\rho^{n}_{AR}\right),\eps/2\right)\ ,
\end{align}
where $\rho^{n}_{AR}\in\cV(\cH_{A}^{\ot n}\ot\cH_{R}^{\ot n})$, and $\cH_{R}\cong\cH_{A}$.\footnote{The factor $1/2$ appears because the one shot entanglement cost for quantum states is defined in terms of the purified distance (Definition~\ref{def:state_cost}), cf.~Lemma~\ref{lem:pdbounds} about the equivalence of distant measures.} Thus, we find 
\begin{align}\label{eq:conv1}
E_{C}(\cE_{A\rightarrow B})\geq\lim_{\eps\rightarrow0}\lim_{n\rightarrow\infty}\frac{1}{n}\cdot\max_{\rho^{n}_{AR}}E_{C}^{(1)}\left(\left(\cE_{A\rightarrow B}^{\ot n}\ot\cI_{R}\right)\left(\rho^{n}_{AR}\right),\eps/2\right)\ ,
\end{align}
where $\rho^{n}_{AR}\in\cV_{\leq}(\cH_{A}^{\ot n}\ot\cH_{R}^{\ot n})$, and $\cH_{R}\cong\cH_{A}$. But for $\omega_{BR}^{n}=\left(\cE^{\ot n}_{A\rightarrow B}\ot\cI_{R}\right)\left(\rho^{n}_{AR}\right)$, the converse for the one-shot entanglement cost for quantum states (Proposition~\ref{prop:ecstates}) implies
\begin{align}\label{eq:conv2}
E_{C}^{(1)}(\omega_{BR}^{n},\eps/2)\geq\min_{\{p_{k},\omega^{k}\}}H_{0}^{\sqrt{2\eps}}(B|K)_{\omega^{n}}\ ,
\end{align}
where the minimum ranges over all pure states decompositions $\omega_{BR}^{n}=\sum_{k}p_{k}^{n}\cdot\omega^{n,k}_{BR}$, and $\omega_{BK}^{n}=\sum_{k}p_{k}^{n}\cdot\omega^{n,k}_{B}\ot\proj{k}_{K}$. Now, let $\bar{\omega}_{BK}^{n}\in\cB_{qc}^{\sqrt{2\eps}}(\omega_{BK}^{n})$ such that $H_{0}^{\sqrt{2\eps}}(B|K)_{\omega^{n}}=H_{0}(B|K)_{\bar{\omega}^{n}}$. Because the conditional R\'enyi zero-entropy is lower bounded by the conditional von Neumann entropy (Lemma~\ref{lem:H0vN}), and since the conditional von Neumann entropy is continuous (Lemma~\ref{lem:fannes}), we find
\begin{align}
H_{0}^{\sqrt{2\eps}}(B|K)_{\omega^{n}}=H_{0}(B|K)_{\bar{\omega}^{n}}\geq H(B|K)_{\bar{\omega}^{n}}\geq H(B|K)_{\omega^{n}}-4n\sqrt{2\eps}\log|B|-2h(\sqrt{2\eps})\ .
\end{align}
Thus, we conclude
\begin{align}
\min_{\{p_{k}^{n},\omega^{n,k}\}}H_{0}^{\sqrt{2\eps}}(B|K)_{\omega^{n}}&\geq\min_{\{p_{k}^{n},\omega^{n,k}\}}H(B|K)_{\omega^{n}}-4n\sqrt{2\eps}\log|B|-2h(\sqrt{2\eps})\notag\\
&=E_{F}(\omega_{BR}^{n})-4n\sqrt{\eps}\log|B|-2h(\sqrt{2\eps})\notag\\
&=E_{F}\left(\left(\cE^{\ot n}_{A\rightarrow B}\ot\cI_{R}\right)\left(\rho^{n}_{AR}\right)\right)-4n\sqrt{2\eps}\log|B|-2h(\sqrt{2\eps})\ ,
\end{align}
where the minimum ranges over all pure states decompositions $\omega_{BR}^{n}=\sum_{k}p_{k}^{n}\cdot\omega^{n,k}_{BR}$, and $\omega_{BK}^{n}=\sum_{k}p_{k}^{n}\cdot\omega^{n,k}_{B}\ot\proj{k}_{K}$, as well as $\rho^{n}_{AR}\in\cV(\cH_{A}^{\ot n}\ot\cH_{R}^{\ot n})$ with $\cH_{R}\cong\cH_{A}$. Together with~\eqref{eq:conv1} and~\eqref{eq:conv2} this implies
\begin{align}
E_{C}(\cE_{A\rightarrow B})&\geq\lim_{\eps\rightarrow0}\lim_{n\rightarrow\infty}\Big(\frac{1}{n}\cdot\max_{\rho^{n}_{AR}}E_{F}\left(\left(\cE^{\ot n}_{A\rightarrow B}\ot\cI_{R}\right)\left(\rho^{n}_{AR}\right)\right)\notag\\
&-4\sqrt{2\eps}\log|B|-\frac{2}{n}\cdot h(\sqrt{2\eps})\Big)\notag\\
&=\lim_{n\rightarrow\infty}\frac{1}{n}\cdot\max_{\rho^{n}_{AR}}E_{F}\left(\left(\cE^{\ot n}_{A\rightarrow B}\ot\cI_{R}\right)\left(\rho^{n}_{AR}\right)\right)\ ,
\end{align}
where $\rho^{n}_{AR}\in\cV(\cH_{A}^{\ot n}\ot\cH_{R}^{\ot n})$, and $\cH_{R}\cong\cH_{A}$. 
\end{proof}


\subsection{Discussion}\label{sec:cost_discussion}

\paragraph{Entanglement Cost of Quantum States.} As the name entanglement cost suggests, $E_{C}(\cE)$ is the channel analogue of the asymptotic iid entanglement cost of quantum states $E_{C}(\rho_{AB})$, which corresponds to the rate of entanglement needed in order to generate a bipartite quantum state $\rho_{AB}^{\ot n}$ in the limit of large $n$~\cite{Bennett96}. Our formula~\eqref{eq:main} can be seen as the analogue of the following well-known formula~\cite{Hayden01}
\begin{align}\label{eq:iidstatesEC}
E_{C}(\rho_{AB})=\lim_{n\rightarrow\infty}\frac{1}{n}\cdot E_{F}(\rho_{AB}^{\ot n})\ .
\end{align}
The gap between $E_C(\cE)$ and $Q_\leftrightarrow(\cE)$ has its analogue in the gap between $E_C(\rho_{AB})$ and $E_D(\rho_{AB})$, the distillable entanglement.

\paragraph{One-Way Classical Communication.} Our result (Theorem~\ref{thm:mainEC}) remains true if we restrict the classical communication to be one-way (forward or backward). This follows from the corresponding result about the entanglement cost of quantum states (Remark~\ref{rmk:datta}). This is also true we think of the problem as simulating a noisy quantum channel from a perfect quantum channel (instead of simulating a noisy quantum channel from perfect entanglement), since in this case, a maximally entangled state can always be distributed by the ideal channel.

\paragraph{Upper and Lower Bounds.} The non-regularized achievability (Lemma~\ref{lem:nonblocking}) together with the converse (Proposition~\ref{prop:ECconverse}) imply the following bounds.

\begin{corollary}\label{cor:non}
Let $\cE_{A\rightarrow B}:\cB(\cH_{A})\rightarrow\cB(\cH_{B})$ be a channel. Then, we have that
\begin{align}\label{cor:prop}
\max_{\rho_{AR}}E_{C}\left(\left(\cE_{A\rightarrow B}\ot\cI_{R}\right)\left(\rho_{AR}\right)\right)\leq E_{C}(\cE_{A\rightarrow B})\leq\max_{\rho_{AR}}E_{F}\left(\left(\cE_{A\rightarrow B}\ot\cI_{R}\right)\left(\rho_{AR}\right)\right)\ ,
\end{align}
where $\rho_{AR}\in\cV(\cH_{A}\ot\cH_{R})$, and $\cH_{R}\cong\cH_{A}$.
\end{corollary}

\paragraph{Entanglement Breaking Channels} Since the right-hand side of~\eqref{cor:prop} vanishes for every entanglement breaking channel,\footnote{A quantum channel $\cE_{A\rightarrow B}$ is called entanglement breaking if $\left(\cE_{A\rightarrow B}\ot\cI_{R}\right)\left(\rho_{AR}\right)$ is separable for all $\rho_{AR}\in\cV(\cH_{A}\ot\cH_{R})$.} and since the left-hand side of~\eqref{cor:prop} is greater than zero if the channel is not entanglement breaking~\cite{Yang05}, this results in the following corollary.

\begin{corollary}
Let $\cE_{A\rightarrow B}:\cB(\cH_{A})\rightarrow\cB(\cH_{B})$ be a channel. Then, $E_{C}(\cE_{A\rightarrow B})=0$ if and only if $\cE_{A\rightarrow B}$ is entanglement breaking.
\end{corollary}

\paragraph{Quantum Feedback Simulation.} Our result (Theorem~\ref{thm:mainEC}) concerns the case of a quantum non-feedback channel simulation. But for the corresponding feedback version, we can just modify the Definitions~\ref{def:simulation} and~\ref{def:main} by exchanging the channel $\cE_{A\rightarrow B}$ in~\eqref{eq:diamondEC} with its Stinespring dilation $U_{A\rightarrow BE}$ (where the register E is at Alice’s side). By going through our proofs (of Lemma~\ref{lem:main_2} and Proposition~\ref{prop:ECconverse}), it is then easily seen that we have the following corollary.

\begin{corollary}
Let $\cE_{A\ra B}:\cB(\cH_{A})\ra\cB(\cH_{B})$ be a channel. Then, the minimal entanglement cost of asymptotic feedback entanglement simulations for $\cE_{A\ra B}$ is given by
\begin{align}\label{eq:ECnf}
E_{C}^{fb}(\cE_{A\rightarrow B})=\max_{\rho}H(B)_{\cE(\rho)}\ ,
\end{align}
where $\rho_{A}\in\cS(\cH_{A})$.
\end{corollary}

We note that~\eqref{eq:ECnf} can also be deduced from the results in~\cite[Figure 2]{Bennett09}. The result is in analogy with the entanglement cost of quantum states~\eqref{eq:iidstatesEC}, which simplifies for pure states $\rho_{AB}$ to
\begin{align}
E_{C}(\rho_{AB})=H(B)_{\rho}=H(A)_{\rho}\ .
\end{align}

\paragraph{Entanglement of Purification.} We might ask if we can derive a similar result using only a sub-linear amount of classical communication. In this case it is known that a quantum non-feedback simulation for fixed iid input states can be done for a quantum communication rate of~\cite{Bennett09,Hayashi06_2}
\begin{align}
q=\lim_{n\rightarrow\infty}\frac{1}{n}\cdot E_{P}\left(\left(\cE_{A\rightarrow B}\ot\cI_{R}\right)\left(\Phi_{AR}\right)^{\ot n}\right)\ ,
\end{align}
with $\Phi_{AR}$ the maximally entangled state, and $E_{P}$ the entanglement of purification \cite{Terhal02}. Now, we might hope to generalize this to a general quantum non-feedback channel simulation for arbitrary input states using the techniques presented here, leading to
\begin{align}
q=\lim_{n\rightarrow\infty}\frac{1}{n}\cdot\max_{\rho^{n}_{AR}}E_{P}\left(\left(\cE^{\ot n}_{A\rightarrow B}\ot\cI_{R}\right)\left(\rho^{n}_{AR}\right)\right)\ ,
\end{align}
where $\rho^{n}_{AR}\in\cV(\cH_{A}^{\ot n}\ot\cH_{R}^{\ot n})$. However, this does not work for the same reason that embezzling states~\cite{vanDam03} are needed in the quantum reverse Shannon theorem; the issue of entanglement spread~\cite{Bennett09,Harrow09}.

\paragraph{Efficiency of Simulation.} An open question concerns the relation of $E_C({\cal E})$ and $Q_{\rightarrow}({\cal E})$. We know that $E_C({\cal E}) \geq Q_{\leftrightarrow}({\cal E})$, with the inequality typically being strict. Can we obtain a characterization of channels for which $E_C({\cal E}) = Q_{\leftrightarrow}({\cal E})$? This is an analogue of the problem of characterizing bipartite states for which the distillable entanglement is equal the entanglement cost, which is still open.

\paragraph{Strong Converse Capacities.} In Section~\ref{se:app}, we show that the entanglement cost of quantum channels is an upper bound on the strong converse two-way classical communication assisted quantum capacity. By the means of our non-regularized achievability (Proposition~\ref{lem:nonblocking}) we can then also give an explicit upper bound for any qubit channel.


\section{Application: Strong Converse Capacities}\label{se:app}

The results in this section have been obtained in collaboration with Fernando Brand\~ao, Matthias Christandl, and Stephanie Wehner, and have appeared in~\cite{Berta12,Berta11_2}. It was realized in~\cite{Bennett09} that the quantum reverse Shannon implies that the entanglement assisted classical and quantum capacities are in fact so-called strong converse capacities. Strong converse capacities are minimal rates above which any attempt to send information necessarily has exponentially small fidelity.\footnote{The strong converse capacity is greater than or equal to the standard capacity, which is defined as the minimal rate above which the fidelity does not approach one.} We do not reproduce this result here, but use the same arguments to show that the entanglement cost of quantum channels gives an upper bound on the strong converse two-way classical communication assisted quantum capacity. Our result also has an application in quantum cryptography, namely, for analyzing security in the noisy-storage model. We can relate security in this model to a problem of sending quantum rather than classical information through the adversary's storage device, and this then improves the range of parameters where security can be shown (cf.~Section~\ref{se:qc}).

Determining strong converse capacities of quantum channels is challenging, and besides the entanglement assisted findings, only partial results are known. The strong converse capacity for sending classical information over a quantum channel is known to be equal to the classical capacity of a quantum channel for a selected number of channels~\cite{Koenig09_2}, or under additional assumptions~\cite{Winter99,Ogawa99}. For many channels there are also upper bounds known~\cite{Bennett09,Datta11,Dorlas11}, but a general formula for the strong converse classical capacity is still lacking. Understanding the strong converse capacity for sending quantum information over a quantum channel turns out to be an even more difficult problem. Apart from the entanglement assisted case, the only previous result is for degradable quantum channels~\cite{Morgan13}.

Similar to the quantum reverse Shannon theorem~\cite{Bennett09}, we employ the idea of a channel simulation to prove that when we send quantum information at a rate exceeding $E_{C}$, then the fidelity decreases exponentially fast. Our bound holds for all quantum channels. To start with, let us first define the notion of quantum capacity more formally.

\begin{definition}
Consider a bipartite system with parties Alice and Bob. Let $\eps\geq0$, and let $\cE:\cB(\cH_{A})\rightarrow\cB(\cH_{B})$ be a channel, where Alice controls $\cH_{A}$ and Bob $\cH_{B}$. An $\eps$-error code for $\cE$ consists of an encoding channel $\Lambda_{\enc}:\left(\mathbb{C}^{2}\right)^{\ot R}\rightarrow\cH_{A}$ on Alice's side, and a decoding channel $\Lambda_{\dec}:\cH_{B}\rightarrow\left(\mathbb{C}^{2}\right)^{\ot R}$ on Bob's side such that
\begin{align}\label{eq:capacity}
\|\Lambda_{\dec}\circ\cE\circ\Lambda_{\enc}-\cI\|_{\Diamond}\leq\eps\ ,
\end{align}
where $\cI:\left(\mathbb{C}^{2}\right)^{\ot R}\rightarrow\left(\mathbb{C}^{2}\right)^{\ot R}$ is the identity channel, and the cost of the code is given by $R$. Furthermore, an asymptotic code for $\cE$ is a sequence of $\eps_{n}$-error codes for $\cE^{\ot n}$ with cost $R_{n}$ such that $\lim_{n\rightarrow\infty}\eps_{n}=0$, and the corresponding asymptotic rate is given by $R=\limsup_{n\rightarrow\infty}\frac{R_{n}}{n}$. The quantum capacity $Q(\cE)$ is then defined as the minimal asymptotic rate of asymptotic codes for $\cE$.
\end{definition}

Note that there are slightly different ways to define the quantum capacity, and we could use other distance measures (like the entanglement fidelity or the channel fidelity) in~\eqref{eq:capacity}. Yet, it was as shown that all definitions lead to the same capacity  (see Lemma~\ref{lem:werner}, taken from~\cite{Kretschmann04}). Similarly, we can define the quantum capacity in the presence of free classical forward communication from the sender to the receiver, denoted by $Q_{\rightarrow}$, the quantum capacity in the presence of free classical backward communication from the receiver to the sender,  denoted by $Q_{\leftarrow}$, and the two-way classical communication assisted quantum capacity $Q_{\leftrightarrow}$.

As our argument makes crucial use of the idea of simulating a noisy channel with perfect, noise-free, channels, we now first establish a strong converse for the identity channel. For the unassisted quantum capacity this is straightforward, and can be understood in terms of the impossibility of compressing $n$ qubits into a smaller storage device.

\begin{lemma}\label{lem:identity}
Let $\cI_{2}$ be the qubit identity channel. Then, we have for every sequence of $\eps_{n}$-error codes for $\cI_{2}^{\ot n}$ with asymptotic rate $R$ that
\begin{align}
\eps_{n}\geq1-2^{-n(R-1)}\ .
\end{align}
\end{lemma}

\begin{proof}
The argument is based on standard estimations, see, e.g., \cite{Koenig09_2,Nielsen00}. For Kraus decompositions $\{E_{j}\}$, $\{D_{k}\}$ of the channel $\Lambda_{\enc}$, $\Lambda_{\dec}$ respectively, we bound the error $\eps_{n}$ by using a lemma about the equivalence of distance measures (Lemma~\ref{lem:werner})
\begin{align}
\eps_{n}&\geq1-\sum_{j,k}\left|\trace\left[D_{k}E_{j}\left(\frac{\1}{2^{nR}}\right)\right]\right|^{2}\notag\\
&\geq1-\sum_{j,k}\trace\left[D_{k}E_{j}\left(\frac{\1}{2^{nR}}\right)E_{j}^{\dagger}D_{k}^{\dagger}\right]\trace\left[\Pi_{k}\left(\frac{\1}{2^{nR}}\right)\right]\notag\\
&=1-\frac{1}{2^{nR}}\sum_{j,k}\trace\left[D_{k}E_{j}\left(\frac{\1}{2^{nR}}\right)E_{j}^{\dagger}D_{k}^{\dagger}\right]\trace\left[\Pi_{k}\right]\notag\\
&\geq1-2^{-n(R-1)}\ ,
\end{align}
where $\Pi_{k}$ denotes the projector onto the subspace to which $D_{k}$ maps, the second inequality follows from the Cauchy-Schwarz inequality, and the last inequality follows from $\trace\left[\Pi_{k}\right]\leq2^{n}$.
\end{proof}

This can be generalized to the case of free classical communication assistance.

\begin{corollary}\label{cor:strong}
Let $\cI_{2}$ be the qubit identity channel. Then, we have for every sequence of classical communication assisted $\eps_{n}$-error codes for $\cI_{2}^{\ot n}$ with asymptotic rate $R$ that
\begin{align}
\eps_{n}\geq1-2^{-n(R-1)}\ .
\end{align}
\end{corollary}

\begin{proof}
Since back communication is allowed, the general form of a protocol consists of potentially many rounds of forward quantum and classical communication as well as backward classical communication. We first analyze one such round, which has without loss of generality the following form:
\begin{enumerate}
\item Map $\cD^{1}$ at the receiver with Kraus operators $\{D_{i}^{1}\}$
\item Classical communication from the receiver to the sender, denoted by $B$
\item Map $\cE$ at the sender with Kraus operators $\{\hat{E}_{j,b}\}=\{E_{j,b}\ot\proj{b}_{B}\}$
\item Classical communication from the sender to the receiver, denoted by $F$
\item Map $\cD^{2}$ at the receiver with Kraus operators $\{\hat{D}_{k,f}^{2}\}=\{D_{k,f}^{2}\ot\proj{f}_{F}\}$.
\end{enumerate}
The channel fidelity after this round can be estimated as before (Lemma~\ref{lem:identity})
\begin{align}
\eps_{n}&\geq1-\sum_{ijkbf}\left|\trace\left[\hat{D}^{2}_{k,f}\hat{E}_{j,b}D_{i}^{1}\left(\frac{\1}{2^{nR}}\right)\right]\right|^{2}\notag\\
&\geq1-\sum_{ijkbf}\trace\left[\hat{D}^{2}_{k,f}\hat{E}_{j,b}D_{i}^{1}\left(\frac{\1}{2^{nR}}\right){D_{i}^{1}}^{\dagger}\hat{E}_{j,b}^{\dagger}\left(\hat{D}_{k,f}^{2}\right)^{\dagger}\right]\trace\left[\Pi_{k,f}\left(\frac{\1}{2^{nR}}\right)\right]\notag\\
&\geq1-2^{-n(R-1)}\ ,
\end{align}
where $\Pi_{k,f}$ denote the projector onto the subspace that $\hat{D}^{2}_{k,f}$ maps. It is now easily seen that adding more rounds does not affect the argument; the projectors $\Pi$ are just chosen such that they project on the subspaces to which the Kraus operators of the last channel at the receiver map to.
\end{proof}

To generalize this to arbitrary quantum channels we need one more ingredient. We need to show that the asymptotic non-feedback entanglement simulation for some channel (as discussed in Theorem~\ref{thm:mainEC}) can be done for an error rate which is exponentially small in $n$.

\begin{lemma}\label{prop:error}
Let $\cE:\cB(\cH_{A})\rightarrow\cB(\cH_{B})$ be a channel, and $\delta_{1}>0$. Then, there exists an asymptotic non-feedback entanglement simulation for $\cE$ with an entanglement cost of $E_{C}+\delta_{1}$, and error
\begin{align}
\alpha_{n}=(n+1)^{|A|^{2}-1}\cdot2^{-n\cdot\frac{\delta_{1}^{2}}{8\left(\log\left(|B|+3\right)\right)^{2}}}\ .
\end{align}
\end{lemma}

\begin{proof}
In the proof of Lemma~\ref{lem:main_2}, we can choose the error $\delta_{n}$ in the one-shot entanglement cost protocol for the de Finetti state $\omega_{BRR'}^{n}$ (as defined in~\eqref{eq:structure}) as
\begin{align}
\delta_{n}=\frac{1}{2}\cdot2^{-n\cdot\frac{\delta_{1}^{2}}{8\left(\log\left(|B|+3\right)\right)^{2}}}\ .
\end{align}
By~\eqref{eq:selection} this leads to a total error rate of
\begin{align}
\alpha_{n}=(n+1)^{|A|^{2}-1}\cdot2^{-n\cdot\frac{\delta_{1}^{2}}{8\left(\log\left(|B|+3\right)\right)^{2}}}
\end{align}
for the asymptotic channel simulation, and by~\eqref{eq:concluding} the entanglement cost for this is
\begin{align}
E_{C}^{(1)}(\omega_{BRR'}^{n},\delta_{n})\leq\;&n\cdot\min_{\{M_{A\rightarrow B}^{k}\}}\max_{j}H(B|K)_{\omega^{j}}\notag\\
&+\sqrt{n}\cdot\log\left(|B|+3\right)\cdot\sqrt{\log\frac{4}{\delta_{n}^{2}}}+2\cdot\log(n+1)\cdot\left(|A|^{2}-1\right)\ .
\end{align}
Since we have
\begin{align}
\lim_{n\rightarrow\infty}\frac{1}{n}\cdot\left(\sqrt{n}\cdot\log\left(|B|+3\right)\cdot\sqrt{\log\frac{4}{\delta_{n}^{2}}}\right)=\delta_{1}\ ,
\end{align}
we get an entanglement cost of $E_{C}+\delta_{1}$ by considering the rest of the proof of the direct part of Theorem~\ref{thm:mainEC} (that is, Lemma~\ref{lem:nonblocking} and Proposition~\ref{prop:ECfinal}).
\end{proof}

Using this, we can now prove the following upper bound on the strong converse quantum capacity. The proof is by contradiction. Since our upper bound holds for any classical communication assistance, we henceforth only talk about $Q_{\leftrightarrow}$.

\begin{theorem}\label{thm:strongconverse}
Let $\cE:\cB(\cH_{A})\rightarrow\cB(\cH_{B})$ be a channel, and $\delta_{2}>\delta_{1}>0$. Then, we have for every sequence of two-way classical communication assisted $\eps_{n}$-error codes for $\cE^{\ot n}$ with asymptotic rate $R=E_{C}(\cE)+\delta_{2}$ that
\begin{align}
\eps_{n}\geq1-(n+1)^{|A|^{2}-1}\cdot2^{-n\cdot\frac{\delta_{1}^{2}}{8\left(\log\left(|B|+3\right)\right)^{2}}}-2^{-n\cdot\frac{\delta_{2}-\delta_{1}}{E_{C}(\cE)+\delta_{1}}-1}=1-2^{-O(n)}\ .
\end{align}
\end{theorem}

\begin{proof}
We start with the perfect qubit identity channel $\cI_{2}$ and do a channel simulation for $\cE$ as defined in Definition~\ref{def:main}. As we have just seen this can be done for an entanglement cost $E_{C}(\cE)+\delta_{1}$ and an exponentially small error (Lemma~\ref{prop:error})
\begin{align}
\alpha_{n}=(n+1)^{|A|^{2}-1}\cdot2^{-n\cdot\frac{\delta_{1}^{2}}{8\left(\log\left(|B|+3\right)\right)^{2}}}\ .
\end{align}
Now suppose that there existed a hypothetical asymptotic code for $\cE$ allowing us to send information at a rate $R=E_{C}+\delta_{2}$ for an error rate $\eps_{n}\geq0$. Hence, in total, we would have an asymptotic code for $\cI_{2}$ at a rate
\begin{align}
\frac{E_{C}(\cE)+\delta_{2}}{E_{C}(\cE)+\delta_{1}}>1
\end{align}
for some error rate $\gamma_{n}>0$. But by the triangle inequality of the metric induced by the diamond norm and Corollary~\ref{cor:strong}, we know that
\begin{align}
(n+1)^{|A|^{2}-1}\cdot2^{-n\cdot\frac{\delta_{1}^{2}}{8\left(\log\left(|B|+3\right)\right)^{2}}}+\eps_{n}\geq\gamma_{n}\geq1-\frac{1}{2}\cdot2^{-n\cdot\left(\frac{E_{C}(\cE)+\delta_{2}}{E_{C}(\cE)+\delta_{1}}-1\right)}\ ,
\end{align}
and thus we are done.
\end{proof}

For generic quantum channels, we expect that this upper bound given by the entanglement cost is far from being tight. But as an easy example, we consider the qubit erasure channel
\begin{align}
\cE_{\eras}(\rho)=(1-p)\rho+p\cdot\proj{e}\ ,
\end{align}
with $p\in[0,1]$. We immediately have $E_{C}(\cE_{\eras})\geq1-p$, and calculate~\cite{Wilde12},
\begin{align}
E_{C}(\cE_{\eras})&\leq\max_{\rho}E_{F}((\cE_{\eras}\ot\cI)(\rho))\leq E_{F}((\cE_{\eras}\ot\cI)(\Phi))\notag\\
&\leq E_{F}\left((1-p)\Phi+p\cdot\proj{e}\ot\frac{\1}{2}\right)\leq (1-p)\cdot E_{F}(\Phi)+p\cdot E_{F}\left(\proj{e}\ot\frac{\1}{2}\right)\notag\\
&=1-p\ ,
\end{align}
where $\Phi$ denotes the maximally entangled state, and we used the non-regularized converse for the entanglement cost (Corollary~\ref{cor:non}), as well as the convexity of the entanglement of formation~\cite{Bennett96}. Hence $E_{C}(\cE_{\eras})=1-p$, and since it is also known that $Q_{\leftrightarrow}(\cE_{\eras})=1-p$~\cite{Bennett97}, we get by Theorem~\ref{thm:strongconverse} that $Q_{\leftrightarrow}(\cE_{\eras})$ is a strong converse capacity. Note that, this argument for the qubit erasure channel was basically already present in~\cite{Bennett97} (see also~\cite{Sharma13}).

In general, our bound may appear rather unsatisfying at first glance. How could we hope to make explicit statements about the strong converse quantum capacities when the formula for $E_{C}$ involves regularization? We now show that even though it is unclear how to calculate $E_{C}$ explicitly, we can nevertheless obtain non-trivial bounds. The key to such bounds is Lemma~\ref{lem:nonblocking}, which gives us
\begin{align}
E_{C}(\cE_{A\rightarrow B})\leq E_{F}(\cE_{A\rightarrow B})=\max_{\rho_{AR}}E_{F}\left(\left(\cE_{A\rightarrow B}\ot\cI_{R}\right)\left(\rho_{AR}\right)\right)\ ,
\end{align}
where $\rho_{AR}\in\cV(\cH_{A}\ot\cH_{R})$, and $\cH_{R}\cong\cH_{A}$. For qubit channels we can use an exact formula for the entanglement of formation as shown in~\cite{Wootters98},
\begin{align}\label{eq:wootters}
E_{F}\left(\left(\cE_{A\rightarrow B}\ot\cI_{R}\right)\left(\rho_{AR}\right)\right)=h\left(\frac{1}{2}+\frac{1}{2}\cdot\sqrt{1-\con^{2}\left(\left(\cE_{A\rightarrow B}\ot\cI_{R}\right)\left(\rho_{AR}\right)\right)}\right)\ ,
\end{align}
with the concurrence
\begin{align}\label{eq:Cexpr}
\con(\rho)=\max\left\{0,\sqrt{\lambda_{1}}-\sqrt{\lambda_{2}}-\sqrt{\lambda_{3}}-\sqrt{\lambda_{4}}\right\}\ ,
\end{align}
$\lambda_{i}$'s the eigenvalues of $\rho\tilde{\rho}$ in decreasing order, $\tilde{\rho}=(\sigma_{y}\ot\sigma_{y})\rho^{*}(\sigma_{y}\ot\sigma_{y})$ with $\rho^{*}$ the complex conjugate of $\rho$ in the canonical basis, and $\sigma_{y}=\begin{pmatrix}0 &-i\\i&0\end{pmatrix}$. Furthermore, we know from~\cite{Verstraete01,Konrad08} that for $\rho_{AR}$ pure
\begin{align}
\con\left(\left(\cE_{A\rightarrow B}\ot\cI_{R}\right)\left(\rho_{AR}\right)\right)=\con\left(\left(\cE_{A\rightarrow B}\ot\cI_{R}\right)\left(\Phi_{AR}\right)\right)\cdot \con(\rho_{AR})\ ,
\end{align}
where $\Phi_{AR}$ denotes the maximally entangled state. Since $\con(\rho_{AR})\leq1$, it follows that
\begin{align}\label{eq:secLimits}
E_{F}(\cE_{A\rightarrow B})=h\left(\frac{1}{2}+\frac{1}{2}\cdot\sqrt{1-\con^{2}\left(\left(\cE_{A\rightarrow B}\ot\cI_{R}\right)\left(\Phi_{AR}\right)\right)}\right)\ ,
\end{align}
that is, it only remains to compute $\con(\cdot)$ for the Choi-Jamiolkowski state of the channel. This can be done explicitly using~\eqref{eq:Cexpr} for any qubit channel of interest.

\begin{figure}[ht]\label{fig:example}
\begin{center}
\includegraphics[scale=1.0]{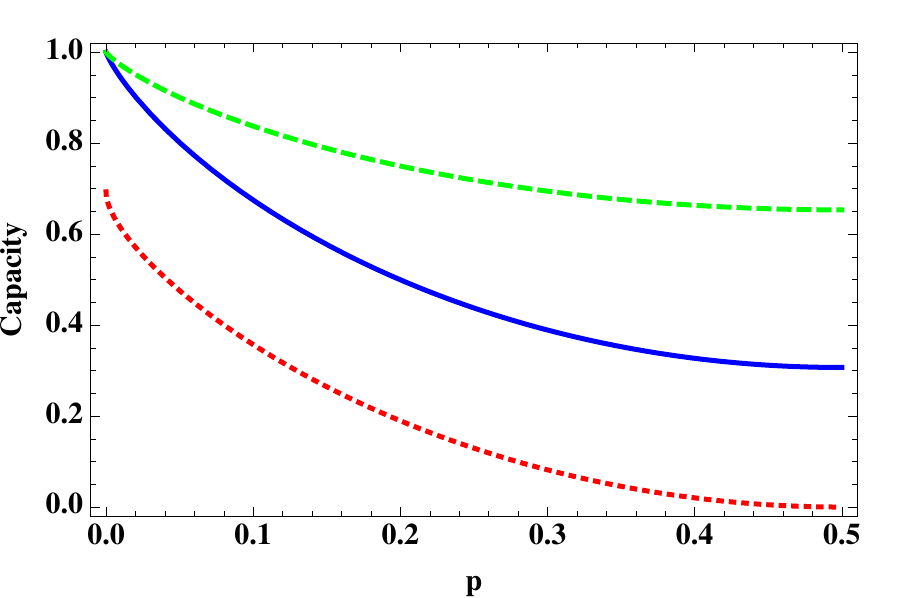}
\end{center}
\caption{Quantum capacities of the qubit dephasing channel.}{The qubit dephasing channel with dephasing parameter $p$ - quantum capacity $Q$ (dotted line) vs. upper bound $E_{F}$ on the entanglement cost (solid line) vs. entanglement assisted quantum capacity $Q_{E}$ (dashed line).}
\end{figure}

To get a feeling for this bound, this can, e.g., be compared to $Q_{\rightarrow}(\cE)$ for all degradable qubit channels~\cite{Devetak05_2,Yard08}. In the case of the the qubit dephasing channel
\begin{align}\label{eq:dephasing}
\cE_{\deph}(\rho)=(1-r)\rho+r\cdot\sigma_{z}\rho\sigma_{z}\ ,
\end{align}
with $\sigma_{z}=\begin{pmatrix} 1&0\\ 0&-1\end{pmatrix}$, we get~\cite{Wilde11},
\begin{align}
Q_{\rightarrow}(\cE_{\deph})=1-h(p)&\leq h\left(\frac{1}{2}+\sqrt{p(1-p)}\right)=E_{F}(\cE_{\deph})\notag\\
&\leq1-\frac{1}{2}\cdot h\left(\frac{p}{2}\right)=Q_{E}(\cE_{\deph})\ .
\end{align}
As shown in Figure~\ref{fig:example}, the upper bound $E_{F}$ is far from being tight, but better than the upper bound $Q_{E}$.\footnote{However, the results in~\cite{Yura03} show that our upper bound $E_{F}$ is not always better than the upper bound $Q_{E}$.} Since $Q_{\leftrightarrow}$ (and also $Q_{\leftarrow}$) can be much larger than $Q_{\rightarrow}$, and since not too much is known about these capacities, the following upper bound, which holds for every qubit channel $\cE_{A\rightarrow B}$ might be useful,
\begin{align}
Q_{\leftrightarrow}(\cE_{A\rightarrow B})\leq h\left(\frac{1}{2}+\frac{1}{2}\cdot\sqrt{1-\con^{2}\left(\left(\cE_{A\rightarrow B}\ot\cI_{R}\right)\left(\Phi_{AR}\right)\right)}\right)\ .
\end{align}

As discussed in Section~\ref{se:qc}, strong converse capacities also have an application in the noisy-storage model~\cite{Wehner08,Schaffner08,Wehner10,Koenig12}. In particular, since $E_{C}$ is an upper bound on the classical communication assisted strong converse quantum capacity, we can immediately apply $E_{C}$ for our results in Section~\ref{se:qc}. However, these results only hold for the six-state encoding, and the entanglement cost of quantum channels even allows for a more direct way of discussing security, which holds in particular also for the BB84 encoding.

\begin{lemma}
Let $m$ be the number of qubits transmitted in the protocol, and let the adversary's storage be of the form $\cF = \cE^{\ot \nu\cdot m}$. Then, for sufficiently large $m$, any two-party cryptographic primitive can be implemented securely in the noisy-storage model if
\begin{align}
E_C(\cE)\cdot\nu<\frac{1}{2}\ .
\end{align}
\end{lemma}

\begin{proof}
Consider the case of bounded, noise-free, memory. Note that~\cite{Damgard05} tells us that security can be achieved for large enough $m$ if the dimension $d$ of the adversary's storage device is strictly smaller than $d < 2^{m/2}$. Now, suppose by contradiction that security could not be achieved with a storage of the form $\cF = \cE^{\ot n}$, where $n=\nu\cdot m$ and $E_C(\cE) \cdot n \leq \log d$. However, then there exists a successful cheating strategy also in the case of bounded storage of dimension $d$: the adversary could simply simulate $\cE^{\ot n}$ using an entangled state of dimension $d$ with $\log d = E_C(\cE)\cdot n$, possibly using additional classical forward communication provided by his unlimited classical storage device. Hence, for large enough $m$, security can be achieved if $E_C(\cE)\cdot\nu<\frac{1}{2}$ as claimed.
\end{proof}

Our analysis improves the range of parameters for which security can be obtained. We illustrate our results with explicit calculations for a number of specific channels. This can be done explicitly using~\eqref{eq:Cexpr} for any qubit channel of interest. To obtain a bound for when security can be achieved we thus can calculate when the condition
\begin{align}
	\nu \cdot h\left(\frac{1}{2}+\frac{1}{2}\cdot\sqrt{1-\con^{2}\left(\left(\cE_{A\rightarrow B}\ot\cI_{A'}\right)\left(\Phi_{AA'}\right)\right)}\right) < \frac{1}{2}
\end{align}
is fulfilled. Figures~\ref{fig:depolImprove} and~\ref{fig:dephaseImprove} illustrate the improvements obtained for depolarizing noise, $\cE_{\mathrm{depol}}(\rho)=r/2\cdot\1_{2}+(1-r)\rho$, and dephasing noise (see~\eqref{eq:dephasing} for the definition). Note that since previous bounds involved the strong converse classical capacity, dephasing noise was no better than mere bounded storage. Using our new bound, however, we obtain non-trivial bounds even for this case. Figure~\ref{fig:damping} provides security bounds for the one qubit amplitude damping channel $\cE_{\rm damp}(\rho) = E_0 \rho E_{0}^{\dagger} + E_1 \rho E_{1}^{\dagger}$ where $E_0 = \begin{pmatrix}  1&0\\0&\sqrt{r}\end{pmatrix}$ and $E_1 = \begin{pmatrix}0 & \sqrt{1-r}\\0&0\end{pmatrix}$. No previous security bound was known for this channel.

\begin{figure}[ht]
\begin{center}
\includegraphics[width=0.4\linewidth]{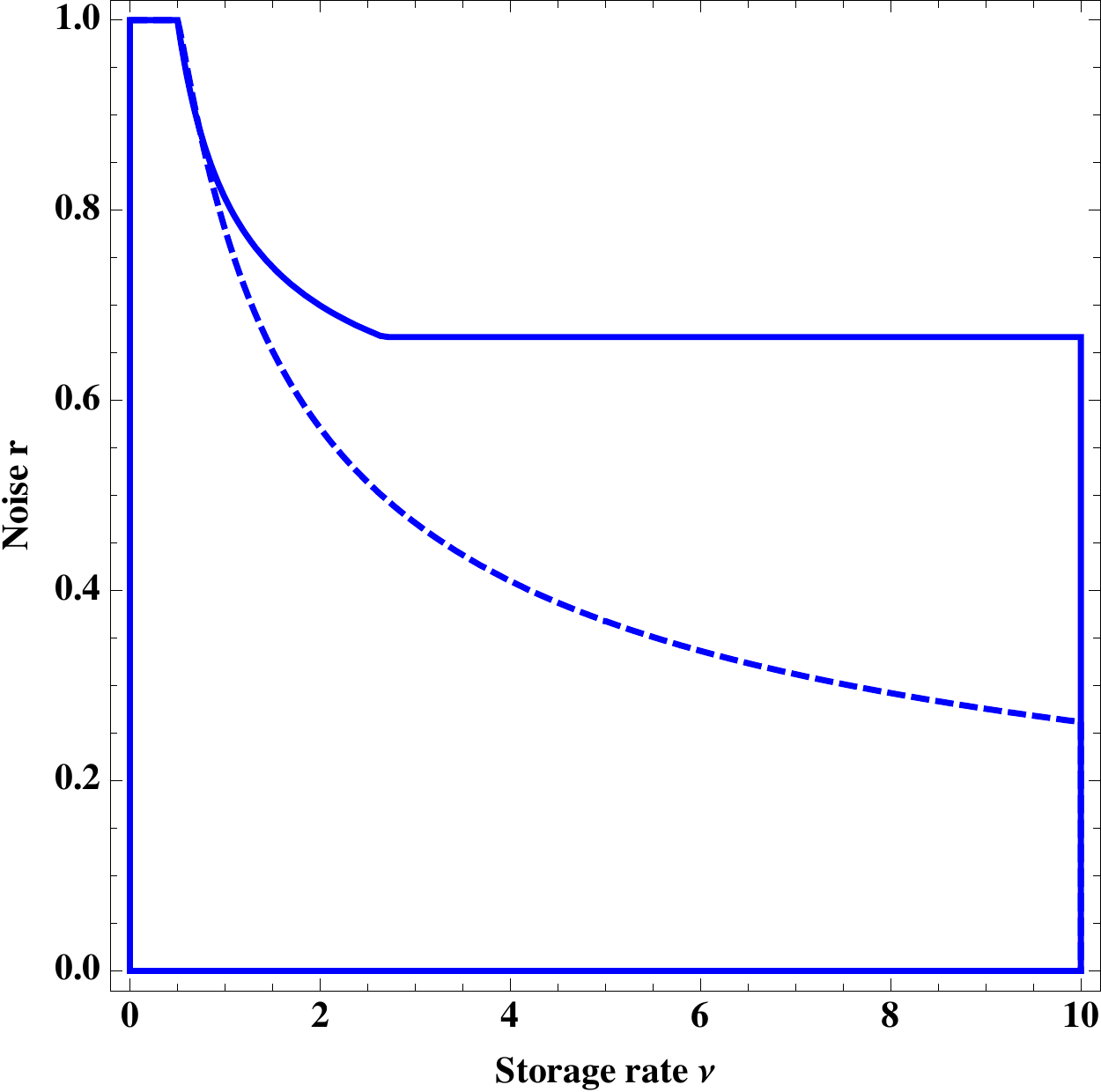}\label{fig:depolImprove}
\end{center}
\caption{Qubit depolarizing channel.}{Security was previously known below the dashed line~\cite{Koenig12}. Now for $(r,\nu)$ inside the solid line.}
\end{figure}

\begin{figure}[ht]
\begin{center}
\includegraphics[width=0.4\linewidth]{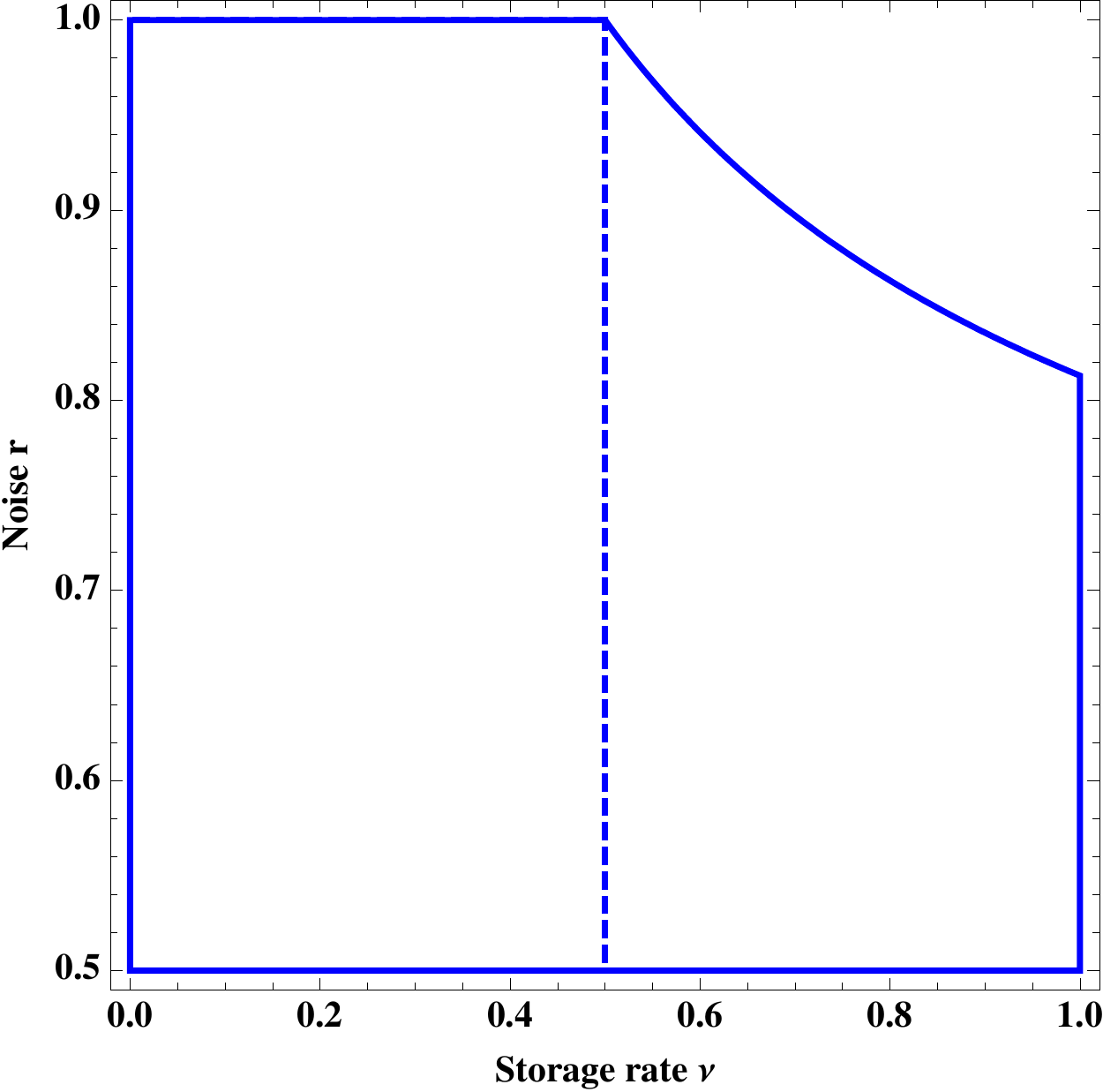}\label{fig:dephaseImprove}
\end{center}
\caption{Qubit dephasing channel.}{Before security was no better than for bounded storage, left of dashed line. Now for $(r,\nu)$ inside the solid line.}
\end{figure}

\begin{figure}[ht]
\begin{center}
\includegraphics[width=0.4\linewidth]{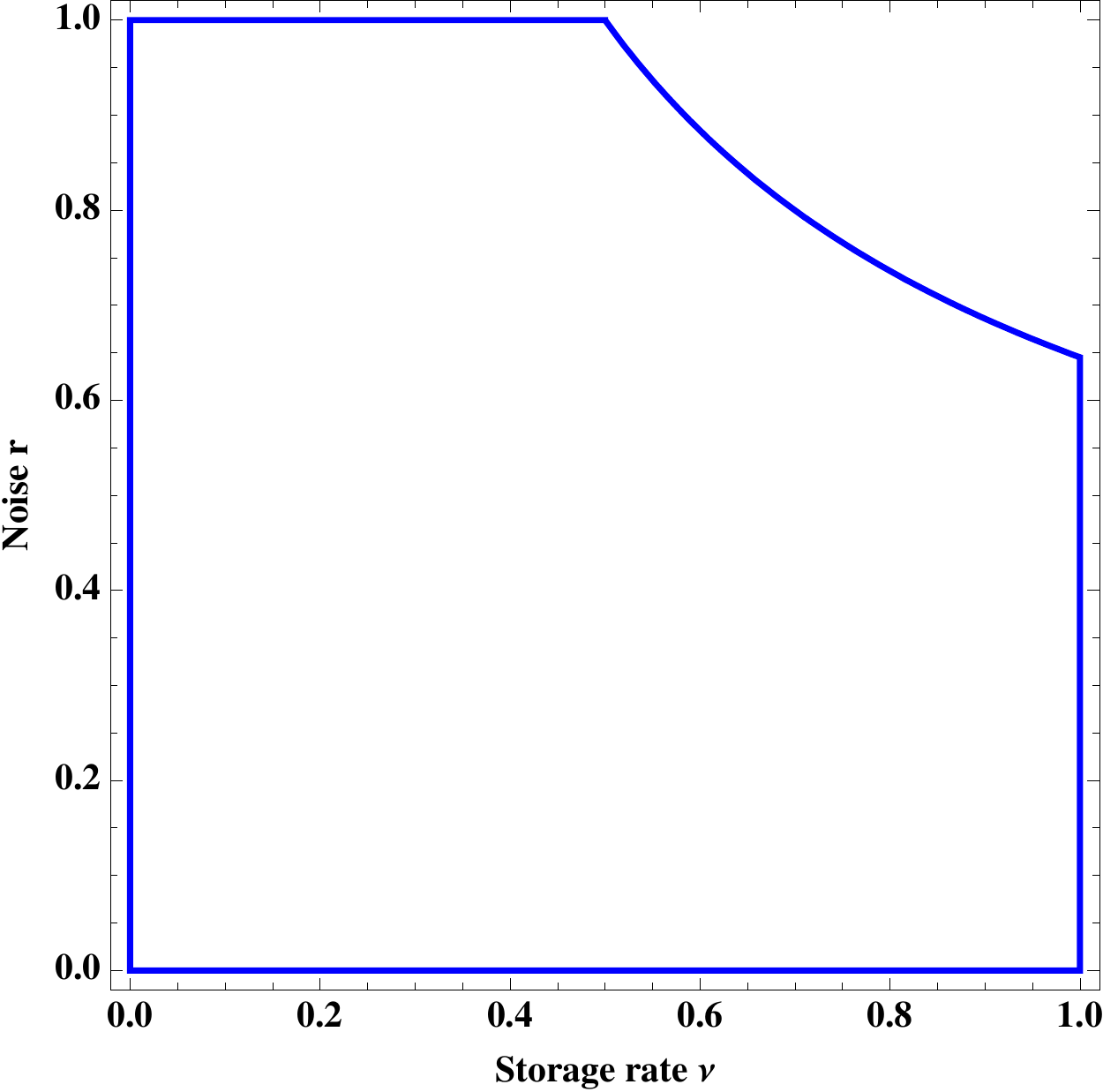}\label{fig:damping}
\end{center}
\caption{Qubit amplitude damping channel.}{No security statement was known previously. Now for $(r,\nu)$ inside the solid line.}
\end{figure}




\cleardoublepage
\phantomsection
\addcontentsline{toc}{chapter}{Bibliography}

\bibliographystyle{alphaarxiv}
\bibliography{library}

\newcommand{\etalchar}[1]{$^{#1}$}
\begin{thebibliography}{KdMT{\etalchar{+}}08}

\bibitem[ADHW09]{Abeyesinghe09}
A.~Abeyesinghe, I.~Devetak, P.~Hayden, and A.~Winter.
\newblock The mother of all protocols: restructuring quantum information's
  family tree.
\newblock {\em Proceedings of the Royal Society A: Mathematical, Physical and
  Engineering Science}, 465(2108): 2537--2563, 2009.
\newblock \\
  \texttt{DOI:\,\href{http://dx.doi.org/10.1098/rspa.2009.0202}{10.1098/rspa.2009.0202}}.

\bibitem[AF04]{Alicki04}
R.~Alicki and M.~Fannes.
\newblock Continuity of quantum conditional information.
\newblock {\em Journal of Physics A: Mathematical and General}, 37(5): L55,
  2004.
\newblock \\
  \texttt{DOI:\,\href{http://dx.doi.org/10.1088/0305-4470/37/5/L01}{10.1088/0305-4470/37/5/L01}}.

\bibitem[Alb83]{Alberti83}
P.~Alberti.
\newblock A note on the transition probability over {C}*-algebras.
\newblock {\em Letters in Mathematical Physics}, 7: 25--32, 1983.
\newblock \\
  \texttt{DOI:\,\href{http://dx.doi.org/10.1007/BF00398708}{10.1007/BF00398708}}.

\bibitem[Ana11]{Ananthaswamy11}
A.~Ananthaswamy.
\newblock Uncertainty entangled: The limits of quantum weirdness.
\newblock {\em New Scientist}, pages 28--31, 2011.

\bibitem[Ara76]{Araki76}
H.~Araki.
\newblock Relative entropy of state of von {N}eumann algebras.
\newblock {\em Publications Research Institute for Mathematical Science Kyoto
  University}, 9: 809--833, 1976.

\bibitem[AS04]{Ambainis04}
A.~Ambainis and A.~Smith.
\newblock Small pseudo-random families of matrices: Derandomizing approximate
  quantum encryption.
\newblock In K.~Jansen, S.~Khanna, J.~Rolim, and D.~Ron, editors, {\em
  Approximation, Randomization, and Combinatorial Optimization. Algorithms and
  Techniques}, volume 3122 of {\em Lecture Notes in Computer Science}, pages
  249--260. Springer Berlin Heidelberg, 2004.
\newblock \\
  \texttt{DOI:\,\href{http://dx.doi.org/10.1007/978-3-540-27821-4$\_$23}{10.1007/978-3-540-27821-4$\_$23}}.

\bibitem[Aub09]{Aubrun09}
G.~Aubrun.
\newblock On almost randomizing channels with a short {K}raus decomposition.
\newblock {\em Communications in Mathematical Physics}, 288: 1103--1116, 2009.
\newblock \\
  \texttt{DOI:\,\href{http://dx.doi.org/10.1007/s00220-008-0695-y}{10.1007/s00220-008-0695-y}}.

\bibitem[Aud07]{Audenaert07}
K.~M.~R. Audenaert.
\newblock A sharp continuity estimate for the von {N}eumann entropy.
\newblock {\em Journal of Physics A: Mathematical and Theoretical}, 40(28):
  8127, 2007.
\newblock \\
  \texttt{DOI:\,\href{http://dx.doi.org/10.1088/1751-8113/40/28/S18}{10.1088/1751-8113/40/28/S18}}.

\bibitem[AW68]{Araki68}
H.~Araki and E.~J. Woods.
\newblock A classification of factors.
\newblock {\em Publications Research Institute for Mathematical Science Series
  A}, 4(1): 51--130, 1968.
\newblock \\
  \texttt{DOI:\,\href{http://dx.doi.org/10.2977/prims/1195195263}{10.2977/prims/1195195263}}.

\bibitem[AW02]{Ahlswede02}
R.~Ahlswede and A.~Winter.
\newblock Strong converse for identification via quantum channels.
\newblock {\em Information Theory, IEEE Transactions on}, 48(3): 569--579,
  2002.
\newblock \\
  \texttt{DOI:\,\href{http://dx.doi.org/10.1109/18.985947}{10.1109/18.985947}}.

\bibitem[BASTS10]{Ben-Aroya10}
A.~Ben-Aroya, O.~Schwartz, and A.~Ta-Shma.
\newblock Quantum expanders: Motivation and construction.
\newblock {\em Theory of Computing}, 6(3): 47--79, 2010.
\newblock \\
  \texttt{DOI:\,\href{http://dx.doi.org/10.4086/toc.2010.v006a003}{10.4086/toc.2010.v006a003}}.

\bibitem[BB84]{Birula84}
I.~Bialynicki-Birula.
\newblock Entropic uncertainty relations.
\newblock {\em Physics Letters A}, 103(5): 253 -- 254, 1984.
\newblock \\
  \texttt{DOI:\,\href{http://dx.doi.org/10.1016/0375-9601(84)90118-X}{10.1016/0375-9601(84)90118-X}}.

\bibitem[BB06]{Birula06}
I.~Bialynicki-Birula.
\newblock Formulation of the uncertainty relations in terms of the {R}\'enyi
  entropies.
\newblock {\em Phys. Rev. A}, 74: 052101, 2006.
\newblock \\
  \texttt{DOI:\,\href{http://dx.doi.org/10.1103/PhysRevA.74.052101}{10.1103/PhysRevA.74.052101}}.

\bibitem[BBaCW13]{Berta11_2}
M.~Berta, F.~Brand\~ao, M.~Christandl, and S.~Wehner.
\newblock Entanglement cost of quantum channels.
\newblock {\em Information Theory, IEEE Transactions on}, 59(10): 6779--6795,
  2013.
\newblock \\
  \texttt{DOI:\,\href{http://dx.doi.org/10.1109/TIT.2013.2268533}{10.1109/TIT.2013.2268533}}.

\bibitem[BBC{\etalchar{+}}93]{Bennett93}
C.~H. Bennett, G.~Brassard, C.~Cr\'epeau, R.~Jozsa, A.~Peres, and W.~K.
  Wootters.
\newblock Teleporting an unknown quantum state via dual classical and
  {E}instein-{P}odolsky-{R}osen channels.
\newblock {\em Phys. Rev. Lett.}, 70: 1895--1899, 1993.
\newblock \\
  \texttt{DOI:\,\href{http://dx.doi.org/10.1103/PhysRevLett.70.1895}{10.1103/PhysRevLett.70.1895}}.

\bibitem[BBM75]{Birula75}
I.~Bia{\l}ynicki-Birula and J.~Mycielski.
\newblock Uncertainty relations for information entropy in wave mechanics.
\newblock {\em Communications in Mathematical Physics}, 44: 129--132, 1975.
\newblock \\
  \texttt{DOI:\,\href{http://dx.doi.org/10.1007/BF01608825}{10.1007/BF01608825}}.

\bibitem[BBR11]{Birula10}
I.~Bialynicki-Birula and {\L}.~Rudnicki.
\newblock Entropic uncertainty relations in quantum physics.
\newblock In K.~Sen, editor, {\em Statistical Complexity}, pages 1--34.
  Springer Netherlands, 2011.
\newblock \\
  \texttt{DOI:\,\href{http://dx.doi.org/10.1007/978-90-481-3890-6$\_$1}{10.1007/978-90-481-3890-6$\_$1}}.

\bibitem[BBRV02]{Bandyopadhyay02}
Bandyopadhyay, Boykin, Roychowdhury, and Vatan.
\newblock A new proof for the existence of mutually unbiased bases.
\newblock {\em Algorithmica}, 34: 512--528, 2002.
\newblock \\
  \texttt{DOI:\,\href{http://dx.doi.org/10.1007/s00453-002-0980-7}{10.1007/s00453-002-0980-7}}.

\bibitem[BCBaW12]{Berta12}
M.~Berta, M.~Christandl, F.~Brand\~ao, and S.~Wehner.
\newblock Entanglement cost of quantum channels.
\newblock In {\em Information Theory Proceedings (ISIT), 2012 IEEE
  International Symposium on}, pages 900 --904, 2012.
\newblock \\
  \texttt{DOI:\,\href{http://dx.doi.org/10.1109/ISIT.2012.6284692}{10.1109/ISIT.2012.6284692}}.

\bibitem[BCC{\etalchar{+}}10]{Berta10}
M.~Berta, M.~Christandl, R.~Colbeck, J.~M. Renes, and R.~Renner.
\newblock The uncertainty principle in the presence of quantum memory.
\newblock {\em Nature Physics}, 6(9): 659--662, 2010.
\newblock \\
  \texttt{DOI:\,\href{http://dx.doi.org/10.1038/nphys1734}{10.1038/nphys1734}}.

\bibitem[BCF{\etalchar{+}}13]{Berta13_2}
M.~Berta, M.~Christandl, F.~Furrer, V.~B. Scholz, and M.~Tomamichel.
\newblock Continuous variable entropic uncertainty relations in the presence of
  quantum memory.
\newblock 2013.
\newblock \\ Online: \url{http://arxiv.org/abs/1308.4527}.

\bibitem[BCH{\etalchar{+}}08]{Buhrman08}
H.~Buhrman, M.~Christandl, P.~Hayden, H.-K. Lo, and S.~Wehner.
\newblock Possibility, impossibility, and cheat sensitivity of quantum-bit
  string commitment.
\newblock {\em Phys. Rev. A}, 78: 022316, 2008.
\newblock \\
  \texttt{DOI:\,\href{http://dx.doi.org/10.1103/PhysRevA.78.022316}{10.1103/PhysRevA.78.022316}}.

\bibitem[BCR11a]{Berta11_3}
M.~Berta, M.~Christandl, and R.~Renner.
\newblock A conceptually simple proof of the quantum reverse {S}hannon theorem.
\newblock In W.~Dam, V.~Kendon, and S.~Severini, editors, {\em Theory of
  Quantum Computation, Communication, and Cryptography}, volume 6519 of {\em
  Lecture Notes in Computer Science}, pages 131--140. Springer Berlin
  Heidelberg, 2011.
\newblock \\
  \texttt{DOI:\,\href{http://dx.doi.org/10.1007/978-3-642-18073-6$\_$11}{10.1007/978-3-642-18073-6$\_$11}}.

\bibitem[BCR11b]{Berta11}
M.~Berta, M.~Christandl, and R.~Renner.
\newblock The quantum reverse {S}hannon theorem based on one-shot information
  theory.
\newblock {\em Communications in Mathematical Physics}, 306: 579--615, 2011.
\newblock \\
  \texttt{DOI:\,\href{http://dx.doi.org/10.1007/s00220-011-1309-7}{10.1007/s00220-011-1309-7}}.

\bibitem[BCW13]{Berta13_3}
M.~Berta, P.~Coles, and S.~Wehner.
\newblock An equality between entanglement and uncertainty.
\newblock 2013.
\newblock \\ Online: \url{http://arxiv.org/abs/1302.5902}.

\bibitem[BD11]{Buscemi11}
F.~Buscemi and N.~Datta.
\newblock Entanglement cost in practical scenarios.
\newblock {\em Phys. Rev. Lett.}, 106: 130503, 2011.
\newblock \\
  \texttt{DOI:\,\href{http://dx.doi.org/10.1103/PhysRevLett.106.130503}{10.1103/PhysRevLett.106.130503}}.

\bibitem[BDF87]{Buchholz87}
D.~Buchholz, C.~D'Antoni, and K.~Fredenhagen.
\newblock The universal structure of local algebras.
\newblock {\em Communications in Mathematical Physics}, 111: 123--135, 1987.
\newblock \\
  \texttt{DOI:\,\href{http://dx.doi.org/10.1007/BF01239019}{10.1007/BF01239019}}.

\bibitem[BDH{\etalchar{+}}09]{Bennett09}
C.~H. Bennett, I.~Devetak, A.~W. Harrow, P.~W. Shor, and A.~Winter.
\newblock Quantum reverse {S}hannon theorem.
\newblock 2009.
\newblock \\ Online: \url{http://arxiv.org/abs/0912.5537}.

\bibitem[BDS97]{Bennett97}
C.~H. Bennett, D.~P. DiVincenzo, and J.~A. Smolin.
\newblock Capacities of quantum erasure channels.
\newblock {\em Phys. Rev. Lett.}, 78: 3217--3220, 1997.
\newblock \\
  \texttt{DOI:\,\href{http://dx.doi.org/10.1103/PhysRevLett.78.3217}{10.1103/PhysRevLett.78.3217}}.

\bibitem[BDSW96]{Bennett96}
C.~H. Bennett, D.~P. DiVincenzo, J.~A. Smolin, and W.~K. Wootters.
\newblock Mixed-state entanglement and quantum error correction.
\newblock {\em Phys. Rev. A}, 54: 3824--3851, 1996.
\newblock \\
  \texttt{DOI:\,\href{http://dx.doi.org/10.1103/PhysRevA.54.3824}{10.1103/PhysRevA.54.3824}}.

\bibitem[Bec75]{Beckner75}
W.~Beckner.
\newblock Inequalities in {F}ourier analysis.
\newblock {\em Annals of Mathematics}, 102(1): 159--182, 1975.
\newblock \\ Online: \url{http://www.jstor.org/stable/1970980}.

\bibitem[Bel87]{Bell87}
J.~S. Bell.
\newblock {\em Speakable and Unspeakable in Quantum Mechanics}.
\newblock Cambridge University Press, 1987.

\bibitem[Ber08]{Berta08}
M.~Berta.
\newblock Single-shot quantum state merging.
\newblock Master's thesis, ETH Zurich, 2008.
\newblock \\ Online: \url{http://arxiv.org/abs/0912.4495}.

\bibitem[BF13]{Brown12}
W.~Brown and O.~Fawzi.
\newblock Decoupling with random quantum circuits.
\newblock 2013.
\newblock \\ Online: \url{http://arxiv.org/abs/1307.0632}.

\bibitem[BFGGS12]{Bouman12}
N.~J. Bouman, S.~Fehr, C.~Gonzales-Guillen, and C.~Schaffner.
\newblock An all-but-one entropic uncertainty relations, and application to
  password-based identification.
\newblock In {\em TQC'12: Proceedings of the 7th conference on Theory of
  quantum computation, communication, and cryptography}. Springer-Verlag, 2012.
\newblock \\
  \texttt{DOI:\,\href{http://dx.doi.org/10.1007/978-3-642-35656-8$\_$3}{10.1007/978-3-642-35656-8$\_$3}}.

\bibitem[BFS11]{Berta11_4}
M.~Berta, F.~Furrer, and V.~B. Scholz.
\newblock The smooth entropy formalism on von {N}eumann algebras.
\newblock 2011.
\newblock \\ Online: \url{http://arxiv.org/abs/1107.5460}.

\bibitem[BFW12]{Berta12_2}
M.~Berta, O.~Fawzi, and S.~Wehner.
\newblock Quantum to classical randomness extractors.
\newblock In R.~Safavi-Naini and R.~Canetti, editors, {\em Advances in
  Cryptology -- CRYPTO 2012}, volume 7417 of {\em Lecture Notes in Computer
  Science}, pages 776--793. Springer Berlin Heidelberg, 2012.
\newblock \\
  \texttt{DOI:\,\href{http://dx.doi.org/10.1007/978-3-642-32009-5$\_$45}{10.1007/978-3-642-32009-5$\_$45}}.

\bibitem[BFW13]{Berta11_5}
M.~Berta, O.~Fawzi, and S.~Wehner.
\newblock Quantum to classical randomness extractors.
\newblock {\em Information Theory, IEEE Transactions on}, to appear, 2013.
\newblock \\ Online: \url{http://arxiv.org/abs/1111.2026}.

\bibitem[BHH08]{Buscemi08}
F.~Buscemi, M.~Hayashi, and M.~Horodecki.
\newblock Global information balance in quantum measurements.
\newblock {\em Phys. Rev. Lett.}, 100: 210504, 2008.
\newblock \\
  \texttt{DOI:\,\href{http://dx.doi.org/10.1103/PhysRevLett.100.210504}{10.1103/PhysRevLett.100.210504}}.

\bibitem[BK02]{Barnum02}
H.~Barnum and E.~Knill.
\newblock Reversing quantum dynamics with near-optimal quantum and classical
  fidelity.
\newblock {\em Journal of Mathematical Physics}, 43(5): 2097, 2002.
\newblock \\
  \texttt{DOI:\,\href{http://dx.doi.org/10.1063/1.1459754}{10.1063/1.1459754}}.

\bibitem[BNS98]{Barnum98}
H.~Barnum, M.~A. Nielsen, and B.~Schumacher.
\newblock Information transmission through a noisy quantum channel.
\newblock {\em Phys. Rev. A}, 57: 4153--4175, 1998.
\newblock \\
  \texttt{DOI:\,\href{http://dx.doi.org/10.1103/PhysRevA.57.4153}{10.1103/PhysRevA.57.4153}}.

\bibitem[BR79]{Bratteli79}
O.~Bratteli and D.~W. Robinson.
\newblock {\em Operator Algebras and Quantum Statistical Mechanics 1}.
\newblock Springer, 1979.

\bibitem[BR81]{Bratteli81}
O.~Bratteli and D.~W. Robinson.
\newblock {\em Operator Algebras and Quantum Statistical Mechanics 2}.
\newblock Springer, 1981.

\bibitem[BRW13]{Berta13}
M.~Berta, J.~M. Renes, and M.~M. Wilde.
\newblock Identifying the information gain of a quantum measurement.
\newblock 2013.
\newblock \\ Online: \url{http://arxiv.org/abs/1301.1594}.

\bibitem[BSST99]{Bennett99}
C.~H. Bennett, P.~W. Shor, J.~A. Smolin, and A.~V. Thapliyal.
\newblock Entanglement-assisted classical capacity of noisy quantum channels.
\newblock {\em Phys. Rev. Lett.}, 83: 3081--3084, 1999.
\newblock \\
  \texttt{DOI:\,\href{http://dx.doi.org/10.1103/PhysRevLett.83.3081}{10.1103/PhysRevLett.83.3081}}.

\bibitem[BSST02]{Bennett02}
C.~Bennett, P.~Shor, J.~Smolin, and A.~Thapliyal.
\newblock Entanglement-assisted capacity of a quantum channel and the reverse
  {S}hannon theorem.
\newblock {\em Information Theory, IEEE Transactions on}, 48(10): 2637 -- 2655,
  2002.
\newblock \\
  \texttt{DOI:\,\href{http://dx.doi.org/10.1109/TIT.2002.802612}{10.1109/TIT.2002.802612}}.

\bibitem[Bur69]{Bures69}
D.~Bures.
\newblock An extension of {K}akutani's theorem on infinite product measures to
  the tensor product of semifinite {W}*-algebras.
\newblock {\em Transactions of the American Mathematical Society}, 135:
  199--212, 1969.

\bibitem[BW92]{Bennett92}
C.~H. Bennett and S.~J. Wiesner.
\newblock Communication via one- and two-particle operators on
  {E}instein-{P}odolsky-{R}osen states.
\newblock {\em Phys. Rev. Lett.}, 69: 2881--2884, 1992.
\newblock \\
  \texttt{DOI:\,\href{http://dx.doi.org/10.1103/PhysRevLett.69.2881}{10.1103/PhysRevLett.69.2881}}.

\bibitem[BW07]{Ballester07}
M.~A. Ballester and S.~Wehner.
\newblock Entropic uncertainty relations and locking: Tight bounds for mutually
  unbiased bases.
\newblock {\em Phys. Rev. A}, 75: 022319, 2007.
\newblock \\
  \texttt{DOI:\,\href{http://dx.doi.org/10.1103/PhysRevA.75.022319}{10.1103/PhysRevA.75.022319}}.

\bibitem[BZ99]{Brukner99}
C.~Brukner and A.~Zeilinger.
\newblock Operationally invariant information in quantum measurements.
\newblock {\em Phys. Rev. Lett.}, 83: 3354--3357, 1999.
\newblock \\
  \texttt{DOI:\,\href{http://dx.doi.org/10.1103/PhysRevLett.83.3354}{10.1103/PhysRevLett.83.3354}}.

\bibitem[CBKG02]{Cerf02}
N.~J. Cerf, M.~Bourennane, A.~Karlsson, and N.~Gisin.
\newblock Security of quantum key distribution using $\mathit{d}$-level
  systems.
\newblock {\em Phys. Rev. Lett.}, 88: 127902, 2002.
\newblock \\
  \texttt{DOI:\,\href{http://dx.doi.org/10.1103/PhysRevLett.88.127902}{10.1103/PhysRevLett.88.127902}}.

\bibitem[CCYZ12]{Coles12}
P.~J. Coles, R.~Colbeck, L.~Yu, and M.~Zwolak.
\newblock Uncertainty relations from simple entropic properties.
\newblock {\em Phys. Rev. Lett.}, 108: 210405, 2012.
\newblock \\
  \texttt{DOI:\,\href{http://dx.doi.org/10.1103/PhysRevLett.108.210405}{10.1103/PhysRevLett.108.210405}}.

\bibitem[Cha05]{Chau05}
H.~F. Chau.
\newblock Unconditionally secure key distribution in higher dimensions by
  depolarization.
\newblock {\em Information Theory, IEEE Transactions on}, 51(4): 1451--1468,
  2005.
\newblock \\
  \texttt{DOI:\,\href{http://dx.doi.org/10.1109/TIT.2005.844076}{10.1109/TIT.2005.844076}}.

\bibitem[Cig12]{Ciganovic12}
N.~Ciganovic.
\newblock Smooth max-mutual information as a generalization of von {N}eumann
  mutual information for the one-shot setting.
\newblock Master's thesis, ETH Zurich, 2012.

\bibitem[CKR09]{Christandl09}
M.~Christandl, R.~K\"onig, and R.~Renner.
\newblock Postselection technique for quantum channels with applications to
  quantum cryptography.
\newblock {\em Phys. Rev. Lett.}, 102: 020504, 2009.
\newblock \\
  \texttt{DOI:\,\href{http://dx.doi.org/10.1103/PhysRevLett.102.020504}{10.1103/PhysRevLett.102.020504}}.

\bibitem[CT91]{Cover91}
T.~M. Cover and J.~A. Thomas.
\newblock {\em Elements of Information Theory}.
\newblock John Wiley \& Sons, 1991.

\bibitem[Cuf08]{Cuff08}
P.~Cuff.
\newblock Communication requirements for generating correlated random
  variables.
\newblock In {\em Information Theory, 2008. ISIT 2008. IEEE International
  Symposium on}, pages 1393 --1397, 2008.
\newblock \\
  \texttt{DOI:\,\href{http://dx.doi.org/10.1109/ISIT.2008.4595216}{10.1109/ISIT.2008.4595216}}.

\bibitem[CYGG11]{Coles11}
P.~J. Coles, L.~Yu, V.~Gheorghiu, and R.~B. Griffiths.
\newblock Information-theoretic treatment of tripartite systems and quantum
  channels.
\newblock {\em Phys. Rev. A}, 83: 062338, 2011.
\newblock \\
  \texttt{DOI:\,\href{http://dx.doi.org/10.1103/PhysRevA.83.062338}{10.1103/PhysRevA.83.062338}}.

\bibitem[Dat09]{Datta09}
N.~Datta.
\newblock Min- and max-relative entropies and a new entanglement monotone.
\newblock {\em Information Theory, IEEE Transactions on}, 55(6): 2816 --2826,
  2009.
\newblock \\
  \texttt{DOI:\,\href{http://dx.doi.org/10.1109/TIT.2009.2018325}{10.1109/TIT.2009.2018325}}.

\bibitem[DBWR10]{Dupuis10}
F.~Dupuis, M.~Berta, J.~Wullschleger, and R.~Renner.
\newblock One-shot decoupling.
\newblock 2010.
\newblock \\ Online: \url{http://arxiv.org/abs/1012.6044}.

\bibitem[DCEL09]{Dankert09}
C.~Dankert, R.~Cleve, J.~Emerson, and E.~Livine.
\newblock Exact and approximate unitary 2-designs and their application to
  fidelity estimation.
\newblock {\em Phys. Rev. A}, 80: 012304, 2009.
\newblock \\
  \texttt{DOI:\,\href{http://dx.doi.org/10.1103/PhysRevA.80.012304}{10.1103/PhysRevA.80.012304}}.

\bibitem[Deu89]{Deutsch89}
D.~Deutsch.
\newblock Quantum computational networks.
\newblock {\em Proceedings of the Royal Society of London. A. Mathematical and
  Physical Sciences}, 425(1868): 73--90, 1989.
\newblock \\
  \texttt{DOI:\,\href{http://dx.doi.org/10.1098/rspa.1989.0099}{10.1098/rspa.1989.0099}}.

\bibitem[Dev05]{Devetak05}
I.~Devetak.
\newblock The private classical capacity and quantum capacity of a quantum
  channel.
\newblock {\em Information Theory, IEEE Transactions on}, 51(1): 44 --55, 2005.
\newblock \\
  \texttt{DOI:\,\href{http://dx.doi.org/10.1109/TIT.2004.839515}{10.1109/TIT.2004.839515}}.

\bibitem[DFR{\etalchar{+}}07]{Damgard07}
I.~Damg{\aa}rd, S.~Fehr, R.~Renner, L.~Salvail, and C.~Schaffner.
\newblock A tight high-order entropic quantum uncertainty relation with
  applications.
\newblock In A.~Menezes, editor, {\em Advances in Cryptology - CRYPTO 2007},
  volume 4622 of {\em Lecture Notes in Computer Science}, pages 360--378.
  Springer Berlin Heidelberg, 2007.
\newblock \\
  \texttt{DOI:\,\href{http://dx.doi.org/10.1007/978-3-540-74143-5$\_$20}{10.1007/978-3-540-74143-5$\_$20}}.

\bibitem[DFSS05]{Damgard05}
I.~Damgard, S.~Fehr, L.~Salvail, and C.~Schaffner.
\newblock Cryptography in the bounded quantum-storage model.
\newblock In {\em Foundations of Computer Science, 2005. FOCS 2005. 46th Annual
  IEEE Symposium on}, pages 449 -- 458, 2005.
\newblock \\
  \texttt{DOI:\,\href{http://dx.doi.org/10.1109/SFCS.2005.30}{10.1109/SFCS.2005.30}}.

\bibitem[DFW13]{Dupuis13}
F.~Dupuis, O.~Fawzi, and S.~Wehner.
\newblock Achieving the limits of the noisy-storage model using entanglement
  sampling.
\newblock In R.~Canetti and J.~A. Garay, editors, {\em Advances in Cryptology
  -- CRYPTO 2013}, volume 8043 of {\em Lecture Notes in Computer Science},
  pages 326--343. Springer Berlin Heidelberg, 2013.
\newblock \\
  \texttt{DOI:\,\href{http://dx.doi.org/10.1007/978-3-642-40084-1$\_$19}{10.1007/978-3-642-40084-1$\_$19}}.

\bibitem[DHW13]{Datta13}
N.~Datta, M.-H. Hsieh, and M.~Wilde.
\newblock Quantum rate distortion, reverse {S}hannon theorems, and
  source-channel separation.
\newblock {\em Information Theory, IEEE Transactions on}, 59(1): 615 --630,
  2013.
\newblock \\
  \texttt{DOI:\,\href{http://dx.doi.org/10.1109/TIT.2012.2215575}{10.1109/TIT.2012.2215575}}.

\bibitem[DHWW13]{Datta12}
N.~Datta, M.-H. Hsieh, M.~M. Wilde, and A.~Winter.
\newblock Quantum-to-classical rate distortion coding.
\newblock {\em Journal of Mathematical Physics}, 54(4): 042201, 2013.
\newblock \\
  \texttt{DOI:\,\href{http://dx.doi.org/10.1063/1.4798396}{10.1063/1.4798396}}.

\bibitem[Die82]{Dieks82}
D.~Dieks.
\newblock Communication by {EPR} devices.
\newblock {\em Physics Letters A}, 92(6): 271 -- 272, 1982.
\newblock \\
  \texttt{DOI:\,\href{http://dx.doi.org/10.1016/0375-9601(82)90084-6}{10.1016/0375-9601(82)90084-6}}.

\bibitem[Dir82]{Dirac82}
P.~Dirac.
\newblock {\em The Principles of Quantum Mechanics (International Series of
  Monographs on Physics)}.
\newblock {Oxford University Press}, 1982.

\bibitem[DJJ01]{Doherty01}
A.~C. Doherty, K.~Jacobs, and G.~Jungman.
\newblock Information, disturbance, and hamiltonian quantum feedback control.
\newblock {\em Phys. Rev. A}, 63: 062306, 2001.
\newblock \\
  \texttt{DOI:\,\href{http://dx.doi.org/10.1103/PhysRevA.63.062306}{10.1103/PhysRevA.63.062306}}.

\bibitem[DL70]{Davies70}
E.~Davies and J.~Lewis.
\newblock An operational approach to quantum probability.
\newblock {\em Communications in Mathematical Physics}, 17: 239--260, 1970.
\newblock \\
  \texttt{DOI:\,\href{http://dx.doi.org/10.1007/BF01647093}{10.1007/BF01647093}}.

\bibitem[DM11]{Dorlas11}
T.~Dorlas and C.~Morgan.
\newblock Invalidity of a strong capacity for a quantum channel with memory.
\newblock {\em Phys. Rev. A}, 84: 042318, 2011.
\newblock \\
  \texttt{DOI:\,\href{http://dx.doi.org/10.1103/PhysRevA.84.042318}{10.1103/PhysRevA.84.042318}}.

\bibitem[DMHB11]{Datta11}
N.~Datta, M.~Mosonyi, M.-H. Hsieh, and F.~Brandao.
\newblock Strong converses for classical information transmission and
  hypothesis testing.
\newblock 2011.
\newblock \\ Online: \url{http://arxiv.org/abs/1106.3089}.

\bibitem[DN06]{Dickinson06}
P.~A. Dickinson and A.~Nayak.
\newblock Approximate randomization of quantum states with fewer bits of key.
\newblock In {\em AIP Conference Proceedings}, volume 864, pages 18--36, 2006.
\newblock \\
  \texttt{DOI:\,\href{http://dx.doi.org/10.1063/1.2400876}{10.1063/1.2400876}}.

\bibitem[DPVR12]{De12}
A.~De, C.~Portmann, T.~Vidick, and R.~Renner.
\newblock Trevisan's extractor in the presence of quantum side information.
\newblock {\em SIAM Journal on Computing}, 41(4): 915--940, 2012.
\newblock \\
  \texttt{DOI:\,\href{http://dx.doi.org/10.1137/100813683}{10.1137/100813683}}.

\bibitem[DS05]{Devetak05_2}
I.~Devetak and P.~W. Shor.
\newblock The capacity of a quantum channel for simultaneous transmission of
  classical and quantum information.
\newblock {\em Communications in Mathematical Physics}, 256: 287--303, 2005.
\newblock \\
  \texttt{DOI:\,\href{http://dx.doi.org/10.1007/s00220-005-1317-6}{10.1007/s00220-005-1317-6}}.

\bibitem[DST12]{Dupuis12}
F.~Dupuis, O.~Szehr, and M.~Tomamichel.
\newblock A decoupling approach to classical data transmission over quantum
  channels.
\newblock 2012.
\newblock \\ Online: \url{http://arxiv.org/abs/1207.0067}.

\bibitem[DSW11]{Duan11}
R.~Duan, S.~Severini, and A.~Winter.
\newblock Zero-error communication via quantum channels and a quantum {L}ovasz
  theta function.
\newblock In {\em Information Theory Proceedings (ISIT), 2011 IEEE
  International Symposium on}, pages 64--68, 2011.
\newblock \\
  \texttt{DOI:\,\href{http://dx.doi.org/10.1109/ISIT.2011.6034211}{10.1109/ISIT.2011.6034211}}.

\bibitem[DSW13]{Duan10}
R.~Duan, S.~Severini, and A.~Winter.
\newblock Zero-error communication via quantum channels, noncommutative graphs,
  and a quantum {L}ovasz number.
\newblock {\em Information Theory, IEEE Transactions on}, 59(2): 1164--1174,
  2013.
\newblock \\
  \texttt{DOI:\,\href{http://dx.doi.org/10.1109/TIT.2012.2221677}{10.1109/TIT.2012.2221677}}.

\bibitem[Dup09]{Dupuis09}
F.~Dupuis.
\newblock {\em The Decoupling Approach to Quantum Information Theory}.
\newblock PhD thesis, Universit\'e de Montr\'eal, 2009.
\newblock \\ Online: \url{http://arxiv.org/abs/1004.1641}.

\bibitem[DW05]{Devetak05_3}
I.~Devetak and A.~Winter.
\newblock Distillation of secret key and entanglement from quantum states.
\newblock {\em Proceedings of the Royal Society A: Mathematical, Physical and
  Engineering Science}, 461(2053): 207--235, 2005.
\newblock \\
  \texttt{DOI:\,\href{http://dx.doi.org/10.1098/rspa.2004.1372}{10.1098/rspa.2004.1372}}.

\bibitem[Ein49]{Einstein49}
{\em Albert {Einstein}: Philosopher-Scientist}, chapter Discussions with
  {Einstein} on Epistemological Problems in Atomic Physics, page 199.
\newblock Cambridge University Press, 1949.

\bibitem[EPR35]{Einstein35}
A.~Einstein, B.~Podolsky, and N.~Rosen.
\newblock Can quantum-mechanical description of physical reality be considered
  complete?
\newblock {\em Phys. Rev.}, 47: 777--780, 1935.
\newblock \\
  \texttt{DOI:\,\href{http://dx.doi.org/10.1103/PhysRev.47.777}{10.1103/PhysRev.47.777}}.

\bibitem[F{\AA}R11]{Furrer11}
F.~Furrer, J.~{\AA}berg, and R.~Renner.
\newblock Min- and max-entropy in infinite dimensions.
\newblock {\em Communications in Mathematical Physics}, 306: 165--186, 2011.
\newblock \\
  \texttt{DOI:\,\href{http://dx.doi.org/10.1007/s00220-011-1282-1}{10.1007/s00220-011-1282-1}}.

\bibitem[Faw12]{Fawzi12}
O.~Fawzi.
\newblock {\em Uncertainty relations for multiple measurements with
  applications}.
\newblock PhD thesis, McGill University, 2012.
\newblock \\ Online: \url{http://arxiv.org/abs/1208.5918}.

\bibitem[Faw13]{Fawzi13}
O.~Fawzi.
\newblock {\em Private Communication}, 2013.

\bibitem[FFB{\etalchar{+}}12]{Furrer12}
F.~Furrer, T.~Franz, M.~Berta, A.~Leverrier, V.~B. Scholz, M.~Tomamichel, and
  R.~F. Werner.
\newblock Continuous variable quantum key distribution: Finite-key analysis of
  composable security against coherent attacks.
\newblock {\em Phys. Rev. Lett.}, 109: 100502, 2012.
\newblock \\
  \texttt{DOI:\,\href{http://dx.doi.org/10.1103/PhysRevLett.109.100502}{10.1103/PhysRevLett.109.100502}}.

\bibitem[FHS11]{Fawzi11}
O.~Fawzi, P.~Hayden, and P.~Sen.
\newblock From low-distortion norm embeddings to explicit uncertainty relations
  and efficient information locking.
\newblock In {\em Proceedings of the 43rd annual ACM symposium on Theory of
  computing}, STOC '11, pages 773--782, New York, NY, USA, 2011. ACM.
\newblock \\
  \texttt{DOI:\,\href{http://dx.doi.org/10.1145/1993636.1993738}{10.1145/1993636.1993738}}.

\bibitem[FJ01]{Fuchs01}
C.~A. Fuchs and K.~Jacobs.
\newblock Information-tradeoff relations for finite-strength quantum
  measurements.
\newblock {\em Phys. Rev. A}, 63: 062305, 2001.
\newblock \\
  \texttt{DOI:\,\href{http://dx.doi.org/10.1103/PhysRevA.63.062305}{10.1103/PhysRevA.63.062305}}.

\bibitem[FL13]{Frank12}
R.~L. Frank and E.~H. Lieb.
\newblock Extended quantum conditional entropy and quantum uncertainty
  inequalities.
\newblock {\em Communications in Mathematical Physics}, 323: 487--495, 2013.
\newblock \\
  \texttt{DOI:\,\href{http://dx.doi.org/10.1007/s00220-013-1775-1}{10.1007/s00220-013-1775-1}}.

\bibitem[Fur09]{Furrer09}
F.~Furrer.
\newblock Min- and max-entropies as generalized entropy measures in
  infinite-dimensional quantum systems.
\newblock Master's thesis, ETH Zurich, 2009.

\bibitem[Fur12]{Furrer12_2}
F.~Furrer.
\newblock {\em Security of Continuous-Variable Quantum Key Distribution and
  Aspects of Device-Independent Security}.
\newblock PhD thesis, Leibniz University Hanover, 2012.

\bibitem[GAE07]{Gross07}
D.~Gross, K.~Audenaert, and J.~Eisert.
\newblock Evenly distributed unitaries: On the structure of unitary designs.
\newblock {\em Journal of Mathematical Physics}, 48(5): 052104, 2007.
\newblock \\
  \texttt{DOI:\,\href{http://dx.doi.org/10.1063/1.2716992}{10.1063/1.2716992}}.

\bibitem[GKK{\etalchar{+}}07]{Gavinsky07}
D.~Gavinsky, J.~Kempe, I.~Kerenidis, R.~Raz, and R.~de~Wolf.
\newblock Exponential separations for one-way quantum communication complexity,
  with applications to cryptography.
\newblock In {\em Proceedings of the thirty-ninth annual ACM symposium on
  Theory of computing}, STOC '07, pages 516--525, New York, NY, USA, 2007. ACM.
\newblock \\
  \texttt{DOI:\,\href{http://dx.doi.org/10.1145/1250790.1250866}{10.1145/1250790.1250866}}.

\bibitem[Gro71]{Groenewold71}
H.~Groenewold.
\newblock A problem of information gain by quantal measurements.
\newblock {\em International Journal of Theoretical Physics}, 4: 327--338,
  1971.
\newblock \\
  \texttt{DOI:\,\href{http://dx.doi.org/10.1007/BF00815357}{10.1007/BF00815357}}.

\bibitem[GUV09]{Guruswami09}
V.~Guruswami, C.~Umans, and S.~Vadhan.
\newblock Unbalanced expanders and randomness extractors from parvaresh--vardy
  codes.
\newblock {\em J. ACM}, 56(4): 20:1--20:34, 2009.
\newblock \\
  \texttt{DOI:\,\href{http://dx.doi.org/10.1145/1538902.1538904}{10.1145/1538902.1538904}}.

\bibitem[GW93]{Gruber93}
P.~M. Gruber and J.~M. Wills.
\newblock {\em Handbook of Convex Geometry, Vol. A}.
\newblock Elsevier Science Publishers, 1993.

\bibitem[Haa92]{Haag92}
R.~Haag.
\newblock {\em Local Quantum Physics: Fields, Particles, Algebras.}
\newblock Springer, 1992.

\bibitem[Hal95]{Hall95}
M.~J.~W. Hall.
\newblock Information exclusion principle for complementary observables.
\newblock {\em Phys. Rev. Lett.}, 74: 3307--3311, 1995.
\newblock \\
  \texttt{DOI:\,\href{http://dx.doi.org/10.1103/PhysRevLett.74.3307}{10.1103/PhysRevLett.74.3307}}.

\bibitem[Har04]{Harrow04}
A.~Harrow.
\newblock Coherent communication of classical messages.
\newblock {\em Phys. Rev. Lett.}, 92: 097902, 2004.
\newblock \\
  \texttt{DOI:\,\href{http://dx.doi.org/10.1103/PhysRevLett.92.097902}{10.1103/PhysRevLett.92.097902}}.

\bibitem[{Har}10]{Harrow09}
A.~W. {Harrow}.
\newblock Entanglement spread and clean resource inequalities.
\newblock In P.~{Exner}, editor, {\em XVITH International Congress on
  Mathematical Physics}, pages 536--540, 2010.
\newblock \\
  \texttt{DOI:\,\href{http://dx.doi.org/10.1142/9789814304634$\_$0046}{10.1142/9789814304634$\_$0046}}.

\bibitem[Has09]{Hastings09}
M.~B. Hastings.
\newblock A counterexample to additivity of minimum output entropy.
\newblock {\em Nature Physics}, 5(4): 255--257, 2009.
\newblock \\
  \texttt{DOI:\,\href{http://dx.doi.org/10.1038/nphys1224}{10.1038/nphys1224}}.

\bibitem[Hay06a]{Hayashi06}
M.~Hayashi.
\newblock {\em Quantum Information: An Introduction}.
\newblock Springer, 2006.

\bibitem[Hay06b]{Hayashi06_2}
M.~Hayashi.
\newblock Optimal visible compression rate for mixed states is determined by
  entanglement of purification.
\newblock {\em Phys. Rev. A}, 73: 060301, 2006.
\newblock \\
  \texttt{DOI:\,\href{http://dx.doi.org/10.1103/PhysRevA.73.060301}{10.1103/PhysRevA.73.060301}}.

\bibitem[Hay11]{Hayden11}
P.~Hayden.
\newblock Quantum information theory via decoupling.
\newblock {\em Tutorial QIP Singapore}, 2011.
\newblock \\ Online:
  \url{http://qip2011.quantumlah.org/images/QIPtutorial1.pdf}.

\bibitem[HDW08]{Hsieh08}
M.-H. Hsieh, I.~Devetak, and A.~Winter.
\newblock Entanglement-assisted capacity of quantum multiple-access channels.
\newblock {\em Information Theory, IEEE Transactions on}, 54(7): 3078 --3090,
  2008.
\newblock \\
  \texttt{DOI:\,\href{http://dx.doi.org/10.1109/TIT.2008.924726}{10.1109/TIT.2008.924726}}.

\bibitem[Hei27]{Heisenberg27}
W.~Heisenberg.
\newblock Uber den anschaulichen {I}nhalt der quantentheoretischen {K}inematik
  und {M}echanik.
\newblock {\em Zeitschrift f{\"u}r Physik}, 43: 172--198, 1927.
\newblock \\
  \texttt{DOI:\,\href{http://dx.doi.org/10.1007/BF01397280}{10.1007/BF01397280}}.

\bibitem[HHR96]{Horodecki96}
M.~Horodecki, P.~Horodecki, and H.~Ryszard.
\newblock Separability of mixed states: necessary and sufficient conditions.
\newblock {\em Physics Letters A}, 223(1-2): 1 -- 8, 1996.
\newblock \\
  \texttt{DOI:\,\href{http://dx.doi.org/10.1016/S0375-9601(96)00706-2}{10.1016/S0375-9601(96)00706-2}}.

\bibitem[HHT01]{Hayden01}
P.~M. Hayden, M.~Horodecki, and B.~M. Terhal.
\newblock The asymptotic entanglement cost of preparing a quantum state.
\newblock {\em Journal of Physics A: Mathematical and General}, 34(35): 6891,
  2001.
\newblock \\
  \texttt{DOI:\,\href{http://dx.doi.org/10.1088/0305-4470/34/35/314}{10.1088/0305-4470/34/35/314}}.

\bibitem[Hir57]{Hirschman57}
I.~I. Hirschman.
\newblock A note on entropy.
\newblock {\em American Journal of Mathematics}, 79(1): 152--156, 1957.
\newblock \\ Online: \url{http://www.jstor.org/stable/2372390}.

\bibitem[HJ85]{Horn85}
R.~A. Horn and C.~R. Johnson.
\newblock {\em Matrix Analysis}.
\newblock Cambridge University Press, 1985.

\bibitem[HL09a]{Harrow09_2}
A.~Harrow and R.~Low.
\newblock Efficient quantum tensor product expanders and k-designs.
\newblock In I.~Dinur, K.~Jansen, J.~Naor, and J.~Rolim, editors, {\em
  Approximation, Randomization, and Combinatorial Optimization. Algorithms and
  Techniques}, volume 5687 of {\em Lecture Notes in Computer Science}, pages
  548--561. Springer Berlin Heidelberg, 2009.
\newblock \\
  \texttt{DOI:\,\href{http://dx.doi.org/10.1007/978-3-642-03685-9$\_$41}{10.1007/978-3-642-03685-9$\_$41}}.

\bibitem[HL09b]{Harrow09_3}
A.~Harrow and R.~Low.
\newblock Random quantum circuits are approximate 2-designs.
\newblock {\em Communications in Mathematical Physics}, 291: 257--302, 2009.
\newblock \\
  \texttt{DOI:\,\href{http://dx.doi.org/10.1007/s00220-009-0873-6}{10.1007/s00220-009-0873-6}}.

\bibitem[HLSW04]{Hayden04}
P.~Hayden, D.~Leung, P.~W. Shor, and A.~Winter.
\newblock Randomizing quantum states: Constructions and applications.
\newblock {\em Communications in Mathematical Physics}, 250: 371--391, 2004.
\newblock \\
  \texttt{DOI:\,\href{http://dx.doi.org/10.1007/s00220-004-1087-6}{10.1007/s00220-004-1087-6}}.

\bibitem[Hol73]{Holevo73}
A.~S. Holevo.
\newblock Bounds for the quantity of information transmitted by a quantum
  communication channel.
\newblock {\em Problems of Information Transmission}, 9(3): 3--11, 1973.

\bibitem[Hol98]{Holevo98}
A.~Holevo.
\newblock The capacity of the quantum channel with general signal states.
\newblock {\em Information Theory, IEEE Transactions on}, 44(1): 269 --273,
  1998.
\newblock \\
  \texttt{DOI:\,\href{http://dx.doi.org/10.1109/18.651037}{10.1109/18.651037}}.

\bibitem[{Hol}02]{Holevo02}
A.~S. {Holevo}.
\newblock {On entanglement-assisted classical capacity}.
\newblock {\em Journal of Mathematical Physics}, 43(9): 4326, 2002.
\newblock \\
  \texttt{DOI:\,\href{http://dx.doi.org/10.1063/1.1495877}{10.1063/1.1495877}}.

\bibitem[Hol12]{Holevo12}
A.~Holevo.
\newblock Information capacity of a quantum observable.
\newblock {\em Problems of Information Transmission}, 48: 1--10, 2012.
\newblock \\
  \texttt{DOI:\,\href{http://dx.doi.org/10.1134/S0032946012010012}{10.1134/S0032946012010012}}.

\bibitem[HOW05]{Horodecki05}
M.~Horodecki, J.~Oppenheim, and A.~Winter.
\newblock Partial quantum information.
\newblock {\em Nature}, 436: 673--676, 2005.
\newblock \\
  \texttt{DOI:\,\href{http://dx.doi.org/10.1038/nature03909}{10.1038/nature03909}}.

\bibitem[HOW07]{Horodecki07}
M.~Horodecki, J.~Oppenheim, and A.~Winter.
\newblock Quantum state merging and negative information.
\newblock {\em Communications in Mathematical Physics}, 269: 107--136, 2007.
\newblock \\
  \texttt{DOI:\,\href{http://dx.doi.org/10.1007/s00220-006-0118-x}{10.1007/s00220-006-0118-x}}.

\bibitem[HR11]{Holenstein11}
T.~Holenstein and R.~Renner.
\newblock On the randomness of independent experiments.
\newblock {\em Information Theory, IEEE Transactions on}, 57(4): 1865 --1871,
  2011.
\newblock \\
  \texttt{DOI:\,\href{http://dx.doi.org/10.1109/TIT.2011.2110230}{10.1109/TIT.2011.2110230}}.

\bibitem[HW94]{Hausladen94}
P.~Hausladen and W.~K. Wootters.
\newblock A `pretty good' measurement for distinguishing quantum states.
\newblock {\em Journal of Modern Optics}, 41(12): 2385--2390, 1994.
\newblock \\
  \texttt{DOI:\,\href{http://dx.doi.org/10.1080/09500349414552221}{10.1080/09500349414552221}}.

\bibitem[ILL89]{Impagliazzo89}
R.~Impagliazzo, L.~A. Levin, and M.~Luby.
\newblock Pseudo-random generation from one-way functions.
\newblock In {\em Proceedings of the twenty-first annual ACM symposium on
  Theory of computing}, STOC '89, pages 12--24, New York, NY, USA, 1989. ACM.
\newblock \\
  \texttt{DOI:\,\href{http://dx.doi.org/10.1145/73007.73009}{10.1145/73007.73009}}.

\bibitem[Iva81]{Ivanovic81}
I.~D. Ivanovic.
\newblock Geometrical description of quantal state determination.
\newblock {\em Journal of Physics A: Mathematical and General}, 14(12): 3241,
  1981.
\newblock \\
  \texttt{DOI:\,\href{http://dx.doi.org/10.1088/0305-4470/14/12/019}{10.1088/0305-4470/14/12/019}}.

\bibitem[Iva92]{Ivanovic92}
I.~D. Ivanovic.
\newblock An inequality for the sum of entropies of unbiased quantum
  measurements.
\newblock {\em Journal of Physics A: Mathematical and General}, 25(7): L363,
  1992.
\newblock \\
  \texttt{DOI:\,\href{http://dx.doi.org/10.1088/0305-4470/25/7/014}{10.1088/0305-4470/25/7/014}}.

\bibitem[Jac03]{Jacobs03}
K.~Jacobs.
\newblock On the properties of information gathering in quantum and classical
  measurements.
\newblock 2003.
\newblock \\ Online: \url{http://arxiv.org/abs/quant-ph/0304200}.

\bibitem[Jac06]{Jacobs06}
K.~Jacobs.
\newblock A bound on the mutual information, and properties of entropy
  reduction, for quantum channels with inefficient measurements.
\newblock {\em Journal of Mathematical Physics}, 47(1): 012102--012102--10,
  2006.
\newblock \\
  \texttt{DOI:\,\href{http://dx.doi.org/10.1063/1.2158433}{10.1063/1.2158433}}.

\bibitem[Jac09]{Jabobs09}
K.~Jacobs.
\newblock Second law of thermodynamics and quantum feedback control: Maxwell's
  demon with weak measurements.
\newblock {\em Phys. Rev. A}, 80: 012322, 2009.
\newblock \\
  \texttt{DOI:\,\href{http://dx.doi.org/10.1103/PhysRevA.80.012322}{10.1103/PhysRevA.80.012322}}.

\bibitem[JNP{\etalchar{+}}11]{Scholz11}
M.~Junge, M.~Navascues, C.~Palazuelos, D.~Perez-Garcia, V.~B. Scholz, and R.~F.
  Werner.
\newblock Connes' embedding problem and {T}sirelson's problem.
\newblock {\em Journal of Mathematical Physics}, 52(1): 012102, 2011.
\newblock \\
  \texttt{DOI:\,\href{http://dx.doi.org/10.1063/1.3514538}{10.1063/1.3514538}}.

\bibitem[Joz94]{Jozsa94}
R.~Jozsa.
\newblock Fidelity for mixed quantum states.
\newblock {\em Journal of Modern Optics}, 41(12): 2315--2323, 1994.
\newblock \\
  \texttt{DOI:\,\href{http://dx.doi.org/10.1080/09500349414552171}{10.1080/09500349414552171}}.

\bibitem[KdMT{\etalchar{+}}08]{Konrad08}
T.~Konrad, F.~de~Melo, M.~Tiersch, C.~Kasztelan, A.~Aragao, and A.~Buchleitner.
\newblock Evolution equation for quantum entanglement.
\newblock {\em Nature Physics}, 4: 99--102, 2008.
\newblock \\
  \texttt{DOI:\,\href{http://dx.doi.org/10.1038/nphys885}{10.1038/nphys885}}.

\bibitem[Kit97]{Kitaev97}
A.~Kitaev.
\newblock Quantum computations: algorithms and error correction.
\newblock {\em Russian Mathematical Surveys}, 52(6): 1191--1249, 1997.
\newblock \\
  \texttt{DOI:\,\href{http://dx.doi.org/10.1070/RM1997v052n06ABEH002155}{10.1070/RM1997v052n06ABEH002155}}.

\bibitem[KMR05]{Koenig05}
R.~Konig, U.~Maurer, and R.~Renner.
\newblock On the power of quantum memory.
\newblock {\em Information Theory, IEEE Transactions on}, 51(7): 2391--2401,
  2005.
\newblock \\
  \texttt{DOI:\,\href{http://dx.doi.org/10.1109/TIT.2005.850087}{10.1109/TIT.2005.850087}}.

\bibitem[KR05]{Klappenecker05}
A.~Klappenecker and M.~Rotteler.
\newblock Mutually unbiased bases are complex projective 2-designs.
\newblock In {\em Information Theory, 2005. ISIT 2005. Proceedings.
  International Symposium on}, pages 1740 --1744, 2005.
\newblock \\
  \texttt{DOI:\,\href{http://dx.doi.org/10.1109/ISIT.2005.1523643}{10.1109/ISIT.2005.1523643}}.

\bibitem[KR11]{Koenig11}
R.~Konig and R.~Renner.
\newblock Sampling of min-entropy relative to quantum knowledge.
\newblock {\em Information Theory, IEEE Transactions on}, 57(7): 4760--4787,
  2011.
\newblock \\
  \texttt{DOI:\,\href{http://dx.doi.org/10.1109/TIT.2011.2146730}{10.1109/TIT.2011.2146730}}.

\bibitem[KRS09]{Koenig09}
R.~Konig, R.~Renner, and C.~Schaffner.
\newblock The operational meaning of min- and max-entropy.
\newblock {\em Information Theory, IEEE Transactions on}, 55(9): 4337 --4347,
  2009.
\newblock \\
  \texttt{DOI:\,\href{http://dx.doi.org/10.1109/TIT.2009.2025545}{10.1109/TIT.2009.2025545}}.

\bibitem[KT08]{Koenig08}
R.~Konig and B.~Terhal.
\newblock The bounded-storage model in the presence of a quantum adversary.
\newblock {\em Information Theory, IEEE Transactions on}, 54(2): 749--762,
  2008.
\newblock \\
  \texttt{DOI:\,\href{http://dx.doi.org/10.1109/TIT.2007.913245}{10.1109/TIT.2007.913245}}.

\bibitem[Kuz10]{Kuznetsova10}
A.~A. Kuznetsova.
\newblock Quantum conditional entropy for infinite-dimensional systems.
\newblock {\em Theory of Probability and its Applications}, 55(4): 709--717,
  2010.
\newblock \\
  \texttt{DOI:\,\href{http://dx.doi.org/10.1137/S0040585X97985121}{10.1137/S0040585X97985121}}.

\bibitem[KW04]{Kretschmann04}
D.~Kretschmann and R.~F. Werner.
\newblock Tema con variazioni: quantum channel capacity.
\newblock {\em New Journal of Physics}, 6(1): 26, 2004.
\newblock \\
  \texttt{DOI:\,\href{http://dx.doi.org/10.1088/1367-2630/6/1/026}{10.1088/1367-2630/6/1/026}}.

\bibitem[KW09]{Koenig09_2}
R.~K\"onig and S.~Wehner.
\newblock A strong converse for classical channel coding using entangled
  inputs.
\newblock {\em Phys. Rev. Lett.}, 103: 070504, 2009.
\newblock \\
  \texttt{DOI:\,\href{http://dx.doi.org/10.1103/PhysRevLett.103.070504}{10.1103/PhysRevLett.103.070504}}.

\bibitem[KW10]{Kiukas10}
J.~Kiukas and R.~F. Werner.
\newblock Maximal violation of {B}ell inequalities by position measurements.
\newblock {\em Journal of Mathematical Physics}, 51(7): 072105, 2010.
\newblock \\
  \texttt{DOI:\,\href{http://dx.doi.org/10.1063/1.3447736}{10.1063/1.3447736}}.

\bibitem[KWW12]{Koenig12}
R.~Konig, S.~Wehner, and J.~Wullschleger.
\newblock Unconditional security from noisy quantum storage.
\newblock {\em Information Theory, IEEE Transactions on}, 58(3): 1962 --1984,
  2012.
\newblock \\
  \texttt{DOI:\,\href{http://dx.doi.org/10.1109/TIT.2011.2177772}{10.1109/TIT.2011.2177772}}.

\bibitem[Lan92]{Landauer92}
R.~Landauer.
\newblock Information is physical.
\newblock In {\em Physics and Computation, 1992. PhysComp '92., Workshop on},
  pages 1--4, 1992.
\newblock \\
  \texttt{DOI:\,\href{http://dx.doi.org/10.1109/PHYCMP.1992.615478}{10.1109/PHYCMP.1992.615478}}.

\bibitem[Lar90]{Larsen90}
U.~Larsen.
\newblock Superspace geometry: the exact uncertainty relationship between
  complementary aspects.
\newblock {\em Journal of Physics A: Mathematical and General}, 23(7): 1041,
  1990.
\newblock \\
  \texttt{DOI:\,\href{http://dx.doi.org/10.1088/0305-4470/23/7/013}{10.1088/0305-4470/23/7/013}}.

\bibitem[LC98]{Lo98}
H.-K. Lo and H.~Chau.
\newblock Why quantum bit commitment and ideal quantum coin tossing are
  impossible.
\newblock {\em Physica D: Nonlinear Phenomena}, 120(1-2): 177 -- 187, 1998.
\newblock \\
  \texttt{DOI:\,\href{http://dx.doi.org/10.1016/S0167-2789(98)00053-0}{10.1016/S0167-2789(98)00053-0}}.

\bibitem[Lin72]{Lindblad72}
G.~Lindblad.
\newblock An entropy inequality for quantum measurements.
\newblock {\em Communications in Mathematical Physics}, 28: 245--249, 1972.
\newblock \\
  \texttt{DOI:\,\href{http://dx.doi.org/10.1007/BF01645778}{10.1007/BF01645778}}.

\bibitem[Llo97]{Lloyd97}
S.~Lloyd.
\newblock Capacity of the noisy quantum channel.
\newblock {\em Phys. Rev. A}, 55: 1613--1622, 1997.
\newblock \\
  \texttt{DOI:\,\href{http://dx.doi.org/10.1103/PhysRevA.55.1613}{10.1103/PhysRevA.55.1613}}.

\bibitem[Lo97a]{Lo97}
H.-K. Lo.
\newblock Insecurity of quantum secure computations.
\newblock {\em Phys. Rev. A}, 56: 1154--1162, 1997.
\newblock \\
  \texttt{DOI:\,\href{http://dx.doi.org/10.1103/PhysRevA.56.1154}{10.1103/PhysRevA.56.1154}}.

\bibitem[Lo97b]{Lo97_2}
H.-K. Lo.
\newblock Is quantum bit commitment really possible?
\newblock {\em Phys. Rev. Lett.}, 78: 3410--3413, 1997.
\newblock \\
  \texttt{DOI:\,\href{http://dx.doi.org/10.1103/PhysRevLett.78.3410}{10.1103/PhysRevLett.78.3410}}.

\bibitem[LS08]{Leung08_2}
D.~Leung and G.~Smith.
\newblock Communicating over adversarial quantum channels using quantum list
  codes.
\newblock {\em Information Theory, IEEE Transactions on}, 54(2): 883--887,
  2008.
\newblock \\
  \texttt{DOI:\,\href{http://dx.doi.org/10.1109/TIT.2007.913433}{10.1109/TIT.2007.913433}}.

\bibitem[LTW08]{Leung08}
D.~Leung, B.~Toner, and J.~Watrous.
\newblock Coherent state exchange in multi-prover quantum interactive proof
  systems.
\newblock 2008.
\newblock \\ Online: \url{http://arxiv.org/abs/0804.4118}.

\bibitem[Luo10]{Luo10}
S.~Luo.
\newblock Information conservation and entropy change in quantum measurements.
\newblock {\em Phys. Rev. A}, 82: 052103, 2010.
\newblock \\
  \texttt{DOI:\,\href{http://dx.doi.org/10.1103/PhysRevA.82.052103}{10.1103/PhysRevA.82.052103}}.

\bibitem[LXX{\etalchar{+}}11]{Li11}
C.-F. Li, J.-S. Xu, X.-Y. Xu, K.~Li, and G.-C. Guo.
\newblock Experimental investigation of the entanglement-assisted entropic
  uncertainty principle.
\newblock {\em Nature Physics}, 7: 752--756, 2011.
\newblock \\
  \texttt{DOI:\,\href{http://dx.doi.org/10.1038/nphys2047}{10.1038/nphys2047}}.

\bibitem[May]{Mayers96}
D.~Mayers.
\newblock The trouble with quantum bit commitment.
\newblock \\ Online: \url{http://arxiv.org/abs/quant-ph/9603015}.

\bibitem[May97]{Mayers97}
D.~Mayers.
\newblock Unconditionally secure quantum bit commitment is impossible.
\newblock {\em Phys. Rev. Lett.}, 78: 3414--3417, 1997.
\newblock \\
  \texttt{DOI:\,\href{http://dx.doi.org/10.1103/PhysRevLett.78.3414}{10.1103/PhysRevLett.78.3414}}.

\bibitem[Mey00]{Meyer00}
C.~D. Meyer.
\newblock {\em Matrix Analysis and Applied Linear Algebra}.
\newblock Cambridge University Press, 2000.

\bibitem[MLDS{\etalchar{+}}13]{Tomamichel13}
M.~M{\"u}ller-Lennert, F.~Dupuis, O.~Szehr, S.~Fehr, and M.~Tomamichel.
\newblock On quantum r\'enyi entropies: a new definition and some properties.
\newblock 2013.
\newblock \\ Online: \url{http://arxiv.org/abs/1306.3142}.

\bibitem[MM90]{Martens90}
H.~Martens and W.~Muynck.
\newblock Nonideal quantum measurements.
\newblock {\em Foundations of Physics}, 20: 255--281, 1990.
\newblock \\
  \texttt{DOI:\,\href{http://dx.doi.org/10.1007/BF00731693}{10.1007/BF00731693}}.

\bibitem[MU88]{Maassen88}
H.~Maassen and J.~B.~M. Uffink.
\newblock Generalized entropic uncertainty relations.
\newblock {\em Phys. Rev. Lett.}, 60: 1103--1106, 1988.
\newblock \\
  \texttt{DOI:\,\href{http://dx.doi.org/10.1103/PhysRevLett.60.1103}{10.1103/PhysRevLett.60.1103}}.

\bibitem[Mur90]{Murphy90}
G.~Murphy.
\newblock {\em C*-algebras and Operator Theory}.
\newblock Academic Press, Inc., Boston, 1990.

\bibitem[MW13]{Morgan13}
C.~Morgan and A.~Winter.
\newblock "{P}retty strong" converse for the quantum capacity of degradable
  channels.
\newblock 2013.
\newblock \\ Online: \url{http://arxiv.org/abs/1301.4927}.

\bibitem[MWW09]{Matthews09}
W.~Matthews, S.~Wehner, and A.~Winter.
\newblock Distinguishability of quantum states under restricted families of
  measurements with an application to quantum data hiding.
\newblock {\em Communications in Mathematical Physics}, 291: 813--843, 2009.
\newblock \\
  \texttt{DOI:\,\href{http://dx.doi.org/10.1007/s00220-009-0890-5}{10.1007/s00220-009-0890-5}}.

\bibitem[NBW12]{Ng12}
N.~H.~Y. Ng, M.~Berta, and S.~Wehner.
\newblock Min-entropy uncertainty relation for finite-size cryptography.
\newblock {\em Phys. Rev. A}, 86: 042315, 2012.
\newblock \\
  \texttt{DOI:\,\href{http://dx.doi.org/10.1103/PhysRevA.86.042315}{10.1103/PhysRevA.86.042315}}.

\bibitem[NC00]{Nielsen00}
M.~A. Nielsen and I.~L. Chuang.
\newblock {\em Quantum computation and quantum information}.
\newblock Cambridge University Press, 2000.

\bibitem[NCPGV12]{Navascues11}
M.~Navascues, T.~Cooney, D.~Perez-Garcia, and I.~Villanueva.
\newblock A physical approach to {T}sirelson's problem.
\newblock {\em Foundations of Physics}, 42(8): 985--995, 2012.
\newblock \\
  \texttt{DOI:\,\href{http://dx.doi.org/10.1007/s10701-012-9641-0}{10.1007/s10701-012-9641-0}}.

\bibitem[Nie00]{Nielsen00_2}
M.~A. Nielsen.
\newblock Continuity bounds for entanglement.
\newblock {\em Phys. Rev. A}, 61: 064301, 2000.
\newblock \\
  \texttt{DOI:\,\href{http://dx.doi.org/10.1103/PhysRevA.61.064301}{10.1103/PhysRevA.61.064301}}.

\bibitem[NPA07]{Navascues07}
M.~Navascues, S.~Pironio, and A.~Acin.
\newblock Bounding the set of quantum correlations.
\newblock {\em Phys. Rev. Lett.}, 98: 010401, 2007.
\newblock \\
  \texttt{DOI:\,\href{http://dx.doi.org/10.1103/PhysRevLett.98.010401}{10.1103/PhysRevLett.98.010401}}.

\bibitem[NZ96]{Nissan96}
N.~Nisan and D.~Zuckerman.
\newblock Randomness is linear in space.
\newblock {\em Journal of Computer and System Sciences}, 51(1): 43 -- 52, 1996.
\newblock \\
  \texttt{DOI:\,\href{http://dx.doi.org/10.1006/jcss.1996.0004}{10.1006/jcss.1996.0004}}.

\bibitem[ON99]{Ogawa99}
T.~Ogawa and H.~Nagaoka.
\newblock Strong converse to the quantum channel coding theorem.
\newblock {\em Information Theory, IEEE Transactions on}, 45(7): 2486 --2489,
  1999.
\newblock \\
  \texttt{DOI:\,\href{http://dx.doi.org/10.1109/18.796386}{10.1109/18.796386}}.

\bibitem[OP93]{Petz93}
M.~Ohya and D.~Petz.
\newblock {\em Quantum Entropy and Its Use}.
\newblock Springer, 1993.

\bibitem[Opp08]{Oppenheim08}
J.~Oppenheim.
\newblock State redistribution as merging: introducing the coherent relay.
\newblock 2008.
\newblock \\ Online: \url{http://arxiv.org/abs/0805.1065}.

\bibitem[OZ99]{Olkiewicz99}
R.~Olkiewicz and B.~Zegarlinski.
\newblock Hypercontractivity in noncommutative lpspaces.
\newblock {\em Journal of Functional Analysis}, 161(1): 246 -- 285, 1999.
\newblock \\
  \texttt{DOI:\,\href{http://dx.doi.org/10.1006/jfan.1998.3342}{10.1006/jfan.1998.3342}}.

\bibitem[Oza86]{Ozawa86}
M.~Ozawa.
\newblock On information gain by quantum measurements of continuous
  observables.
\newblock {\em Journal of Mathematical Physics}, 27(3): 759--763, 1986.
\newblock \\
  \texttt{DOI:\,\href{http://dx.doi.org/10.1063/1.527179}{10.1063/1.527179}}.

\bibitem[Par83]{Partovi83}
M.~H. Partovi.
\newblock Entropic formulation of uncertainty for quantum measurements.
\newblock {\em Phys. Rev. Lett.}, 50: 1883--1885, 1983.
\newblock \\
  \texttt{DOI:\,\href{http://dx.doi.org/10.1103/PhysRevLett.50.1883}{10.1103/PhysRevLett.50.1883}}.

\bibitem[Pau02]{Paulsen02}
V.~I. Paulsen.
\newblock {\em Completely bounded maps and operator algebras}.
\newblock Cambridge University Press, 2002.

\bibitem[Per96]{Peres96}
A.~Peres.
\newblock Separability criterion for density matrices.
\newblock {\em Phys. Rev. Lett.}, 77: 1413--1415, 1996.
\newblock \\
  \texttt{DOI:\,\href{http://dx.doi.org/10.1103/PhysRevLett.77.1413}{10.1103/PhysRevLett.77.1413}}.

\bibitem[PHC{\etalchar{+}}11]{Prevedel11}
R.~Prevedel, D.~R. Hamel, R.~Colbeck, K.~Fisher, and K.~J. Resch.
\newblock Experimental investigation of the uncertainty principle in the
  presence of quantum memory and its application to witnessing entanglement.
\newblock {\em Nature Physics}, 7: 757--761, 2011.
\newblock \\
  \texttt{DOI:\,\href{http://dx.doi.org/10.1038/nphys2048}{10.1038/nphys2048}}.

\bibitem[R\'61]{Renyi60}
A.~R\'{e}nyi.
\newblock On measures of information and entropy.
\newblock In {\em Proceedings of the 4th Berkeley Symposium on Mathematics,
  Statistics and Probability}, pages 547--561. University of California Press,
  1961.

\bibitem[RB09]{Boileau09}
J.~M. Renes and J.-C. Boileau.
\newblock Conjectured strong complementary information tradeoff.
\newblock {\em Phys. Rev. Lett.}, 103: 020402, 2009.
\newblock \\
  \texttt{DOI:\,\href{http://dx.doi.org/10.1103/PhysRevLett.103.020402}{10.1103/PhysRevLett.103.020402}}.

\bibitem[Ren05]{Renner05}
R.~Renner.
\newblock {\em Security of Quantum Key Distribution}.
\newblock PhD thesis, ETH Zurich, 2005.
\newblock \\
  \texttt{DOI:\,\href{http://dx.doi.org/10.1142/S0219749908003256}{10.1142/S0219749908003256}}.

\bibitem[RK05]{Renner05_2}
R.~Renner and R.~K{\"o}nig.
\newblock Universally composable privacy amplification against quantum
  adversaries.
\newblock In J.~Kilian, editor, {\em Theory of Cryptography}, volume 3378 of
  {\em Lecture Notes in Computer Science}, pages 407--425. Springer Berlin
  Heidelberg, 2005.
\newblock \\
  \texttt{DOI:\,\href{http://dx.doi.org/10.1007/978-3-540-30576-7$\_$22}{10.1007/978-3-540-30576-7$\_$22}}.

\bibitem[Rob29]{Robertson29}
H.~P. Robertson.
\newblock The uncertainty principle.
\newblock {\em Phys. Rev.}, 34: 163--164, 1929.
\newblock \\
  \texttt{DOI:\,\href{http://dx.doi.org/10.1103/PhysRev.34.163}{10.1103/PhysRev.34.163}}.

\bibitem[RTS00]{Radhakrishnan00}
J.~Radhakrishnan and A.~Ta-Shma.
\newblock Bounds for dispersers, extractors, and depth-two superconcentrators.
\newblock {\em SIAM Journal on Discrete Mathematics}, 13(1): 2--24, 2000.
\newblock \\
  \texttt{DOI:\,\href{http://dx.doi.org/10.1137/S0895480197329508}{10.1137/S0895480197329508}}.

\bibitem[Rud10]{Rudnicki10}
{\L}.~Rudnicki.
\newblock Uncertainty related to position and momentum localization of a
  quantum state.
\newblock 2010.
\newblock \\ Online: \url{http://arxiv.org/abs/1010.3269}.

\bibitem[Rud11]{Rudnicki11}
{\L}.~Rudnicki.
\newblock Shannon entropy as a measure of uncertainty in positions and momenta.
\newblock {\em Journal of Russian Laser Research}, 32: 393--399, 2011.
\newblock \\
  \texttt{DOI:\,\href{http://dx.doi.org/10.1007/s10946-011-9227-x}{10.1007/s10946-011-9227-x}}.

\bibitem[RvE13]{Ray13}
M.~R. Ray and S.~van Enk.
\newblock Missing data outside the detector range: application to continuous
  variable entanglement verification and quantum cryptography.
\newblock 2013.
\newblock \\ Online: \url{http://arxiv.org/abs/1302.5087}.

\bibitem[RW04]{Renner04}
R.~Renner and S.~Wolf.
\newblock Smooth {R}enyi entropy and applications.
\newblock In {\em Information Theory, 2004. ISIT 2004. Proceedings.
  International Symposium on}, page 233, 2004.
\newblock \\
  \texttt{DOI:\,\href{http://dx.doi.org/10.1109/ISIT.2004.1365269}{10.1109/ISIT.2004.1365269}}.

\bibitem[RWT12]{Rudnicki12}
L.~Rudnicki, S.~P. Walborn, and F.~Toscano.
\newblock Optimal uncertainty relations for extremely coarse-grained
  measurements.
\newblock {\em Phys. Rev. A}, 85: 042115, 2012.
\newblock \\
  \texttt{DOI:\,\href{http://dx.doi.org/10.1103/PhysRevA.85.042115}{10.1103/PhysRevA.85.042115}}.

\bibitem[Sch07]{Schaffner07}
C.~Schaffner.
\newblock {\em Cryptography in the Bounded-Quantum-Storage Model}.
\newblock PhD thesis, University of Aarhus, 2007.
\newblock \\ Online: \url{http://arxiv.org/abs/0709.0289}.

\bibitem[SDH{\etalchar{+}}13]{Schneeloch12}
J.~Schneeloch, P.~B. Dixon, G.~A. Howland, C.~J. Broadbent, and J.~C. Howell.
\newblock Violation of continuous-variable einstein-podolsky-rosen steering
  with discrete measurements.
\newblock {\em Phys. Rev. Lett.}, 110: 130407, 2013.
\newblock \\
  \texttt{DOI:\,\href{http://dx.doi.org/10.1103/PhysRevLett.110.130407}{10.1103/PhysRevLett.110.130407}}.

\bibitem[SDTR13]{Szehr11_2}
O.~Szehr, F.~Dupuis, M.~Tomamichel, and R.~Renner.
\newblock Decoupling with unitary almost two-designs.
\newblock {\em New Journal of Physics}, 15(5): 053022, 2013.
\newblock \\
  \texttt{DOI:\,\href{http://dx.doi.org/10.1088/1367-2630/15/5/053022}{10.1088/1367-2630/15/5/053022}}.

\bibitem[Sha48]{Shannon48}
C.~E. Shannon.
\newblock A mathematical theory of communication.
\newblock {\em Bell System Technical Journal}, 27: 379--423, 623--656, 1948.

\bibitem[Sha02]{Shaltiel02}
R.~Shaltiel.
\newblock Recent developments in explicit constructions of extractors.
\newblock {\em Bulletin of the EATCS}, 77: 67--95, 2002.
\newblock \\ Online:
  \url{http://kam.mff.cuni.cz/~matousek/cla/shaltiel-extractors-survey.ps}.

\bibitem[{Shi}11]{Shirokov11}
M.~E. {Shirokov}.
\newblock {Entropy reduction of quantum measurements}.
\newblock {\em Journal of Mathematical Physics}, 52(5): 052202, 2011.
\newblock \\
  \texttt{DOI:\,\href{http://dx.doi.org/10.1063/1.3589831}{10.1063/1.3589831}}.

\bibitem[Sho02]{Shor02}
P.~W. Shor.
\newblock The quantum channel capacity and coherent information.
\newblock {\em Lecture notes, MSRI Workshop on Quantum Computation}, 2002.
\newblock \\ Online:
  \url{http://www.msri.org/publications/ln/msri/2002/quantumcrypto/shor/1/}.

\bibitem[Sho04]{Shor04}
P.~W. Shor.
\newblock Equivalence of additivity questions in quantum information theory.
\newblock {\em Communications in Mathematical Physics}, 246: 453--472, 2004.
\newblock \\
  \texttt{DOI:\,\href{http://dx.doi.org/10.1007/s00220-003-0981-7}{10.1007/s00220-003-0981-7}}.

\bibitem[Sio58]{Sion58}
M.~Sion.
\newblock On general minimax theorems.
\newblock {\em Pacific Journal of Mathematics}, 8(1): 171--176, 1958.

\bibitem[Sip88]{Sipser88}
M.~Sipser.
\newblock Expanders, randomness, or time versus space.
\newblock {\em Journal of Computer and System Sciences}, 36(3): 379 -- 383,
  1988.
\newblock \\
  \texttt{DOI:\,\href{http://dx.doi.org/10.1016/0022-0000(88)90035-9}{10.1016/0022-0000(88)90035-9}}.

\bibitem[SP64]{Slepian64}
D.~Slepian and H.~O. Pollak.
\newblock Prolate spheroidal wave functions, {F}ourier analysis and
  uncertainty-{I}.
\newblock {\em The Bell System Technical Journal}, 40: 43, 1964.

\bibitem[SR93]{Sanchez93}
J.~Sanchez-Ruiz.
\newblock Entropic uncertainty and certainty relations for complementary
  observables.
\newblock {\em Physics Letters A}, 173(3): 233 -- 239, 1993.
\newblock \\
  \texttt{DOI:\,\href{http://dx.doi.org/10.1016/0375-9601(93)90269-6}{10.1016/0375-9601(93)90269-6}}.

\bibitem[SR95]{Sanchez95}
J.~Sanchez-Ruiz.
\newblock Improved bounds in the entropic uncertainty and certainty relations
  for complementary observables.
\newblock {\em Physics Letters A}, 201(2/3): 125 -- 131, 1995.
\newblock \\
  \texttt{DOI:\,\href{http://dx.doi.org/10.1016/0375-9601(95)00219-S}{10.1016/0375-9601(95)00219-S}}.

\bibitem[SR98]{Sanchez98}
J.~Sanchez-Ruiz.
\newblock Optimal entropic uncertainty relation in two-dimensional hilbert
  space.
\newblock {\em Physics Letters A}, 244(4): 189 -- 195, 1998.
\newblock \\
  \texttt{DOI:\,\href{http://dx.doi.org/10.1016/S0375-9601(98)00292-8}{10.1016/S0375-9601(98)00292-8}}.

\bibitem[Sti55]{Stinespring55}
W.~Stinespring.
\newblock Positive function on {C}*-algebras.
\newblock {\em Proceedings of American Mathematical Society}, 6: 211--216,
  1955.
\newblock \\
  \texttt{DOI:\,\href{http://dx.doi.org/10.1090/S0002-9939-1955-0069403-4}{10.1090/S0002-9939-1955-0069403-4}}.

\bibitem[STW08]{Schaffner08}
C.~Schaffner, B.~Terhal, and S.~Wehner.
\newblock Robust cryptography in the noisy-quantum-storage model.
\newblock {\em Quantum Information \& Computation}, 9: 963--996, 2008.

\bibitem[SW73]{Slepian73}
D.~Slepian and J.~Wolf.
\newblock Noiseless coding of correlated information sources.
\newblock {\em Information Theory, IEEE Transactions on}, 19(4): 471 -- 480,
  1973.
\newblock \\
  \texttt{DOI:\,\href{http://dx.doi.org/10.1109/TIT.1973.1055037}{10.1109/TIT.1973.1055037}}.

\bibitem[SW97]{Schumacher97}
B.~Schumacher and M.~D. Westmoreland.
\newblock Sending classical information via noisy quantum channels.
\newblock {\em Phys. Rev. A}, 56: 131--138, 1997.
\newblock \\
  \texttt{DOI:\,\href{http://dx.doi.org/10.1103/PhysRevA.56.131}{10.1103/PhysRevA.56.131}}.

\bibitem[SW08]{Scholz08}
V.~B. Scholz and R.~F. Werner.
\newblock Tsirelson's problem.
\newblock 2008.
\newblock \\ Online: \url{http://arxiv.org/abs/0812.4305}.

\bibitem[SW13]{Sharma13}
N.~Sharma and N.~A. Warsi.
\newblock Fundamental bound on the reliability of quantum information
  transmission.
\newblock {\em Phys. Rev. Lett.}, 110: 080501, 2013.
\newblock \\
  \texttt{DOI:\,\href{http://dx.doi.org/10.1103/PhysRevLett.110.080501}{10.1103/PhysRevLett.110.080501}}.

\bibitem[Sze11]{Szehr11}
O.~Szehr.
\newblock Decoupling theorems.
\newblock Master's thesis, ETH Zurich, 2011.
\newblock \\ Online: \url{http://arxiv.org/abs/1207.3927}.

\bibitem[Tak01]{Takesaki01}
M.~Takesaki.
\newblock {\em Theory of Operator Algebras I}.
\newblock Springer, 2001.

\bibitem[Tak02a]{Takesaki02}
M.~Takesaki.
\newblock {\em Theory of Operator Algebras II}.
\newblock Springer, 2002.

\bibitem[Tak02b]{Takesaki02_2}
M.~Takesaki.
\newblock {\em Theory of Operator Algebras II}.
\newblock Springer, 2002.

\bibitem[TCR09]{Tomamichel09}
M.~Tomamichel, R.~Colbeck, and R.~Renner.
\newblock A fully quantum asymptotic equipartition property.
\newblock {\em Information Theory, IEEE Transactions on}, 55(12): 5840 --5847,
  2009.
\newblock \\
  \texttt{DOI:\,\href{http://dx.doi.org/10.1109/TIT.2009.2032797}{10.1109/TIT.2009.2032797}}.

\bibitem[TCR10]{Tomamichel10}
M.~Tomamichel, R.~Colbeck, and R.~Renner.
\newblock Duality between smooth min- and max-entropies.
\newblock {\em Information Theory, IEEE Transactions on}, 56(9): 4674 --4681,
  2010.
\newblock \\
  \texttt{DOI:\,\href{http://dx.doi.org/10.1109/TIT.2010.2054130}{10.1109/TIT.2010.2054130}}.

\bibitem[THLD02]{Terhal02}
B.~M. Terhal, M.~Horodecki, D.~W. Leung, and D.~P. DiVincenzo.
\newblock The entanglement of purification.
\newblock {\em Journal of Mathematical Physics}, 43(9): 4286--4298, 2002.
\newblock \\
  \texttt{DOI:\,\href{http://dx.doi.org/10.1063/1.1498001}{10.1063/1.1498001}}.

\bibitem[TLGR12]{Tomamichel12_2}
M.~Tomamichel, C.~C.~W. Lim, N.~Gisin, and R.~Renner.
\newblock Tight finite-key analysis for quantum cryptography.
\newblock {\em Nature Communications}, 3(634), 2012.
\newblock \\
  \texttt{DOI:\,\href{http://dx.doi.org/10.1038/ncomms1631}{10.1038/ncomms1631}}.

\bibitem[Tom12]{Tomamichel12}
M.~Tomamichel.
\newblock {\em A Framework for Non-Asymptotic Quantum Information Theory}.
\newblock PhD thesis, ETH Zurich, 2012.
\newblock \\ Online: \url{http://arxiv.org/abs/1203.2142}.

\bibitem[TR11]{Tomamichel11_2}
M.~Tomamichel and R.~Renner.
\newblock Uncertainty relation for smooth entropies.
\newblock {\em Phys. Rev. Lett.}, 106: 110506, 2011.
\newblock \\
  \texttt{DOI:\,\href{http://dx.doi.org/10.1103/PhysRevLett.106.110506}{10.1103/PhysRevLett.106.110506}}.

\bibitem[Tre99]{Trevisan99}
L.~Trevisan.
\newblock Construction of extractors using pseudo-random generators (extended
  abstract).
\newblock In {\em Proceedings of the thirty-first annual ACM symposium on
  Theory of computing}, STOC '99, pages 141--148, New York, NY, USA, 1999. ACM.
\newblock \\
  \texttt{DOI:\,\href{http://dx.doi.org/10.1145/301250.301289}{10.1145/301250.301289}}.

\bibitem[Tro12]{Tropp12}
J.~Tropp.
\newblock User-friendly tail bounds for sums of random matrices.
\newblock {\em Foundations of Computational Mathematics}, 12: 389--434, 2012.
\newblock \\
  \texttt{DOI:\,\href{http://dx.doi.org/10.1007/s10208-011-9099-z}{10.1007/s10208-011-9099-z}}.

\bibitem[Tsa88]{Tsallis88}
C.~Tsallis.
\newblock Possible generalization of {B}oltzmann-{G}ibbs statistics.
\newblock {\em Journal of Statistical Physics}, 52: 479--487, 1988.
\newblock \\
  \texttt{DOI:\,\href{http://dx.doi.org/10.1007/BF01016429}{10.1007/BF01016429}}.

\bibitem[Tsi93]{Tsirelson93}
B.~S. Tsirelson.
\newblock Some results and problems on quantum {B}ell-type inequalities.
\newblock {\em Hadronic Journal Supplement}, 8: 329, 1993.

\bibitem[TSSR11]{Tomamichel11}
M.~Tomamichel, C.~Schaffner, A.~Smith, and R.~Renner.
\newblock Leftover hashing against quantum side information.
\newblock {\em Information Theory, IEEE Transactions on}, 57(8): 5524 --5535,
  2011.
\newblock \\
  \texttt{DOI:\,\href{http://dx.doi.org/10.1109/TIT.2011.2158473}{10.1109/TIT.2011.2158473}}.

\bibitem[Uhl76]{Uhlmann76}
A.~Uhlmann.
\newblock The transition probability in the state space of a *-algebra.
\newblock {\em Reports on Mathematical Physics}, 9(2): 273 -- 279, 1976.
\newblock \\
  \texttt{DOI:\,\href{http://dx.doi.org/10.1016/0034-4877(76)90060-4}{10.1016/0034-4877(76)90060-4}}.

\bibitem[Uhl98]{Uhlmann98}
A.~Uhlmann.
\newblock Entropy and optimal decompositions of states relative to a maximal
  commutative subalgebra.
\newblock {\em Open Systems \& Information Dynamics}, 5: 209--228, 1998.
\newblock \\
  \texttt{DOI:\,\href{http://dx.doi.org/10.1023/A:1009664331611}{10.1023/A:1009664331611}}.

\bibitem[Vad07]{Vadhan07}
S.~Vadhan.
\newblock The unified theory of pseudorandomness: guest column.
\newblock {\em SIGACT News}, 38(3): 39--54, 2007.
\newblock \\
  \texttt{DOI:\,\href{http://dx.doi.org/10.1145/1324215.1324225}{10.1145/1324215.1324225}}.

\bibitem[Vad11]{Vadhan11}
S.~P. Vadhan.
\newblock Pseudorandomness.
\newblock {\em Lecture notes}, 2011.
\newblock \\ Online:
  \url{http://people.seas.harvard.edu/~salil/pseudorandomness/}.

\bibitem[VDD01]{Verstraete01}
F.~Verstraete, J.~Dehaene, and B.~DeMoor.
\newblock Local filtering operations on two qubits.
\newblock {\em Phys. Rev. A}, 64: 010101, 2001.
\newblock \\
  \texttt{DOI:\,\href{http://dx.doi.org/10.1103/PhysRevA.64.010101}{10.1103/PhysRevA.64.010101}}.

\bibitem[vDH03]{vanDam03}
W.~van Dam and P.~Hayden.
\newblock Universal entanglement transformations without communication.
\newblock {\em Phys. Rev. A}, 67: 060302, 2003.
\newblock \\
  \texttt{DOI:\,\href{http://dx.doi.org/10.1103/PhysRevA.67.060302}{10.1103/PhysRevA.67.060302}}.

\bibitem[Wal13]{Walter13}
M.~Walter.
\newblock {\em Private Communication}, 2013.

\bibitem[WCSL10]{Wehner10}
S.~Wehner, M.~Curty, C.~Schaffner, and H.-K. Lo.
\newblock Implementation of two-party protocols in the noisy-storage model.
\newblock {\em Phys. Rev. A}, 81: 052336, 2010.
\newblock \\
  \texttt{DOI:\,\href{http://dx.doi.org/10.1103/PhysRevA.81.052336}{10.1103/PhysRevA.81.052336}}.

\bibitem[WDHW13]{Wilde12_3}
M.~Wilde, N.~Datta, M.-H. Hsieh, and A.~Winter.
\newblock Quantum rate-distortion coding with auxiliary resources.
\newblock {\em Information Theory, IEEE Transactions on}, 59(10): 6755--6773,
  2013.
\newblock \\
  \texttt{DOI:\,\href{http://dx.doi.org/10.1109/TIT.2013.2271772}{10.1109/TIT.2013.2271772}}.

\bibitem[WF89]{Wootters89}
W.~K. Wootters and B.~D. Fields.
\newblock Optimal state-determination by mutually unbiased measurements.
\newblock {\em Annals of Physics}, 191(2): 363 -- 381, 1989.
\newblock \\
  \texttt{DOI:\,\href{http://dx.doi.org/10.1016/0003-4916(89)90322-9}{10.1016/0003-4916(89)90322-9}}.

\bibitem[WHBH12]{Wilde12_2}
M.~M. Wilde, P.~Hayden, F.~Buscemi, and M.-H. Hsieh.
\newblock The information-theoretic costs of simulating quantum measurements.
\newblock {\em Journal of Physics A: Mathematical and Theoretical}, 45(45):
  453001, 2012.
\newblock \\
  \texttt{DOI:\,\href{http://dx.doi.org/10.1088/1751-8113/45/45/453001}{10.1088/1751-8113/45/45/453001}}.

\bibitem[Wil12]{Wilde12}
M.~M. Wilde.
\newblock {\em Private Communication}, 2012.

\bibitem[Wil13]{Wilde11}
M.~M. Wilde.
\newblock {\em From Classical to Quantum {S}hannon Theory}.
\newblock Cambridge University Press, 2013.
\newblock \\
  \texttt{DOI:\,\href{http://dx.doi.org/10.1017/CBO9781139525343}{10.1017/CBO9781139525343}}.

\bibitem[Win99]{Winter99}
A.~Winter.
\newblock Coding theorem and strong converse for quantum channels.
\newblock {\em Information Theory, IEEE Transactions on}, 45(7): 2481 --2485,
  1999.
\newblock \\
  \texttt{DOI:\,\href{http://dx.doi.org/10.1109/18.796385}{10.1109/18.796385}}.

\bibitem[Win02]{Winter02}
A.~Winter.
\newblock Compression of sources of probability distributions and density
  operators.
\newblock 2002.
\newblock \\ Online: \url{http://arxiv.org/abs/quant-ph/0208131}.

\bibitem[Win04]{Winter04}
A.~Winter.
\newblock ``{E}xtrinsic'' and ``intrinsic'' data in quantum measurements:
  Asymptotic convex decomposition of positive operator valued measures.
\newblock {\em Communications in Mathematical Physics}, 244: 157--185, 2004.
\newblock \\
  \texttt{DOI:\,\href{http://dx.doi.org/10.1007/s00220-003-0989-z}{10.1007/s00220-003-0989-z}}.

\bibitem[Win10]{Winter10}
A.~Winter.
\newblock Quantum information: Coping with uncertainty.
\newblock {\em Nature Physics}, 6(9): 640--641, 2010.
\newblock \\
  \texttt{DOI:\,\href{http://dx.doi.org/10.1038/nphys1771}{10.1038/nphys1771}}.

\bibitem[Wol64]{Wolfowitz64}
J.~Wolfowitz.
\newblock {\em Coding Theorems of Information Theory}.
\newblock Springer, New York, 1964.

\bibitem[Woo98]{Wootters98}
W.~K. Wootters.
\newblock Entanglement of formation of an arbitrary state of two qubits.
\newblock {\em Phys. Rev. Lett.}, 80: 2245--2248, 1998.
\newblock \\
  \texttt{DOI:\,\href{http://dx.doi.org/10.1103/PhysRevLett.80.2245}{10.1103/PhysRevLett.80.2245}}.

\bibitem[Wor72]{Woronowicz72}
S.~Woronowicz.
\newblock On the purification of factor states.
\newblock {\em Communications in Mathematical Physics}, 28: 221--235, 1972.
\newblock \\
  \texttt{DOI:\,\href{http://dx.doi.org/10.1007/BF01645776}{10.1007/BF01645776}}.

\bibitem[WST08]{Wehner08}
S.~Wehner, C.~Schaffner, and B.~M. Terhal.
\newblock Cryptography from noisy storage.
\newblock {\em Phys. Rev. Lett.}, 100: 220502, 2008.
\newblock \\
  \texttt{DOI:\,\href{http://dx.doi.org/10.1103/PhysRevLett.100.220502}{10.1103/PhysRevLett.100.220502}}.

\bibitem[WW10]{Wehner09}
S.~Wehner and A.~Winter.
\newblock Entropic uncertainty relations---a survey.
\newblock {\em New Journal of Physics}, 12(2): 025009, 2010.
\newblock \\
  \texttt{DOI:\,\href{http://dx.doi.org/10.1088/1367-2630/12/2/025009}{10.1088/1367-2630/12/2/025009}}.

\bibitem[WZ82]{Wootters82}
W.~K. Wootters and W.~H. Zurek.
\newblock A single quantum cannot be cloned.
\newblock {\em Nature}, 299: 802 -- 803, 1982.
\newblock \\
  \texttt{DOI:\,\href{http://dx.doi.org/10.1038/299802a0}{10.1038/299802a0}}.

\bibitem[YHD08]{Yard08}
J.~Yard, P.~Hayden, and I.~Devetak.
\newblock Capacity theorems for quantum multiple-access channels:
  classical-quantum and quantum-quantum capacity regions.
\newblock {\em Information Theory, IEEE Transactions on}, 54(7): 3091 --3113,
  2008.
\newblock \\
  \texttt{DOI:\,\href{http://dx.doi.org/10.1109/TIT.2008.924665}{10.1109/TIT.2008.924665}}.

\bibitem[YHHSR05]{Yang05}
D.~Yang, M.~Horodecki, R.~Horodecki, and B.~Synak-Radtke.
\newblock Irreversibility for all bound entangled states.
\newblock {\em Phys. Rev. Lett.}, 95: 190501, 2005.
\newblock \\
  \texttt{DOI:\,\href{http://dx.doi.org/10.1103/PhysRevLett.95.190501}{10.1103/PhysRevLett.95.190501}}.

\bibitem[Yur03]{Yura03}
F.~Yura.
\newblock Entanglement cost of three-level antisymmetric states.
\newblock {\em Journal of Physics A: Mathematical and General}, 36(15): L237,
  2003.
\newblock \\
  \texttt{DOI:\,\href{http://dx.doi.org/10.1088/0305-4470/36/15/104}{10.1088/0305-4470/36/15/104}}.

\bibitem[Zur09]{Zurek09}
W.~H. Zurek.
\newblock Quantum {Darwinism}.
\newblock {\em Nature Physics}, 5(3): 181--188, 2009.
\newblock \\
  \texttt{DOI:\,\href{http://dx.doi.org/10.1038/nphys1202}{10.1038/nphys1202}}.

\end{thebibliography}



\cleardoublepage
\phantomsection
\addcontentsline{toc}{chapter}{List of Figures}

\listoffigures






\appendix

\chapter{Entropy - Revisited}\label{ap:entropy}

\section{Von Neumann Entropy}\label{app:vN}

\begin{lemma}\cite[Corollary 5.12 (ii)]{Petz93}\label{lem:petz1}
Let $\omega\in\cS_{\leq}(\cM)$, and $\sigma,\gamma\in\cP^{+}(\cM)$ with $\sigma\geq\gamma$. Then, we have that
\begin{align}
D(\omega\|\gamma)\geq D(\omega\|\sigma)\ .
\end{align}
\end{lemma}

\begin{lemma}\cite[Proposition 5.1]{Petz93}\label{lem:petz2}
Let $\omega\in\cS_{\leq}(\cM)$, $\sigma\in\cP^{+}(\cM)$, and $c>0$. Then, we have that
\begin{align}
D(\omega\|c\cdot\sigma)=D(\omega\|\sigma)+\log\frac{1}{c}\ .
\end{align}
\end{lemma}

It is known that the von Neumann entropy can be approximated in the following sense.

\begin{lemma}\cite[Corollary 5.12 (iv)]{Petz93}\label{lem:vNconv}
Let $\{M_{i}\}_{i\in\mathbb{N}}$ be a sequence of von Neumann subalgebras of $\cM$ such that their union is weakly dense in $\cM$. If $\omega,\sigma\in\cP^{+}(\cM)$, then the increasing sequence $D(\omega_{\cM_{i}}\|\sigma_{\cM_{i}})$ converges to $D(\omega\|\sigma)$, where $\omega_{\cM_{i}},\sigma_{\cM_{i}}$ denote the restriction of $\omega,\sigma$ onto the subalgebra $\cM_{i}$, respectively.
\end{lemma}

The following is a chain rule for the von Neumann entropy, known as the conditional expectation property.

\begin{lemma}\cite[Corollary 5.20]{Petz93}\label{lem:petz3}
Let $\cM_{AB}=\cM_{A}\ot\cM_{B}$, $\omega_{AB}\in\cS(\cM_{AB})$, $\sigma_{A}\in\cS(\cM_{A})$, and $\sigma_{B}\in\cS(\cM_{B})$. Then, we have that
\begin{align}
D(\omega_{AB}\|\sigma_{A}\ot\sigma_{B})=D(\omega_{A}\|\sigma_{A})+D(\omega_{AB}\|\omega_{A}\ot\sigma_{B})\ .
\end{align}
\end{lemma}


\section{Max-Relative Entropy}\label{app:Imax}

The proofs of following two lemmas readily follow from the definition of the max-relative entropy.

\begin{lemma}\label{lem:minmax_elementary1}
Let $\omega\in\cS_{\leq}(\cM)$, and $\sigma,\gamma\in\cP^{+}(\cM)$ with $\sigma\geq\gamma$. Then, we have that
\begin{align}
D_{\max}(\omega\|\gamma)\geq D_{\max}(\omega\|\sigma)\ .
\end{align}
\end{lemma}

\begin{lemma}\label{lem:minmax_elementary2}
Let $\omega\in\cS_{\leq}(\cM)$, $\sigma\in\cP^{+}(\cM)$, and $c>0$. Then, we have that
\begin{align}
D_{\max}(\omega\|c\cdot\sigma)=D_{\max}(\omega\|\sigma)+\log\frac{1}{c}\ .
\end{align}
\end{lemma}

From now on, all systems in this appendix are finite-dimensional. The max-information is monotone under application of local channels.

\begin{lemma}\label{lem:imaxmono}
Let $\rho_{AB}\in\cS_{\leq}(\cH_{AB})$, and let $\cE_{A\ra C}:\cB(\cH_{A})\ra\cB(\cH_{C})$ and $\cF_{B\ra D}:\cB(\cH_{B})\ra\cB(\cH_{D})$ be channels. Then, we have that
\begin{align}
I_{\max}(A:B)_{\rho}\geq I_{\max}(C:D)_{(\cE\ot\cF)(\rho)}\ .
\end{align}
\end{lemma}

\begin{proof}
Let $\tilde{\sigma}_{B}\in\cS(\cH_{B})$. Using the monotonicity of the max-relative entropy under application of channels (Lemma~\ref{lem:maxmono}) we obtain
\begin{align}
I_{\max}(A:B)_{\rho}=D_{\max}(\rho_{AB}\|\rho_{A}\ot\tilde{\sigma}_{B})&\geq D_{\max}(\cE(\rho_{AB})\|\cE_{A}(\rho_{A})\ot\cE_{B}(\tilde{\sigma}_{B}))\notag\\
&\geq\min_{\omega_{B}\in\cS(\cH_{B})}D_{\max}(\cE(\rho_{AB})\|\cE_{A}(\rho_{A})\ot\omega_{B})\notag\\
&=I_{\max}(A:B)_{\cE(\rho)}\ .
\end{align}
\end{proof}

The max-information can be lower and upper bounded in terms of a difference between two entropic quantities.

\begin{lemma}\label{lem:maxbounds}
Let $\rho_{AB}\in\cS_{\leq}(\cH_{AB})$. Then, we have that
\begin{align}
H_{\min}(A)_{\rho}-H_{\min}(A|B)_{\rho}\leq I_{\max}(A:B)_{\rho}\leq H_{R}(A)_{\rho}-H_{\min}(A|B)_{\rho}\ .
\end{align}
\end{lemma}

\begin{proof}
We first prove the left inequality. Let $\tilde{\sigma}_{B}\in\cS(\cH_{B})$ such that $I_{\max}(A:B)_{\rho}=D_{\max}(\rho_{AB}\|\rho_{A}\ot\tilde{\sigma}_{B})=\log\lambda$. That is, $\lambda$ is minimal such that $\lambda\cdot\rho_{A}\ot\tilde{\sigma}_{B}\geq\rho_{AB}$. Furthermore, let $\mu$ be minimal such that $\mu\cdot\lambda_{1}(\rho_{A})\cdot\1_{A}\ot\tilde{\sigma}_{B}\geq\rho_{AB}$. Since $\lambda_{1}(\rho_{A})\cdot\1_{A}\geq\rho_{A}$, we have that $\lambda\geq\mu$. Now, let $\nu$ be minimal such that $\nu\cdot\1_{A}\ot\tilde{\sigma}_{B}\geq\rho_{AB}$. Thus, $\nu=\mu\cdot\lambda_{1}(\rho_{A})$, and we conclude
\begin{align}
I_{\max}(A:B)_{\rho}=\log\lambda\geq\log\mu&=-\log\lambda_{1}(\rho_{A})+\log\nu\notag\\
&=H_{\min}(A)_{\rho}+D_{\max}(\id_{A}\ot\tilde{\sigma}_{B}\|\rho_{AB})\notag\\
&\geq H_{\min}(A)_{\rho}-H_{\min}(A|B)_{\rho}\ .
\end{align}
The right inequality is as follows. Let $\tilde{\sigma}_{B}\in\cS(\cH_{B})$ such that
\begin{align}
H_{\min}(A|B)_{\rho}=-D_{\max}(\rho_{AB}\|\id_{A}\ot\tilde{\sigma}_{B})=-\log\mu\ ,
\end{align}
that is, $\mu$ is minimal such that $\mu\cdot\1_{A}\ot\tilde{\sigma}_{B}\geq\rho_{AB}$. Since multiplication by $\rho_{A}^{0}\ot\1_{B}$ does not affect $\rho_{AB}$ (note that the support of $\rho_{AB}$ is contained in the support of $\rho_{A}\ot\rho_{B}$), we also have $\mu\cdot\rho_{A}^{0}\ot\tilde{\sigma}_{B}\geq\rho_{AB}$. Furthermore, $\rho_{A}\geq\lambda_{\min}(\rho_{A})\cdot\rho_{A}^{0}$, where $\lambda_{\min}(\rho)$ denotes the smallest non-zero eigenvalue of $\rho$. Therefore, $\frac{\mu}{\lambda_{\min}(\rho_{A})}\cdot\rho_{A}\ot\tilde{\sigma}_{B}\geq\rho_{AB}$. Now, let $\lambda$ be minimal such that $\lambda\cdot\rho_{A}\ot\tilde{\sigma}_{B}\geq\rho_{AB}$. Hence, we have $\lambda\leq\mu\cdot\lambda^{-1}_{\min}(\rho_{A})$, and thus
\begin{align}
I_{\max}(A:B)_{\rho}\leq D_{\max}(\rho_{AB}|\rho_{A}\ot\tilde{\sigma}_{B})\leq H_{R}(A)_{\rho}-H_{\min}(A|B)_{\rho}\ .
\end{align}
\end{proof}

The upper bound can be generalized to a version for the smooth max-information.

\begin{lemma}\label{lem:imaxupsmooth}
Let $\eps>0$, and $\rho_{AB}\in\cS(\cH_{AB})$. Then, we have that
\begin{align}
I_{\max}^{\eps}(A:B)_{\rho}\leq H_{\max}^{\eps^{2}/48}(A)_{\rho}-H_{\min}^{\eps^{2}/48}(A|B)_{\rho}-2\log\frac{\eps^{2}}{24}\ .
\end{align}
\end{lemma}

\begin{proof}
By the entropy measure upper bound for the max-information (Lemma~\ref{lem:maxbounds}) we obtain
\begin{align}
I_{\max}^{\eps}(A:B)_{\rho}&\leq\min_{\bar{\rho}_{AB}\in\cB^{\eps}(\rho_{AB})}\big(H_{R}(A)_{\bar{\rho}}-H_{\min}(A|B)_{\bar{\rho}}\big)\notag\\
&\leq\min_{\omega_{AB}\in\cB^{\eps^{2}/48}(\rho_{AB})}\left\{\min_{\Pi_{A}}\big(H_{R}(A)_{\Pi_{A}\omega\Pi_{A}}-H_{\min}(A|B)_{\Pi_{A}\omega\Pi_{A}}\big)\right\}\ ,
\end{align}
where the minimum ranges over all $0\leq\Pi_{A}\leq\1_{A}$ such that $\Pi_{A}\omega_{AB}\Pi_{A}\approx_{\eps/2}\omega_{AB}$. Now we choose $\tilde{\sigma}_{B}\in\cS(\cH_{B})$ such that $H_{\min}(A|B)_{\omega}=-D_{\max}(\id_{A}\ot\tilde{\sigma}_{B}\|\omega_{AB})$, and use Lemma~\ref{lem:hminproj} to obtain
\begin{align}
I_{\max}^{\eps}(A:B)_{\rho}&\leq\min_{\omega_{AB}\in\cB^{\eps^{2}/48}(\rho_{AB})}\left\{\min_{\Pi_{A}}\big(H_{R}(A)_{\Pi_{A}\omega\Pi_{A}}+D_{\max}(\Pi_{A}\omega_{AB}\Pi_{A}\|\id_{A}\ot\tilde{\sigma}_{B})\big)\right\}\notag\\
&\leq\min_{\omega_{AB}\in\cB^{\eps^{2}/48}(\rho_{AB})}\left\{\min_{\Pi_{A}}\big(H_{R}(A)_{\Pi_{A}\omega\Pi_{A}}\big)+D_{\max}(\omega_{AB}\|\id_{A}\ot\tilde{\sigma}_{B})\right\}\notag\\
&=\min_{\omega_{AB}\in\cB^{\eps^{2}/48}(\rho_{AB})}\left\{\min_{\Pi_{A}}\big(H_{R}(A)_{\Pi_{A}\omega\Pi_{A}}\big)-H_{\min}(A|B)_{\omega}\right\}\ ,
\end{align}
where the minimum ranges over all $0\leq\Pi_{A}\leq\1_{A}$ such that $\Pi_{A}\omega_{AB}\Pi_{A}\approx_{\eps/2}\omega_{AB}$. As a next step we choose $\omega_{AB}=\tilde{\omega}_{AB}\in\cB^{\eps^{2}/48}(\rho_{AB})$ such that $H_{\min}^{\eps^{2}/48}(A|B)_{\rho}=H_{\min}(A|B)_{\tilde{\omega}}$. Hence, we get
\begin{align}
I_{\max}^{\eps}(A:B)_{\rho}\leq \min_{\Pi_{A}}\big(H_{R}(A)_{\Pi_{A}\tilde{\omega}\Pi_{A}}\big)-H_{\min}^{\eps^{2}/48}(A|B)_{\rho}\ ,
\end{align}
where the minimum ranges over all $0\leq\Pi_{A}\leq\1_{A}$ such that $\Pi_{A}\tilde{\omega}_{AB}\Pi_{A}\approx_{\eps/2}\tilde{\omega}_{AB}$. Using Lemma~\ref{lem:hrhmaxequiv}, we can choose $0\leq\Pi_{A}\leq\1_{A}$ with $\Pi_{A}\tilde{\omega}_{AB}\Pi_{A}\approx_{\eps/2}\tilde{\omega}_{AB}$ such that
$H_{R}(A)_{\Pi_{A}\tilde{\omega}\Pi_{A}}\leq H_{\max}^{\eps^{2}/24}(A)_{\tilde{\omega}}-2\cdot\log\frac{\eps^{2}}{24}$. From this we finally obtain
\begin{align}
I_{\max}^{\eps}(A:B)_{\rho}&\leq H_{\max}^{\eps^{2}/24}(A)_{\tilde{\omega}}-H_{\min}^{\eps^{2}/48}(A|B)_{\rho}-2\log\frac{\eps^{2}}{24}\notag\\
&\leq H_{\max}^{\eps^{2}/48}(A)_{\rho}-H_{\min}^{\eps^{2}/48}(A|B)_{\rho}-2\log\frac{\eps^{2}}{24}\ .
\end{align}
\end{proof}

The max-information of classical-quantum states can be estimated as follows.

\begin{lemma}\label{lem:maxcq}
Let $\rho_{ABI}\in\cS_{\leq}(\cH_{ABI})$ with $\rho_{ABI}=\sum_{i\in I}p_{i}\rho_{AB}^{i}\ot\proj{i}_{I}$ classical on $I$ with respect to the basis $\{\ket{i}\}_{i\in I}$, and $p_{i}>0$ for $i\in I$. Then, we have that
\begin{align}
I_{\max}(AI:B)_{\rho}\geq\max_{i\in I}I_{\max}(A:B)_{\rho^{i}}\ .
\end{align}
\end{lemma}

\begin{proof}
Let $\tilde{\sigma}_{B}\in\cS(\cH_{B})$ such that $I_{\max}(AI:B)_{\rho}=D_{\max}(\rho_{ABI}\|\rho_{AI}\ot\tilde{\sigma}_{B})=\log\lambda$. That is, $\lambda$ is minimal such that
\begin{align}
\lambda\cdot\sum_{i}p_{i}\cdot\rho_{A}^{i}\ot\tilde{\sigma}_{B}\ot\proj{i}\geq\sum_{i}p_{i}\cdot\rho_{AB}^{i}\ot\proj{i}\ .
\end{align}
Since the $\ket{i}$ are mutually orthogonal and $p_{i}>0$ for $i\in I$, this is equivalent to $\forall i\in I:\lambda\cdot\rho_{A}^{i}\ot\tilde{\sigma}_{B}\geq\rho_{AB}^{i}$. Set $D_{\max}(\rho_{AB}^{i}\|\rho_{A}^{i}\ot\tilde{\sigma}_{B})=\log\lambda_{i}$, i.e., $\lambda_{i}$ is minimal such that $\lambda_{i}\cdot\rho_{A}^{i}\ot\tilde{\sigma}_{B}\geq\rho_{AB}^{i}$. Hence, $\lambda\geq\max_{i\in I}\lambda_{i}$, and therefore we find
\begin{align}
I_{\max}(AI:B)_{\rho} =\log\lambda\geq\max_{i\in I}\lambda_{i}=\max_{i\in I}D_{\max}(\rho_{AB}^{i}\|\rho_{A}^{i}\ot\tilde{\sigma}_{B})
\geq\max_{i\in I}I_{\max}(A:B)_{\rho^{i}}\ .
\end{align}
\end{proof}

From this, we obtain the following corollary about the behavior of the max-information under projective measurements.

\begin{corollary}\label{cor:maxproj}
Let $\rho_{AB}\in\cS_{\leq}(\cH_{AB})$, and let $P=\left\{P_{A}^{i}\right\}_{i\in I}$ be a collection of projectors that describe a projective measurement on $A$. For $\trace\left[P^{i}_{A}\rho_{AB}\right]\neq0$, let $p_{i}=\trace\left[P^{i}_{A}\rho_{AB}\right]$, and $\rho_{AB}^{i}=\frac{1}{p_{i}}\cdot P^{i}_{A}\rho_{AB}P^{i}_{A}$. Then, we have that
\begin{align}
I_{\max}(A:B)_{\rho}\geq\max_{i}I_{\max}(A:B)_{\rho^{i}}\ ,
\end{align}
where the maximum is over all $i$ for which $\rho_{AB}^{i}$ is defined.
\end{corollary}

\begin{proof}
We define the channel $\cE:\cB(\cH_{AB})\mapsto\cB(\cH_{ABI})$ with $\cE(\cdot)=\sum_{i}(P_{A}^{i}(\cdot)P_{A}^{i})\ot\proj{i}_{I}$, where the $\ket{i}$ are mutually orthogonal. Then, the monotonicity of the max-information under local channels (Lemma~\ref{lem:imaxmono}) combined with the preceding lemma about the max-information of classical-quantum states (Lemma~\ref{lem:maxcq}) show that
\begin{align}
I_{\max}(A:B)_{\rho}\geq I_{\max}(A:B)_{\cE(\rho)}\geq\max_{i}I_{\max}(A:B)_{\rho^{i}}\ .
\end{align}
\end{proof}

The following is a bound on the increase of the max-information when an additional subsystem is added.

\begin{lemma}\label{lem:maxdbound}
Let $\eps\geq0$, and $\rho_{ABC}\in\cS(\cH_{ABC})$. Then, we have that
\begin{align}
I_{\max}^{\eps}(A:BC)_{\rho}\leq I_{\max}^{\eps}(A:B)_{\rho}+2\cdot\log|C|\ .
\end{align}
\end{lemma}

\begin{proof}
Let $\tilde{\rho}_{AB}\in\cB^{\eps}(\rho_{AB})$, and $\tilde{\sigma}_{B}\in\cS(\cH_{B})$ be such that
\begin{align}
I_{\max}^{\eps}(A:B)_{\rho}=D_{\max}\left(\tilde{\rho}_{AB}\|\tilde{\rho}_{A}\ot\tilde{\sigma}_{B}\right)=\log\mu\ ,
\end{align}
that is, $\mu\in\bR$ is minimal such that $\mu\cdot\tilde{\rho}_{A}\ot\tilde{\sigma}_{B}\geq\tilde{\rho}_{AB}$. This implies $\mu\cdot\tilde{\rho}_{A}\ot\tilde{\sigma}_{B}\ot\frac{\1_{C}}{|C|}\geq\frac{1}{|C|}\cdot\tilde{\rho}_{AB}\ot\1_{C}$. Furthermore, define $\tilde{\rho}_{ABC}\in\cB^{\eps}(\rho_{ABC})$ such that $\trace_{C}\left[\tilde{\rho}_{ABC}\right]=\tilde{\rho}_{AB}$ (by Uhlmann's theorem such a state exists~\cite{Uhlmann76,Jozsa94}). We have $|C|\cdot\tilde{\rho}_{AB}\ot\1_{C}\geq\tilde{\rho}_{ABC}$, and hence $\mu\cdot\tilde{\rho}_{A}\ot\tilde{\sigma}_{B}\ot\frac{\1_{C}}{|C|}\geq\frac{1}{|C|^{2}}\cdot\tilde{\rho}_{ABC}$. Now let $D_{\max}\big(\tilde{\rho}_{ABC}\|\tilde{\rho}_{A}\ot\tilde{\sigma}_{B}\ot\frac{\1_{C}}{|C|}\big)=\log\lambda$, that is, $\lambda$ is minimal such that $\lambda\cdot\tilde{\rho}_{A}\ot\tilde{\sigma}_{B}\ot\frac{\1_{C}}{|C|}\geq\tilde{\rho}_{ABC}$. Thus, it follows that $\lambda\leq\mu\cdot|C|^{2}$, and from this we get
\begin{align}
I_{\max}^{\eps}(A:BC)_{\rho}&\leq D_{\max}\left(\tilde{\rho}_{ABC}\|\tilde{\rho}_{A}\ot\tilde{\sigma}_{B}\ot\frac{\1_{C}}{|C|}\right)\notag\\
&\leq D_{\max}\left(\tilde{\rho}_{AB}\|\tilde{\rho}_{A}\ot\tilde{\sigma}_{B}\right)+2\cdot\log|C|\notag\\
&=I_{\max}^{\eps}(A:B)_{\rho}+2\cdot\log|C|\ .
\end{align}
\end{proof}

\begin{remark}
In general, there is no dimension upper bound for $D_{\max}(\rho_{AB}\|\rho_{A}\ot\rho_{B})$ as can be seen by the following example. 
Let $\rho_{AB}\in\cV(\cH_{AB})$ with Schmidt-decomposition $\ket{\rho}_{AB}=\sum_{i}\lambda_{i}\ket{i}_{A}\ot\ket{i}_{B}$. Then, we have that
\begin{align}
D_{\max}(\rho_{AB}\|\rho_{A}\ot\rho_{B})=\log\Big(\sum_{i}\lambda_{i}^{-1}\Big)\ ,
\end{align}
where the sum ranges over all $i$ with $\lambda_{i}>0$.
\end{remark}

The following is a strengthening of the bound in Lemma~\ref{lem:maxdbound} when the additional system is classical.

\begin{corollary}\label{cor:maxdbound}
Let $\rho_{ABX}\in\cS(\cH_{ABX})$ be classical on $X$ with respect to the basis $\{\ket{x}\}_{x\in X}$. Then, we have that
\begin{align}
I_{\max}(A:BX)_{\rho}\leq I_{\max}(A:B)_{\rho}+\log|X|\ .
\end{align}
\end{corollary}

\begin{proof}
The proof is immediate from the proof of Lemma~\ref{lem:maxdbound} and noting that for $X$ classical $D_{\max}(\rho_{XAB}\|\id_{X}\ot\rho_{AB})\leq0$~\cite[Lemma 3.1.9]{Renner05}.
\end{proof}

The max-relative entropy is quasi-convex in the following sense.

\begin{lemma}\cite[Lemma 9]{Datta09}\label{lem:maxqconvex}
Let $\rho=\sum_{i\in I}p_{i}\rho_{i}\in\cS_{\leq}(\cH)$, and $\sigma=\sum_{i\in I}p_{i}\sigma_{i}\in\cS_{\leq}(\cH)$ with $\rho_{i},\sigma_{i}\in\cS_{\leq}(\cH)$ for $i\in I$. Then, we have that
\begin{align}
D_{\max}(\rho\|\sigma)\leq\max_{i\in I}D_{\max}(\rho_{i}\|\sigma_{i})\ .
\end{align}
\end{lemma}

From this we find the following quasi-convexity type lemma for the smooth max-information.

\begin{lemma}\label{lem:imaxqconvex}
Let $\eps\geq0$, and $\rho_{AB}=\sum_{i\in I}p_{i}\rho_{AB}^{i}\in\cS_{\leq}(\cH_{AB})$ with $\rho_{AB}^{i}\in\cS_{\leq}(\cH_{AB})$ for $i\in I$. Then, we have that
\begin{align}
I_{\max}^{\eps}(A:B)_{\rho}\leq\max_{i\in I}I_{\max}^{\eps}(A:B)_{\rho^{i}}+\log|I|\ .
\end{align}
\end{lemma}

\begin{proof}
Let $\tilde{\rho}_{AB}^{i}\in\cB^{\epsilon}(\rho_{AB}^{i})$, and $\tilde{\sigma}^{i}_{B}\in\cS(\cH_{B})$ for $i\in I$. Using the quasi-convexity of the max-relative entropy (Lemma~\ref{lem:maxqconvex}) we obtain
\begin{align}
\max_{i\in I}I_{\max}^{\eps}(A:B)_{\rho^{i}}&+\log|I|=\max_{i\in I}D_{\max}\left(\tilde{\rho}_{AB}^{i}\|\tilde{\rho}_{A}^{i}\ot\tilde{\sigma}_{B}^{i}\right)+\log|I|\notag\\
&\geq D_{\max}\left(\sum_{i\in I}p_{i}\tilde{\rho}_{AB}^{i}\|\sum_{i\in I}p_{i}\tilde{\rho}_{A}^{i}\ot\tilde{\sigma}_{B}^{i}\right)+\log|I|\notag\\
&=\log\min\left\{\lambda\in\mathbb{R}:\sum_{i\in I}p_{i}\tilde{\rho}_{AB}^{i}\leq\lambda\cdot\sum_{i\in I}p_{i}\tilde{\rho}_{A}^{i}\ot\tilde{\sigma}_{B}^{i}\right\}+\log|I|\notag\\
&\stackrel{\mathrm{(i)}}{\geq}\log\min\left\{\mu\in\mathbb{R}:\sum_{i\in I}p_{i}\tilde{\rho}_{AB}^{i}\leq\mu\cdot\sum_{i\in I}p_{i}\tilde{\rho}_{A}^{i}\ot\sum_{j\in I}\tilde{\sigma}_{B}^{j}\right\}+\log|I|\notag\\
&=\log\min\left\{\mu\in\mathbb{R}:\sum_{i\in I}p_{i}\tilde{\rho}_{AB}^{i}\leq\mu\cdot\sum_{i\in I}p_{i}\tilde{\rho}_{A}^{i}\ot\sum_{j\in I}\frac{1}{l}\cdot\tilde{\sigma}_{B}^{j}\right\}\ ,
\end{align}
where step (i) holds because $\left(\sum_{i\in I}p_{i}\tilde{\rho}_{A}^{i}\right)\ot\left(\sum_{j\in I}\tilde{\sigma}_{B}^{j}\geq\right)\sum_{i\in I}p_{i}\tilde{\rho}_{A}^{i}\ot\tilde{\sigma}_{B}^{i}$. Now set $\tilde{\sigma}_{B}=\sum_{j\in I}\frac{1}{l}\cdot\tilde{\sigma}_{B}^{j}$ and $\tilde{\rho}_{AB}=\sum_{i\in I}p_{i}\tilde{\rho}_{AB}^{i}$. By the convexity of the purified distance in its arguments (Lemma~\ref{lem:pdconvex}) we obtain
\begin{align}
\max_{i\in I}I_{\max}^{\eps}(A:B)_{\rho^{i}}&+\log|I|\geq\log\min\left\{\mu\in\mathbb{R}:\tilde{\rho}_{AB}\leq\mu\cdot\tilde{\rho}_{A}\ot\tilde{\sigma}_{B}\right\}\notag\\
&\geq\min_{\bar{\rho}_{AB}\in\cB^{\eps}(\rho_{AB})}\min_{\sigma_{B}\in\cS(\cH_{B})}\log\min\left\{\nu\in\mathbb{R}:\bar{\rho}_{AB}\leq\nu\cdot\bar{\rho}_{A}\ot\sigma_{B}\right\}\notag\\
&=I_{\max}^{\eps}(A:B)_{\rho}\ .
\end{align}
\end{proof}

The asymptotic equipartition property for the smooth max-information is as follows.

\begin{lemma}\label{lem:aepimax}
Let $\eps>0$, $n\geq2\cdot(1-\eps^{2})$, and $\rho_{AB}\in\cS(\cH_{AB})$. Then, we have that
\begin{align}
\frac{1}{n}I_{\max}^{\eps}(A:B)_{\rho^{\ot n}}\leq I(A:B)_{\rho}+\frac{\xi(\eps)}{\sqrt{n}}-\frac{2}{n}\cdot\log\frac{\eps^{2}}{24}\ ,
\end{align}
where $\xi(\eps)=8\cdot\sqrt{13-4\cdot\log\eps}\cdot(2+\frac{1}{2}\cdot\log|A|)$.
\end{lemma}

\begin{proof}
Using the entropy measure upper bound for the smooth max-information (Lemma \ref{lem:imaxupsmooth}) together with the asymptotic equipartition property for the smooth conditional min- and max-entropies (Lemma~\ref{lem:aepminmax}) we obtain
\begin{align}
\frac{1}{n}I_{\max}^{\eps}(A:B)_{\rho^{\ot n}}&\leq\frac{1}{n}H_{\max}^{\eps^{2}/48}(A)_{\rho^{\ot n}}-\frac{1}{n}H_{\min}^{\eps^{2}/48}(A|B)_{\rho^{\ot n}}-\frac{2}{n}\cdot\log\frac{\eps^{2}}{24}\notag\\
&\leq H(A)_{\rho}-H(A|B)_{\rho}-\frac{2}{n}\cdot\log\frac{\eps^{2}}{24}\notag\\
&+2\cdot\frac{4}{\sqrt{n}}\sqrt{1-\log\left(\frac{\eps^{2}}{48}\right)^{2}}\cdot\log(2+\frac{1}{2}\cdot\log|A|)\notag\\
&\leq I(A:B)_{\rho}+\frac{\xi(\eps)}{\sqrt{n}}-\frac{2}{n}\cdot\log\frac{\eps^{2}}{24}\ .
\end{align}
\end{proof}

The following two lemmas are here for purely technical reasons.

\begin{lemma}\label{lem:nonnegative}
Let $\rho_{AB}\in\cS(\cH_{AB})$, and $\Pi_{A}\in\cP^{+}(\cH_{A})$ with $\Pi_{A}\leq\id_{A}$. Then, we have that
\begin{align}
I_{\max}(A:B)_{\rho}\geq I_{\max}(A:B)_{\Pi\rho\Pi}\ .
\end{align}
\end{lemma}

\begin{proof}
Let $\sigma_{B}\in\cS(\cH_{B})$, and $\lambda\in\mathbb{R}$ be such that
\begin{align}
I_{\max}(A:B)_{\rho}=D_{\max}(\rho_{AB}\|\rho_{A}\ot\sigma_{B})=\log\lambda\ .
\end{align}
Then, we have that $\lambda\cdot\rho_{A}\ot\sigma_{B}\geq\rho_{AB}$, and with this
\begin{align}
\lambda\cdot\rho_{A}\ot\sigma_{B}\geq\lambda\cdot\Pi_{A}\rho_{A}\Pi_{A}\ot\sigma_{B}\geq\Pi_{A}\rho_{AB}\Pi_{A}\ .
\end{align}
Hence, we have $\log\lambda\geq D_{\max}(\Pi_{A}\rho_{AB}\Pi_{A}\|\Pi_{A}\rho_{A}\Pi_{A}\ot\sigma_{B})\geq I_{\max}(A:B)_{\Pi\rho\Pi}$.
\end{proof}

\begin{lemma}\label{lem:hminproj}
Let $\rho_{AB}\in\cS_{\leq}(\cH_{AB})$, $\sigma_{B},\omega_{B}\in\cS(\cH_{B})$, and $\Pi_{AB}\in\cP^{+}(\cH_{AB})$ with $\Pi_{AB}\leq\1_{AB}$ and $\1_{A}\ot\omega_{B}-\Pi_{AB}(\1_{A}\ot\sigma_{B})\Pi_{AB}\geq0$. Then, we have that
\begin{align}
D_{\max}(\id_{A}\ot\omega_{B}\|\Pi_{AB}\rho_{AB}\Pi_{AB})\leq D_{\max}(\id_{A}\ot\sigma_{B}\|\rho_{AB})\ .
\end{align}
\end{lemma}

\begin{proof}
Let $\lambda$ be minimal such that $\lambda\cdot\1_{A}\ot\sigma_{B}-\rho_{AB}\geq0$. Hence, we have $\lambda\cdot\Pi_{AB}(\1_{A}\ot\sigma_{B})\Pi_{AB}-\Pi_{AB}\rho_{AB}\Pi_{AB}\geq0$. Using $\1_{A}\ot\omega_{B}-\Pi_{AB}(\1_{A}\ot\sigma_{B})\Pi_{AB}\geq0$, we obtain
\begin{align}
\lambda\cdot\1_{A}\ot\omega_{B}-\Pi_{AB}\rho_{AB}\Pi_{AB}&=\lambda\cdot(\1_{A}\ot\omega_{B}-\Pi_{AB}(\1_{A}\ot\sigma_{B})\Pi_{AB})\notag\\
&+\lambda\cdot\Pi_{AB}(\1_{A}\ot\sigma_{B})\Pi_{AB}-\Pi_{AB}\rho_{AB}\Pi_{AB}\geq0\ .
\end{align}
The claim then follows by the definition of the max-relative entropy (Definition~\ref{def:maxrelative}).
\end{proof}


\section{Conditional Min- and Max-Entropy}\label{app:Hminmax}

In this appendix, all systems are finite-dimensional. The conditional min-entropy is non-negative for separable states.

\begin{lemma}\label{lem:Hmin_sep}
Let $\rho_{AB}\in\cS(\cH_{AB})$ be separable between $A$ and $B$. Then, we have that
\begin{align}
H_{\min}(A|B)_{\rho}\geq0\ .
\end{align}
\end{lemma}

\begin{proof}
Immediate from the definition of the conditional min-entropy (Definition~\ref{def:Hmin}).
\end{proof}

We have the following chain rule for the conditional min-entropy.

\begin{lemma}\cite[Lemma 3.1.10]{Renner05}\label{lem:minlower}
Let $\rho_{AB}\in\cS_{\leq}(\cH_{AB})$. Then, we have that
\begin{align}
H_{\min}(A|B)_{\rho}\geq H_{\min}(AB)_{\rho}-H_{0}(B)_{\rho}\ .
\end{align}
\end{lemma}

The following is a quasi-convexity property of the smooth max-entropy.

\begin{lemma}\label{lem:maxquasicon}
Let $\eps\geq0$, and $\rho_{A}=\sum_{x}p_{x}\rho_{A}^{x}\in\cS(\cH_{A})$ with $p_{x}>0$ for all $x$. Then, we have that
\begin{align}
H_{\max}^{\eps}(A)_{\rho}\leq\max_{x}H_{\max}^{\eps}(A)_{\rho^{x}}+\log|X|\ .
\end{align}
\end{lemma}

\begin{proof}
Let $\rhob_{A}^{x}\in\cB^{\eps}(\rho_{A}^{x})$ such that $H_{\max}^{\eps}(A)_{\rho^{x}}=H_{\max}(A)_{\rhob^{x}}$, and define $\rhob_{AX}=\sum_{x}p_{x}\cdot\rhob_{A}^{x}\ot\proj{x}_{X}$. Note that by Lemma~\ref{lem:pdconvex}, $\rhob_{A}\in\cB^{\eps}(\rho_{A})$. Then, it follows by the definition of the smooth max-entropy as well as ~\cite[Proposition 5.10]{Tomamichel12} and~\cite[Proposition 4.6]{Tomamichel12} that
\begin{align}
H_{\max}^{\eps}(A)_{\rho}&\leq H_{\max}(A)_{\rhob}\leq H_{\max}(A|X)_{\rhob}+\log|X|\notag\\
&=\log\big(\sum_{x}p_{x}\cdot2^{H_{\max}(A)_{\rhob^{x}}}\big)+\log|X|\notag\\
&\leq\max_{x}H_{\max}^{\eps}(A)_{\rho^{x}}+\log|X|\ .
\end{align}
\end{proof}

The following shows that the entropies $H_{R}$ and $H_{0}$ can be smoothed by a projection, and that their smoothed versions are equivalent to $H_{\max}$.

\begin{lemma}\label{lem:hrhmaxequiv}
Let $\eps>0$, and $\rho_{A}\in\cS_{\leq}(\cH_{A})$. Then, there exists $\Pi_{A}\in\cP^{+}(\cH_{A})$ with $\Pi_{A}\leq\1_{A}$ such that
\begin{align}
H_{\max}^{\eps}(A)_{\rho}\geq H_{R}(A)_{\Pi\rho\Pi}-2\log\frac{1}{\eps}\geq H_{0}(A)_{\Pi\rho\Pi}-2\log\frac{1}{\eps}\ ,
\end{align}
and $\Pi_{A}\rho_{A}\Pi_{A}\in\cB^{\sqrt{6\eps}}(\rho_{A})$.
\end{lemma}

\begin{proof}
The second inequality is immediate since $H_{R}(A)_{\sigma}\geq H_{0}(A)_{\sigma}$ for all $\sigma_{A}\in\cS_{\leq}(\cH_{A})$. For the first inequality, we note that by~\cite[Lemma A.15]{Berta10} we have for $\eps>0$ and $\rho_{A}\in\cS_{\leq}(\cH_{A})$, that there exists $\Pi_{A}\in\cP^{+}(\cH_{A})$ with $\Pi_{A}\leq\1_{A}$ such that $\trace\left[\left(\1_{A}-\Pi_{A}^{2}\right)\rho_{A}\right]\leq3\eps$, and
\begin{align}
H_{\max}^{\eps}(A)_{\rho}\geq H_{R}(A)_{\Pi\rho\Pi}-2\log\frac{1}{\eps}\ .
\end{align}
Then, Lemma~\ref{lem:projpurd} shows that $\trace\left[\left(\1_{A}-\Pi_{A}^{2}\right)\rho_{A}\right]\leq3\eps$ implies $\rho_{A}\approx_{\sqrt{6\eps}}\Pi_{A}\rho_{A}\Pi_{A}$.
\end{proof}


\section{Conditional R\'enyi 2-Entropy}\label{app:H2}

In this appendix, all systems are finite-dimensional. The conditional R\'enyi 2-entropy is monotone under application of channels on the conditioning system.

\begin{lemma}\cite{Tomamichel13}\label{lem:h2data}
Let $\rho_{AB}\in\cS(\cH_{AB})$, and let $\cE_{B\rightarrow C}:\cB(\cH_{B})\ra\cB(\cH_{C})$ be a channel. Then, we have that
\begin{align}
H_{2}(A|B)_{\rho}\leq H_{2}(A|C)_{\cE(\rho)}\ .
\end{align}
\end{lemma}

The conditional R\'enyi 2-entropy is equivalent to the smooth conditional min-entropy in the following sense.

\begin{lemma}\label{lem:h2hmin_equiv}
Let $\rho_{AB}\in\cS(\cH_{AB})$, and $\eps>0$. Then, we have that
\begin{align}
H_{\min}^{\eps}(A|B)_{\rho}+\log\frac{2}{\eps^{2}}\geq H_{2}(A|B)_{\rho}\geq H_{\min}(A|B)_{\rho}\ .
\end{align}
\end{lemma}

\begin{proof}
The proof of the second inequality is immediate by invoking the operational form of the conditional min-entropy (Proposition~\ref{prop:HminDualForm}) and the conditional R\'enyi 2-entropy (Proposition~\ref{prop:h2operational}), since $\Lambda^{\pg}_{B\ra A'}$ in~\eqref{eq:h2operational} is a particular channel, and the conditional min-entropy involves an optimization over all channels $\Lambda_{B\ra A'}$ in~\eqref{eq:HminDualForm}. For the first inequality, the exact proof technique of~\cite[Theorem 7]{Tomamichel09} together with the data-processing inequality for the conditional R\'enyi 2-entropy (Lemma~\ref{lem:h2data}) leads to the claim~\cite{Tomamichel13}.
\end{proof}

The conditional R\'enyi 2-entropy is upper bounded by the conditional von Neumann entropy.

\begin{lemma}\label{lem:h2vN_bounds}
Let $\rho_{AB}\in\cS(\cH_{AB})$. Then, we have that
\begin{align}
H_{2}(A|B)_{\rho}\leq H(A|B)_{\rho}\ .
\end{align}
\end{lemma}

\begin{proof}
The idea for this proof is from~\cite{Tomamichel13}. We evaluate the upper bound on the conditional R\'enyi 2-entropy from Lemma \ref{lem:h2hmin_equiv} for $\rho_{AB}^{\ot n}$,
\begin{align}
H_{2}(A|B)_{\rho^{\ot n}}&\leq H_{\min}^{\eps}(A|B)_{\rho^{\ot n}}+\log\frac{2}{\eps^{2}}\leq\sup_{\rhob^{n}\in\cB^{\eps}(\rho^{\ot n})}H(A|B)_{\rhob^{n}}+\log\frac{2}{\eps^{2}}\notag\\
&\leq H(A|B)_{\rho^{\ot n}}+8\eps\cdot n\cdot\log|A|+2\cdot h(2\eps)+\log\frac{2}{\eps^{2}}\ ,
\end{align}
where we have used the definition of the smooth conditional min-entropy (Definition~\ref{def:smooth_entropy}), that the conditional min-entropy is upper bounded by the conditional von Neumann entropy~\cite[Lemma 2]{Tomamichel09}, and the continuity of the conditional von Neumann entropy (Lemma~\ref{lem:fannes}). Since the conditional R\'enyi 2-entropy and the conditional von Neumann entropy are both additive (by inspection), we arrive at the claim by letting $n\ra\infty$ and $\eps\ra0$.
\end{proof}

The following lemma is from the collaboration~\cite{Berta13_3} and written by Patrick Coles.

\begin{lemma}\label{lem:onebasis}
Let $\rho_{AB}\in\cS(\cH_{AB})$, and $\{\ket{k}\}_{k=1}^{|A|}$ be an orthonormal basis of $\cH_{A}$. Then, we have that
\begin{align}\label{eqn21}
P^{\pg}_{\guess}(K|B)_{\rho}\geq F^{\pg}(A|B)_{\rho}\ ,
\end{align}
where $\rho_{KB}=\sum_{k=1}^{|A|}(\proj{k}\ot\id_{B})\rho_{AB}(\proj{k}\ot\id_{B})$.
\end{lemma}

\begin{proof}
We calculate
\begin{align}\label{eqn22}
P^{\pg}_{\guess}(K|B)&=\trace\big[\rho_{KB}\rho_B^{-1/2}\rho_{KB}\rho_B^{-1/2}\big]=\trace\big[\rho_{AB}\rho_B^{-1/2}\rho_{KB}\rho_B^{-1/2}\big]\notag\\
&=|A|\cdot\trace\big[\Phi_{AA'}(\cI_{A}\ot\Lambda^{\pg}_{B\ra A'})(\rho_{KB})\big]\notag\\
&=F\left(\Phi_{AA'},|A|\cdot(\cI_{A}\ot\Lambda^{\pg}_{B\rightarrow A'})(\rho_{KB})\right) \notag \\
&\geq F\left(\Phi_{AA'}, (\cI_{A}\ot\Lambda^{\pg}_{B\ra A'})(\rho_{AB})\right)=F^{\pg}(A|B)_{\rho}\ .
\end{align}
where $\Phi_{AA'}$ denotes the maximally entangled state, and the inequality step invokes the property that the fidelity decreases upon decreasing one of its arguments. Hence, the result would follow from
\begin{align}\label{eqn32049}
|A|\cdot(\cI_{A}\ot\Lambda^{\pg}_{B\ra A'})(\rho_{KB})\geq(\cI_{A}\ot\Lambda^{\pg}_{B\ra A'})(\rho_{AB})\ .
\end{align}
To show~\eqref{eqn32049}, we denote the positive operator $\gamma_{AA'}=(\cI\ot\Lambda^{\pg}_{B\ra A'})(\rho_{AB} )$, and note that the measurement in $K$ on the $A$-system commutes with $\cI\ot\Lambda^{\pg}_{B\ra A'}$. We get
\begin{align}
&|A|\cdot(\cI_{A}\ot \Lambda^{\pg}_{B\ra A'})(\rho_{KB})- (\cI_{A}\ot\Lambda^{\pg}_{B\ra A'})(\rho_{AB})\notag \\ 
&=|A|\cdot\sum_{k}\big(\proj{k}\ot\id\big)\gamma_{AA'}\big(\proj{k}\ot\id\big)-\sum_{k,k'}\big(\proj{k}\ot\id\big)\gamma_{AA'}\big(\proj{k'}\ot\id\big)\notag\\
&=(|A|-1)\Big(\sum_{k}\big(\proj{k}\ot\id\big)\gamma_{AA'}\big(\proj{k}\ot\id\big)\notag\\
&-\frac{1}{|A|-1}\cdot\sum_{k,k'\neq k}\big(\proj{k}\ot\id\big)\gamma_{AA'}\big(\proj{k'}\ot\id\big)\Big)\notag\\
&=\big(|A|-1\big)\cdot(\cF\ot\cI)(\gamma_{AA'})\ ,
\end{align}
where we set in the last line
\begin{align}
\cF(\cdot)=\frac{1}{|A|-1}\sum_{m=1}^{|A|-1} Z^{m} (\cdot) (Z^{m})^{\dagger},\quad Z = \sum_{k=0}^{|A|-1} \omega^{k} \proj{k},\quad\omega = e^{2\pi i /|A|}\ .
\end{align}
Since $\cF$ is a channel, the claim follows.
\end{proof}


\section{Conditional Zero-Entropy}\label{app:H0}

In this appendix, all systems are finite-dimensional. The conditional zero-entropy is lower bounded by the conditional von Neumann entropy.

\begin{lemma}\cite[Lemma 10]{Datta09}\label{lem:H0vN}
Let $\rho_{AB}\in\cS_{\leq}(\cH_{AB})$. Then, we have that
\begin{align}
H_{0}(A|B)_{\rho}\geq H(A|B)_{\rho}\ .
\end{align}
\end{lemma}

The following is a quasi-convexity property of the zero-entropy.

\begin{lemma}\label{lem:logn}
Let $\rho_{A}=\sum_{j=1}^{N}p_{j}\rho_{A}^{j}$ with $\rho_{A}^{j}\in\cP^{+}(\cH_{A})$ and $p_{j}>0$ for $j=1,\ldots,N$. Then, we have that
\begin{align}
H_{0}(A)_{\rho}\leq\max_{j}H_{0}(A)_{\rho^{j}}+\log N\ .
\end{align}
\end{lemma}

\begin{proof}
We have $\rank(M+N)\leq\rank(M)+\rank(N)$ for $M,N\in\cB(\cH_{A})$~\cite[Proposition 0.4.5]{Horn85}, and hence we get
\begin{align}
H_{0}(A)_{\rho}&=\log\Big(\rk\Big(\sum_{j=1}^{N}p_{j}\rho_{A}^{j}\Big)\Big)\leq\log\Big(\sum_{j=1}^{N}\rk\big(p_{j}\rho_{A}^{j}\big)\Big)=\log\Big(\sum_{j=1}^{N}\rk\big(\rho_{A}^{j}\big)\Big)\notag\\
&\leq\log\big(N\cdot\max_{j}\rk\big(\rho_{A}^{j}\big)\big)=\max_{j}H_{0}(A)_{\rho^{j}}+\log N\ .
\end{align}
\end{proof}

This can be generalized to a quasi-convexity property of the smooth conditional zero-entropy.

\begin{lemma}\label{lem:h0qconvex}
Let $\eps\geq0$, $\rho_{AK}=\sum_{j=1}^{N}p_{j}\rho_{AK}^{j}\in\cS(\cH_{AK})$, and $p_{j}>0$ as well as $\rho_{AK}^{j}\in\cP^{+}(\cH_{AK})$ classical on $K$ with respect to the basis $\{\ket{k}\}_{k\in K}$ for $j=1,\ldots,N$. Then, we have that
\begin{align}
H_{0}^{\eps}(A|K)_{\rho}\leq\max_{j}H_{0}^{\eps}(A|K)_{\rho^{j}}+\log N\ .
\end{align}
\end{lemma}

\begin{proof}
Let $\bar{\rho}_{AK}^{j}\in\cB^{\eps}_{\qc}(\rho_{AK}^{j})$ be such that $H_{0}^{\eps}(A|K)_{\rho^{j}}=H_{0}(A|K)_{\bar{\rho}^{j}}$ for $j=1,\ldots,N$. Now, define $\bar{\rho}_{AK}=\sum_{j=1}^{N}p_{j}\bar{\rho}_{AK}^{j}$ and note that $\bar{\rho}_{AK}\in\cB^{\eps}_{\qc}(\rho_{AK})$. Using the definition of the smooth conditional zero-entropy and its form on quantum-classical states (Lemma~\ref{lem:h0class}), it follows that
\begin{align}
H_{0}^{\eps}(A|K)_{\rho}\leq H_{0}(A|K)_{\bar{\rho}}=\max_{k}H_{0}(A)_{\sum_{j}p_{j}\bar{\rho}^{j,k}}\ .
\end{align}
Using the preceding lemma (Lemma~\ref{lem:logn}), and again the structure of quantum-classical states (Lemma~\ref{lem:h0class}), we conclude
\begin{align}
\max_{k}H_{0}(A)_{\sum_{j}p_{j}\bar{\rho}_{A}^{j,k}}&\leq\max_{k}\max_{j}H_{0}(A)_{\bar{\rho}^{j,k}}+\log N=\max_{j}H_{0}(A|K)_{\bar{\rho}^{j}}+\log N\notag\\
&=\max_{j}H_{0}^{\eps}(A|K)_{\rho^{j}}+\log N\ .
\end{align}
\end{proof}

\begin{lemma}\label{lem:smoothing}
Let $\eps>0$, and let $\rho_{AK}\in\cS(\cH_{AK})$ be classical on $K$ with respect to the basis $\{\ket{k}\}_{k\in K}$. Then, the smoothing in $H_{0}^{\eps}(A|K)_{\rho}$ can without lost of generality be restricted to states that commute with $\rho_{AK}$.
\end{lemma}

\begin{proof}
We have $\rho_{AK}=\sum_{k\in K}\rho_{A}^{k}\ot\proj{k}_{K}$. The crucial step is to see that for every $\sigma_{AK}=\sum_{k}\sigma_{A}^{k}\ot\proj{k}_{K}\in\cB^{\eps}_{\qc}(\rho_{AB})$, there exists a unitary $U_{AK}=\sum_{k}U_{A}^{k}\ot\proj{k}_{K}$ such that $U_{AK}\sigma_{AK}U_{AK}^{\dagger}\in\cB^{\eps}_{\qc}(\rho_{AK})$ and $[U_{AK}\sigma_{AK}U_{AK}^{\dagger},\rho_{AK}]=0$. For this, just choose $U_{A}^{k}$ to be the unitary that maps the eigenbasis of $\sigma_{A}^{k}$ to the eigenbasis of $\rho_{A}^{k}$. Therefore, $[U_{AK}\sigma_{AK}U_{AK}^{\dagger},\rho_{AK}]=0$, and furthermore by Lemma~\ref{lem:eigenvalue},
\begin{align}
\eps\geq\|\rho_{AK}-\sigma_{AK}\|_{1}&=\sum_{k}\|\rho_{A}^{k}-\sigma_{A}^{k}\|_{1}\geq\sum_{k}\|P_{A}^{k}-Q_{A}^{k}\|_{1}=\sum_{k}\|\rho_{A}^{k}-U_{A}^{k}\sigma_{A}^{k}(U_{A}^{k})^{\dagger}\|_{1}\notag\\
&=\|\rho_{AK}-U_{AK}\sigma_{AK}U_{AK}^{\dagger}\|_{1}\ ,
\end{align}
where $P_{A}^{k},Q_{A}^{k}$ denote the eigenvalue distributions of $\rho_{A}^{k},\sigma_{A}^{k}$, respectively.
\end{proof}

The definition of the smooth conditional zero-entropy can be specialized canonically to classical probability distributions.

\begin{definition}\label{def:classical}
Let $\eps\geq0$, and let $X$ and $Y$ be random variables with range $\cX$ and $\cY$ respectively, and joint probability distribution $P_{XY}$. The conditional zero-entropy of $X$ given $Y$ is defined as
\begin{align}
H_{0}(X|Y)_{P}=\max_{y\in\cY}\log|\supp\left(P_{X}^{y}\right)|\ ,
\end{align}
where $P_{X}^{y}$ denotes the function $P_{X}^{y}:x\mapsto P_{XY}(x,y)$. The smooth conditional zero-entropy of $X$ given $Y$ is defined as
\begin{align}
H_{0}^{\eps}(X|Y)_{P}=\inf_{\bar{P}_{XY}\in\cB^{\eps}_{c}(P_{XY})}H_{0}(X|Y)_{\bar{P}}\ ,
\end{align}
where $\cB^{\eps}_{\mathrm{c}}(P_{XY})$ denotes the set of non-negative linear functions $\bar{P}_{XY}:\cX\times\cY\rightarrow\mathbb{R}^{+}$ such that $\|P_{XY}-\bar{P}_{XY}\|_{1}\leq\eps$.
\end{definition}

The following is an entropic formulation of the classical asymptotic equipartition property.

\begin{lemma}\cite[Theorem 1]{Holenstein11}\label{lem:classicalaep}
Let $X$ and $Y$ be random variables with range $\cX$ and $\cY$ respectively, and joint probability distribution $P_{XY}$. Furthermore, let $\eps>0$, $n\geq1$, and let $P_{X^{n}Y^{n}}^{n}= P_{X_{1}Y_{1}}\times\ldots\times P_{X_{n}Y_{n}}$ be the $n$-fold product probability distribution over $\cX^{n}\times\cY^{n}$. Then, we have that
\begin{align}
\frac{1}{n}H_{0}^{\eps}(X^{n}|Y^{n})_{P^{n}}\leq H(X|Y)_{P}+\frac{\log\left(|X|+3\right)\cdot\sqrt{\log\frac{1}{\eps^{2}}}}{\sqrt{n}}\ .
\end{align}
\end{lemma}

This can be generalized to the following quantum-classical asymptotic equipartition property.

\begin{lemma}\label{lem:h0aep}
Let $\eps>0$, $n\geq1$, $\rho_{AK}\in\cS(\cH_{AK})$ classical on $K$ with respect to the basis $\{\ket{k}\}_{k\in K}$. Then, we have that
\begin{align}
\frac{1}{n}H_{0}^{\eps}(A|K)_{\rho^{\ot n}}\leq H(A|K)_{\rho}+\frac{\log\left(|A|+3\right)\cdot\sqrt{\log\frac{1}{\eps^{2}}}}{\sqrt{n}}\ .
\end{align}
\end{lemma}

\begin{proof}
The basic idea is that by Lemma~\ref{lem:smoothing}, the smoothing of the conditional zero-entropy can be restricted to states that commute with the initial state, and hence all states that appear are diagonal in the same basis. Working in this basis, this then allows us to use the classical asymptotic equipartition property (Lemma~\ref{lem:classicalaep}). In more detail, we calculate
\begin{align}
\frac{1}{n}H_{0}^{\eps}(A|K)_{\rho^{\ot n}}=\frac{1}{n}\min_{\bar{\rho}^{n}_{AK}\in\cB^{\eps}_{\qc}(\rho_{AB}^{\ot n})}H_{0}(A|K)_{\bar{\rho}^{n}}=\frac{1}{n}\min_{\bar{P}_{AK}^{n}\in\cB^{\eps}_{\mathrm{c}}(P_{AK}^{n})}H_{0}(A|K)_{\bar{P}^{n}}\ ,
\end{align}
where the second equality is due to Lemma~\ref{lem:smoothing}, $P_{AK}^{n}$ is the eigenvalue distribution of $\rho_{AK}^{\ot n}$, and $\cB^{\eps}_{\mathrm{c}}(\cdot)$ is as in Definition~\ref{def:classical}. Moreover,  we conclude by the definition of the classical smooth conditional zero-entropy (Definition~\ref{def:classical}), and the classical asymptotic equipartition property (Lemma~\ref{lem:classicalaep}) that
\begin{align}
\frac{1}{n}\min_{\bar{P}_{AK}^{n}\in\cB^{\eps}_{\mathrm{c}}(P_{AK}^{n})}H_{0}(A|K)_{\bar{P}^{n}}=\frac{1}{n}H_{0}^{\eps}(A|K)_{P^{n}}&\leq H(A|K)_{P}+\frac{\log\left(|A|+3\right)\cdot\sqrt{\log\frac{1}{\eps^{2}}}}{\sqrt{n}}\notag\\
&=H(A|K)_{\rho}+\frac{\log\left(|A|+3\right)\cdot\sqrt{\log\frac{1}{\eps^{2}}}}{\sqrt{n}}\ ,
\end{align}
where $P_{AK}$ denotes the eigenvalue distribution of $\rho_{AK}$.
\end{proof}


\chapter{The Post-Selection Technique}\label{ap:postselect}

In this appendix, all systems are finite-dimensional. The following proposition lies at the heart of the post-selection technique for quantum channels.

\begin{proposition}\cite{Christandl09}\label{prop:postselect}
Let $\eps>0$, and let $\cE^{n}_{A}$ and $\cF^{n}_{A}$ be channels from $\cB(\cH_{A}^{\ot n})$ to $\cB(\cH_{B})$. If there exists a channel $K_{\pi}$ for any permutation $\pi$ such that $(\cE^{n}_{A}-\cF^{n}_{A})\circ\pi=K_{\pi}\circ(\cE^{n}_{A}-\cF^{n}_{A})$, then $\|\cE^{n}_{A}-\cF^{n}_{A}\|_{\Diamond}\leq\eps$ whenever
\begin{align}
\left\|((\cE^{n}_{A}-\cF^{n}_{A})\ot\cI_{RR'})(\zeta^{n}_{ARR'})\right\|_{1}\leq\eps(n+1)^{-(|A|^{2}-1)}\ ,
\end{align}
where $\zeta^{n}_{ARR'}$ is a purification of the de Finetti state
\begin{align}\label{eq:definetti}
\zeta_{AR}^{n}=\int\sigma_{AR}^{\ot n}\;d(\sigma_{AR})\ ,
\end{align}
with $\sigma_{AR}\in\cV(\cH_{AR})$, $\cH_{R}\cong\cH_{A}$ and $d(\cdot)$ the measure on the normalized pure states on $AR$ induced by the Haar measure on the unitary group acting on $AR$, normalized to $\int d(\cdot)=1$. Furthermore, we can assume without loss of generality that
\begin{align}
|R'|\leq(n+1)^{|A|^{2}-1}\ .
\end{align}
\end{proposition}

\begin{theorem}\cite[Carath\'{e}odory]{Gruber93}\label{thm:cara}
Let $d\in\mathbb{N}$, and $x$ be a point that lies in the convex hull of a set $P$ of points in $\mathbb{R}^{d}$. Then, there exists a subset $P'$ of $P$ consisting of $d+1$ or fewer point such that $x$ lies in the convex hull of $P'$.
\end{theorem}

A straightforward application of Carath\'eodory's theorem gives the following.

\begin{corollary}\label{cor:cara}
Let $\zeta_{AR}^{n}=\int\sigma_{AR}^{\ot n}\;d(\sigma_{AR})$ be as in~\eqref{eq:definetti}. Then, we have that
\begin{align}
\zeta_{AR}^{n}=\sum_{i}p_{i}\left(\omega^{i}_{AR}\right)^{\ot n}\ ,
\end{align}
with $\omega^{i}_{AR}\in\cV(AR)$, $i\in\{1,2,\ldots,(n+1)^{2|A||R|-2}\}$, and $\{p_{i}\}$ a probability distribution.
\end{corollary}

\begin{proof}
We can think of $\zeta_{AR}^{n}$ as a normalized state on the symmetric subspace
\begin{align}
\mathrm{Sym}^{n}(\cH_{AR})\subset\cH_{AR}^{\ot n}\ ,
\end{align}
and the dimension of $\mathrm{Sym}^{n}(\cH_{AR})$ is bounded by $k=(n+1)^{|A||R|-1}$. Furthermore, the normalized states on $\mathrm{Sym}^{n}(\cH_{AR})$ can be seen as living in a $m$-dimensional real vector space, where $m=k-1+2\cdot\frac{k(k-1)}{2}=k^{2}-1$. Now define $S$ as the set of all $\xi_{AR}^{n}=\omega_{AR}^{\ot n}$, where $\omega_{AR}\in\cV(\cH_{AR})$. Then, $\zeta^{n}_{AR}$ lies in the convex hull of the set $S\subset\mathbb{R}^{k^{2}-1}$. Using Carath\'eodory's theorem (Theorem~\ref{thm:cara}), we have that $\zeta_{AR}^{n}$ lies in the convex hull of a set $S'\subset S$ where $S'$ consists of at most $p=k^{2}-1+1=k^{2}$ points. Hence, we can write $\zeta_{AR}^{n}$ as a convex combination of $p=(n+1)^{2|A||R|-2}$ extremal points in $S'$, i.e.,
\begin{align}
\zeta_{AR}^{n}=\sum_{i}p_{i}(\omega^{i}_{AR})^{\ot n}\ ,
\end{align}
where $\omega^{i}_{AR}\in\cV(\cH_{AR})$, $i\in\{1,2,\ldots,(n+1)^{2|A||R|-2}\}$, and $\{p_{i}\}$ a probability distribution.
\end{proof}


\chapter{Typical Projectors}\label{app:typical}

This appendix is from the collaboration~\cite{Berta13} and written by Mark Wilde. All systems are finite-dimensional. 

A sequence $x^{n}$ is typical with respect to some probability distribution $p_{X}\left(x\right)$ if its empirical distribution has maximum deviation $\delta$ from $p_{X}\left(x\right)$. The typical set $T_{\delta}^{X^{n}}$ is the set of all such sequences
\begin{align}
T_{\delta}^{X^{n}}=\left\{x^{n}:\left\vert \frac{1}{n}N\left(x|x^{n}\right)-p_{X}\left(x\right)\right\vert\leq\delta\ \ \ \ \forall x\in\mathcal{X}\right\}\ ,
\end{align}
where $N\left(  x|x^{n}\right)  $ counts the number of occurrences of the letter $x$ in the sequence $x^{n}$. The above notion of typicality is the strong notion (as opposed to the weaker entropic version of typicality sometimes employed~\cite{Cover91}). The typical set enjoys three useful properties: its probability approaches unity in the large $n$ limit, it has exponentially smaller cardinality than the set of all sequences, and every sequence in the typical set has approximately uniform probability. That is, suppose that $X^{n}$ is a random variable distributed according to $p_{X^{n}}\left(x^{n}\right)= p_{X}\left(x_{1}\right)\cdots p_{X}\left(x_{n}\right)$, $\eps$ is positive number that becomes arbitrarily small as $n$ becomes large, and $c$ is some positive constant. Then the following three properties hold~\cite{Cover91}
\begin{align}
\Pr\left\{X^{n}\in T_{\delta}^{X^{n}}\right\}&\geq1-\eps\ ,\label{eq:typ-1}\\
\left\vert T_{\delta}^{X^{n}}\right\vert&\leq2^{n\cdot\left(H\left(X\right)+c\delta\right)}\ ,\label{eq:typ-2}\\
\forall x^{n}\in T_{\delta}^{X^{n}}:\ \ \ 2^{-n\cdot\left(H\left(X\right)+c\delta\right)}&\leq p_{X^{n}}\left(x^{n}\right)\leq2^{-n\cdot\left(H\left(X\right)-c\delta\right)}\ .\label{eq:typ-3}
\end{align}
These properties translate straightforwardly to the quantum setting by applying the spectral theorem to a density matrix $\rho$. That is, suppose that
\begin{align}
\rho=\sum_{x}p_{X}\left(x\right)\left\vert x\right\rangle \left\langle x\right\vert\ ,
\end{align}
for some orthonormal basis $\left\{\left\vert x\right\rangle \right\}_{x}$. Then there is a typical subspace defined as follows
\begin{align}
T_{\rho,\delta}^{n}=\text{span}\left\{  \left\vert x^{n}\right\rangle:\left\vert \frac{1}{n}N\left(x|x^{n}\right)-p_{X}\left(x\right)
\right\vert \leq\delta\ \ \ \ \forall x\in\mathcal{X}\right\}\ ,
\end{align}
and let $\Pi_{\rho,\delta}^{n}$ denote the projector onto it. Then properties analogous to (\ref{eq:typ-1}-\ref{eq:typ-3}) hold for the typical subspace.
The probability that a tensor power state $\rho^{\ot n}$\ is in the typical subspace approaches unity as $n$ becomes large, the rank of the typical projector is exponentially smaller than the rank of the full $n$-fold tensor-product Hilbert space of $\rho^{\ot n}$, and the state $\rho^{\ot n}$ looks approximately maximally mixed on the typical subspace
\begin{align}
\trace\left[\Pi_{\rho,\delta}^{n}\ \rho^{\ot n}\right]&\geq1-\eps\ ,\label{eq:typ-q-1}\\
\trace\left[\Pi_{\rho,\delta}^{n}\right]&\leq2^{n\cdot\left(H\left(B\right)+c\delta\right)}\ ,\label{eq:typ-q-2}\\
2^{-n\cdot\left(H\left(B\right)+c\delta\right)}\ \Pi_{\rho,\delta}^{n}&\leq\Pi_{\rho,\delta}^{n}\ \rho^{\ot n}\ \Pi_{\rho,\delta}^{n}\leq2^{-n\cdot\left(H\left(B\right)-c\delta\right)}\ \Pi_{\rho,\delta}^{n}\ .\label{eq:typ-q-3}
\end{align}
Suppose now that we have an ensemble of the form $\left\{  p_{X}\left(x\right)  ,\rho_{x}\right\}  $, and suppose that we generate a typical sequence $x^{n}$ according to a pruned distribution (defined as a normalized version of $p_{X^n}(x^n)$ with support on its typical set and zero otherwise),
leading to a tensor product state $\rho_{x^{n}}=\rho_{x_{1}}\ot\cdots\ot\rho_{x_{n}}$. Then there is a conditionally typical
subspace with a conditionally typical projector defined as follows
\begin{align}
\Pi_{\rho_{x^{n}},\delta}^{n}=\bigotimes\limits_{x\in\mathcal{X}}\Pi_{\rho_{x},\delta}^{I_{x}}\ ,
\end{align}
where $I_{x}=\left\{  i:x_{i}=x\right\}  $ is an indicator set that selects the indices $i$\ in the sequence $x^{n}$ for which the $i^{\text{th}}$
symbol $x_{i}$\ is equal to $x\in\mathcal{X}$ and $\Pi_{\rho_{x},\delta}^{I_{x}}$ is the typical projector for the state $\rho_{x}$. The
conditionally typical subspace has the three following properties
\begin{align}
\trace\left[\Pi_{\rho_{x^{n}},\delta}^{n}\ \rho_{x^{n}}\right]&\geq1-\eps\ ,\\
\trace\left[\Pi_{\rho_{x^{n}},\delta}^{n}\right]&\leq2^{n\cdot\left(H\left(B|X\right)+c\delta\right)}\ ,\\
2^{-n\cdot\left(H\left(B|X\right)+c\delta\right)}\ \Pi_{\rho_{x^{n}},\delta}^{n}&\leq\Pi_{\rho_{x^{n}},\delta}^{n}\ \rho_{x^{n}}\ \Pi_{\rho_{x^{n}},\delta}^{n}\leq2^{-n\cdot\left(H\left(B|X\right)-c\delta\right)}\ \Pi_{\rho_{x^{n}},\delta}^{n}\ .
\end{align}
Let $\rho$ be the expected density matrix of the ensemble $\left\{p_{X}\left(  x\right)  ,\rho_{x}\right\}  $ so that $\rho=\sum_{x}p_{X}\left(x\right)  \rho_{x}$. The following properties are proved in~\cite{Winter99,Devetak05,Wilde11}
\begin{align}\label{eq:prune-avg-op-ineq}
\forall x^{n} \in T^{X^{n}}_{\delta}: \trace\left[\rho_{x^{n}}\ \Pi_{\rho}\right]&\geq1-\eps\ ,\\
\sum_{x^{n}}p_{X^{\prime n}}^{\prime}\left(  x\right)\rho_{x^{n}}&\leq\left(1-\eps\right)^{-1}\rho^{\ot n}\ .
\end{align}

In order to justify some of the estimates made in Section~\ref{sec:meas_main}, we use the above estimates on eigenvalues and support
sizes. For the classical communication cost, we consider
\begin{align}
H_{R}\left(W^{n}\right)_{\bar{\gamma}^{i}}-H_{\min}\left(W^{n}R^{n}\right)_{\bar{\gamma}^{i}}+H_{0}\left(R^{n}\right)_{\bar{\gamma}^{i}}\ .
\end{align}
The smallest nonzero eigenvalue of the reduced state on $W^{n}$ is larger than $2^{-n\cdot\left(H\left(W\right)+c\delta\right)}$ due to the typical projection on $W^{n}$. Thus, we have that
\begin{align}
H_{R}\left(W^{n}\right)\leq n\cdot\left(H\left(W\right)+c\delta\right)\ .
\end{align}
The largest eigenvalue of $\bar{\gamma}_{WR}^{i,n}$ is bounded by
\begin{align}
2^{-n\cdot\left(H\left(W\right)_{\bar{\gamma}^{i}}-c\delta\right)}2^{-n\cdot\left[H\left(R|W\right)_{\bar{\gamma}^{i}}-c\delta\right]}\ ,
\end{align}
due to the typical projection on $W^{n}$ and the conditionally typical projection on $R^{n}$. So we have that
\begin{align}
H_{\min}\left(W^{n}R^{n}\right)_{\bar{\gamma}^{i}}\geq n\cdot\left(H\left(WR\right)_{\bar{\gamma}^{i}}+2c\delta\right)\ .
\end{align}
The size of the support of $R^{n}$ is bounded from above by $2^{n\cdot\left(H\left(R\right)_{\bar{\gamma}^{i}}+\delta\right)}$, due to the outermost projection on $R^{n}$. Thus, we have that
\begin{align}
H_{0}\left(R^{n}\right)_{\bar{\gamma}^{i}}\leq n\cdot\left(H\left(R\right)_{\bar{\gamma}^{i}}+2c\delta\right)\ .
\end{align}
The above development then gives the following bound
\begin{align}
H_{R}\left(W^{n}\right)_{\bar{\gamma}^{i}}-H_{\min}\left(W^{n}R^{n}\right)_{\bar{\gamma}^{i}}+H_{0}\left(R^{n}\right)_{\bar{\gamma}^{i}}\leq n\cdot\left(I\left(W:R\right)_{\bar{\gamma}^{i}}+5c\delta\right)\ .
\end{align}
We have similar arguments for bounding the shared randomness cost
\begin{align}
H_{R}\left(  W^{n}\right)  _{\bar{\gamma}^{i}}-H_{\min}\left(  W^{n}X^{n}R^{n}\right)  _{\bar{\gamma}^{i}}+H_{0}\left(  R^{n}X^{n}\right)
_{\bar{\gamma}^{i}}\ .
\end{align}
By the same argument as above, we have that
\begin{align}
H_{R}\left(W^{n}\right)_{\bar{\gamma}^{i}}\leq n\cdot\left(H\left(W\right)_{\bar{\gamma}^{i}}+c\delta\right)\ .
\end{align}
The largest eigenvalue of $\bar{\gamma}_{WXR}^{i,n}$ is bounded by
\begin{align}
&  2^{-n\cdot\left(H\left(W\right)_{\bar{\gamma}^{i}}-c\delta\right)}2^{-n\cdot\left(H\left(X|W\right)_{\bar{\gamma}^{i}}-c\delta\right)}2^{-n\cdot\left(H\left(R|W\right)_{\bar{\gamma}^{i}}-c\delta\right)}\notag\\
&=2^{-n\cdot\left(H\left(WX\right)_{\bar{\gamma}^{i}}-2c\delta\right)}2^{-n\cdot\left(H\left(R|WX\right)_{\bar{\gamma}^{i}}-c\delta\right)}\notag\\
&=2^{-n\cdot\left(H\left(WXR\right)_{\bar{\gamma}^{i}}-3c\delta\right)}\ ,
\end{align}
where we have used the fact that $H\left(R|W\right)_{\bar{\gamma}^{i}}=H\left(R|WX\right)_{\bar{\gamma}^{i}}$ because the state on $R$ is independent of $X$. Thus, we have that
\begin{align}
H_{\min}\left(W^{n}X^{n}R^{n}\right)_{\bar{\gamma}^{i}}\geq n\cdot\left(H\left(WXR\right)_{\bar{\gamma}^{i}}-3c\delta\right)\ .
\end{align}
Finally, the support of $R^{n}X^{n}$ is bounded again by $2^{n\cdot\left(H\left(RX\right)_{\bar{\gamma}^{i}}+2c\delta\right)}$, due to the typical projections, so that we have
\begin{align}
H_{0}\left(R^{n}X^{n}\right)_{\bar{\gamma}^{i}}\leq n\cdot\left(H\left(RX\right)_{\bar{\gamma}^{i}}+2c\delta\right)\ .
\end{align}
The above development then gives the following bound
\begin{align}
H_{R}\left(W^{n}\right)_{\bar{\gamma}^{i}}-H_{\min}\left(W^{n}X^{n}R^{n}\right)_{\bar{\gamma}^{i}}+H_{0}\left(  R^{n}X^{n}\right)_{\bar{\gamma}^{i}}\leq n\cdot\left(I\left(W:XR\right)_{\bar{\gamma}^{i}}+6c\delta\right)\ .
\end{align}


\chapter{Technical Lemmas}\label{ap:lemmas}

In this appendix, all systems are finite-dimensional. The purified distance is convex in its arguments in the following sense.

\begin{lemma}\cite[Lemma A.3]{Berta11}\label{lem:pdconvex}
Let $\rho_{i}$, $\sigma_{i}\in\cS_{\leq}(\cH)$ with $\rho_{i}\approx_{\eps}\sigma_{i}$ for $i\in I$, and $\{p_{i}\}$ a probability distribution. Then, we have that
\begin{align}
\sum_{i\in I}p_{i}\rho_{i}\approx_{\eps}\sum_{i\in I}p_{i}\sigma_{i}\ .
\end{align}
\end{lemma}

The metric induced by the 1-norm can be written in the following alternative form.

\begin{lemma}~\cite[Lemma 9.1.1]{Wilde11}\label{lem:1norm}
Let $\rho,\sigma\in\cS(\cH)$. Then, we have that
\begin{align}
\|\rho-\sigma\|_{1}=2\cdot\max_{0\leq X\leq\id}\trace[(\rho-\sigma)X]\ .
\end{align}
\end{lemma}

The 1-norm can be upper and lower bounded by the 2-norm.

\begin{lemma}\cite{Horn85}\label{lem:12norm}
Let $M\in\mathbb{C}^{a\times b}$ for $a,b\in\mathbb{N}$. Then, we have that
\begin{align}
\|M\|_{2}\leq\|M\|_{1}\leq\sqrt{\rank(M)}\cdot\|M\|_{2}\ .
\end{align}
\end{lemma}

The 2-norm is sub-multiplicative.

\begin{lemma}\cite[Section 5.2]{Meyer00}\label{lem:sub}
Let $M\in\mathbb{C}^{a\times b}$ and $N\in\mathbb{C}^{b\times c}$ for $a,b,c\in\mathbb{N}$. Then, we have that
\begin{align}
\|M\cdot N\|_{2}\leq\|M\|_{2}\|N\|_{2}\ .
\end{align}
\end{lemma}

The following lemma is about an $\eps$-net construction in the 2-norm.

\begin{lemma}\label{lem:net}
Let $0<\eps<1$, and $D,d>0$. Furthermore, let
\begin{align}
\cN_{D}^{d}=\left\{w\in\mathbb{C}^{d}\mid\|w\|_{2}\leq D\right\}\ ,
\end{align}
and let $\cT$ be some subset of $\cN_{D}^{d}$. Then, there exists a subset $\cT_{\eps}\subset\cT$ with $|\cT_{\eps}|\leq\left(\frac{2D}{\eps}+1\right)^{2d}$, such that for every vector $v\in\cT$, there exists a vector $v_{\eps}\in\cT_{\eps}$ with $\|v-v_{\eps}\|_{2}\leq\eps$.
\end{lemma}

\begin{proof}
The proof is inspired by~\cite[Lemma II.4]{Hayden04}. Let $\cT_{\eps}=\{v_{i}\}_{i=1,\ldots,m}$ be a maximal subset of $v\in\cT$ satisfying $\|v_{i}-v_{j}\|_{2}\geq\eps$ for all $i,j$.\footnote{Such a subset can be constructed by starting with an arbitrary vector $v_{1}\in\cT$, as a next step taking another vector $v_{2}\in\cT$ with $\|v_{1}-v_{2}\|_{2}\geq\eps$, and then $v_{3}\in\cT$ with $\|v_{1}-v_{3}\|_{2}\geq\eps$, $\|v_{2}-v_{3}\|_{2}\geq\eps$ etc. A subset constructed like this becomes maximal as soon as it is not possible to add another vector $v_{k}\in\cT$, such that $\|v_{k}-v_{i}\|_{2}\geq\eps$ for all vectors $v_{i}$ that are already in the subset.} It remains to estimate $m$. As subsets of $\mathbb{R}^{2d}$, the open balls of radius $\eps/2$ about each $v_{i}\in\cT_{\eps}$ are pairwise disjoint, and all contained in the ball of radius $D+\eps/2$ centered at the origin. Hence, we have that
\begin{align}
m\cdot\left(\eps/2\right)^{2d}\leq\left(D+\eps/2\right)^{2d}\ .
\end{align}
\end{proof}

The von Neumann entropy is continuous.

\begin{lemma}\cite[Theorem 1]{Audenaert07}\label{lem:neumann}
Let $\rho_{A},\sigma_{A}\in\cS(\cH_{A})$ with $\rho_{A}\approx_{\eps}\sigma_{A}$ for some $\eps\geq0$. Then, we have that
\begin{align}
|H(A)_{\rho}-H(A)_{\sigma})|\leq\eps\cdot\log(|A|-1)+h(\eps)\ .
\end{align}
\end{lemma}

The conditional von Neumann entropy is continuous.

\begin{lemma}\cite{Alicki04}\label{lem:fannes}
Let $\rho_{AB},\sigma_{AB}\in\cS(\cH_{AB})$ with $\|\rho_{AB}-\sigma_{AB}\|_{1}\leq\eps$ for some $\eps\geq0$. Then, we have that
\begin{align}
|H(A|B)_{\rho}-H(A|B)_{\sigma}|\leq4\eps\cdot\log|A|+2h(\eps)\ .
\end{align}
\end{lemma}

The entanglement of formation is continuous.

\begin{lemma}\label{lem:nielsen}
Let $\rho_{AB},\sigma_{AB}\in\cS(\cH_{AB})$ with $\rho_{AB}\approx_{\eps}\sigma_{AB}$ for some $\eps\geq0$. Then, we have that
\begin{align}
|E_{F}(\rho_{AB})-E_{F}(\sigma_{AB})|\leq8\eps\cdot\log|A|+2h(2\eps)\ .
\end{align}
\end{lemma}

\begin{proof}
The proof is the same as the original proof for the continuity of the entanglement of formation~\cite{Nielsen00_2}, but uses the (improved) continuity of the conditional von Neumann entropy (Lemma~\ref{lem:fannes}) instead of the continuity of the unconditional von Neumann entropy (Lemma~\ref{lem:neumann}).
\end{proof}

\begin{lemma}\cite[Lemma 1]{Uhlmann98}\label{lem:uhlman}
Let $\rho_{AB}\in\cS(\cH_{AB})$. Then, the minimization over all pure states decompositions $\rho_{AB}=\sum_{i}p_{i}\rho_{AB}^{i}$ in the entanglement of formation $E_{F}(\rho_{AB})=\min_{\{p_{i},\rho^{i}\}}\sum_{i}p_{i}H(A)_{\rho^{i}}$ is taken for a decomposition with at least $\rank(\rho_{AB})$ and at most $\rank(\rho_{AB})^{2}$ elements.
\end{lemma}

The following lemma shows that the different ways of the defining the quantum capacity are all equivalent.

\begin{lemma}\cite[Proposition 4.3]{Kretschmann04}\label{lem:werner}
Let $\cE:\cB(\cH_{A})\mapsto\cB(\cH_{B})$ be a channel. Then, we have that
\begin{align}
1-\min_{\rho\in\cS(\cH_{A})}F_{e}(\rho,\cE)\leq4\sqrt{1-F_{c}(\cE)}\leq4\sqrt{\|\cE-\cI\|_{\diamond}}\leq8\left(1-\min_{\rho\in\cS(\cH_{A})}F_{e}(\rho,\cE)\right)^{1/4}\ ,
\end{align}
where $F_{c}(\cE)=\bra{\Phi}(\cE\ot\cI)(\Phi)\ket{\Phi}$ with $\Phi_{AA'}$ the maximally entangled state on $\cH_{A}\ot\cH_{A'}$, and $F_{e}(\rho,\cE)=\bra{\rho}(\cE\ot\cI)(\rho)\ket{\rho}$ with $\rho_{AA'}\in\cV(\cH_{A}\ot\cH_{A'})$.
\end{lemma}

The following is one aspect of the Perron–Frobenius theorem.

\begin{lemma}\cite{Meyer00}\label{lem:perron_frobenius}
Let $M\in\nR^{d\times d}$ with $d\in\nN$, and $M=(m_{ij})_{i,j=1}^{d}$ with $m_{ij}\geq0$. Then, we have that the largest eigenvalue of $M$ has an eigenvector with non-negative entries.
\end{lemma}

The following is Sion's minimax theorem.

\begin{lemma}\cite[Corollary 3.3]{Sion58}\label{lem:minimax}
Let $X$ be a compact, convex subset of a linear topological space, $Y$ a convex subset of a linear topological space, and $f$ a real valued function on $X\times Y$, that is quasi-convex and lower semicontinuous in the first argument, and quasi-concave and upper semicontinuous in the second argument. Then, we have that
\begin{align}
\inf_{x\in X}\sup_{y\in Y}f(x,y)=\sup_{y\in Y}\inf_{x\in X}f(x,y)\ .
\end{align}
\end{lemma}

The last three lemmas are purely technical statements.

\begin{lemma}\cite{Nielsen00}\label{lem:eigenvalue}
Let $\rho,\sigma\in\cP^{+}(\cH)$, and denote the corresponding eigenvalue distribution by $P_{X},Q_{X}$ respectively. Then, we have that
\begin{align}
\|\rho-\sigma\|_{1}\geq\|P_{X}-Q_{X}\|_{1}\ .
\end{align}
\end{lemma}

\begin{lemma}~\cite[Lemma A.7]{Berta10}\label{lem:projpurd}
Let $\rho\in\cS_{\leq}(\cH)$, and let $\Pi\in\cP^{+}(\cH)$ be such that $\Pi\leq\1$. Then, we have that
\begin{align}
P(\rho,\Pi\rho\Pi)\leq\trace[\rho]^{-1/2}\cdot\sqrt{\trace[\rho]^{2}-\trace[\Pi^{2}\rho]^{2}}\ .
\end{align}
\end{lemma}

\begin{lemma}\label{lem:unitary_group}
Let $A=A_{1}A_{2}$, and let $F_{AA'}$ be the swap operator on $\cH_{AA'}$ with $\cH_{A'}\cong\cH_{A}$. Then, we have that
\begin{align}
\int(U\ot U)^\dagger(\id_{A_{2}A_{2'}}\ot F_{A_{1}A_{1}'})(U\ot U)dU\leq\frac{1}{|A_{1}|}\id_{AA'}+\frac{1}{|A_{2}|}F_{AA'}\ ,
\end{align}
where the integration is with respect to the Haar measure on the unitary group $U(|A|)$.
\end{lemma}

\begin{proof}
For any $X\in\cB(\cH_{AA'})$ that is Hermitian, it follows from Schur's lemma that
\begin{align}
\int_{U(A)}(U^\dagger\ot U^\dagger)X(U\ot U)dU=a_{+}(X)\Pi_{A}^{+}+a_{-}(X)\Pi_{A}^{-}\ ,
\end{align}
where $a_{\pm}(X)\cdot\rank(\Pi_{A}^{\pm})=\trace[X\Pi_{A}^{\pm}]$, and $\Pi_{A}^{\pm}$ denotes the projector on the symmetric and the anti-symmetric subspace of $\cH_{AA'}$, respectively. Choosing $X=\1_{A_{2}A_{2}'}\ot F_{A_{1}A_{1}'}$ we get
\begin{align}
\trace\left[\Pi_{A}^{\pm}(\id_{A_{2}A_{2}'}\ot F_{A_{1}A_{1}'})\right]&=\frac{1}{2}\cdot\trace\left[(\id_{AA'}\pm F_{AA'})(\id_{A_{2}A_{2}'}\ot F_{A_{1}A_{1}'})\right]\notag\\
&=\frac{1}{2}\trace\left[\id_{A_{2}A_{2}'}\ot F_{A_{1}A_{1}'}\right]\pm\frac{1}{2}\cdot\trace\left[F_{AA'}(\id_{A_{2}A_{2}'}\ot F_{A_{1}A_{1}'})\right]\notag\\
&=\frac{1}{2}\cdot|A_{2}|^{2}\cdot|A_{1}|\pm\frac{1}{2}\cdot|A_{2}|\cdot|A_{1}|^{2}\ .
\end{align}
Since $\rank(\Pi_{A}^{\pm})=\frac{1}{2}\cdot|A|(|A|\pm1)$ we get
\begin{align}
a_{\pm}(X)=\frac{|A_{2}|^{2}|A_{1}|\pm|A_{2}||A_{1}|^{2}}{|A|(|A|\pm1)}=\frac{|A_{2}|\pm|A_{1}|}{|A|\pm1}\ .
\end{align}
From
\begin{align}
\frac{a_{+}(X)+a_{-}(X)}{2}&=\frac{1}{2}\cdot\left(\frac{|A_{2}|+|A_{1}|}{|A|+1}+\frac{|A_{2}|-|A_{1}|}{|A|-1}\right)=\frac{|A_{2}||A|-|A_{1}|}{|A|^{2}-1}\notag\\
&=\frac{1}{|A_{1}|}\cdot\frac{|A|^{2}-|A_{1}|^{2}}{|A|^{2}-1}\leq\frac{1}{|A_{1}|}\ ,
\end{align}
and
\begin{align}
\frac{a_{+}(X)-a_{-}(X)}{2}&=\frac{1}{2}\cdot\left(\frac{|A_{2}|+|A_{1}|}{|A|+1}-\frac{|A_{2}|-|A_{1}|}{|A|-1}\right)=\frac{|A_{1}||A|-|A_{2}|}{|A|^{2}-1}\notag\\
&=\frac{1}{|A_{2}|}\cdot\frac{|A|^{2}-|A_{2}|^{2}}{|A|^{2}-1}\leq\frac{1}{|A_{2}|}\ ,
\end{align}
it follows that
\begin{align}
\int(U^\dagger\ot U^\dagger)X(U\ot U)dU&=a_{+}(X)\Pi_{A}^{+}+a_{-}(X)\Pi_{A}^{-}\notag\\
&=\frac{a_{+}(X)+a_{-}(X)}{2}\id_{AA'}+\frac{a_{+}(X)-a_{-}(X)}{2}F_{AA'}\notag\\
&\leq\frac{1}{|A_{1}|}\cdot\id_{AA'}+\frac{1}{|A_{2}|}\cdot F_{AA'}\ .
\end{align}
\end{proof}








\end{document}